\documentclass[11pt]{article}
\usepackage{makeidx}  
\usepackage{fullpage}
\usepackage{times}
\usepackage{float}
\usepackage{subfigure}
\usepackage{color}
\usepackage{url}
\usepackage{graphicx}
\usepackage{amsmath}
\usepackage{amssymb}
\usepackage{stackengine}

\usepackage{tikz}
\usepackage{ifthen}
\usepackage{fp}
\usetikzlibrary{arrows,shapes}
\usepackage{algorithmic}
\usepackage[ruled]{algorithm2e}
\usepackage[geometry]{ifsym}

\usepackage{accents}

\newcommand{\qed}{\mbox{}\hspace*{\fill}\nolinebreak\mbox{$\rule{0.6em}{0.6em}$}}

\newcommand{\expect}{{\bf \mbox{\bf E}}}
\newcommand{\prob}{{\bf \mbox{\bf Pr}}}

\definecolor{gray}{rgb}{0.5,0.5,0.5}
\newcommand{\e}{{\epsilon}}

\newcommand{\B}{{\bf B}}

\newtheorem{theorem}{Theorem}
\newtheorem{lemma}[theorem]{Lemma}

\newtheorem{claim}[theorem]{Claim}
\newtheorem{corollary}[theorem]{Corollary}
\newtheorem{definition}[theorem]{Definition}
\newtheorem{property}[theorem]{Property}
\newtheorem{remark}[theorem]{Remark}

\newenvironment{proof}{{\bf Proof:}}{$\qed$\par}

\newenvironment{proofof}[1]{\noindent{\bf Proof of #1:}}{$\qed$\par}
\newenvironment{proofsketch}{{\sc{Proof Outline:}}}{$\qed$\par}
\usepackage{hyperref}
\hypersetup{
bookmarksnumbered
}

\renewcommand{\b}{{\mbox {\bf b}}}

\newcommand{\I}{\mathbf I}

\newcommand{\bool}{\{0, 1\}}
\newcommand{\F}{{\mathcal F}}
\newcommand{\wt}[1]{\widetilde{#1}}
\newcommand{\wh}[1]{\widehat{#1}}
\newcommand{\ac}[1]{\accentset{\circ}{#1}}

\usepackage{tkz-graph}
\usepackage{tikz}

{\makeatletter
 \gdef\xxxmark{%
   \expandafter\ifx\csname @mpargs\endcsname\relax 
     \expandafter\ifx\csname @captype\endcsname\relax 
       \marginpar{xxx}
     \else
       xxx 
     \fi
   \else
     xxx 
   \fi}
 \gdef\xxx{\@ifnextchar[\xxx@lab\xxx@nolab}
 \long\gdef\xxx@lab[#1]#2{{\bf [\xxxmark #2 ---{\sc #1}]}}
 \long\gdef\xxx@nolab#1{{\bf [\xxxmark #1]}}
\long\gdef\xxx@lab[#1]#2{}\long\gdef\xxx@nolab#1{}%
}

\renewcommand{\B}{{\mathbf B}}
\renewcommand{\j}{{\mathbf j}}

\def\delequal{\asymp}

\usepackage{nameref}

\makeatletter

\let\orgdescriptionlabel\descriptionlabel
\renewcommand*{\descriptionlabel}[1]{%
  \let\orglabel\label
  \let\label\@gobble
  \phantomsection
  \edef\@currentlabel{#1}%
  \let\label\orglabel
  \orgdescriptionlabel{#1}%
}
\makeatother

\newcommand{\lbl}{\text{label}}
\newcommand{\Sp}{\Psi}
\newcommand{\Spt}{\widetilde{\Psi}}
\newcommand{\downset}{\textsc{DownSet}}

\newcommand{\weight}{\text{wt}}
\newcommand{\Ext}{\text{Ext}}

\begin{document}
\setcounter{page}{0}
\title{\Large Space Lower Bounds for Approximating Maximum Matching in the Edge Arrival Model}

\author{Michael Kapralov\\EPFL}

\maketitle              

\begin{abstract}
The bipartite matching problem in the online and streaming settings has received a lot of attention recently. The classical vertex arrival setting, for which the celebrated Karp, Vazirani and Vazirani (KVV) algorithm achieves a $1-1/e$ approximation, is rather well understood:  the $1-1/e$ approximation is optimal in both the online and semi-streaming setting, where the algorithm is constrained to use $n\cdot \log^{O(1)} n$ space. The more challenging the edge arrival model has seen significant progress recently in the online algorithms literature. For the strictly online model (no preemption) approximations better than trivial factor $1/2$ have been ruled out [Gamlath et al'FOCS'19]. For the less restrictive online preemptive model a better than $\frac1{1+\ln 2}$-approximation [Epstein et al'STACS'12]  and even a better than $(2-\sqrt{2})$-approximation[Huang et al'SODA'19] have been ruled out. 

The recent hardness results for online preemptive matching in the edge arrival model are based on the idea of stringing together multiple copies of a KVV hard instance using edge arrivals.  In this paper, we show how to implement such constructions using ideas developed in the literature on Ruzsa-Szemer\'edi graphs. As a result, we show that any single pass streaming algorithm that approximates the maximum matching in a bipartite graph with $n$ vertices to a factor better than $\frac1{1+\ln 2}\approx 0.59$ requires $n^{1+\Omega(1/\log\log n)}\gg n \log^{O(1)} n$ space. This gives the first separation between the classical one sided vertex arrival setting and the edge arrival setting in the semi-streaming model.
\end{abstract}

\newpage
\tableofcontents
\setcounter{page}{0}
\newpage


\newcommand{\duallabel}[1]{\label{#1-intro}}
\newcommand{\dualref}[1]{\ref{#1-intro}}
\newcommand{\dualeqref}[1]{\eqref{#1-intro}}

\section{Introduction}\duallabel{sec:intro}

Large datasets are common in modern data analysis, and processing them requires algorithms with space complexity sublinear in the size of the input. The streaming model of computation, originally introduced in the seminar work of~\cite{AlonMS96}, captures this setting, has received a lot of attention in the literature recently. In this paper we study the space complexity of the bipartite matching problem in the streaming model: the edges of a bipartite graph $G=(P, Q, E)$ are presented in an adversarial order as a stream, and the algorithm must output a matching $M_{ALG}\subseteq E$ at the end of the stream such that $|M_{ALG}|\geq \gamma |M_{OPT}|$ with high constant probability, for some approximation ratio $\gamma\in (0, 1]$. The algorithm is constrained to use $n\log^{O(1)} n$ space, where $n$ is the number of vertices in the input graph. This is a common assumption, and the streaming model with this space restriction is often referred to as the {\em semi-streaming} model of computation~\cite{feigenbaum2005graph}. In this model the simple greedy algorithm, which maintains a maximal matching in the graph received so far, achieves a $\gamma=\frac1{2}$ approximation by storing $O(n)$ edges (and therefore using only $O(n\log n)$ bits of space). Despite a considerable amount of research over the past decade, it is still not known whether it is possible to achieve a better than $\frac1{2}$ approximation in this model using a single pass over the stream. The best hardness result so far is due to~\cite{Kapralov13}, ruling out a $(1-1/e+\eta)$-approximation for any constant $\eta>0$ in less than $n^{1+\Omega(1/\log\log n)}$ space, and thereby showing that no semi-streaming algorithm can do significantly better than $1-1/e$. The lower bound of~\cite{Kapralov13} applies (and is tight for) a more restricted model, where vertices on one side of the input graph $G=(P, Q, E)$ arrive in the stream in an arbitrary order and reveal their edges upon arrival. This model is inspired by the classical online matching problem studied in the seminal work of Karp, Vazirani and Vazirani~\cite{KarpVV90}, where vertices on one side of a bipartite graph $G=(P, Q, E)$ arrive online, and the algorithm must match an arriving vertex irrevocably to one of its neighbors upon arrival or discard it. The competitive ratio of $1-1/e$ is achievable and tight for the online model with one sided vertex arrivals as well. The online version of the matching problem in the edge arrival setting has recently been resolved, the work of~\cite{GamlathKMSW19} showing that no strictly online (i.e., without preemption) algorithm can do better than greedy in the edge arrival model. The same question remains open for the online model with preemption, which is close to the semi-streaming model that we are interested in. Several new hardness results for this model have been shown recently~\cite{EpsteinLSW13,WangW15,DBLP:conf/soda/0002PTTWZ19},  ruling out the possibility of a $1-1/e$ approximation in the online model with preemption. In this work we extend one of these results, due to Epstein et al~\cite{EpsteinLSW13}, to the streaming setting. Specifically, our main result is
\begin{theorem}\label{thm:main}
Any single-pass streaming algorithm that finds a $(\frac1{1+\ln 2}+\eta)$-approximate matching  in an $n$-vertex bipartite graph for a constant $\eta>0$  with probability at least $1/2$ must use $n^{1+\Omega(1/\log\log n)}\gg n\log^{O(1)} n$ bits of space.
\end{theorem}
This gives the first separation between the classical one sided vertex arrival setting and the edge arrival setting in the semi-streaming model. We note that the best hardness result for online preemptive matching at the moment is a $2-\sqrt{2}$-hardness, due to~\cite{DBLP:conf/soda/0002PTTWZ19}. Our techniques in this paper can probably be extended to their instance, but we prefer to use the earlier instance of~\cite{EpsteinLSW13} to simplify exposition.

\subsection{Related Work}
Over the past decade, matchings have been extensively studied in the context of streaming and related settings. The prior work closest to ours is the aforementioned $1-1/e$ lower bound of~\cite{Kapralov13} (see also~\cite{GoelKK12} and~\cite{FischerLNRRS02}). Strong lower bounds for approximating matchings in the sketching model have been proposed in~\cite{AssadiKLY16,AssadiK17}. Multipass lower bounds for exact matching computation are given in~\cite{GuruswamiO16}.

Good approximations using a small number of passes have been presented in~\cite{KaleT17}, and algorithmic results on the weighted version of the problem are given in~\cite{CrouchS14,PazS17}. Besides the most stringent adversarial edge arrival model, the relaxed random order streaming model has seen a lot of attention, where small space approximations to matching size have been given~\cite{KapralovKS14,CormodeJMM17,MonemizadehMPS17,KapralovMNT20}. The problem of approximating the size of the maximum matching in adversarially ordered streams has also received significant attention in the literature:\cite{EsfandiariHLMO15,BuryS15,AssadiKL17,V18,BuryGMMSVZ19,McGregorV16,ChitnisCEHMMV16,EsfandiariHM16}.

\section{Technical overview}\label{sec:tech-overview}

We start by defining $\frac1{1+\ln 2}$-hard instance from~\cite{EpsteinLSW13}. We first define an $\alpha$-KVV gadget $G=(S, T, E)$. 

\begin{definition}[$\alpha$-KVV gadget] We define a $\alpha$-KVV gadget as a bipartite graph $G=(S, T, E)$ with $|T|=N$ vertices on the $T$ side of the bipartition and $|S|=\alpha\cdot N$ vertices on the $S$ side of the biparitition as follows. We think of vertices in $T$ as being numbered with integers in $[N]=\{0, 1, \ldots, N-1\}$ and vertices in $S$ as being numbered with integers in $[\alpha\cdot N]=\{0, 1,\ldots, \alpha\cdot N-1\}$. The graph $G$ is parameterized by a permutation $\pi: [N]\to [N]$ of vertices in $T$: every vertex $j\in S$ is connected to all vertices $i\in T$ such that $\pi(i)\geq j$.
\end{definition} 

Note that the original $(1-1/e)$-hard instance of Karp, Vazirani and Vazirani~\cite{KarpVV90} is a $1$-KVV instance as above with the permutation $\pi$ chosen uniformly at random. A version of this construction was implemented using techniques from the literature on Ruzsa-Szemer\'edi  graphs in~\cite{Kapralov13}, showing that any algorithm that finds a better than $(1-1/e)$-approximation using a single pass over the input stream with high constant probability must use $n^{1+\Omega(1/\log\log n)}$ bits of space. In this work we show how to combine such implementations of a KVV-gadget to achieve the stronger hardness result of $\frac1{1+\ln 2}$ for single pass streaming algorithms in the more general edge arrival model. To achieve our result, we implement the construction of~\cite{EpsteinLSW13}, which we now describe.

\paragraph{Combining $1/2$-KVV gadgets: the hard instance of~\cite{EpsteinLSW13}}.  The hard instance of~\cite{EpsteinLSW13} uses a combination of $L$ independent copies of the $1/2$-KVV gadget for a large constant $L$ as follows. For $\ell\in [L]=\{0, 1,\ldots, L-1\}$\footnote{We use the notation $[a]=\{0,1,\ldots, a-1\}$ throughout the paper.} let $G^\ell=(S^\ell, T^\ell, E^\ell)$ be an independent $1/2$-KVV gadget. Let $\pi^\ell$ denote the $\ell$-th permutation, selected independently and uniformly at random.  For every $\ell\in [L]$ define the {\em terminal subset} of $T^\ell$ as the set of vertices that are assigned the largest values by $\pi^\ell$. Namely, let
$$
T_*^\ell=\{i\in T^\ell: \pi(i)\geq N/2\}.
$$
Note that $|T_*^\ell|=N/2$ for every $\ell$ (we assume that $N$ is even).  The actual input graph $\wh{G}=(P, Q, \wh{E})$ of ~\cite{EpsteinLSW13} is defined as follows. First, one lets 
$$
P=\bigcup_{\text{even~}\ell\in [L]} T^\ell
$$
and 
$$
Q=S^0\cup\bigcup_{\text{odd~}\ell\in [L]} T^\ell.
$$
For every $\ell\in [L], \ell>1$ one lets 
$$
\tau^\ell: S^\ell \to T_*^{\ell-1}
$$
denote an arbitrary bijective mapping between the $S^\ell$ side of the biparition of $G^\ell$ and the terminal subset $T_*^{\ell-1}$ of $G^{\ell-1}$ -- we refer to such maps as {\em glueing maps}.    See Fig.~\dualref{fig:g-ell} for an illustration. Let $\tau^0$ denote the identity map for convenience. The mapping $\tau^\ell, \ell\in [L], \ell>0,$ is naturally extended to edges $e=(u, v)\in E^\ell$, where $u\in S^\ell$ and $v\in T^\ell$ by letting
$$
\tau^\ell(e)=(\tau^\ell(u), v).
$$
In other words, one simply applies the map $\tau^\ell$ to the $S^\ell$ endpoint of $e$.  The edge set $\wh{E}$ of $\wh{G}$ is now defined as follows: for every $\ell\in [L]$ one adds, for every edge $e\in E^\ell$, the edge $\tau^\ell(e)$ to $\wh{E}$. In other words, for every $\ell\in [L], \ell>0,$ one simply grows the $1/2$-KVV instance $G^\ell$ with the $S^\ell$ side of the bipartition identified with the terminal subset $T_*^\ell$ of the previous instance $G^{\ell-1}$ -- see Fig.~\dualref{fig:g-hat} for an illustration. In~\cite{EpsteinLSW13} the authors show that no online preemptive algorithm can find a better than $(\frac1{1+\ln 2}+o(1))$-competitive matching on this instance in expectation (and with nontrivial probability).

\begin{figure}[H]
	\begin{center}
		\begin{tikzpicture}[scale=1.2]

		\draw (-6, 0) rectangle (-3, 1);
		\draw (-2, 0) rectangle (+1, 1);
		\draw (+2, 0) rectangle (+5, 1);
		
		\draw (-6, -2) rectangle (-4.5, -1);
		\draw (-2, -2) rectangle (-0.5, -1);
		\draw (+2, -2) rectangle (+3.5, -1);

		\draw[fill=gray!20] (-4.5, 0) rectangle (-3, +1);
		\draw[fill=gray!20] (-0.5, 0) rectangle (1, +1);
		\draw[fill=gray!20] (+3.5, 0) rectangle (+5.0, +1);
		
		\draw (-5.75, -2+0.5) -- (-5.75, 0.5);
		\draw (-5.25, -2+0.5) -- (-5.25, 0.5);
		\draw (-4.75, -2+0.5) -- (-4.75, 0.5);		

		\draw (-5.75, -2+0.5) -- (-5.75+1, 0.5);
		\draw (-5.25, -2+0.5) -- (-5.25+1, 0.5);
		\draw (-4.75, -2+0.5) -- (-4.75+1, 0.5);		
		
		\draw (-5, 1.5) node {\large $T^{\ell-1}$};
		\draw (-3.5, 1.5) node {\large $T_*^{\ell-1}$};		
		\draw (-5, -2.5) node {\large $S^{\ell-1}$};		

		\draw[->] (-6.5, -1.2) -- (-7, -0.2);
		\draw (-7, -1.2) node {\large $\tau^{\ell-1}$};

		\draw (-5.75+4, -2+0.5) -- (-5.75+4, 0.5);
		\draw (-5.25+4, -2+0.5) -- (-5.25+4, 0.5);
		\draw (-4.75+4, -2+0.5) -- (-4.75+4, 0.5);		

		\draw (-5.75+4, -2+0.5) -- (-5.75+1+4, 0.5);
		\draw (-5.25+4, -2+0.5) -- (-5.25+1+4, 0.5);
		\draw (-4.75+4, -2+0.5) -- (-4.75+1+4, 0.5);		

		\draw (-5+4, 1.5) node {\large $T^{\ell}$};
		\draw (-3.5+4, 1.5) node {\large $T_*^{\ell}$};
		\draw (-5+4, -2.5) node {\large $S^{\ell}$};		

		\draw[->] (-6.5+4, -1.2) -- (-7+4, -0.2);
		\draw (-7+4, -1.2) node {\large $\tau^{\ell}$};

		\draw (-5.75+8, -2+0.5) -- (-5.75+8, 0.5);
		\draw (-5.25+8, -2+0.5) -- (-5.25+8, 0.5);
		\draw (-4.75+8, -2+0.5) -- (-4.75+8, 0.5);		

		\draw (-5.75+8, -2+0.5) -- (-5.75+1+8, 0.5);
		\draw (-5.25+8, -2+0.5) -- (-5.25+1+8, 0.5);
		\draw (-4.75+8, -2+0.5) -- (-4.75+1+8, 0.5);		

		\draw (-5+8, 1.5) node {\large $T^{\ell+1}$};
		\draw (-3.5+8, 1.5) node {\large $T_*^{\ell+1}$};
		\draw (-5+8, -2.5) node {\large $S^{\ell+1}$};		
		
		\draw[->] (-6.5+8, -1.2) -- (-7+8, -0.2);
		\draw (-7+8, -1.2) node {\large $\tau^{\ell+1}$};	
		\draw[->] (-6.5+12, -1.2) -- (-7+12, -0.2);
		\draw (-7+12, -1.2) node {\large $\tau^{\ell+2}$};

		\end{tikzpicture}
		\caption{Illustration of the basic gadgets $G^\ell=(S^\ell, T^\ell, E^\ell)$ and the glueing maps $\tau^\ell: S^\ell\to T_*^{\ell-1}$. The terminal sets in each gadget are shaded.} 	\duallabel{fig:g-ell}
	\end{center}
	
\end{figure}
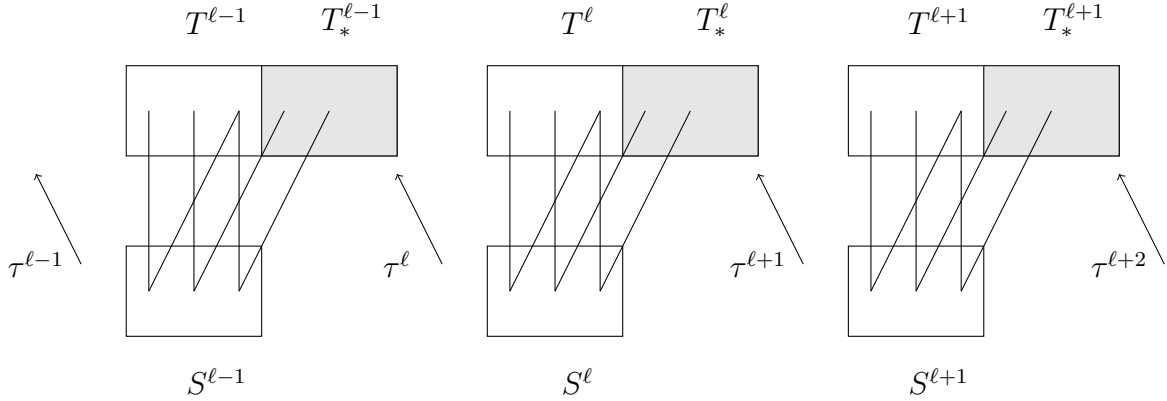

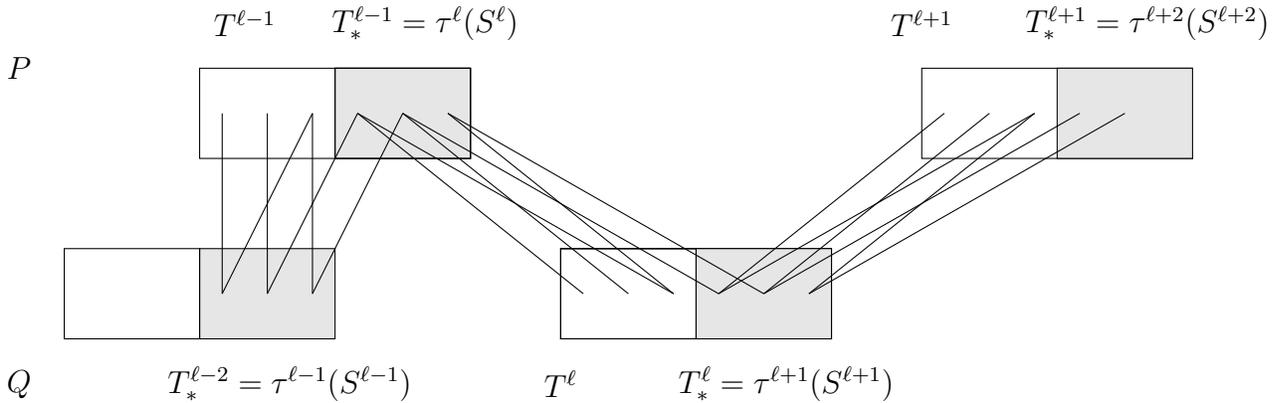
\begin{figure}[H]
	\begin{center}
		\begin{tikzpicture}[scale=1.2]

 \begin{scope}[shift={(0,-2)}]


		\draw (-6, 0) rectangle (-3, 1);
		\draw[fill=gray!20] (-6, -2) rectangle (-4.5, -1);
		\draw (-7.5, -2) rectangle (-6, -1);
		\draw[fill=gray!20] (-4.5, 0) rectangle (-3, +1);
		\draw (-5.75, -2+0.5) -- (-5.75, 0.5);
		\draw (-5.25, -2+0.5) -- (-5.25, 0.5);
		\draw (-4.75, -2+0.5) -- (-4.75, 0.5);		

		\draw (-5.75, -2+0.5) -- (-5.75+1, 0.5);
		\draw (-5.25, -2+0.5) -- (-5.25+1, 0.5);
		\draw (-4.75, -2+0.5) -- (-4.75+1, 0.5);		
		
		\draw (-5.5, 1.5) node {\large $T^{\ell-1}$};
		\draw (-3.5, 1.5) node {\large $T_*^{\ell-1}=\tau^{\ell}(S^\ell)$};	
		\draw (-5, -2.5) node {\large $T_*^{\ell-2}=\tau^{\ell-1}(S^{\ell-1})$};

		\draw (-2-2.5, 0) rectangle (-0.5-2.5, +1);
		
		\draw (-2, -2) rectangle (+1, -1);
		\draw (-2, -2) rectangle (-0.5, -1);		
		\draw[fill=gray!20] (-0.5, -2) rectangle (1, -1);
		
		\draw (-5.75+4-2.5, +0.5) -- (-5.75+4, -2+0.5);
		\draw (-5.25+4-2.5, +0.5) -- (-5.25+4, -2+0.5);
		\draw (-4.75+4-2.5, +0.5) -- (-4.75+4, -2+0.5);		

		\draw (-5.75+4-2.5, +0.5) -- (-5.75+1+4, -2+0.5);
		\draw (-5.25+4-2.5, +0.5) -- (-5.25+1+4, -2+0.5);
		\draw (-4.75+4-2.5, +0.5) -- (-4.75+1+4, -2+0.5);

		\draw (-5+4-1, -2.5) node {\large $T^{\ell}$};
		\draw (-3.5+4, -2.5) node {\large $T_*^{\ell}=\tau^{\ell+1}(S^{\ell+1})$};

		\draw (+2, 0) rectangle (+5, 1);
		\draw (+2-4, -2) rectangle (+3.5-4, -1);		

		\draw[fill=gray!20] (+3.5, 0) rectangle (+5.0, +1);
		
		\draw (-5.75+8-2.5, -2+0.5) -- (-5.75+8, 0.5);
		\draw (-5.25+8-2.5, -2+0.5) -- (-5.25+8, 0.5);
		\draw (-4.75+8-2.5, -2+0.5) -- (-4.75+8, 0.5);		

		\draw (-5.75+8-2.5, -2+0.5) -- (-5.75+1+8, 0.5);
		\draw (-5.25+8-2.5, -2+0.5) -- (-5.25+1+8, 0.5);
		\draw (-4.75+8-2.5, -2+0.5) -- (-4.75+1+8, 0.5);		

		\draw (-5+8-1, 1.5) node {\large $T^{\ell+1}$};
		\draw (-3.5+8, 1.5) node {\large $T_*^{\ell+1}=\tau^{\ell+2}(S^{\ell+2})$};
		\draw (-8, -2.5) node {\large $Q$};
		\draw (-8, +1) node {\large $P$};		
		
            \end{scope}		

		\end{tikzpicture}
		\caption{Illustration of the final graph $\wh{G}=(P, Q, \wh{E})$ obtained by gluing together basic gadgets $G^\ell=(S^\ell, T^\ell, E^\ell)$ from Fig.~\dualref{fig:g-ell} using maps $\tau^\ell$.}	\duallabel{fig:g-hat}
	\end{center}
	
\end{figure}

\paragraph{Our construction: a geometric implementation of the instance of~\cite{EpsteinLSW13}.} We present our lower bound construction in two steps. First, two illustrate our construction, we consider a simpler model than streaming, namely the {\em generalized online} model that we define below. The intuition behind the model is simple. In this model the edges of the graph are presented to the algorithm as a stream, and the algorithm must output a large matching at the end of the stream. Upon receiving an edge in the stream the algorithm can arbitrarily choose to either remember the edge or discard it (in this case the algorithm may not use the edge as part of the final output matching), and is constrained to remember at most $s$ edges overall, i.e. the algorithm can only use edges that it remember when they arrived. The algorithm is not allowed to forget edges, i.e. the budget of $s$ bounds the total number of edges remembered upon their arrival. The formal definition of the model is given in

\begin{definition}[Generalized online algorithms]\duallabel{def:gen-online} In the {\em generalized online} setting the algorithm is presented with edges of a graph $G=(P, Q, E)$ as a stream of edges, and at every point must either commit to remembering the edge that has been presented to it, or discard the edge irrevocably. The total number of edges that the algorithm can remember is bounded by a parameter $s$. At the end of the stream the algorithm must output a matching $M_{ALG}$ in the subset of edges that it remembered upon their arrival.
\end{definition}
This setting is easier than streaming, where the algorithm may maintain any small state. On the other hand, this setting is quite a bit more general than the online model, as the the algorithm may maintain significantly more edges than are needed to find a matching. It is not hard to see that this power renders standard hard instances for online algorithms very easy. In particular, 
\begin{lemma}\duallabel{lm:alg-gen-online}
For every constant $\e\in (0, 1)$ there exists a generalized online algorithm that remembers $s=O(n/\e^2)$ edges and achieves a $(1-\e)$-approximation on the instance above.
\end{lemma}
The algorithm is simple -- one simply maintains a random sample of $O(n/\e^2)$ edges of the input graph and outputs a maximum matching in the sample at the end. We include the proof of Lemma~\dualref{lm:alg-gen-online} in Appendix~\ref{app:gen-online} for completeness. A similar claim is true for the $1$-KVV instance -- simply maintaining a uniform sample of $O(n/\e^2)$ edges of the input graph will result in a $(1-\e)$-approximation.

To illustrate the techniques that lead to a proof of our main result (Theorem~\ref{thm:main}) in an easier setting, in Section~\ref{sec:toy-construction} we prove the following: 

\begin{theorem}\label{thm:main-simple}
There exists a distribution $\mathcal{D}$ on input graphs $G=(V, E)$ with $n$ vertices such that any generalized algorithm that finds a $(\frac1{1+\ln 2}+\eta)$-approximation to the maximum matching in a graph $G$ sampled from $\mathcal{D}$ with probability at least $0.9$ must remember $\Omega(n\log n)$ edges.
\end{theorem}
Note that even though the lower bound of $\Omega(n\log n)$ edges is not very strong from the standpoint of streaming algorithms, the result is interesting in light of Lemma~\dualref{lm:alg-gen-online}. A major advantage of this setting is that it {\bf (a)} allows for a rather clean construction and {\bf (b)} illustrates all the central ideas of our main construction that leads to a proof of Theorem~\ref{thm:main}.

In what follows we give an outline of the proof of Theorem~\ref{thm:main-simple}. The construction consists of two pieces. First, we define a construction of a basic gadget $G^\ell=(S^\ell, T^\ell, E^\ell)$, which is a geometric version of the $\alpha$-KVV gadgets defined above. Second, we define maps $\tau^\ell$ that glue together these gadgets to obtain the final input instance. Finally, we prove the $\frac1{1+\ln 2}$-hardness result in the generalized online model we defined above. The main ideas behind the first step are implicit in~\cite{Kapralov13}, but we are able to present the construction in a different and arguably cleaner way (in particular, all bounds on the sizes of various sets are exact in our construction, which significantly simplifies presentation). The main contribution of the present paper lies in the second and third steps.

\paragraph{A geometric version of the $\alpha$-KVV gadget $G=(S, T, E)$.}  Let $K, m\geq 1$ be large constant integers. Let 
$$
T=[m]^n
$$ 
i.e. vertices in $T$ are vectors of dimension $n$, with each co-ordinate taking values in $[m]=\{0, 1, 2,\ldots, m-1\}.$ This way we have $N:=|T|=m^n$, so $n=\Omega(\log N)$ for every constant $m$. The vertices on the $S$ side of the bipartition will also be associated with points on the hypercube $[m]^n$, as we define below. We often treat vertices in $T$ or $S$ and points in $[m]^n$ interchangeably where this does not create confusion. The set $S$ will consist of $\alpha K$ disjoint sets (we will use $\alpha=1/2$ for our main result here, since we are implementing an instance that uses $1/2$-KVV gadgets). We will have 
$$
S=S_0\uplus S_1\uplus \ldots \uplus S_{\alpha K-1}.
$$
\begin{remark}[Use of $\uplus$ instead of $\cup$]
Note that we use the $\uplus$ as opposed to $\cup$ above. The reason for this is as follows. It is convenient to view vertices in $G$ as points in the hypercube $[m]^n$. Formally, this means that our vertices are labeled by points in the hypercube. For example, the set $T$ contains all of the hypercube $[m]^n$, and there is no confusion in using points in $[m]^n$ and vertices in $T$ interchangeably. Vertices in $S$ are also labeled by points in the hypercube, as we define below, but the labels are not distinct -- there can be two vertices, say one in $S_i$ and one in $S_j$ for $i\neq j$, whose labels are the same (but labels are distinct within one set $S_i, i\in [\alpha\cdot K]$).  Thus, we use the $\uplus$ sign to stress the fact that the union above is disjoint, even if different sets $S_i$ may contain vertices with the same labels, to avoid confusion. Also, for two vertices $x, y\in S\cup T$ we write $x\delequal y$ to denote the relation `the label of $x$ equals the label of $y$' -- see definition of $S_k$ in~\dualeqref{eq:def-sk}.
\end{remark}

 Before we define $S$, however, recall that an $\alpha$-KVV gadget is parameterized by a permutation $\pi:[N]\to [N]$, and then every vertex  $j\in S=\{0, 1,\ldots, N-1\}$ has an edge to vertices $i\in T$ such that $\pi(i)\geq j$. In our basic gadget the role of this permutation $\pi$ is played by a nested sequence of subsets of $T$ that we denote by 
$$
T=T_0\supset T_1\supset\ldots \supset T_{\alpha K}=T_*,
$$
where the outermost set in the nested sequence is the entire $T$ side of the bipartition, and the innermost set is the {\em terminal subset $T_*$}, which we refer to as the {\em terminal subcube} for reasons that will become clear shortly. These $\alpha K+1$ nested sets will correspond to $\alpha K$ phases over which the gadget will be revealed to the algorithm (nothing is revealed in the last phase for certain technical reasons). In every phase $k\in [\alpha K]$ the algorithm will receive a carefully crafted subset of edges in $S_k\times T_k$, i.e. a subgraph induced by the $k$-th set $S_k$ and the $k$-th set $T_k$ in the nested sequence above. Once $\alpha K$ rounds are done, the next gadget will be presented, with vertices in the terminal subcube $T_*=T_{\alpha K}$ serving as the $S$ side of the new gadget, as in the~\cite{EpsteinLSW13} construction outlined above.  

We now define the nested sequence $T_0\supset T_1\supset\ldots \supset T_{\alpha K}=T_*$.  Choose a subset $\B\subseteq [n]$ of coordinates to be used by our basic $\alpha$-KVV gadget $G=(S, T, E)$ (we need to reserve other coordinate blocks for the other $L-1$ gadgets -- see below). Partition $\B$ into $\alpha K+1$ disjoint roughly equal size subsets as 
$$
\B=\B_0\cup\ldots\cup \B_{\alpha K},
$$ 
where $|\B_k|=\frac{|\B|}{(\alpha K+1)}$. The nested sequence in $T$ is parameterized by a vector
$$
J\in \B_0\times \ldots\times  \B_{\alpha K}.
$$
In other words, for every $k\in [\alpha K+1]$ we have $J_k\in \B_k$.
We use the notation $J_{<k}:=(J_0,\ldots, J_{k-1})$ and $J_{\geq k}:=(J_k,\ldots, J_{K/2}).$ Let $T_0=T$, and for every $k\in [\alpha K]$ let 
\begin{equation}\duallabel{eq:def-tk}
\begin{split}
T_{k+1}=\left\{y\in T_k: y_{J_k}/m\in \left[0, 1-\frac{1}{K-k}\right)\right\},\\
\end{split}
\end{equation}
so that
\begin{equation}\duallabel{eq:def-tk-allconstraints}
\begin{split}
T_k=\left\{y\in [m]^n: y_{J_s}/m\in \left[0, 1-\frac{1}{K-s}\right)\text{~~for all~}s\in \{0, 1, \ldots, k-1\}\right\}.\\
\end{split}
\end{equation}
One can show that $|T_k|=(1-\frac{k}{K})\cdot |T_0|$ for every $k\in [\alpha K+1]$, i.e. the sizes of $T_k$ decrease linearly in $k$ -- see Lemma~\ref{lm:size-bounds}\footnote{We note that lemma proved in Section~\ref{sec:toy-construction} are presented for the setting of $\alpha=1/2$. However, all of these bounds extend to other settings of $\alpha$ that are bounded away from $1$.}.
The set $S$ of vertices is naturally partitioned into disjoint subsets 
\begin{equation}\duallabel{eq:def-s}
S=S_0\uplus S_1\uplus \ldots \uplus S_{\alpha K-1}
\end{equation}
as follows. For every $k\in [\alpha K]$ we let
\begin{equation}\duallabel{eq:def-sk}
\begin{split}
S_k\delequal \{x\in T_k: &\weight(x) \in \left[0, \frac{1}{K-k}\right)\cdot W \pmod{ W}\}
\end{split}
\end{equation}
In the definition above $W$ is an integer parameter that we choose so that $W$ divides $m$ and $\weight(x)=\sum_{j\in [n]} x_j$.  Recall that for a pair of vertices $x, y\in P\cup Q$ we write $x\delequal y$ if their labels (vertices of the hypercube $[m]^n$ assigned to them) are the same. The notation in~\dualeqref{eq:def-sk} above stands for $S_k$ being a copy of the set of vertices on the rhs.  Intuitively, the set $S_k$ is a {\em subsample} of the set $T_k$ that contains a $\frac1{K-k}$ fraction of points in $T_k$. A similar effect was achieved in~\cite{Kapralov13} by sampling vertices in $T_k$ independently, but we find this deterministic construction cleaner to present.  The key reason why we include vertices in $S_k$ depending on the residue class of their weight modulo $W$ is that we need this `sampling mechanism' to `accept' exactly a $\frac1{K-k}$ fraction of vertices along every coordinate aligned line as defined in~\dualeqref{eq:line-def} below; this property is crucial for establishing the existence of a large matching in our gadget -- see Lemma~\ref{lm:matching}. Using the weight of a point ensures that this property is satisfied. Another important observation is that $|S_k|=\frac1{K}\cdot |T_0|$ for every $k\in [\alpha K]$ (see Lemma~\ref{lm:size-bounds} in Section~\ref{sec:toy-construction}). Intuitively, this means that in every round $k\in [\alpha K]$ the number of vertices arriving on the $S$ side of the bipartition and revealing their edges to vertices in $T_k$ is the same.    We also define, for every $k\in [\alpha K]$ and $j\in \B_k$ 
\begin{equation}\duallabel{eq:def-tkj}
\begin{split}
T_k^j&=\left\{y\in T_k: y_j/m\in \left[0, 1-\frac{1}{K-k}\right)\right\}\\
S_k^j&=\left\{x\in S_k: x_j/m\in \left[0, 1-\frac{1}{K-k}\right)\right\}.\\
\end{split}
\end{equation}
Note that $T_k^{J_k}=T_{k+1}$ as per~\dualeqref{eq:def-tk}. The intuition behind these sets is that during phase $k$, i.e. when the edge set induced by $S_k$ and $T_k$ is revealed to the algorithm, the algorithm is presented with several possible options for the next set $T_{k+1}$ in the nested sequence defined above. In other words, given the edge set presented in round $k$ and before (we define the edge set below), any of the sets $T_k^j$ for all $j\in \B_k$ (rather, any $j$ in a subset $\ac{\B}_j$ of $\B_j$ of comparale size)  look like perfectly valid continuations for the nested sequence. The algorithm does not know which of them is important and hence most likely misses important edges in phase $k$ -- see below for more details.

\paragraph{Edges of $G$.}   Fix $k\in [\alpha K]$. For each coordinate $j\in \B_k$ for each $x\in [m]^n$ we denote the line in direction $j$ going through $x$ by 
\begin{equation}\duallabel{eq:line-def}
\text{line}_j(x)=\{x'\in [m]^n: x'_{-j}=x_{-j}\},
\end{equation}
where we write $x_{-j}$ to denote the restriction of $x$ to coordinates $[n]\setminus \{j\}$. Note that for every $y, y'$ and every $j$ one has either $\text{line}_j(y)=\text{line}_j(y')$ or $\text{line}_j(y)\cap \text{line}_j(y')=\emptyset$, i.e. lines in direction $j$ partition $T_k$, and consequently also partition $S_k$.  We now define the edges of $G$ incident on $S_k$ for every $k\in [\alpha K]$. For that we first need

\begin{definition}[Line cover in direction $j$]\duallabel{def:line-cover}
For every $j\in \B_k$ a collection $C_k^j\subseteq T_k$ of representative points is called a {\em line cover of $T_k$ in direction $j$} if 
$T_k=\bigcup_{y\in C_k^j} \text{line}_j(y)$
and $\text{line}_j(y)\cap \text{line}_j(y')=\emptyset$ for every $y, y'\in C_k^j$, $y\neq y'$.
\end{definition}
Note that a line cover of $T_k$ in direction $j$ can be constructed by picking points $y\in T_k$ greedily until the union of lines in direction $j$ through these points covers $T_k$. 

The edge set induced by $S_k\cup T_k$ is defined as follows. First, we leave out a few coordinates from the current coordinate block $\B_k$, letting  
$$
\ac{\B}_k\subset \B_k.
$$ 
The reason for this will be clear once we define the glueing maps $\tau^\ell$ below. For now it is only important that $|\ac{\B}_k|\approx |\B_k|$, i.e. we did not lose too many coordinates by passing to $\ac{\B}_k$. Now for every $j\in \ac{\B}_k$ fix a line cover $C_k^j$ of $T_k$ in direction $j$ (as per Definition~\dualref{def:line-cover}). 

\begin{center}
\fbox{\parbox{0.8\textwidth}{
\begin{center}
For every $y\in C_k^j$  we include a complete bipartite graph between $\text{line}_j(y)\cap S_k^j$ and $\text{line}_j(y)\cap (T_k\setminus T_k^j)$. 
\end{center}
}}
\end{center}

In other words, let
\begin{equation}\duallabel{eq:edges-ei-k-def}
E_k=\bigcup_{j\in \ac{\B}_k} E_{k, j},
\end{equation}
where
\begin{equation}\duallabel{eq:edges-ei-def}
E_{k, j}=\bigcup_{y\in C_k^j} (\text{line}_j(y)\cap S_k^j) \times (\text{line}_j(y) \cap (T_k\setminus T_k^j)).
\end{equation}
One can show that the edge sets $E_{k, j}$ are disjoint (see Lemma~\ref{lm:disjoint}).  Note that $E_k$ is fully determined by the first $k-1$ values of $J$, namely by the prefix $J_{<k}$.  At this point we note that for our hard input distribution we choose 
$$
J\sim UNIF(\ac{\B}_0\times\ldots\times \ac{\B}_{\alpha K}).
$$
Crucially, conditioned on all edges received up to phase $k$, i.e. on $\bigcup_{s=0}^{k} E_s$, one has $J_k\sim UNIF(\ac{\B}_k)$. This property is important for a key structural lemma, Lemma~\dualref{lm:key} below. 

For every choice of $J\in \ac{\B}_0\times\ldots\times \ac{\B}_{\alpha K}$ the edge set of $G$ that we defined satisfies
\begin{lemma}[Matching of $S$ to $T\setminus T_*$; see Lemma~\ref{lm:matching} in Section~\ref{sec:toy-construction}]\duallabel{lm:matching}
There exists a matching of a $(1-O(1/K))$ fraction of $S$ to $T\setminus T_*$ in $E$.
\end{lemma}
This is a very natural property since the instance that we are defining is a version of the $\alpha$-KVV gadget defined at the beginning of this section: those gadgets admitted a perfect matching of $S$ to $T\setminus T_*$ (the optimal matching in that instance).

We now define the key property that underlies our lower bound (and, similarly, that of~\cite{Kapralov13}). For that we need the definition of a downset of a subset $U$ of $T$:
\begin{definition}[Down-set of a set in $T$]\duallabel{def:downset}
For every $U\subseteq T$, $k\in [\alpha K]$, we define the {\em downset of $U$ in $S_k$} by 
\begin{equation*}
\begin{split}
\downset_k(U)=\{x\in S_k: \exists y\in U: y\delequal x\}
\end{split}
\end{equation*}
and  define
$$
\downset(U)=\bigcup_{k\in [\alpha K]} \downset_k(U).
$$
\end{definition}
Note that a given point in $U$ has anywhere between $0$ and $\alpha K$ images under the $\downset$ map: indeed, the downset of a point $x\in U$ is simply the set of points in the vertex sets $S_0, S_1,\ldots, S_{\alpha K-1}$ whose labels match the label of $x$.  There can be up to $\alpha K$ such points, since labels are distinct within every single $S_i, i\in [\alpha K]$. It is also good to note that $S_k=\downset_k(T_k)$ for every $k\in [\alpha K]$. Finally, it is important to note that for appropriately `nice' subsets $U\subseteq T$  (see Lemma~\ref{lm:u-subsampling} in Section~\ref{sec:toy-construction}) one has 
$$
|\downset_k(U)|=\frac1{K-k}|U|,
$$
which is consistent with the idea that our weight condition in the definition of $S_k$ (see~\dualeqref{eq:def-sk}) essentially `samples' points at rate $\frac1{K-k}$. The more important property of the $\downset$ map and the terminal subcube $T_*$ is
\begin{lemma}[Key structural property]\duallabel{lm:key}
For $G=(S, T, E)$ defined as above, for every $E'\subset E$ one has 
$$
E'\cap (\downset(T_*) \times (T\setminus T_*)) \subseteq \bigcup_{k\in [\alpha K]} E'\cap E_{k, J_k}.
$$ 
\end{lemma}
Note that intuitively the lemma above shows that only very special edges in $E$, namely the ones in $E_{k, J_k}$, cross from $\downset(T_*)$ to the complement of $T_*$ in $T$. This is intuitively useful since, as we verify below, for the right setting of parameters the cardinality of $\downset(T_*)$ is quite a bit higher than that of $T_*$, meaning that if $E'\cap E_{k, J_k}$ is small, one gets a Hall's theorem witness set certifying that $E'$ does not contain a large matching (as without these special edges all neighbors of $\downset(T_*)$ are in $T_*$, a set of size significantly smaller than $\downset(T_*)$). 

\begin{proofof}{Lemma~\dualref{lm:key}}
Consider an edge $(a, b)\in E'\cap E_k$ with $a\in \downset_k(T_*)$ for some $k\in [\alpha K]$ and $b\in T\setminus T_*$.  Recalling that $T_*=T_{\alpha K+1}$ and using~\dualeqref{eq:def-tk-allconstraints}, we get
\begin{equation*}
\begin{split}
T_*=\left\{y\in [m]^n: y_{J_s}/m\in \left[0, 1-\frac{1}{K-s}\right)\text{~~for all~}s\in \{0, 1, \ldots, \alpha K\}\right\}.\\
\end{split}
\end{equation*}
Further, since by Definition~\dualref{def:downset} the set $\downset_k(T_*)$ is the set of vertices in $S_k$  whose label matches the label of some vertex in $T_*$, the assumption that $a\in \downset_k(T_*)$ implies
$$
a_{J_s}/m\in \left[0, 1-\frac{1}{K-s}\right)\text{~~for all~}s\in \{0, 1, \ldots, \alpha K\}.
$$
On the other hand, since $b\in T\setminus T_*$ by assumption, there exists an index $r\in [\alpha K+1]$ such that
$$
b_{J_r}/m\in \left[1-\frac{1}{K-r}, 1\right).
$$
Thus, $a$ and $b$ differ on coordinate $J_r$. At the same time, we have $(a, b)\in E_k$ by assumption, which means by~\dualeqref{eq:edges-ei-k-def} that there exists a point $y\in C_k^j$ (a line cover of $T_k$ in direction $j$) such that  $a\in \text{line}_j(y)$ and $b\in \text{line}_j(y)$, which by definition of a line in direction $j$ (see~\dualeqref{eq:line-def}) implies that  $a_{-j}=y_{-j}=b_{-j}.$
Since $a$ and $b$ differ on coordinate $J_r$, we now get that $j=J_r$. It remains to note that edges in $E_k$ are all generated by lines in directions $j\in \B_k$. Since the blocks $\B_i$ are disjoint for different $i$ by construction, we get that $J_r\in \B_k$. Thus, $r=k$ and $(a, b)\in E_{k, J_k}$, as required.
\end{proofof}

We now explain the significance of the key structural property above. Suppose that $E'$ is the set of edges maintained by a generalized online algorithm that remembers at most  $s=o(N\log N)$ edges. Taking the expectation of the rhs in Lemma~\dualref{lm:key} above with respect to $J_k$ and conditioning $J_{<k}$, we get
\begin{equation}\duallabel{eq:293ht92hg}
\begin{split}
\expect_{J_k\sim UNIF(\ac{\B}_k)} \left[|E'\cap E_{k, J_k}|\right]&=\sum_{j\in \ac{\B}_k} |E'\cap E_{k, j}|\cdot \prob[J_k=j]\\
&=\frac1{|\ac{\B}_k|}\sum_{j\in \ac{\B}_k} |E'\cap E_{k, j}|\\
&\leq \frac{s}{|\ac{\B}_k|},\\
\end{split}
\end{equation}
where we used the fact that $\sum_{j\in \ac{\B}_k} |E'\cap E_{k, j}|\leq |E'|\leq s$ for an algorithm that remembers at most $s$ edges (since $E_{k, j}$ are disjoint), as well as the fact that $J_k$ is independent of $J_{<k}$. At the same time we have $|\ac{\B}_k|\geq n/2K$ or so, since we have $n$ coordinates altogether, and $\alpha K\leq K$ phases to present the gadget to the algorithm over. If the parameter $m$ is a constant (which it is by our parameter setting), we get $|\ac{\B}_k|\geq n/2K=\Omega_K(\log N)$.  Substituting into~\dualeqref{eq:293ht92hg}, we thus get
\begin{equation*}
\begin{split}
\expect_{J_k\sim UNIF(\ac{\B}_k)} \left[|E'\cap E_{k, J_k}|\right]=O_K(s/\log N).
\end{split}
\end{equation*}
Summing over all $\alpha K\leq K$ phases, we get that any generalized algorithm that remembers at most $s$ edges can remember at most 
$$
O_K(s/\log N)=o(N)
$$ 
edges from $E'\cap (\downset(T_*)\times (T\setminus T_*))$ (the lhs in Lemma~\dualref{lm:key}) since $s=o(N\log N)$ by assumption. Now we can upper bound the size of the matching that the set $E'$ contains by exhibiting a vertex cover as follows:
\begin{lemma}[Small vertex cover]\duallabel{lm:vertex-cover-simple}
The size of the maximum matching in $E'\subseteq E$ is upper bounded by 
$$
|E'\cap (\downset(T_*) \times (T\setminus T_*))|+|S\setminus \downset(T_*)|+|T_*|.
$$
\end{lemma}
\begin{proof}
We construct a vertex cover by first adding one endpoint of every edge $e\in (\downset(T_*) \times (T\setminus T_*))$. Then add $T_*$ and $S\setminus \downset(T_*)$. This is indeed a vertex cover: every edge $(u, v)\in E'$ either belongs to the first edge set, or has one endpoint in at least one of the other two. The lemma now follows.
\end{proof}

Since, as we established above, 
\begin{equation}\duallabel{eq:low-order}
\expect_J[|E'\cap (\downset(T_*) \times (T\setminus T_*))|]=o(N),
\end{equation}
it suffices to upper bound $|S\setminus \downset(T_*)|$ and $|T_*|$. Since $T_*=T_{\alpha K}$, and we have $|T_k|=(1-\frac{k}{K})|T_0|$ (see Lemma~\ref{lm:size-bounds} in Section~\ref{sec:toy-construction}), we get that $|T_*|=(1-\alpha) |T_0|$. We also have
\begin{equation*}
\begin{split}
|S\setminus \downset(T_*)|&=|S|-|\downset(T_*)|\\
&=\sum_{k\in [\alpha K]} |S_k|-\sum_{k\in [\alpha K]} |\downset_k(T_*)|\\
&=\alpha K\cdot \frac1{K}|T_0|-\sum_{k\in [\alpha K]} \frac1{K-k} |T_*|\\
&=\left(\alpha-\int_0^\alpha \frac1{1-x} dx \cdot (1-\alpha)\right)|T_0|+O(1/K)|T_0|.\\
\end{split}
\end{equation*}
We used the fact that $|S_k|=\frac1{K}|T_0|$ (see Lemma~\ref{lm:size-bounds}) and $|\downset_k(T_*)|=\frac1{K-k}|T_*|$ (see Lemma~\ref{lm:u-subsampling}). Letting $\alpha=1-e^{-1}$ and using the fact that $\int_0^\alpha \frac1{1-x}dx=\ln \frac1{1-\alpha}=1$, we get
$$
|S\setminus \downset(T_*)|=(2\alpha-1+O(1/K))|T_0|,
$$
and our upper bound on the size of the matching constructed by the algorithm becomes
\begin{equation*}
\begin{split}
|E'\cap (\downset(T_*)\times (T\setminus T_*))|+|S\setminus \downset(T_*)|+|T_*|&\approx \left((2\alpha-1)+(1-\alpha)\right)|T_0|\\
&=\alpha |T_0|\\
&=(1-e^{-1})|T_0|.
\end{split}
\end{equation*}
Thus, by Markov's inequality applied to $|E'\cap (\downset(T_*)\times (T\setminus T_*))|$ no generalized online online algorithm that remembers $s=o(N\log N)$ edges can construct a matching of size larger than $(1-e^{-1})|T_0|$ with any nontrivial probability. At the same time, if the terminal set $T_*$ is perfectly matched to a separate set of vertices by an extra matching that arrives last in the stream, the size of the maximum matching in the graph is $\approx |T_0|$. This is exactly what happens in the $1-e^{-1}$ hardness result of~\cite{Kapralov13}, and brings us to our main challenge: how does one ensure that the terminal subset $T_*$ is not merely matched to an unstructured set of vertices by a perfect matching, like in~\cite{Kapralov13}, but rather that we are able to attach another $\alpha$-KVV instance (in this case, for $\alpha=1/2$) and continue?

\paragraph{Our contribution: glueing maps $\tau^\ell$ and vertex cover construction in the $\frac1{1+\ln 2}$-hard instance.} First, it is useful to observe that the construction of gluing maps $\tau^\ell$ that associate different instances is nontrivial. For example, simply choosing $\tau^\ell$ to be a random bijection from $S^\ell$ to $T_*^\ell$ will not work, as it will certainly destroy the delicate coordinate structure of the good vertex cover that we defined above.  From now on we let $\alpha=1/2$, as this discussion directly corresponds to our construction in Section~\ref{sec:toy-construction}.

We use $L$ basic gadgets $G^\ell=(S^\ell, T^\ell, E^\ell)$, $\ell\in [L]$, and assume that $L$ is even throughout the paper. We now define maps $\tau^\ell$ identifying vertices in $S^\ell$ in $G^\ell=(S^\ell, T^\ell, E^\ell)$ with vertices in the terminal subcube $T_*^{\ell-1}$ of $G^{\ell-1}=(S^{\ell-1}, T^{\ell-1}, E^{\ell-1})$. For simplicity of notation we let $G'=(S', T', E')$ denote $G^\ell=(S^\ell, T^\ell, E^\ell)$, let $G=(S, T, E)$ denote $G^{\ell-1}=(S^{\ell-1}, T^{\ell-1}, E^{\ell-1})$, and adopt similar notation for all other relevant quantities. Specifically, let the disjoint coordinate blocks dedicated to these two gadgets be denoted by $\B':=\B^\ell, \B:=\B^{\ell-1}$, and let the coordinate vectors be denoted by $J\in \B_0\times \B_1\times \ldots\times \B_{K/2}$, and $J'=\B'_0\times \B'_1\times \ldots \times \B'_{K/2}$ respectively. Thus, we define a bijection $\tau$ from $S'$ to $T_*$:
$$
\tau: S'\to T_*.
$$

We start by defining $\tau$ on the sets $S'_k$ for $k\in [K/2]$ (recall that $S'=S_0'\uplus \ldots \uplus S_{K/2-1}'$ as per~\dualeqref{eq:def-s}). The restriction of $\tau$ to $S'_k$ is denoted by $\tau_k$:
$$
\tau_k: S'_k\to T_*.
$$ 
The images of $\tau_k$ that we define will be disjoint for different $k\in [K/2]$, i.e. these maps extend naturally to an injective map from the union of $S'_k$ over all $k\in [K/2]$ to $T_*$.  At this point our task is to map {\em subsampled} cubes $S'_k$ (as per~\dualeqref{eq:def-sk}) to a {\em non-sampled} terminal subcube $T_*$ (as per~\dualeqref{eq:def-tk}).  Our first step is to design an intermediate mapping $\rho$ that maps $S'_k$ to a {\em non-sampled} subcube.  We ensure that such a map uses only one special coordinate direction -- for every $k\in [K/2]$ we denote this coordinate by $q_k\in \B_k$ and refer to it as the {\em compression index}.  We refer to the corresponding map as the {\em densifying map} $\rho$, defined below. Note that in order to `densify' $S'_k$ we need to set the densification parameter to $\lambda=K-k$ (i.e., the inverse of the subsampling rate).

\begin{definition}[$(\lambda, r)$-densifying map]\duallabel{def:rho}
For $r\in [n]$ and integer $\lambda>0$ the $(\lambda, r)$-densifying map $\rho:[m]^n\to [m]^n$ is defined as follows.  We let $x\in [m]^n$, and write $x=(x', x''), x'\in [m]^{[n]\setminus \{r\}}, x''\in [m]$. Write $x''=aW+b W/\lambda+c,$
where $a\in \{0, 1,\ldots, m/W-1\}, b\in \{0, 1,\ldots,\lambda-1\}$ and $c\in \{0, 1,\ldots, W/\lambda-1\}$. We define $\rho(x)$ by letting, for $j\in [n]$:
\begin{equation*}
\begin{split}
(\rho(x))_j:=\left\lbrace
\begin{array}{ll}
a W/\lambda+c&\text{~if~}j=r\\
x_j&\text{~o.w.~}
\end{array}
\right.
\end{split}
\end{equation*}
\end{definition}

The following lemma formalizes the densification property:
\begin{lemma}[Densification of a subsampled set; see Lemma~\ref{lm:u-subsampling} in Section~\ref{sec:toy-construction}]\duallabel{lm:densification}
For every integer $\lambda\geq 2$, every $r\in [n]$, every $U\subseteq [m]^n$ that does not depend on coordinate $r$ the $(\lambda, r)$-densifying map $\rho$ (see Definition~\dualref{def:rho}) maps 
$$
\left\{x\in U: \weight(x) \in \left[0, 1/\lambda\right)\cdot W \pmod{ W} \right\}
$$ bijectively to $\left\{x\in U: x_r/m\in \left[0, 1/\lambda\right)\right\}$.
\end{lemma}

Letting $\rho_k$ be a $(K-k, q_k)$-densifying map, we get by Lemma~\dualref{lm:densification} that
\begin{equation}\duallabel{eq:s-prime-k}
\rho_k(S'_k)=\left\{x\in T'_k: x_{q_k}/m\in \left[0, \frac1{K-k}\right) \right\}.
 \end{equation}
 This is progress, since now we need to design a map that maps the subcube $\rho_k(S'_k)$ above to $T_*$, which is also a subcube. We would like to design a mapping from $S'_k$ to $T_*$ that {\bf (a)} `uses' as few coordinates as possible and {\bf (b)} maps entire subcubes of $\rho_k(S'_k)$ to subcubes of $T_*$. This second property {\bf (b)} ensures that the structure of the good vertex cover we defined in Lemma~\dualref{lm:vertex-cover} above can be translated from one instance of a basic gadget to another. A basic issue that we are facing now is that $\rho_k(S'_k)$ and $T_*$ are subcubes, but have a different number of `active dimensions': the former constrains variables in $J'_{<k}$ and $q_k$, while the latter constrains variables in $J$. To equalize this number let  $\Ext_k\subseteq \B'_k$ is a subset of size $K/2+1-k$ referred to as the {\em extension indices} in phase $k$ -- we will artificially add them to the index set $J'_{<k}$ and $q_k$ to equalize the number of `active' coordinates.
 $q_k\in \B'_k\setminus \Ext_k$ referred to as the {\em compression index} for the $k$-th phase of round $\ell$. Now we can define the set $\ac{\B}_k$ that we already used formally: 
$\ac{\B}'_k=\B'_k\setminus (\Ext_k\cup \{q_k\})$. We also need an index $r\in \B_{K/2}$ that we refer to as the compression index for the terminal subcube $T_*$. Define index sets 
\begin{equation}\duallabel{eq:i-def}
I=J\cup \{r\}\subset [n]
\end{equation}
and 
\begin{equation}\duallabel{eq:i-prime-def}
I'=\{J'_0,\ldots, J'_{k-1}\}\cup \Ext_k\cup \{q_k\}.
\end{equation}
Note that  $|I|=|I'|$ -- this is exactly why we defined the relevant compression and extension indices. This allows us to write 
\begin{equation}\duallabel{eq:t-star-product}
T_*\delequal A\times [m]^{[n]\setminus I},
\end{equation}
where
$$
A=\left\{x\in [m]^I: x_{J_s}/m\in \left[0, \frac1{K-i}\right) \text{~for all~}s\in [K/2]\right\}.
$$
Similarly, we write 
\begin{equation}\duallabel{eq:rho-s-prime-k}
\rho_k(S'_k) \delequal D_k\times [m]^{[n]\setminus I'},
\end{equation}
where as per~\dualeqref{eq:s-prime-k}
\begin{equation*}
\begin{split}
D_k&=\left\{x\in [m]^{I'}: x_{i_s}/m\in \left[0, 1-\frac1{K-s}\right)\text{~for all~}s\in [k]\text{~and~}x_{q_k}/m\in \left[0, \frac1{K-k}\right) \right\}.
\end{split}
\end{equation*}

We choose a bijection 
\begin{equation}\duallabel{eq:m-def}
M: \biguplus_{k\in [K/2]} D_k \to A,
\end{equation}
which exists since the cardinalities of these sets are indeed equal (see derivation after~\eqref{eq:m-def} in Section~\ref{sec:toy-construction}). This lets us define another auxiliary transformation, referred to as the {\em subcube permutation map} $\Pi_k$. The composition of the densifying map $\rho_k$ and the subcube permutation map $\Pi_k$ gives us the glueing map $\tau$. 
\begin{definition}[Subcube permutation map $\Pi_k$]\duallabel{def:pi-k}
Define an injective map
\begin{equation}\duallabel{eq:pi-k-def}
\Pi_k: \rho_k(S'_k) \to T_*
\end{equation}
as follows. First, let $\eta: I\to I'$ be an arbitrary bijection. Given $x\in \rho_k(S'_k)$, write 
$$
x=(a, b, c),
$$ 
where $a=x_{I'}\in D_k\subseteq [m]^{I'}$, $b=x_I\in [m]^I$ and $c=x_{[n]\setminus (I\cup I')}\in [m]^{[n]\setminus (I\cup I')}$. We let $
\Pi_k(x):=z,$ 
where 
\begin{equation*}
z_j=\left\lbrace
\begin{array}{ll}
b_{\eta^{-1}(j)}&\text{~if~}j\in I'\\
(M(a))_{\eta(j)}&\text{~if~}j\in I\\
c_j&\text{o.w.}
\end{array}
\right.
\end{equation*}

In other words, $\Pi_k(x)$ replaces $x_I$ with $x_{I'}$, replaces $x_{I'}$ with $M(x_I)$ and leaves coordinates outside of $I\cup I'$ untouched, so that 
$$
\Pi_k(z)=(b, M(a), c).
$$ 
See Fig.~\dualref{fig:pi-k}  for an illustration.
\end{definition}

A key property of the map $\Pi_k$ is Lemma~\ref{lm:pi-prop} (see Section~\ref{sec:toy-construction}). Intuitively, this lemma says that $\Pi$ maps entire subspace to subspaces, which is a key property that we need our glueing maps to satisfy. This is because, as described in Section~\ref{sec:tech-overview}, if we were to upper bound the size of the maximum matching constructed by algorithm on a single gadget (like~\cite{Kapralov13} does), we would need to consider a vertex cover that is defined by the terminal subcube $T_*$ and its downset. Our construction of a vertex cover in the concatenation of basic gadgets will use this approach, and we need (the downset of) the terminal subcube in one gadget to have `nice structure' when mapped to another gadget using the glueing map $\tau$. Our mapping $\Pi_k$ is useful for this purpose, because the terminal subcube of a subsequent gadget is a subcube defined by coordinates  in $[n]\setminus (I\cup I')$ (these coordinates are the set $\Lambda$ above), and Lemma~\ref{lm:pi-prop} shows that  this set is still a subcube after an application of $\Pi_k$.

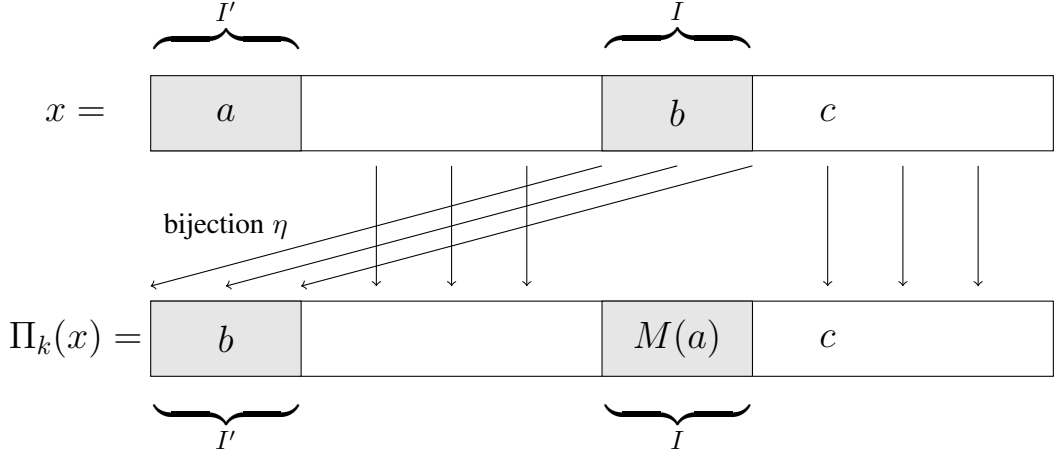
\begin{figure}
	\begin{center}
		\begin{tikzpicture}
		
		\draw (-6, 0) rectangle (+6, 1);		
		\draw[fill=gray!20] (-6, 0) rectangle (-4, 1);		
		\draw[fill=gray!20] (0, 0) rectangle (+2, 1);
		\draw (-7, 0.5) node {\Large $x=$};		
		\draw (-5, 0.5) node {\Large $a$};
		\draw (1, 0.5) node {\Large $b$};
		\draw (3, 0.5) node {\Large $c$};
		
		\draw (-5, 1.5) node {\Large $\overbrace{\phantom{a+b+c}}^{I'}$};
		\draw (+1, 1.5) node {\Large $\overbrace{\phantom{a+b+c}}^{I}$};		

		\draw (-6, 0-3) rectangle (+6, 1-3);
		\draw[fill=gray!20] (-6, 0-3) rectangle (-4, 1-3);
		\draw[fill=gray!20] (0, 0-3) rectangle (+2, 1-3);
		\draw (-7, 0.5-3) node {\Large $\Pi_k(x)=$};		
		\draw (-5, 0.5-3) node {\Large $b$};
		\draw (1, 0.5-3) node {\Large $M(a)$};
		\draw (3, 0.5-3) node {\Large $c$};
		
		\draw (-5, -3-0.5) node {\Large $\underbrace{\phantom{a+b+c}}_{I'}$};
		\draw (+1, -3-0.5) node {\Large $\underbrace{\phantom{a+b+c}}_{I}$};		
		
		\draw[->] (-3, -0.2) -- (-3, -1.8);
		\draw[->] (-2, -0.2) -- (-2, -1.8);
		\draw[->] (-1, -0.2) -- (-1, -1.8);		
		\draw[->] (+3, -0.2) -- (+3, -1.8);		
		\draw[->] (+4, -0.2) -- (+4, -1.8);	
		\draw[->] (+5, -0.2) -- (+5, -1.8);	

		\draw[->] (+0, -0.2) -- (+0-6, -1.8);	
		\draw[->] (+1, -0.2) -- (+1-6, -1.8);		
		\draw[->] (+2, -0.2) -- (+2-6, -1.8);	
		
		\draw (-5, -1) node {bijection $\eta$};

		\end{tikzpicture}
		\caption{Illustration of the map $\Pi_k$. Note that $\Pi_k$ simply leaves coordinates in $[n]\setminus (I\cup I')$, denoted by $c$, unchanged, copies coordinates in $I$ to coordinates in $I'$ using an arbitrarily chosen but fixed bijection $\eta$, and applies the map $M$ to coordinates in $I'$, assigning the result to coordinates in $I$.}	\duallabel{fig:pi-k}
	\end{center}
	
\end{figure}

We can now define

\begin{definition}[Glueing map $\tau$]\duallabel{def:tau}
For every $k\in [K/2]$ we define $\tau_k: S'_k\to T_*$ by letting $\tau_k(x):=\Pi_k(\rho_k(x))$.
Define 
$$
\tau: \biguplus_{k\in [K/2]} S'_k\to T_*
$$ by letting $\tau(x)=\tau_k(x)$ for $x\in S'_k$. 
\end{definition}

This completes the definition of the glueing map $\tau$. Note that so far we defined, for every $\ell\in [L]$, as basic gadget $G^\ell=(S^\ell, T^\ell, E^\ell)$ that is a geometric version of a $1/2$-KVV gadget used in~\cite{EpsteinLSW13}. The gadgets are parameterized by sequences $J^\ell\in \ac{\B}^\ell_0\times \dots \times \B^\ell_{K/2}$, where 
$$
[n]=\B^0\cup \B^1\cup \ldots\cup \B^{L-1}
$$
is a partition of $[n]$ into coordinate blocks, one for each round $\ell\in [L]$, which are in turn partitioned into disjoint subblocks $\B^\ell=\B^\ell_0\cup \ldots \cup \B^\ell_{K/2}$ corresponding to phases in which a given gadget is revealed to the algorithm. The sets $\ac{\B}^\ell_k$ are subsets of $\B^\ell_k$ of comparable size, equal to $\B^\ell_k$ minus the extension and compression indices for the corresponding phase.  We also defined bijections
$$
\tau^\ell: S^\ell \to T_*^{\ell-1}
$$
that we refer to as glueing maps. We now define our hard input distribution $\mathcal{D}$ on graphs $\wh{G}=(P, Q, \wh{E})$. A graph $\wh{G} \sim \mathcal{D}$ is sampled as follows.

First, for every round $\ell\in [L]$ and phase $k\in [K/2]$ one {\bf arbitrarily} selects the extension indices and compression indices appropriately -- we do not dwell on this here and refer the reader to Section~\ref{sec:simple-lb} for details. More importantly, one selects, for every $\ell\in [L]$ and $k\in [K/2]$,
$J^\ell_k\sim UNIF(\ac{\B}^\ell_k).$ Note that the vectors $J^\ell$ are the only random variables in the construction.

\paragraph{Edge set of $\wh{G}=(P, Q, \wh{E})$.} We first define
\begin{equation}\duallabel{eq:tau-star-def}
\tau_*(x)=\left\lbrace
\begin{array}{ll}
\tau^\ell(x)&\text{~if~}x\in S^\ell\text{~for~}\ell>0\\
x&\text{o.w.}.
\end{array}
\right.
\end{equation}
and  for every edge $e=(u, v)\in E^\ell, u\in S^\ell, v\in T^\ell,\ell\in [L]$ define $\tau_*(e)=(\tau_*(u), v)$. We now let 
\begin{equation}\duallabel{eq:g-hat-edges}
\wh{E}=\bigcup_{\ell\in [L]} \wh{E}^\ell,
\end{equation}
where 
\begin{equation}\duallabel{eq:g-hat-edges-ell}
\wh{E}^\ell=\bigcup_{k\in [K/2]}\bigcup_{j\in \ac{\B}^\ell_k} \tau_*(E^\ell_{k, j}),
\end{equation}
and $E^\ell_{k, j}$ is as in~\dualeqref{eq:edges-ei-def}.

\paragraph{Ordering of edges of $\wh{G}$ in the stream.} The graph $G^\ell=(S^\ell, T^\ell, E^\ell)$ is presented in the stream over $L$ {\em rounds} and $K/2$ {\em phases} as follows. For every $\ell\in \{1, \ldots, L-1\}$, for every $k\in [K/2]$, the edges in $\tau^\ell(E^\ell_k)$ are presented in the stream; the ordering within $\tau^\ell(E^\ell_k)$ is arbitrary.

The graph $\wh{G}$ contains a nearly perfect matching (intuitively, this is because in our gadgets the set $S^\ell$ can always be nearly perfectly matched to $T^\ell\setminus T_*^\ell$):
\begin{lemma}\duallabel{lm:large-matching-ghat}
The graph $\wh{G}=(P, Q, \wh{E})$ contains a matching of size $(1-O(1/L))|P|$.
\end{lemma}

The central part of our analysis is consists of designing a convenient vertex cover that lets us upper bound the size of matching constructed by a low space algorithm.  The key concept underlying our analysis here is a map $\nu$ that we refer to as the {\em predecessor map}. The intuition for this map is a combination of the analysis of iterated KVV constructions in~\cite{EpsteinLSW13,DBLP:conf/soda/0002PTTWZ19}, which essentially amount to a fixed point computation (the one in \cite{EpsteinLSW13} is not phrased this way, but it appears that for our purposes this view is more useful). 

\begin{definition}[Predecessor map $\nu$]\duallabel{def:nu}
We define the map $\nu_{\ell, j}$ mapping {\em subsets $U\subseteq T^\ell$} to subsets of $T^{\ell-j}$ by induction on $j\geq 0$ as follows. For $j=0$ let
$\nu_{\ell, 0}(U):=U$. For $j>0$ let
$$
\nu_{\ell, j}(U):=\tau^{\ell-(j-1)}(\downset^{\ell-(j-1)}(\nu_{\ell, j-1}(U))).
$$
We define the {\em closure} map $\nu_{\ell, *}$ by
$$
\nu_{\ell, *}(U):=\bigcup_{\substack{j=0\\j \text{~even}}}^\ell \nu_{\ell, j}(U).
$$

We define the map $\mu_{\ell, j}$ mapping {\em subsets $U\subseteq T^\ell$} to subsets of $S^{\ell-j}$ by letting
$$
\mu_{\ell, j}(U):=\downset^{\ell-j}(\nu_{\ell, j}(U))
$$
for $j=0,\ldots, \ell$. We let
$$
\mu_{\ell, *}(U):=\bigcup_{\substack{j=0\\j \text{~even}}}^\ell \mu_{\ell, j}(U).
$$
\end{definition}
In the definition above we write $\downset^\ell$ to denote the downset map of the $\ell$-th basic gadget, and $\tau^\ell$ the glueing map of the $\ell$-th basic gadget.

We note that $\nu_{\ell, j}(U)$ is the set of vertices that the set $U$ can be traced back to through $j$ applications of the glueing maps $\tau$, interleaved with applications of the $\downset$ map (which is the reason we refer to $\nu$ as the predecessor map). Intuitively, this map is useful for our purposes because it allows us to find a vertex cover similar  to what we obtained in Lemma~\dualref{lm:vertex-cover} above, but at the same time consistent with the fixed point type argument implicit in~\cite{EpsteinLSW13} and explicit in~\cite{DBLP:conf/soda/0002PTTWZ19}. First let
\begin{equation}\duallabel{eq:apm-def}
\begin{split}
A_P&=\bigcup_{\substack{\ell\in [L]\\ \ell \text{~even}}} \nu_{\ell, *}(T^\ell\setminus T_*^\ell)\\
A_Q&=\bigcup_{\substack{\ell\in [L]\\ \ell\text{~odd}}} \nu_{\ell, *}(T^\ell\setminus T_*^\ell).
\end{split}
\end{equation}
and
\begin{equation}\duallabel{eq:bpm-def}
\begin{split}
B_Q&=\bigcup_{\substack{\ell\in [L]\\ \ell \text{~even}}} \tau_*(\mu_{\ell, *}(T^\ell\setminus T_*^\ell))\\
B_P&=\bigcup_{\substack{\ell\in [L]\\ \ell\text{~odd}}} \tau_*(\mu_{\ell, *}(T^\ell\setminus T_*^\ell)).
\end{split}
\end{equation}
We prove that $A_P\cap B_P=\emptyset$ and $A_P\cup B_P\approx P$ (similarly for $A_Q$ and $B_Q$), as well as prove the following upper bounds on the cardinality of $B_P$ and $B_Q$ (which translates to the size of our vertex cover; recall that $N$ is the number of vertices in $T^\ell$):
\begin{lemma}[See Lemma~\ref{lm:sizeof-ab} in Section~\ref{sec:toy-construction}]\duallabel{lm:sizeof-ab}
One has $|B_P|\leq (1+O(1/L))\cdot \frac{L}{2}\cdot \frac{N}{2}\cdot\frac1{1+\ln 2}$ and 
$|B_Q|\leq (1+O(1/L))\cdot \frac{L}{2}\cdot \frac{N}{2}\cdot\frac1{1+\ln 2}$.
\end{lemma}

\begin{lemma}[See Lemma~\ref{lm:vertex-cover} in Section~\ref{sec:toy-construction}]\duallabel{lm:vertex-cover}
For every matching $M$ in $G$ one has 
$$
|M\cap (A_P\times (Q\setminus B_Q))|+\frac1{1+\ln 2}|P|+O(|P|/L).
$$
\end{lemma}
\begin{proofsketch}
Similarly to our analysis above with a single $\alpha$-KVV gadget, we exhibit a vertex cover for $M$. Specifically, we add to the vertex cover one endpoint of every edge in
$$
M\cap (A_P\times (Q\setminus B_Q)),
$$
as well as all vertices in $P\setminus A_P\approx B_P$ and $B_Q$ to the vertex cover. Note that this is indeed a vertex cover: $A_P\cap B_P=\emptyset$ and $A_Q\cap B_Q=\emptyset$, so every edge of $M$ either has an endpoint in $P\setminus A_P$, or belongs to $A_P\times (Q\setminus B_Q)$, or belongs to $A_P\times B_Q$, in which case it has an endpoint in $B_Q$.
The size of the vertex cover is 
\begin{equation}\duallabel{eq:23t77Gy8uGBBYVF}
\begin{split}
&|M\cap (A_P\times (Q\setminus B_Q))|+|P\setminus A_P|+|B_Q|\\
\approx &|M\cap (A_P\times (Q\setminus B_Q))|+|B_P|+|B_Q|,
\end{split}
\end{equation}
where we used the fact that $P\setminus A_P\approx B_P$ (see~Lemma~\dualref{lm:vertex-cover} for the precise version of this statement). By Lemma~\dualref{lm:sizeof-ab} we have 
\begin{equation*}
\begin{split}
|B_P|&\leq \frac{L}{2}\cdot \frac{N}{2}\frac1{1+\ln 2}(1+O(1/L))\\
&\text{and}\\
|B_Q|&\leq \frac{L}{2}\cdot \frac{N}{2}\frac1{1+\ln 2}(1+O(1/L)).
\end{split}
\end{equation*}
Putting the above together with~\dualeqref{eq:23t77Gy8uGBBYVF} and recalling that
$$
|P|=\left|\bigcup_{\substack{\text{even~}\ell\in [L]}} T^\ell\right|=L\cdot N/2
$$
gives the result.
\end{proofsketch}

The equivalent of our key lemma(Lemma~\dualref{lm:key} in the simple analysis above) is given by
\begin{lemma}[See Lemma~\ref{lm:special-edges} in Section~\ref{sec:toy-construction}]\duallabel{lm:special-edges}
For every matching $M\subseteq \wh{E}$ one has 
$$
M\cap (A_P\times (Q\setminus B_Q))\subseteq \bigcup_{\ell\in [L], k\in [K/2]} \tau^\ell(E^\ell_{k, J^\ell_k}).
$$
\end{lemma}

The proof relies on the fact that our glueing maps $\tau^\ell$ only use a small number of coordinates, and map entire (sampled) subspaces in $S'_k$ to subspaces in $T_*$.

\paragraph{Comparison of our vertex cover with that from Lemma~\dualref{lm:vertex-cover-simple}.} We note that, naturally, there are similarities between the vertex cover that we use to obtain the $\frac1{1+\ln 2}$ hardness and the vertex cover from Lemma~\dualref{lm:vertex-cover-simple}. Indeed, as per the proof (sketch) of Lemma~\dualref{lm:vertex-cover-simple} the vertex cover contains one endpoint of every edge in 
\begin{equation}\duallabel{eq:239yt9t}
M\cap (A_P\times (Q\setminus B_Q)),
\end{equation}
as well as all vertices in $P\setminus A_P\approx B_P$ and $B_Q$.  Similarly to~\dualeqref{eq:low-order} above, we show that using Lemma~\dualref{lm:special-edges} that the contribution of~\dualeqref{eq:239yt9t} can essentially be ignored. Thus, up to lower order terms, the vertex cover is the union of $B_P$ and $B_Q$. Recall that as per~\dualeqref{eq:bpm-def}
\begin{equation*}
\begin{split}
B_Q&=\bigcup_{\substack{\ell\in [L]\\ \ell \text{~even}}} \tau_*(\mu_{\ell, *}(T^\ell\setminus T_*^\ell))\\
&\text{and}\\
B_P&=\bigcup_{\substack{\ell\in [L]\\ \ell\text{~odd}}} \tau_*(\mu_{\ell, *}(T^\ell\setminus T_*^\ell)).
\end{split}
\end{equation*}
The application of $\tau_*$ above can be ignored for intuition, and we consider terms of the form 
$$
\mu_{\ell, *}(T^\ell\setminus T_*^\ell)=\bigcup_{\substack{j=0\\j \text{~even}}}^\ell \mu_{\ell, j}(T^\ell\setminus T_*^\ell)
$$
in the definition of $B_Q$ above, where as per Definition~\dualref{def:nu} one has
$$
\mu_{\ell, j}(T^\ell\setminus T_*^\ell):=\downset^{\ell-j}(\nu_{\ell, j}(T^\ell\setminus T_*^\ell))
$$
for $j=0,\ldots, \ell$.  Note that the $j=0$ term above in particular gives 
\begin{equation*}
\begin{split}
\mu_{\ell, 0}(T^\ell\setminus T_*^\ell)&=\downset^\ell(\nu_{\ell, 0}(T^\ell\setminus T_*^\ell))\\
&=\downset^\ell(T^\ell\setminus T_*^\ell)\\
&=\downset^\ell(T^\ell)\setminus \downset^\ell(T_*^\ell)\\
&=S^\ell\setminus \downset^\ell(T_*^\ell),\\
\end{split}
\end{equation*}
where we used the fact that $\nu_{\ell, 0}$ is the identity map and the fact that $S^\ell=\downset^\ell(T^\ell)$. Note that the last term in the equation above matches the second term in Lemma~\dualref{lm:vertex-cover-simple} (the first term is neglible, as we established). The other term in Lemma~\dualref{lm:vertex-cover-simple} is the terminal subcube itself, and is not present in our vertex cover since it is carefully split into different subsets, only some of which are added to the vertex cover -- see  Lemma~\ref{lm:t-star-ell-rec} for a formal statement supporting this intuition (only a subset of the terms on the rhs of that lemma contribute to the vertex cover that we defined above, due to parity constraints).

\paragraph{Overview of the main construction.} Our main construction, presented in Section~\ref{sec:main-result-full} onwards, basically follows the logic outlined above and made precise in Section~\ref{sec:toy-construction}. However, instead of using orgthogonal directions, we use nearly orthogonal vectors, which gives the stronger lower bound of $N^{1+\Omega(1/\log\log N)}$ even for the streaming model of computation (as opposed to just our stylized generalized online algorithms model from Section~\ref{sec:toy-construction}). We made an effort to make the exposition of the main construction follow quite closely the simple model we present in Section~\ref{sec:toy-construction}. Still, the setting is different and new technical ideas are needed, mostly revolving around the fact that in the real construction we lose product structure, which leads to multiple error terms that need to be handled carefully. At a high level, we resolve this issue by defining various relevant maps `locally'. Specifically, the definition of the {\em local permutation map} $\Pi$ is roughly equivalent to (a concatenation of) our maps $\Pi_k$ above, but works by first partitioning the space into appropriately defined low dimensional `subspaces' and defining the map on every such subspace (see Section~\ref{sec:loc-pi-full} for details). Intuitively, the reason for this is the fact that we only use nearly orthogonal vectors to define the edge set of the graph, and therefore all our maps need to be performing rather local operations, in order to avoid a degradation in the amount of orthogonality that we have.

\paragraph{Organization.} The rest of the paper is organized as follows. In Section~\ref{sec:toy-construction} we prove Theorem~\ref{thm:main-simple}. Then main construction is then presented in Sections~\ref{sec:main-result-full} onwards. We have invested effort into ensuring that the structure of the proof in Section~\ref{sec:toy-construction} follows quite closely the structure of the main proof. As a consequence, subsections of Section~\ref{sec:toy-construction} are in rather good correspondence with sections~\ref{sec:main-result-full} onwards of the main paper.

\section{Warm-up: a toy construction for the generalized online model}\label{sec:toy-construction}

In this section we provide a toy version of our lower bound instance that shows that no generalized online algorithm (see Definition~\ref{def:gen-online-intro}) with space $s=o(|P|\log |P|)$ can obtain a better (by an absolute constant) than $\frac1{1+\ln 2}$ approximation. Formally, we prove Theorem~\ref{thm:main-simple}, restated here for convenience of the reader:

\noindent{\em {\bf Theorem~\ref{thm:main-simple}}
There exists a distribution $\mathcal{D}$ on input graphs $G=(V, E)$ with $n$ vertices such that any generalized algorithm that finds a $(\frac1{1+\ln 2}+\eta)$-approximation to the maximum matching in $G$ with probability at least $0.9$ must remember $\Omega(n\log n)$ edges.}

\newcommand{\E}{{\mathcal E}}

We start by defining basic gadget graphs $G^\ell=(S^\ell, T^\ell, E^\ell)$ for $\ell\in [L]$ for an even integer $L$. Then input graph $G=(P, Q, E)$ is then an edge disjoint (but not vertex disjoint) union of graphs $G^\ell$. Specifically, the $P$ side of the bipartition will be 
\begin{equation}\label{eq:p-def}
P=\bigcup_{\text{even~}\ell\in [L]} T^\ell
\end{equation}
and the $Q$ side of the bipartition will be
\begin{equation}\label{eq:q-def}
Q=S^0\cup\bigcup_{\text{odd~}\ell\in [L]} T^\ell.
\end{equation}
Note that among the sets $S^\ell$ only the set $S^0$ belongs to the vertex set of $G$.This is because we obtain $G$ by glueing together instances of $G^\ell, \ell\in [L]$, using carefully designed maps $\tau^\ell$. For every $\ell=1,\ldots, L/2-1$ the map $\tau^\ell$ maps $S^\ell$ bijectively to a special subset $T^{\ell-1}_*$ of $T^{\ell-1}$ that we refer to as the {\em terminal subcube}.  Thus we have
$$
\tau^\ell: S^\ell\to T^{\ell-1}_*.
$$

\paragraph{Organization.} In what follows we first set up basic notation in Section~\ref{sec:notation}, then specify global parameter setting in Section~\ref{sec:params}. We then define our basic gadgets $G^\ell$ in Section~\ref{sec:basic-gadgets}. We then define auxiliary transformations, namely sparsification and densification operations, in Section~\ref{sec:densification}. We then define the glueing maps $\tau^\ell$ in Section~\ref{sec:tau-ell}. Another key object in our analysis, the predecessor map $\nu$, is defined in Section~\ref{sec:nu}  -- this map is key to defining a good upper bound for the matching $M_{ALG}$ constructed by a small space algorithm. Finally we put the pieces together and give a proof of Theorem~\ref{thm:main-simple} in Section~\ref{sec:simple-lb}.

\subsection{Notation and preliminaries}\label{sec:notation}

We start by setting up notation for the construction of basic gadgets $G^\ell=(S^\ell, T^\ell, E^\ell)$.  For every $\ell \in [L]$ we have $|T^\ell|=N=m^n$, and have $|S^\ell|=N/2$. For every $T^\ell$ we select a subset (referred to as the terminal subcube of $T^\ell)$, denoted by $T^\ell_*$, and for $\ell>0$ carefully map vertices of $S^\ell$ bijectively to $T^{\ell-1}_*$.  For every $\ell$ every vertex in $T^\ell$ and $S^\ell$ is equipped with a label from $[m]^n$ that we denote by 
$$
\lbl: P\cup Q\to [m]^n.
$$
For a pair of vertices $x\in P$ and $y\in Q$ we write $x\delequal y$ if $\lbl(x)=\lbl(y)$. For every $\ell\in [L]$ vertices in $T^\ell$ have distinct labels. The set $S^\ell$ will be partitioned into disjoint sets $S^\ell=S^\ell_0\cup\ldots\cup S^\ell_{K/2-1}$, and for every $k\in [K/2]$ vertices in $S^\ell_k$ also have distinct labels (their labels are a subset of the labels of $T^\ell$).  Thus, we will often think of vertices in $G^\ell$ as points in the hypercube when we think of vertices in $T^\ell$, or vertices in $S^\ell_k$ and $k$ is fixed. Throughout the paper we use the notation $[a]=\{0, 1,\ldots, a-1\}$ for a positive integer $a$. We partition $[n]$ into disjoint subsets 
$$
[n]=\B^0\cup \B^1\cup \ldots \cup \B^{L-1}
$$ of equal size, i.e. $|\B^\ell|=n/L$ for every $\ell \in [L]$. For every $\ell$ we further partition $\B^\ell$ as 
$$
\B^\ell=\B^\ell_0\cup\ldots\cup \B^\ell_{K/2},
$$ where $|\B^\ell_k|=\frac{n}{L(K/2+1)}$, corresponding to $K/2$ {\em phases} in which the graph $G^\ell$ will be presented in the stream.

\paragraph{Special indices.} The $\ell$-th graph $G^\ell$ is parameterized by a vector 
$$
J^\ell\in \B^\ell_0\times \ldots\times \B^\ell_{K/2}
$$ of indices.  For $\ell\in [L]$ and $k\in [K]$ we use the notation 
$J^\ell_{<k}:=(J^\ell_0,\ldots, J^\ell_{k-1})$ and $J^\ell_{\geq k}:=(J^\ell_k,\ldots, J^\ell_{K/2}).$

\begin{definition}[Compression and extension indices]\label{def:cext} For every $\ell\in [L], \ell>0,$ the map $\tau^\ell$ is parameterized by index $r^\ell\in \B^\ell_{K/2}$, referred to as the compression index for the terminal subcube $T^\ell_*$, as well as a collection of auxiliary coordinates for every $k\in [K/2]$:
\begin{itemize}
\item $\Ext^\ell_k\subseteq \B^\ell_k$ is a subset of size $K/2+1-k$ referred to as the {\em extension indices} in phase $k$ of round $\ell$;
\item $q^\ell_k\in \B^\ell_k\setminus \Ext^\ell_k$ referred to as the {\em compression index} for the $k$-th phase of round $\ell$.
\end{itemize}
We let $\ac{\B}^\ell_k=\B^\ell_k\setminus (\Ext^\ell_k\cup \{q^\ell_k\})$ and let $\ac{\B}^\ell_{K/2}=\B^\ell_{K/2}\setminus \{r^\ell\}$.
\end{definition}

\begin{property}\label{prop:q-k}
We will ensure that for every $\ell$ and $k\in [K/2+1]$ one has $J^\ell_k\in \ac{\B}^\ell_k$, and in particular $J^\ell\cap \left(\{r^\ell\}\cup \bigcup_{k\in [K/2]} \Ext^\ell_k\cup \{q^\ell_k\}\right)=\emptyset$.
\end{property}

For convenience of notation we introduce 
\begin{definition}[Special coordinates]\label{def:special-coords}
For every $\ell$ we define the {\em special coordinates} in $\B^\ell$ by $\Sp(\B^\ell):=J^\ell\cup \{r^\ell\}.$
We also let $\Sp(\B^{\geq \ell}):=\bigcup_{j\geq \ell} \Sp(\B^\ell)$. 

\end{definition}

\begin{definition}[Weight of $x\in {[m]^n}$]\label{def:weight}
For every $x\in [m]^n$ we define $\weight(x)=\sum_{j\in [n]} x_j$.
\end{definition}

We will use 
\begin{claim}\label{cl:sum-int}
There exists an absolute constant $C>0$ such that for every integer $K>0$ greater than an absolute constant one has 
$\ln 2-1/K\leq \sum_{k\in [K/2]} \frac1{K-k}\leq \ln 2.$
\end{claim}
\begin{proof}
One has for every integer $k\geq 0$, $\frac1{K}\cdot \frac1{1-(k+1)/K} \leq \int_{k/K}^{(k+1)/K} \frac1{1-x} dx\leq \frac1{K}\cdot \frac1{1-k/K},$
and hence 
$\sum_{k\in [K/2]} \frac1{K-k}=\frac1{K}\sum_{k\in [K/2]} \frac1{1-k/K}\leq \sum_{k\in [K/2]} \int_{k/K}^{(k+1)/K} \frac1{1-x} dx=\int_0^{1/2} \frac1{1-x}dx=\ln 2,$
establishing the upper bound. Similarly, 
\begin{equation*}
\begin{split}
\sum_{k\in [K/2]} \frac1{K-k}&=\frac1{K}-\frac1{K/2}+\frac1{K}\sum_{k=1}^{K/2} \frac1{1-k/K}\\
&\geq -\frac1{K}+\sum_{k\in [K/2]} \int_{k/K}^{(k+1)/K} \frac1{1-x} dx\\
&=-\frac1{K}+\int_0^{1/2} \frac1{1-x}dx=\ln 2-\frac1{K},
\end{split}
\end{equation*}
\end{proof}

\subsection{Parameter setting}\label{sec:params}
We assume throughout this section that parameters $m$, $W$, $K$ and $L$ satisfy the following properties:
\begin{description}
\item[(p0)\label{p0}] $(K-s) \mid W$ for all $s\in [K/2]$
\item[(p1)\label{p1}] $W \mid m/(K-s)$ for all $s\in [K/2]$
\item[(p2)\label{p2}] $L=K$
\end{description}
In the above we write $a \mid b$ if $b/a$ is an integer.

Such a setting is possible:
\begin{lemma}
For every constant $K$ there exists a setting of parameters $W,  L$ and $m$ that satisfies~\ref{p0}-\ref{p2}.
\end{lemma}
\begin{proof}
Let $m=(\text{lcm}(K, K-1, \ldots, 3, 2, 1))^2$ and $W=\text{lcm}(K, K-1, \ldots, 3, 2, 1)$, where $\text{lcm}$ stands for the least common multiple.
\end{proof}

In what follows we define the individual instances $G^\ell$ and state their main properties in Section~\ref{sec:basic-gadgets}, then define the maps $\tau^\ell$ in Section~\ref{sec:tau-ell}. We then give the proof of the lower bound in Section~\ref{sec:simple-lb}. 

\subsection{Basic gadgets $G^\ell$}\label{sec:basic-gadgets}
We give the construction of $G^\ell=(S^\ell, T^\ell, E^\ell)$ in this section. Since $\ell\in [L]$ is fixed, we write $T=T^\ell, S=S^\ell, E=E^\ell$ to simplify notation. We let $\B=\B^\ell$ and $\B=\B_0\cup \ldots\cup \B_{K/2}$ denote the partition of $\B$.

\paragraph{Vertices of $G$: the $T$ side of the bipartition.} Let $K\geq 1$ be a large constant integer, let $m\geq 1$ be a large integer. Let 
$$
T=[m]^n
$$ 
i.e. vertices in $T$ are vectors of dimension $n$, with each co-ordinate taking values in $[m]=\{0, 1, 2,\ldots, m-1\}.$ This way we have $N:=|T|=m^n$, so $n=\Omega(\log N)$ for every constant $m$. The vertices on the $S$ side of the bipartition will also be associated with points on the hypercube $[m]^n$, as defined below. 

\renewcommand{\I}{{\mathcal I}}

Let $T_0=T$, and for every $k\in [K/2]$ let 
\begin{equation}\label{eq:def-tk}
\begin{split}
T_{k+1}=\left\{y\in T_k: y_{J_k}/m\in \left[0, 1-\frac{1}{K-k}\right)\right\},\\
\end{split}
\end{equation}
so that
\begin{equation}\label{eq:def-tk-allconstraints}
\begin{split}
T_k=\left\{y\in [m]^n: y_{J_s}/m\in \left[0, 1-\frac{1}{K-s}\right)\text{~~for all~}s\in \{0, 1, \ldots, k-1\}\right\}.\\
\end{split}
\end{equation}

\paragraph{Vertices of $G$: the $S$ side of the bipartition.} The set $S$ of vertices is naturally partitioned into disjoint subsets 
\begin{equation}\label{eq:def-s}
S=S_0\uplus S_1\uplus \ldots \uplus S_{K/2-1}
\end{equation}
as follows. For every $k\in [K/2]$ we let
\begin{equation}\label{eq:def-sk}
\begin{split}
S_k\delequal \{x\in T_k: &\weight(x) \in \left[0, \frac{1}{K-k}\right)\cdot W \pmod{ W}\}
\end{split}
\end{equation}
In the definition above $W$ is an integer parameter that we choose so that $W \mid m$ as per~\ref{p1}, and $\weight(x)=\sum_{j\in [n]} x_j$ as per Definition~\ref{def:weight}.

\begin{definition}[Down-set of a set in $T$]\label{def:downset}
For every $U\subseteq T$, $k\in [K/2]$, we define the {\em downset of $U$ in $S_k$} by 
\begin{equation*}
\begin{split}
\downset_k(U)=\{x\in S_k: \exists y\in U: y\delequal x\}
\end{split}
\end{equation*}
and  define
$$
\downset(U)=\bigcup_{k\in [K/2]} \downset_k(U).
$$
\end{definition}
We note that in the definition above the union on the rhs is a union of disjoint sets.
\begin{remark}
We note that a given point in $U$ has anywhere between $0$ and $K/2$ images under the $\downset$ map.
\end{remark}
\begin{remark}
Note that $S_k=\downset_k(T_k)$ for every $k\in [K/2]$.
\end{remark}
\begin{remark}\label{rm:truncated-downset}
Note that if $U\subset T_k\setminus T_{k+1}$ for some $k\in [K/2]$, then $\downset_s(U)=\emptyset$ for all $s\in \{k+1, \ldots, K/2-1\}$. Thus, in that case we have
$$
\downset(U)=\bigcup_{s=0}^{k} \downset_s(U).
$$
\end{remark}

We also let, for every $k\in [K/2]$ and $j\in \B_k$ 
\begin{equation}\label{eq:def-tkj}
\begin{split}
T_k^j&=\left\{y\in T_k: y_j/m\in \left[0, 1-\frac{1}{K-k}\right)\right\}\\
S_k^j&=\left\{x\in S_k: x_j/m\in \left[0, 1-\frac{1}{K-k}\right)\right\}.\\
\end{split}
\end{equation}

\begin{definition}[Terminal subcube]\label{def:terminal-subcube}
We refer to $T_*:=T_{K/2}$ as the {\em terminal subcube} of $T$.
\end{definition}

We gather basic bounds on the size of $T_k$'s and $S_k$'s in 
\begin{lemma}\label{lm:size-bounds}
One has
\begin{itemize}
\item[{\bf (1)}] $|T_k|=(1-\frac{k}{K})\cdot |T_0|$ for every $k\in [K/2+1]$;
\item[{\bf (2)}] $|S_k|=\frac1{K}\cdot |T_0|$ for every $k\in [K/2]$.
\end{itemize}
\end{lemma}
The proof is given in Appendix~\ref{app:size-bounds}. We now define the edge set of $G$.

\paragraph{Edges of $G$.}   Fix $k\in [K/2]$. For each coordinate $j\in \B_k$ for each $x\in [m]^n$ we denote the line in direction $j$ going through $x$ by 
\begin{equation}\label{eq:line-def}
\text{line}_j(x)=\{x'\in [m]^n: x'_{-j}=x_{-j}\},
\end{equation}
where we write $x_{-j}$ to denote the restriction of $x$ on coordinates $[n]\setminus \{j\}$.
We have 
\begin{lemma}\label{lm:line-properties}
For all $s\in [K/2]$,  for every $k\in [K/2]$, every $J_{<k} \in \B_{<k}$ for each $y\in T_k$ one has for each $j\in \B_k$
\begin{itemize}
\item[{\bf (1)}] $|\text{line}_j(y)|=m$ and $\text{line}_j(y)\subseteq T_k$;
\item[{\bf (2)}] $|\text{line}_j(y)\setminus T_k^j|=\frac1{K-k}\cdot|\text{line}_j(y)|$;
\item[{\bf (3)}] for every $y\in T_k$ one has $|\text{line}_j(y)\cap S_k|=\frac1{K-k}\cdot |\text{line}_j(y)|$;
\item[{\bf (4)}] for every $y\in T_k$ one has $|\text{line}_j(y)\cap S_k^j|=\frac1{K-k}\cdot |\text{line}_j(y)|\cdot (1-1/(K-k))$.
\end{itemize}
\end{lemma}
The proof of the lemma is given in Appendix~\ref{app:line-properties}.

\begin{remark}\label{lm:lines-intersection-structure}
Note that for every $y, y'$ and every $j$ one has either $\text{line}_j(y)=\text{line}_j(y')$ or $\text{line}_j(y)\cap \text{line}_j(y')=\emptyset$, i.e. lines in direction $j$ partition $T_k$, and consequently also partition $S_k$. 
\end{remark}

We now define the edges of $G$ incident on $S_k$ for every $k\in [K/2]$. 

\begin{definition}[Line cover in direction $j$]\label{def:line-cover}
For every $j\in \B_k$ a collection $C_k^j\subseteq T_k$ of representative points is called a {\em line cover of $T_k$ in direction $j$} if 
$$
T_k=\bigcup_{y\in C_k^j} \text{line}_j(y)
$$
and $\text{line}_j(y)\cap \text{line}_j(y')=\emptyset$ for every $y, y'\in C_k^j$, $y\neq y'$.
\end{definition}
Note that a line cover of $T_k$ in direction $j$ can be constructed by picking points $y\in T_k$ greedily until the union of lines in direction $j$ through these points covers $T_k$. Every such line belongs to $T_k$ by Lemma~\ref{lm:line-properties}, {\bf (1)}, and every two lines either are disjoint or coincide as per Remark~\ref{lm:lines-intersection-structure}.

Now for every $j$ {\em except} the extension indices $\Ext_k$ or the compression index $q_k$ (see Definition~\ref{def:cext}), i.e. for all 
$$
j\in \ac{\B}_k=\B_k\setminus \left(\Ext_k\cup \{q_k\}\right),
$$
for every $y\in C_k^j$ for a line cover $C_k^j$ of $T_k$ in direction $j$ (as per Definition~\ref{def:line-cover}), we include a complete bipartite graph between $\text{line}_j(y)\cap (T_k\setminus T_k^j)$ and $\text{line}_j(y)\cap S_k^j$. In other words, let $E=\bigcup_{k\in [K/2]} E_k$, where
\begin{equation}\label{eq:edges-ei-k-def}
E_k=\bigcup_{j\in \ac{\B}_k} E_{k, j}
\end{equation}
and
\begin{equation}\label{eq:edges-ei-def}
E_{k, j}=\bigcup_{y\in C_k^j} (\text{line}_j(y)\cap S_k^j) \times (\text{line}_j(y) \cap (T_k\setminus T_k^j)).
\end{equation}
Note that $E_k$ is fully determined by the first $k-1$ values of $J$, namely by the prefix $J_{<k}$.

We have 
\begin{lemma}\label{lm:disjoint}
For every $k \in [K/2]$, every $i, j\in \B_k, i\neq j$, every $x\in C_k^i, y\in C_k^j$, where $C_k^i$ and $C_k^j$ are minimal line covers of $T_k$ in direction $i$ and $j$ respectively, the edge sets
$$
(\text{line}_i(x)\cap S_k^i) \times (\text{line}_i(x) \cap (T_k\setminus T_k^i))
$$
and 
$$
(\text{line}_j(y)\cap S_k^j) \times (\text{line}_j(y) \cap (T_k\setminus T_k^j))
$$
are disjoint.
\end{lemma}
\begin{proof}
We argue by contradiction. Note that the complete graphs above have a nonempty intersection if and only if there exist $a, b$ such that
\begin{equation}\label{eq:0293yt88gfdSF}
a\in (\text{line}_i(x)\cap S_k^i)\cap (\text{line}_j(y)\cap S_k^j)
\end{equation}
and 
\begin{equation}\label{eq:92yt99gfef}
b\in (\text{line}_i(x) \cap (T_k\setminus T_k^i))\cap (\text{line}_j(y) \cap (T_k\setminus T_k^j)).
\end{equation}
Since $a, b\in \text{line}_j(y)$, we have
\begin{equation}\label{eq:abi-290gh9hg}
b_{-j}=a_{-j}
\end{equation}

On the other hand, since $a\in \text{line}_i(x)\cap S_k^i \subseteq S_k^i$, we have by~\eqref{eq:def-tkj} that $a_i/m\in \left[0, 1-\frac{1}{K-k}\right)$, and since $b\in  \text{line}_i(x) \cap (T_k\setminus T_k^i)\subseteq T_k\setminus T_k^i$, we have by~\eqref{eq:def-tk} that $b_i/m\in \left[1-\frac{1}{K-k}, 1\right)$.  On the other hand, we have $a_i=b_i$ by~\eqref{eq:abi-290gh9hg}, a contradiction.
\end{proof}

\begin{lemma}[Matching of $S$ to $T\setminus T_*$]\label{lm:matching}
There exists a matching of a $(1-O(1/K))$ fraction of $S$ to $T\setminus T_*$ in $E$.
\end{lemma}
\begin{proof}
For every $k\in [K/2]$ and $j=J_k$ we match almost all of $S_k$ to $T_k\setminus T_k^j$ as follows. First note that for every $y, y'\in T_k$ one has either $\text{line}_j(y)=\text{line}_j(y')$ or $\text{line}_j(y)\cap \text{line}_j(y')=\emptyset$, i.e. lines in direction $j$ partition $T_k$, and consequently also partition $S_k$. Thus, it suffices to define the matching on all lines in direction $j=J_k$ for each $y\in T_k$. By Lemma~\ref{lm:line-properties}, {\bf (2)}, we have 
$$
|\text{line}_j(y)\setminus T_k^j|=\frac1{K-k}\cdot|\text{line}_j(y)|
$$
and by Lemma~\ref{lm:line-properties}, {\bf (3)}, one has 
$$
|\text{line}_j(y)\cap S_k^j|=\frac1{K-k}\cdot |\text{line}_j(y)|(1+1/(K-k)).
$$
We match $\text{line}_j(y)\cap S_k$ to $\text{line}_j(y)\setminus T_k^j$ using the edges 
$$
(\text{line}_j(y)\cap S_k^j) \times (\text{line}_j(y) \cap (T_k\setminus T_k^j)),
$$
which belong to $E_k^j$ as per~\eqref{eq:edges-ei-def}. This defines a matching of a $1-O(1/K)$ fraction of $S_k$ to $T_k^j$, where $j=J_k$.  

Equipped with a matching of a $1-O(1/K)$ fraction of $S_k$ to $T_k^j$, where $j=J_k$, we now note that
$T_{k+1}=T_k^j$ when $j=J_k$ for every $k\in [K/2]$ (see~\eqref{eq:def-tkj} and~\eqref{eq:def-tk}), and therefore
$$
T\setminus T_*=T\setminus T_{K/2}=\bigcup_{k\in [K/2]} T_k\setminus T_{k+1}=\bigcup_{k\in [K/2]} T_k\setminus T^{J_k}_k,
$$
where the union on the right hand side contains disjoint sets. Since $S_0, S_1,\ldots, S_{K/2-1}$ are disjoint, the union of constructed matchings is a matching of a $1-O(1/K)$ fraction of $S$ to $T\setminus T_*$, as required.
\end{proof}

\subsection{Subsampling and densification}\label{sec:densification}

\begin{definition}\label{def:u-dep}
For a subset $F\subseteq [m]^n$ and a coordinate $r\in [n]$ we say that {\em $F$ does not depend on $r$} if for every $x=(x_r, x_{-r})\in F$ one has $(y, x_{-r})\in F$ for every $y\in [m]$. Equivalently, $F$ does not depend on $r$ if 
$$
F=\{z\in [m]^{[n]\setminus \{r\}}: z=x_{-r} \text{~for some~}x\in F\}\times [m].
$$
\end{definition}

\begin{lemma}[Subsampling]\label{lm:u-subsampling}
For every $U\subseteq [m]^n$ and $r\in [n]$ such that $U$ does not depend on $r$ (as per Definition~\ref{def:u-dep}), every integer $\lambda$ such that $\lambda \mid W$, as long as $W \mid m$, one has
$$
\left|\left\{x\in U: \weight(x) \in \left[0, 1/\lambda\right)\cdot W \pmod{ W} \right\}\right|=\frac1{\lambda} |U|.
$$
\end{lemma}
\begin{proof}
Let 
$U_{-r}:=\{x_{-r}: x\in U\},$
where $x_{-r}\in [m]^{[n]\setminus \{r\}}$ stands for the projection of $x$ to $[n]\setminus \{r\}$, and note that $
U=U_{-r}\times [m]$. Furthermore, we have 
\begin{equation*}
\begin{split}
&\left\{x\in U: \weight(x) \in \left[0, 1/\lambda\right)\cdot W \pmod{ W} \right\}\\
&=\left\{(x', x'')\in U_{-r}\times  [m]: \weight(x')+x''\in \left[0, 1/\lambda\right)\cdot W \pmod{ W} \right\}\\
\end{split}
\end{equation*}

This in turn implies
\begin{equation}\label{eq:8g8g8gFYGSYFG}
\begin{split}
&\left|\left\{(x', x'')\in U_{-r}\times  [m]: \weight(x')+x''\in \left[0, 1/\lambda\right)\cdot W \pmod{ W} \right\}\right|\\
&=\sum_{x'\in U_{-r}} m\cdot \prob_{x''\sim UNIF([m])}[\weight(x')+x''\in \left[0, 1/\lambda\right)\cdot W \pmod{ W}]\\
\end{split}
\end{equation}
Since $W \mid m$ by assumption of the lemma, when $x''$ is uniformly random in $[m]$, $(\weight(x')+x'')\pmod{W}$ is uniformly random in $[W]$. Thus, for any $x'\in [m]^{[n]\setminus \{r\}}$ one has
$$
\prob_{x''\sim UNIF([m])}[\weight(x')+x''\in \left[0, 1/\lambda\right)\cdot W \pmod{ W}]=\frac1{\lambda}.
$$
Substituting this into~\eqref{eq:8g8g8gFYGSYFG}, we get
\begin{equation*}
\begin{split}
&\left|\left\{(x', x'')\in U_{-r}\times  [m]: \weight(x')+x''\in \left[0, 1/\lambda\right)\cdot W \pmod{ W} \right\}\right|\\
&=\sum_{x'\in U_{-r}} m\cdot \frac1{\lambda}\\
&=\frac1{\lambda} |U_{-r}|\cdot m\\
&=\frac1{\lambda} |U|, 
\end{split}
\end{equation*}
as required.
\end{proof}

\begin{definition}[$(\lambda, r)$-densifying map]\label{def:rho}
For $r\in [n]$ and integer $\lambda>0$ the $(\lambda, r)$-densifying map $\rho:[m]^n\to [m]^n$ is defined as follows.  We let $x\in [m]^n$, and write $x=(x', x''), x'\in [m]^{[n]\setminus \{r\}}, x''\in [m]$. Write 
$$
x''=aW+b W/\lambda+c,
$$ 
where $a\in \{0, 1,\ldots, m/W-1\}, b\in \{0, 1,\ldots,\lambda-1\}$ and $c\in \{0, 1,\ldots, W/\lambda-1\}$. We define $\rho(x)$ by letting, for $j\in [n]$:
\begin{equation*}
\begin{split}
(\rho(x))_j:=\left\lbrace
\begin{array}{ll}
a W/\lambda+c&\text{~if~}j=r\\
x_j&\text{~o.w.~}
\end{array}
\right.
\end{split}
\end{equation*}
\end{definition}
\begin{remark}
Note that equivalently, one lets, for $j\in [n]$,
\begin{equation*}
\begin{split}
(\rho(x))_j:=\left\lbrace
\begin{array}{ll}
(x'' \pmod{W/\lambda})+\frac{W}{\lambda}\cdot \lfloor x''/W)\rfloor&\text{~if~}j=r\\
x_j&\text{~o.w.~}
\end{array}
\right.
\end{split}
\end{equation*}

\end{remark}

We have 
\begin{lemma}[Densification of a subsampled set]\label{lm:densification}
For every integer $\lambda\geq 2$, every $r\in [n]$, every $U\subseteq [m]^n$ that does not depend on coordinate $r$ (as per Definition~\ref{def:u-dep}) the $(\lambda, r)$-densifying map $\rho$ (see Definition~\ref{def:rho}) maps 
$$
\left\{x\in U: \weight(x) \in \left[0, 1/\lambda\right)\cdot W \pmod{ W} \right\}
$$ bijectively to $\left\{x\in U: x_r/m\in \left[0, 1/\lambda\right)\right\}$.
\end{lemma}
\begin{proof}
We first prove that 
\begin{equation}\label{eq:932ghg8hUGFUGFshfmlwhladsdf}
\rho\left(\left\{x\in U: \weight(x) \in \left[0, 1/\lambda\right)\cdot W \pmod{ W} \right\}\right)\subseteq \left\{x\in U: x_r/m\in \left[0, 1/\lambda\right)\right\}.
\end{equation}
We let $x\in [m]^n$, and write $x=(x', x''), x'\in [m]^{[n]\setminus \{r\}}, x''\in [m]$. Write 
$$
x''=aW+b W/\lambda+c,
$$ 
where $a\in \{0, 1,\ldots, m/W-1\}, b\in \{0, 1,\ldots,\lambda-1\}$ and $c\in \{0, 1,\ldots, W/\lambda-1\}$. As per Definition~\ref{def:rho} one has for $j\in [n]$
\begin{equation*}
\begin{split}
(\rho(x))_j:=\left\lbrace
\begin{array}{ll}
(a W/\lambda+c) \pmod{m/\lambda}&\text{~if~}j=r\\
x''_j&\text{~o.w.~}
\end{array}
\right.
\end{split}
\end{equation*}
Since $U$ does not depend on $r$ by assumption and
$$
0\leq a (W/\lambda)+c\leq (m/W-1) (W/\lambda)+(W/\lambda-1)=m/\lambda-1,
$$
we have $\rho(x)\in  \left\{x\in U: x_r/m\in \left[0, 1/\lambda\right)\right\}$, as required. This establishes~\eqref{eq:932ghg8hUGFUGFshfmlwhladsdf}.

We now establish injectivity. Let  $U_{-r}:=\{x_{-r}: x\in U\},$ where $x_{-r}\in [m]^{[n]\setminus \{r\}}$ stands for the projection of $x$ to $[n]\setminus \{r\}$, and note that $
U=U_{-r}\times [m]$ since $U$ does not depend on $r$ by assumption. Furthermore, we have 
\begin{equation}
\begin{split}
&\left\{x\in U: \weight(x) \in \left[0, 1/\lambda\right)\cdot W \pmod{ W} \right\}\\
&=\left\{(x', x'')\in U_{-r}\times  [m]: \weight(x')+x''\in \left[0, 1/\lambda\right)\cdot W \pmod{ W} \right\}\\
\end{split}
\end{equation}

Now pick $x''_1, x''_2\in [m]$ such that 
\begin{equation}\label{eq:290g9hhihfsf}
\left\{(x', x''_i)\in U_{-r}\times  [m]: \weight(x')+x''_i\in \left[0, 1/\lambda\right)\cdot W \pmod{ W} \right\}\text{~for~}i\in \{1, 2\}.
\end{equation}
Write 
$$
x''_1=a_1W+b_1 W/\lambda+c_1\text{~~and~~}x''_2=a_2W+b_2 W/\lambda+c_2,
$$ 
where $a_i\in \{0, 1,\ldots, m/W-1\}, b_i\in \{0, 1,\ldots,\lambda-1\}$ and $c_i\in \{0, 1,\ldots, W/\lambda-1\}$, $i\in \{1, 2\}$.  Suppose towards a contradiction that $\rho((x', x''_1))=\rho((x', x''_2))$, i.e., that $a_1=a_2$ and $c_1=c_2$. We show that $b_1=b_2$. We  have
\begin{equation}
\begin{split}
((\weight(x')+x''_1)-(\weight(x')+x''_2)) \pmod {W}&=(x''_1-x''_2) \pmod {W}\\
&=((b_1-b_2) W/\lambda) \pmod {W}\\
\end{split}
\end{equation}
Since $|b_1-b_2|<\lambda$,  $b_1\neq b_2$ would contradict~\eqref{eq:290g9hhihfsf}. Thus, we have $b_1=b_2$, and the map $\rho$ is injective. 

Finally, note that by Lemma~\ref{lm:u-subsampling} one has 
$$
\left|\left\{x\in U: \weight(x) \in \left[0, 1/\lambda\right)\cdot W \pmod{ W} \right\}\right|=\frac1{\lambda}\cdot |U|=\left|\left\{x\in U: x_r/m\in \left[0, 1/\lambda\right)\right\}\right|,
$$
and hence $\rho$ is a bijection.
\end{proof}

\newcommand{\q}{\mathbf{q}}
\newcommand{\Int}{\text{Int}}

\subsection{Maps $\tau^\ell$ identifying the basic gadgets}\label{sec:tau-ell}
In this section we define the map $\tau^\ell$ identifying vertices in $S^\ell$ in $G^\ell=(S^\ell, T^\ell, E^\ell)$ with vertices in $T_*^{\ell-1}$ of $G^{\ell-1}=(S^{\ell-1}, T^{\ell-1}, E^{\ell-1})$ for every $\ell\in [L], \ell>0$:
$$
\tau^\ell: S^\ell\to T_*^{\ell-1}.
$$

Since $\ell$ is fixed for most of the section, we omit the superscript $\ell$. We let $G'=(S', T', E')$ denote $G^\ell=(S^\ell, T^\ell, E^\ell)$, let $G=(S, T, E)$ denote $G^{\ell-1}=(S^{\ell-1}, T^{\ell-1}, E^{\ell-1})$, adopting similar notation for all other relevant quantities.
Specifically, let $\B':=\B^\ell, \B:=\B^{\ell-1}$, and let the special coordinate vectors be denoted by $J\in \B_0\times \B_1\times \ldots\times \B_{K/2}$, and $J'=\B'_0\times \B'_1\times \ldots \times \B'_{K/2}$ respectively. Thus, we define a bijection $\tau$ from $S'$ to $T_*$:
$$
\tau: S'\to T_*.
$$

We start by defining $\tau$ on the sets $S'_k$ for $k\in [K/2]$ (recall that $S'=S_0'\uplus \ldots \uplus S_{K/2-1}'$). The restriction of $\tau$ to $S'_k$ is denoted by $\tau_k$:
$$
\tau_k: S'_k\to T_*.
$$ 
The images of $\tau_k$ that we define will be disjoint for different $k\in [K/2]$, i.e. these maps extend naturally to an injective map from the union of $S'_k$ over all $k\in [K/2]$ to $T_*$. 

\paragraph{Defining $\tau_k$.} Fix $k\in [K/2]$. Let $r\in \B$ denote the compression index of the terminal subcube $T_*$. 
Let $\Ext_k\subseteq \B'_k$ and $q_k\in \B'_k$ denote the $k$-th extension and compression indices (see Definition~\ref{def:cext}). Let $\rho_k$ be a $(K-k, q_k)$-densifying map as per Definition~\ref{def:rho}.  Now note that $T'_k$ does not depend on $q_k$ by Property~\ref{prop:q-k}.  We thus have by Lemma~\ref{lm:densification} that
\begin{equation}\label{eq:s-prime-k}
\rho_k(S'_k)=\left\{x\in T'_k: x_{q_k}/m\in \left[0, \frac1{K-k}\right) \right\}
 \end{equation}
 and $\rho_k$ maps $S'_k$ to the set on the rhs of~\eqref{eq:s-prime-k} bijectively.

Now define index sets 
\begin{equation}\label{eq:i-def}
I=J\cup \{r\}\subset [n]
\end{equation}
and 
\begin{equation}\label{eq:i-prime-def}
I'_k=\{J'_0,\ldots, J'_{k-1}\}\cup \Ext_k\cup \{q_k\}.
\end{equation} 
We sometimes write $I'$ instead of $I'_k$ when the value of $k$ is clear from context.
Note that $I=\Sp(\B)$ and  for every $k$

\begin{equation}\label{eq:i-i-prime-size}
|I|=|I'_k|=K/2+2.
\end{equation}
This allows us to write 
\begin{equation}\label{eq:t-star-product}
T_*\delequal A\times [m]^{[n]\setminus I},
\end{equation}
where
$$
A=\left\{x\in [m]^I: x_{J_s}/m\in \left[0, \frac1{K-s}\right) \text{~for all~}s\in [K/2]\right\}.
$$
Similarly, we write 
\begin{equation}\label{eq:rho-s-prime-k}
\rho_k(S'_k) \delequal D_k\times [m]^{[n]\setminus I'_k},
\end{equation}
where as per~\eqref{eq:s-prime-k}
\begin{equation*}
\begin{split}
D_k&=\left\{x\in [m]^{I'_k}: x_{i_s}/m\in \left[0, 1-\frac1{K-s}\right)\text{~for all~}s\in [k]\text{~and~}x_{q_k}/m\in \left[0, \frac1{K-k}\right) \right\}.
\end{split}
\end{equation*}

Choose a bijection 
\begin{equation}\label{eq:m-def}
M: \biguplus_{k\in [K/2]} D_k \to A.
\end{equation}
This is possible because 
\begin{equation*}
\begin{split}
\sum_{k\in [K/2]} |D_k|&=\sum_{k\in [K/2]} \frac{|\rho_k(S'_k)|}{m^{n-|I'_k|}}\text{~~~~~~~~~~~~~(by~\eqref{eq:rho-s-prime-k})}\\
&=\sum_{k\in [K/2]} \frac{|S'_k|}{m^{n-|I'_k|}}\text{~~~~~~~~~~~~~(since $\rho_k$ is a bijection)}\\
&=\frac1{K}\sum_{k\in [K/2]} \frac{|T_0|}{m^{n-|I'_k|}}\text{~~~~~~~(since $|S'_k|=\frac1{K}|T_0|$ for all $k\in [K/2]$ by Lemma~\ref{lm:size-bounds}, {\bf (2)})}\\
&=\frac1{K}\sum_{k\in [K/2]} \frac{|T_0|}{m^{n-|I|}}\text{~~~~~~~~(since $|I|=|I'_k|$ by ~\eqref{eq:i-i-prime-size})}\\
&=\frac1{2}\frac{|T_0|}{m^{n-|I|}}\\ 
&=\frac{|T_*|}{m^{n-|I|}}\text{~~~~~~~~~~~~~~~~~~~~~~~~(since $|T_*|=|T_{K/2}|=\frac1{2}|T_0|$ by Lemma~\ref{lm:size-bounds}, {\bf (1)})}\\ 
&=|A|.\text{~~~~~~~~~~~~~~~~~~~~~~~~~~~~~~(by~\eqref{eq:t-star-product})}
\end{split}
\end{equation*}

\begin{figure}[H]
	\begin{center}
		\begin{tikzpicture}
		
		\draw (-6, 0) rectangle (+6, 1);		
		\draw[fill=gray!20] (-6, 0) rectangle (-4, 1);		
		\draw[fill=gray!20] (0, 0) rectangle (+2, 1);
		\draw (-7, 0.5) node {\Large $x=$};		
		\draw (-5, 0.5) node {\Large $a$};
		\draw (1, 0.5) node {\Large $b$};
		\draw (3, 0.5) node {\Large $c$};
		
		\draw (-5, 1.5) node {\Large $\overbrace{\phantom{a+b+c}}^{I'}$};
		\draw (+1, 1.5) node {\Large $\overbrace{\phantom{a+b+c}}^{I}$};		

		\draw (-6, 0-3) rectangle (+6, 1-3);
		\draw[fill=gray!20] (-6, 0-3) rectangle (-4, 1-3);
		\draw[fill=gray!20] (0, 0-3) rectangle (+2, 1-3);
		\draw (-7, 0.5-3) node {\Large $\Pi_k(x)=$};		
		\draw (-5, 0.5-3) node {\Large $b$};
		\draw (1, 0.5-3) node {\Large $M(a)$};
		\draw (3, 0.5-3) node {\Large $c$};
		
		\draw (-5, -3-0.5) node {\Large $\underbrace{\phantom{a+b+c}}_{I'}$};
		\draw (+1, -3-0.5) node {\Large $\underbrace{\phantom{a+b+c}}_{I}$};		
		
		\draw[->] (-3, -0.2) -- (-3, -1.8);
		\draw[->] (-2, -0.2) -- (-2, -1.8);
		\draw[->] (-1, -0.2) -- (-1, -1.8);		
		\draw[->] (+3, -0.2) -- (+3, -1.8);		
		\draw[->] (+4, -0.2) -- (+4, -1.8);	
		\draw[->] (+5, -0.2) -- (+5, -1.8);	

		\draw[->] (+0, -0.2) -- (+0-6, -1.8);	
		\draw[->] (+1, -0.2) -- (+1-6, -1.8);		
		\draw[->] (+2, -0.2) -- (+2-6, -1.8);	
		
		\draw (-5, -1) node {bijection $\eta$};

		\end{tikzpicture}
		\caption{Illustration of the map $\Pi_k$. Note that $\Pi_k$ simply leaves coordinates in $[n]\setminus (I\cup I')$, denoted by $c$, unchanged, copies coordinates in $I$ to coordinates in $I'$ using an arbitrarily chosen but fixed bijection $\eta$, and applies the map $M$ to coordinates in $I'$, assigning the result to coordinates in $I$.}	\label{fig:pi-k}
	\end{center}
	
\end{figure}
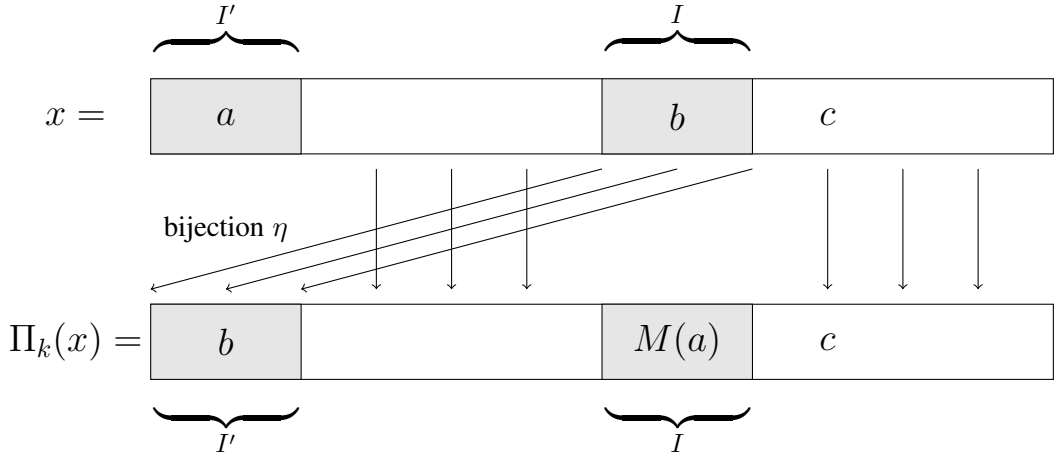

\begin{remark}\label{rm:tau-prefix}
Note that for every $k\in [K/2]$ the set $D_k$ is determined by $I$ and $I'_k$, and $A$ is determined by $I$. Thus, we can construct the map $M$ incrementally, by fixing  $M|_{D_k}: D_k \to A$ as soon as $I'_k$ becomes known. The latter in fact amounts to knowing $\{J'_0, J'_1,\ldots, J'_{k-1}\}$, since we fix $Ext_k$ and $q_k$ for our hard input distribution.
\end{remark}

\begin{definition}[Subcube permutation map $\Pi_k$]\label{def:pi-k}
Define an injective map
\begin{equation}\label{eq:pi-k-def}
\Pi_k: \rho_k(S'_k) \to T_*
\end{equation}
by letting $\Pi_k(x)$ replace $x_I$ with $x_{I'}$ (where $I'=I'_k)$, replace $x_{I'}$ with $M(x_I)$ and leave coordinates outside of $I\cup I'$ untouched, so that 
$$
\Pi_k(z)=(b, M(a), c).
$$ 
See Fig.~\ref{fig:pi-k}  for an illustration.

Formally, we first let $\eta: I\to I'$ be an arbitrary bijection. Given $x\in \rho_k(S'_k)$, write 
$$
x=(a, b, c),
$$ 
where $a=x_{I'}\in D_k\subseteq [m]^{I'}$, $b=x_I\in [m]^I$ and $c=x_{[n]\setminus (I\cup I')}\in [m]^{[n]\setminus (I\cup I')}$. We let
$$
\Pi_k(x):=z,
$$ 
where 
\begin{equation*}
z_j=\left\lbrace
\begin{array}{ll}
b_{\eta^{-1}(j)}&\text{~if~}j\in I'\\
(M(a))_{\eta(j)}&\text{~if~}j\in I\\
c_j&\text{o.w.}
\end{array}
\right.
\end{equation*}
\end{definition}

We will use
\begin{definition}[Rectangle]\label{def:rectangle}
We say that a set $R$ is a rectangle in $I\subset [n]$ if $R\subseteq [m]^I=\prod_{i\in I} A_i$, where $A_i\subseteq [m]$ for every $i\in I$ (i.e., $R$ is the direct product of the sets $A_i, i\in I$).
\end{definition}
The following lemma establishes a key property of the map $\Pi_k$:

\begin{lemma}[Basic properties of the permutation maps $\Pi_k$]\label{lm:pi-prop}
For every $k\in [K/2]$, every $x\in [m]^{I'}$ (where $I'=I'_k$), every rectangle $R$ in $\Lambda\subseteq [n]\setminus (I\cup I')$ the rectangle 
$$
F=\{x\}\times R\times [m]^{[n]\setminus (\Lambda\cup I')}
$$ 
satisfies 
$$
\Pi_k(F)=\{M(x)\}\times R\times [m]^{[n]\setminus (\Lambda\cup I)},
$$
where $M(x)\in [m]^I$ (as per~\eqref{eq:m-def}).
\end{lemma}
\begin{remark}
Intuitively, this lemma says that $\Pi$ maps entire subspace to subspaces, which is a key property that we need our glueing maps to satisfy. This is because, as described in Section~\ref{sec:tech-overview}, if we were to upper bound the size of the maximum matching constructed by algorithm on a single gadget (like~\cite{Kapralov13} does), we would need to consider a vertex cover that is defined by the terminal subcube $T_*$ and its downset. Our construction of a vertex cover in the concatenation of basic gadgets will use this approach, and we need (the downset of) the terminal subcube in one gadget to have `nice structure' when mapped to another gadget using the glueing map $\tau$. Our mapping $\Pi_k$ is useful for this purpose, because the terminal subcube of a subsequent gadget is a subcube defined by coordinates  in $[n]\setminus (I\cup I')$ (these coordinates are the set $\Lambda$ above), and Lemma~\ref{lm:pi-prop} shows that  this set is still a subcube after an application of $\Pi_k$.

This lemma is crucially used in Lemma~\ref{lm:rect-nu-j}, and the rectangle $R$ in question there is the following. We first take some rectangle $F$ in $\Sp(\B^\ell)$ for some $\ell>0$. Then we apply $j$ iterations of the predecessor map $\nu$ to it, namely apply the map $\nu_{\ell, j}$. This results in a rectangle in some subset $\Lambda$ of coordinates with $\Lambda\subseteq \Sp(\B^{\ell+j})$ -- this rectangle essentially contains information about the trajectory of $F$ through repeated invocations of the predecessor map $\nu$. Lemma~\ref{lm:pi-prop} essentially shows that the permutation map $\Pi_k$ does not interfere with this information as long as $\Lambda$ does not overlap with $I'$, which is the case in the application in Lemma~\ref{lm:rect-nu-j}.
\end{remark}
\begin{proof}
We first rewrite the input rectangle $F$ as
$$
F=\{x\}\times [m]^I \times \left(R\times [m]^{[n]\setminus (\Lambda\cup I' \cup I)}\right).
$$
We can thus express every 
$$
z\in \{x\}\times [m]^I \times \left(R\times [m]^{[n]\setminus (\Lambda\cup I' \cup I)}\right)
$$
as 
$$
z=(a, b, c),
$$
where $a\in [m]^{I'}, b\in [m]^I$ and $c\in [m]^{[n]\setminus (I'\cup I)}$, and get by Definition~\ref{def:pi-k}
$$
\Pi_k(z)=(b, M(a), c).
$$
Thus, as $z$ ranges over $\{x\}\times [m]^I \times \left(R\times [m]^{[n]\setminus (\Lambda\cup I' \cup I)}\right)$, the parameter $a$ always equals $x$, $b$ ranges over $[m]^I$ and $c$ ranges over $R\times [m]^{[n]\setminus (\Lambda\cup I' \cup I)}$.
Hence,
\begin{equation*}
\begin{split}
\Pi_k(F)&=\{M(a)\}\times [m]^I\times R\times [m]^{[n]\setminus (\Lambda\cup I' \cup I)}=\{M(a)\}\times R\times [m]^{[n]\setminus (\Lambda\cup I')},
\end{split}
\end{equation*}
as required. \end{proof}

We can now define

\begin{definition}[Glueing map $\tau$]\label{def:tau}
For every $k\in [K/2]$ we define $\tau_k: S'_k\to T_*$ by letting
\begin{equation}\label{eq:tau-def}
\tau_k(x):=\Pi_k(\rho_k(x)).
\end{equation}
Define 
$$
\tau: \biguplus_{k\in [K/2]} S'_k\to T_*
$$ by letting $\tau(x)=\tau_k(x)$ for $x\in S'_k$. 
\end{definition}

\begin{lemma}[Basic properties of $\tau$]\label{lm:tau-prop}
The map $\tau$ is a bijective map from $\biguplus_{k\in [K/2]} S'_k$ to $T_*$.
\end{lemma}
\begin{proof}
The map $\rho_k$ is bijective by Lemma~\ref{lm:densification} (see discussion after~\eqref{eq:s-prime-k} for more details) . The map $\Pi_k$ is injective, since $M$ is injective. Bijectivity of $\tau$ follows from the fact that images of $D_k$ under $M$ are disjoint, and since $\sum_{k\in [K/2]} |D_k|=|A|$ by~\eqref{eq:m-def}.
\end{proof}
\begin{definition}[Basic coordinates $\Gamma$]\label{def:basic-coordinates}
We define the set of {\em basic coordinates} as
$$
\Gamma=\bigcup_{\ell \in [L]} \left(\Sp(\B^\ell)\cup \bigcup_{k\in [K/2]} \left(\Ext_k^\ell\cup \{q_k^\ell\}\right)\right).
$$
\end{definition}
\begin{remark}
Note that $\Gamma\subseteq [n]$ contains all coordinates that densifying maps (Definition~\ref{def:rho}) and permutation maps (Definition~\ref{def:pi-k}) use across all $L$ gadgets $G^\ell$. This fact is crucial for the following lemma (Lemma~\ref{lm:cut-structure}).
\end{remark}

\begin{lemma}\label{lm:tau-inverse-prop}
For every $\ell\in [L], \ell<L-1$, every $x, y\in T_*^\ell$ such that $x_i=y_i$ for all $i\in \Gamma$ there exists $a\in [K/2]$ and $u, v\in S_a^{\ell+1}$ such that
\begin{description}
\item[(1)] $x=\tau^{\ell+1}(u)$ and $y=\tau^{\ell+1}(v)$;
\item[(2)] $u_\Gamma=v_\Gamma$.
\end{description}
\end{lemma}
\begin{proof} We let $\tau:=\tau^{\ell+1}$, $T_*:=T_*^\ell$, $S':=S^{\ell+1}$. We also let $J:=J^\ell, J':=J^{\ell+1}$, $r:=r^{\ell+1}$, and $q_i:=q^{\ell+1}_i$ and $\Ext_i:=\Ext_i^{\ell+1}$ for $i\in [K/2]$ to simplify notation. 

Since $\tau$ maps $S'=\bigcup_{f\in [K/2]} S_f'$ bijectively to $T_*$, there exist $f, g\in [K/2]$ and $u\in S_f', v\in S_g'$ such that $x=\tau(u)$ and $y=\tau(v)$.
Define
\begin{equation*}
\begin{split}
I&:=J^\ell\cup \{r^\ell\}\\
I'_f&:=J'_{<f} \cup \Ext_f \cup \{q_f\}\\
I'_g&:=J'_{<g} \cup \Ext_g \cup \{q_g\}.
\end{split}
\end{equation*}

 By Definition~\ref{def:tau} this means that 
$$
x=\Pi_f(\rho_f(u))\text{~~and~~}y=\Pi_g(\rho_g(v)),
$$
where $\Pi_f, \Pi_g$ are subcube permutation maps as per Definition~\ref{def:pi-k} and $\rho_f$ and $\rho_g$ are $(K-f, q_f)$ and $(K-g, q_g)$-densifying maps as per Definition~\ref{def:rho} respectively. Now write
$$
\rho_f(u)=(a_u, b_u, c_u), \text{~where~}a_u\in D_f\subseteq [m]^{I'_f}, b_u\in [m]^I\text{~and~}c_u\in [m]^{[n]\setminus (I\cup I'_f)}.
$$
Similarly, write 
$$
\rho_g(v)=(a_v, b_v, c_v), \text{~where~}a_v\in D_g\subseteq [m]^{I'_g}, b_v\in [m]^I\text{~and~}c_v\in [m]^{[n]\setminus (I\cup I'_g)}.
$$

Then we have by~Definition~\ref{def:pi-k}
\begin{equation}\label{eq:932yt23y-1}
x=\Pi_f(\rho_f(u))=\Pi_f((a_u, b_u, c_u))=(b_u, M(a_u), c_u),
\end{equation}
where the partition of coordinates on the right hand side is $I'_f\cup I\cup ([n]\setminus (I'_f\cup I))$ and
\begin{equation}\label{eq:932yt23y-2}
y=\Pi_g(\rho_g(v))=\Pi_g((a_v, b_v, c_v))=(b_v, M(a_v), c_v),
\end{equation}
where the partition on the right hand side is $I'_g\cup I\cup ([n]\setminus (I'_g\cup I))$. Here $M: \biguplus_{k\in [K/2]} D_k\to A$ is the bijective map from~\eqref{eq:m-def}.  Since $x_\Gamma=y_\Gamma$ and $I=\Sp(\B^\ell)\subseteq \Gamma$ (as per Definition~\ref{def:basic-coordinates}), we have $M(a_u)=M(a_v)$, and since $M$ is a bijection from $\biguplus_{k\in [K/2]} D_k$ to $A$, this means that $g=f$ and $a_u=a_v$. This in turn means that $I'_f=I'_g$, and we let $I':=I'_f=I'_g$ to simplify notation. In particular, we now have that the partition of coordinates on the right hand side of~\eqref{eq:932yt23y-1} and~\eqref{eq:932yt23y-2} is the same. Since $I'\subseteq \Sp(\B^\ell)\cup \Ext'_f\cup \{q_f\}\subseteq \Gamma$ and $x_\Gamma=y_\Gamma$ by assumption, we have $b_u=b_v$. The assumption $x_\Gamma=y_\Gamma$ also implies 
$$
(c_u)_\Gamma=x_{\Gamma\cap ([n]\setminus (I\cup I'))}=y_{\Gamma\cap ([n]\setminus (I\cup I'))}=(c_v)_\Gamma.
$$
Thus, we have $u_\Gamma=v_\Gamma$, as required.
\end{proof}

\subsection{The predecessor map $\nu$ and its properties}\label{sec:nu}
The predecessor map $\nu$, defined below, is our main tool in defining a vertex cover that lets us bound the size of the matching constructed by a small space algorithm. Intuitively, the predecessor map $\nu_{\ell, j}$ maps a subset of $T^\ell$ for some $\ell\in [L]$ through $j$ repeated applications of the glueing map $\tau^\ell$ interleaved with applications of the $\downset$ map. This is a natural object, since our construction is motivated by the fact that for appropriately defined `nice' subsets $U\subseteq T^\ell$, namely for appropriately defined rectangles (see Lemma~\ref{lm:key-intro} in Section~\ref{sec:tech-overview}), the edge boundary of the set $U\cup \downset^\ell(U)$ is very sparse, which is the basis of our hard input instance.

\begin{definition}[Predecessor map $\nu$]\label{def:nu}
We define the map $\nu_{\ell, j}$ mapping {\em subsets $U\subseteq T^\ell$} to subsets of $T^{\ell-j}$ by induction on $j\geq 0$ as follows. For $j=0$ let
$\nu_{\ell, 0}(U):=U$. For $j>0$ let
$$
\nu_{\ell, j}(U):=\tau^{\ell-(j-1)}(\downset^{\ell-(j-1)}(\nu_{\ell, j-1}(U))).
$$
We define the {\em closure} map $\nu_{\ell, *}$ by
$$
\nu_{\ell, *}(U):=\bigcup_{\substack{j=0\\j \text{~even}}}^\ell \nu_{\ell, j}(U).
$$

We define the map $\mu_{\ell, j}$ mapping {\em subsets $U\subseteq T^\ell$} to subsets of $S^{\ell-j}$ by letting
$$
\mu_{\ell, j}(U):=\downset^{\ell-j}(\nu_{\ell, j}(U))
$$
for $j=0,\ldots, \ell$. We let
$$
\mu_{\ell, *}(U):=\bigcup_{\substack{j=0\\j \text{~even}}}^\ell \mu_{\ell, j}(U).
$$
\end{definition}
\begin{remark}
We stress that the maps $\nu_{\ell, j}$ as well as $\mu_{\ell, j}$ are defined as mapping {\em subsets} of $T^\ell$ to subsets of $T^{\ell-j}$ (resp. $S^{\ell-j}$). This is somewhat more convenient, as otherwise they would not be one to one maps from elements of $T^\ell$ to elements of $T^{\ell-j}$ (resp. $S^{\ell-j}$), because the $\downset$ function is not one to one as per Definition~\ref{def:downset}).
\end{remark}
\begin{remark}
Note that the closure map $\nu_{\ell, *}$ takes a set $U$ to a union of sets $\nu_{\ell, j}(U)$ for even $j$. The significance of the parity constraint on $j$ lies, in particular, in the fact that if $U$ is entirely contained in either the $P$ or the $Q$ side of the bipartition defined in~\eqref{eq:p-def} and~\eqref{eq:q-def}, the closure of $U$, namely $\nu_{\ell, *}(U)$, belongs to the same side of the bipartition. At the same time, the set $\mu_{\ell, *}(U)$ belongs to the other side of the bipartition due to the application of the $\downset$ map in Definition~\ref{def:nu} above. 
\end{remark}

The main results of this section are the following two lemmas. 

The first lemmas is central to establishing the required upper bound on the size of the vertex cover (see Lemma~\ref{lm:vertex-cover}) that bounds the performance of a small space algorithm in Section~\ref{sec:simple-lb}.
\begin{lemma}\label{lm:mu-ell-j}
For every $\ell \in [L]$, every $j=0,\ldots, \ell$, one has 
$$
(\ln 2-C/K)^j \frac1{2}(1-\ln 2)|T^\ell|\leq |\mu_{\ell, j}(T^\ell\setminus T_*^\ell)|\leq (\ln 2+C/K)^j \frac1{2}(1-\ln 2)|T^\ell|.
$$
\end{lemma}

The next lemma establishes the key structural property analogous to Lemma~\ref{lm:key-intro} in Section~\ref{sec:tech-overview}. The lemma is crucially used to upper bound the size of the matching that a low space algorithm can construct in Lemma~\ref{lm:special-edges} in Section~\ref{sec:simple-lb} below. 

\begin{lemma}\label{lm:cut-structure}
For every $\ell \in [L]$, every $j=0,\ldots, \ell$ the following conditions hold. For every 
$$
x\in \nu_{\ell+j, j}(T^{\ell+j}\setminus T_*^{\ell+j})\subset T^\ell,
$$
if $y\in T^\ell$ is such $y_i=x_i$ for all $i\in \Gamma$ (the set of basic coordinates as per Definition~\ref{def:basic-coordinates}), then 
$$
y\in \nu_{\ell+j, j}(T^{\ell+j}\setminus T_*^{\ell+j}).
$$
\end{lemma}

\begin{corollary}\label{cor:cut-structure}
For every $\ell \in [L]$, every $j=0,\ldots, \ell$ for every 
$$
x\in \nu_{\ell+j, j}(T^{\ell+j}\setminus T_*^{\ell+j})\subset T^\ell,
$$
if $y\in S^\ell$ is such that $y_i=x_i$ for all $i\in \Gamma$ (the set of basic coordinates as per Definition~\ref{def:basic-coordinates}), then
$$
y\in \mu_{\ell+j, j}(T^{\ell+j}\setminus T_*^{\ell+j}).
$$
\end{corollary}

In what follows we start by establishing some basic properties of the predecessor map in Section~\ref{sec:predecessor-basic},  then prove Lemma~\ref{lm:rect-nu-j} and Lemma~\ref{lm:mu-ell-j} in Section~\ref{sec:rect-nu} and finally prove the key structural property provided by Lemma~\ref{lm:cut-structure} in Section~\ref{sec:cut-structure}.

\subsubsection{Basic properties of the predecessor map}\label{sec:predecessor-basic}

\begin{claim}\label{cl:nu-prop-basic}
For every $\ell\in [L]$, every $j=1,\ldots, \ell$ and every $U\subseteq T^\ell$ one has {\bf (1)} $\nu_{\ell, j}(U)=\nu_{\ell-1, j-1}(\tau^\ell(\downset^\ell(U))$
 and {\bf (2)}  $\mu_{\ell, j}(U)=\mu_{\ell-1, j-1}(\tau^\ell(\downset^\ell(U)))$. Furthermore, for every $\ell\in [L]$, every $j=0,\ldots, \ell$, $a=0,\ldots, j$ and every $U\subseteq T^\ell$ one has
 {\bf (3)} $\nu_{\ell, j}(U)=\nu_{\ell-a, j-a}(\nu_{\ell, a}(U))$.
\end{claim}
\begin{proof}
For {\bf (1)} we have by Definition~\ref{def:nu} 
\begin{equation*}
\begin{split}
\nu_{\ell, j}(U)&=\tau^{\ell-(j-1)}(\downset^{\ell-(j-1)}(\nu_{\ell, j-1}(U)))\\
&=\tau^{\ell-(j-1)}(\downset^{\ell-(j-1)}(\tau^{\ell-(j-2)}(\downset^{\ell-(j-2)}(\nu_{\ell, j-2}(U)))))\\
&=\tau^{\ell-(j-1)}(\downset^{\ell-(j-1)}(\tau^{\ell-(j-2)}(\downset^{\ell-(j-2)}(\ldots \tau^{\ell}(\downset^\ell(U))\ldots))))\\
&=\nu_{\ell-1, j-1}(\tau^{\ell}(\downset^\ell(U)))\\
\end{split}
\end{equation*}

For {\bf (2)} we have by Definition~\ref{def:nu} and using {\bf (1)}
\begin{equation*}
\begin{split}
\mu_{\ell, j}(U)&=\downset^{\ell-j}(\nu_{\ell, j}(U))\\
&=\downset^{\ell-j}(\nu_{\ell-1, j-1}(\tau^{\ell}(\downset^\ell(U))))\\
&=\mu_{\ell-1, j-1}(\tau^{\ell}(\downset^\ell(U))).
\end{split}
\end{equation*}

Finally, {\bf (3)} follows since by Definition~\ref{def:nu} 
\begin{equation*}
\begin{split}
\nu_{\ell, j}(U)&=\tau^{\ell-(j-1)}(\downset^{\ell-(j-1)}(\nu_{\ell, j-1}(U)))\\
&=\tau^{\ell-(j-1)}(\downset^{\ell-(j-1)}(\tau^{\ell-(j-2)}(\downset^{\ell-(j-2)}(\nu_{\ell, j-2}(U)))))\\
&=\tau^{\ell-(j-1)}(\downset^{\ell-(j-1)}(\tau^{\ell-(j-2)}(\downset^{\ell-(j-2)}(\ldots \tau^{\ell}(\downset^\ell(U))\ldots))))\\
&=\nu_{\ell-a, j-a}(\nu_{\ell, a}(U)))\\
\end{split}
\end{equation*}

\end{proof}

The following claim will help simplify our notation:
\begin{claim}\label{cl:mu-prop}
For every $\ell\in [L]$, every $j\in 1,\ldots, \ell-1$ and every $U\subseteq T^\ell$ one has $|\mu_{\ell, j}(U)|=|\nu_{\ell, j+1}(U)|$.
\end{claim}
\begin{proof}
One has by Definition~\ref{def:nu} 
\begin{equation*}
\begin{split}
\nu_{\ell, j+1}(U)&=\tau^{\ell-j}(\downset^{\ell-j}(\nu_{\ell, j}(U)))\\
&=\tau^{\ell-j}(\mu_{\ell, j}(U))),\\
\end{split}
\end{equation*}
and the claim follows since $\tau^{\ell-j}$ is injective by Lemma~\ref{lm:tau-prop}.
\end{proof}

We note that the increment of the index $j$ on the right hand side in the claim above is crucial, as in general $|\mu_{\ell, j}(U)|$ is very different from $|\nu_{\ell, j}(U)|$ (since $\downset$ is not a one to one map).

We also need 
\begin{lemma}[Basic properties of the maps $\nu_{\ell, j}$ and $\mu_{\ell, j}$]\label{lm:nu-prop}
The following conditions hold for the maps $\nu$ and $\mu$ defined above:
\begin{description}
\item[(1)] for every $\ell\in [L]$ and every $0\leq j\leq \ell$ the maps $\nu_{\ell, j}$ and $\mu_{\ell, j}$ are injective;
\item[(2)] every $\ell, \ell'\in [L]$ every $0\leq j\leq \ell$, $0\leq j'\leq \ell'$ one has 
$$
\nu_{\ell, j}(T^\ell\setminus T_*^\ell)\cap \nu_{\ell', j'}(T^{\ell'}\setminus T_*^{\ell'})=\emptyset
$$ unless $\ell=\ell'$ and $j=j'$.
\item[(3)] every $\ell, \ell'\in [L]$ every $0\leq j\leq \ell$, $0\leq j'\leq \ell'$ one has 
$$
\mu_{\ell, j}(T^\ell\setminus T_*^\ell)\cap \mu_{\ell', j'}(T^{\ell'}\setminus T_*^{\ell'})=\emptyset
$$ unless $\ell=\ell'$ and $j=j'$.
\end{description}
\end{lemma}
\begin{proof}
\noindent {\bf (1)}  follows since $\tau$ is injective by Lemma~\ref{lm:tau-prop} and $\downset$ is injective by construction (Definition~\ref{def:downset}).

We now show {\bf (2)}. First note that  $\nu_{\ell, j}(T^\ell\setminus T_*^\ell)\subseteq T^{\ell-j}$ and $\nu_{\ell', j'}(T^{\ell'}\setminus T_*^{\ell'})\subseteq T^{\ell'-j'}$, and hence the two sets are disjoint if $\ell-j\neq \ell'-j'$. Now suppose that $\ell-j=\ell'-j'$ and assume without loss of generality that $\ell\leq \ell'$. Furthermore, we can assume that $\ell<\ell'$, since if $\ell=\ell'$, one must have $j=j'$ as otherwise the sets are disjoint by the previous argument. Now note that 
\begin{equation*}
\begin{split}
\nu_{\ell', \ell'-\ell}(T^{\ell'}\setminus T_*^{\ell'})&=\tau^{\ell+1}(\downset^{\ell+1}(\nu_{\ell', \ell'-\ell-1}(T^{\ell'}\setminus T_*^{\ell'}))\subseteq T_*^\ell,
\end{split}
\end{equation*}
since the range of $\tau^{\ell+1}$ is $T_*^\ell$ by~Definition~\ref{def:tau}. 
This means that
\begin{equation*}
\begin{split}
\nu_{\ell', j'}(T^{\ell'}\setminus T_*^{\ell'})&=\nu_{\ell, j}(\nu_{\ell', \ell'-\ell}(T^{\ell'}\setminus T_*^{\ell'}))\text{~~~~~~~~(by Claim~\ref{cl:nu-prop-basic}, {\bf (3)}, and using $\ell-j=\ell'-j'$)}\\
&\subseteq \nu_{\ell, j}(T_*^\ell),
\end{split}
\end{equation*}
and we get that
$$
\nu_{\ell, j}(T^\ell\setminus T_*^\ell)\cap  \nu_{\ell', j'}(T^{\ell'}\setminus T_*^{\ell'})\subseteq \nu_{\ell, j}(T^\ell\setminus T_*^\ell)\cap \nu_{\ell, j}(T_*^\ell)=\emptyset
$$
since $\nu_{\ell, j}$ is injective by {\bf (1)}.

We now prove {\bf (3)}. First note that by Definition~\ref{def:nu}
$$
\mu_{\ell, j}(T^\ell\setminus T_*^\ell)\subseteq S^{\ell-j}
$$
and 
$$
\mu_{\ell', j'}(T^{\ell'}\setminus T_*^{\ell'})\subseteq S^{\ell'-j'},
$$
and hence similarly to above the two sets are disjoint unless $\ell-j=\ell'-j'$. {\bf (3)} now follows by noting that, again using Definition~\ref{def:nu}, we get, since $\ell-j=\ell'-j'$ and $\downset$ is injective,
\begin{equation*}
\begin{split}
\mu_{\ell, j}(T^\ell\setminus T_*^\ell)\cap \mu_{\ell', j'}(T^{\ell'}\setminus T_*^{\ell'})&=\downset^{\ell-j}(\nu_{\ell, j}(T^\ell\setminus T_*^\ell))\cap \downset^{\ell'-j'}(\nu_{\ell', j'}(T^{\ell'}\setminus T_*^{\ell'}))\\
&=\downset^{\ell-j}(\nu_{\ell, j}(T^\ell\setminus T_*^\ell) \cap \nu_{\ell', j'}(T^{\ell'}\setminus T_*^{\ell'}))\\
&=\emptyset,
\end{split}
\end{equation*}
where we used {\bf (2)} in the last transition.
\end{proof}

We will use
\begin{lemma}\label{lm:t-star-ell-rec}
For every $\ell \in [L]$ one has
\begin{equation*}
T_*^\ell=\nu_{L-1, L-1-\ell}(T_*^{L-1})\cup \bigcup_{j=1}^{L-1-\ell} \nu_{\ell+j, j}(T^{\ell+j}\setminus T_*^{\ell+j})
\end{equation*}
and 
\begin{equation*}
T^\ell=\nu_{L-1, L-1-\ell}(T_*^{L-1})\cup \bigcup_{j=0}^{L-1-\ell} \nu_{\ell+j, j}(T^{\ell+j}\setminus T_*^{\ell+j}).
\end{equation*}
\end{lemma}
\begin{proof}
We start by establishing the first result of the lemma, namely
\begin{equation}\label{eq:t-star-2ug2438}
T_*^\ell=\nu_{L-1, L-1-\ell}(T_*^{L-1})\cup \bigcup_{j=1}^{L-1-\ell} \nu_{\ell+j, j}(T^{\ell+j}\setminus T_*^{\ell+j})
\end{equation}
by induction on $\ell=L-1,\ldots, 0$.

\noindent{\bf Base: $\ell=L-1$.} One has $T_*^\ell=\nu_{L-1, 0}(T_*^L)$, as required, since $\nu_{L-1, 0}$ is the identity map by definition (see Definition~\ref{def:nu}).

\noindent{\bf Inductive step: $\ell\to \ell-1$.}   By the inductive hypothesis we have 
$$
T_*^\ell=\nu_{L-1, L-1-\ell}(T_*^{L-1})\cup\bigcup_{\substack{j=1}}^{L-1-\ell} \nu_{\ell+j, j}(T^{\ell+j}\setminus T_*^{\ell+j}).
$$
Applying $\tau^\ell(\downset^\ell(\cdot))$ to both sides of the equation above, we get, letting $Q=\nu_{L-1, L-1-\ell}(T_*^{L-1})$ to simplify notation,
\begin{equation}\label{eq:o4329hg394hg}
\begin{split}
\tau^\ell(\downset^\ell(T_*^\ell))&=\tau^\ell\left(\downset^\ell\left(Q\cup \bigcup_{\substack{j=1}}^{L-1-\ell} \nu_{\ell+j, j}(T^{\ell+j}\setminus T_*^{\ell+j})\right)\right)\\
&= \tau^\ell\left(\downset^\ell\left(Q\right)\right)\cup \bigcup_{\substack{j=1}}^{L-1-\ell} \tau^\ell\left(\downset^\ell\left(\nu_{\ell+j, j}(T^{\ell+j}\setminus T_*^{\ell+j})\right)\right)\\
&= \nu_{L-1, L-\ell}(T_*^{L-1})\cup \bigcup_{\substack{j=1}}^{L-1-\ell} \nu_{\ell+j, j+1}(T^{\ell+j}\setminus T_*^{\ell+j})\\
&= \nu_{L-1, L-\ell}(T_*^{L-1})\cup \bigcup_{\substack{j=1}}^{L-1-\ell} \nu_{(\ell-1)+(j+1), j+1}(T^{(\ell-1)+(j+1)}\setminus T_*^{(\ell-1)+(j+1)})\\
&= \nu_{L-1, L-1-(\ell-1)}(T_*^{L-1})\cup \bigcup_{\substack{j=2}}^{L-\ell} \nu_{(\ell-1)+j, j}(T^{(\ell-1)+j}\setminus T_*^{(\ell-1)+j})
\end{split}
\end{equation}
We also have
\begin{equation}\label{eq:9h429gh23g}
\tau^\ell(\downset^\ell(T^\ell\setminus T_*^\ell))=\nu_{\ell, 1}(T^\ell\setminus T_*^\ell)=\nu_{(\ell-1)+1, 1}(T^\ell\setminus T_*^\ell).
\end{equation}

We now recall that $\tau^\ell$ maps $S^\ell$ bijectively to $T_*^{\ell-1}$ and $S^\ell=\downset(T^\ell)$ (this follows by putting together the fact  that $S^\ell=\biguplus_{k\in [K/2]} S^\ell_k$ with~\eqref{eq:def-sk} and Definition~\ref{def:downset}), which implies
\begin{equation}\label{eq:239hg92hg}
\begin{split}
T_*^{\ell-1}&=\tau^\ell(S^\ell)\text{~~~~~~~~~~~~~~~(since $\tau^\ell$ bijectively maps $S^\ell$ to $T_*^{\ell-1}$)}\\
&=\tau^\ell(\downset^\ell(T^\ell))\\
&=\tau^\ell(\downset^\ell(T^\ell\setminus T_*^\ell))\cup \tau^\ell(\downset^\ell(T_*^\ell))\\
\end{split}
\end{equation}

Substituting~\eqref{eq:o4329hg394hg} and~\eqref{eq:9h429gh23g} into~\eqref{eq:239hg92hg}, we get
$$
T_*^{\ell-1}=\nu_{L-1, L-\ell}(T_*^{L-1})\cup\bigcup_{\substack{j\geq 1}} \nu_{\ell-1+j, j}(T^{\ell-1+j}\setminus T_*^{\ell-1+j}), 
$$
as required. This completes the inductive claim and establishes the first result of the lemma.

Now in order to obtain the second result of the lemma we take the union of both sides of~\eqref{eq:t-star-2ug2438} with $T^\ell\setminus T_*^\ell$, writing $T^\ell\setminus T_*^\ell=\nu_{\ell+0, 0}(T^{\ell+0}\setminus T_*^{\ell+0})$ on the rhs. This results in
$$
T^\ell=(T^\ell\setminus T_*^\ell)\cup T_*^\ell=\nu_{L-1, L-1-\ell}(T_*^{L-1})\cup\bigcup_{\substack{j\geq 0}} \nu_{\ell+j, j}(T^{\ell+j}\setminus T_*^{\ell+j}),
$$
as required.
\end{proof}

\subsubsection{Proof of Lemma~\ref{lm:mu-ell-j}}\label{sec:rect-nu}

We start with
\begin{definition}[Rectangle consistent with a terminal subcube]\label{def:consistent}
For every $\ell\in [L]$, every $I\subseteq [n]$, every fixing $f$ of coordinates in $I$ we say that $f$ is {\em consistent with $T_*^\ell$} if 
$$
\{f\}\times [m]^{[n]\setminus I}\subseteq T_*^\ell.
$$ We say that a rectangle $R$ in $I\subseteq [n]$ (as per Definition~\ref{def:rectangle}) is consistent with $T_*^\ell$ if 
$$
R\times [m]^{[n]\setminus I}\subseteq T_*^\ell.
$$
\end{definition}

We first prove an auxiliary
\begin{claim}\label{cl:fixing}
For every $\ell\in [L]$, every $I\subseteq [n]$, every rectangle $R$ in $I$ that is consistent with the terminal subcube $T_*^\ell$ the following conditions hold. If $I'=I\cap \Sp(\B^\ell)$ (see Definition~\ref{def:special-coords}) and $R=R_0\times R_1$, where $R_0$ is a rectangle in $I'$ and $R$ is a rectangle in $I\setminus I'$, then for every $f_0\in R_0$ one has that $\{f_0\}\times R_1$ is consistent with the terminal subcube $T_*^\ell$.
\end{claim}
\begin{proof}
Since $R$ is consistent with $T_*^\ell$ by assumption, we have $R\times [m]^{[n]\setminus I}=R_0\times R_1\times [m]^{[n]\setminus I}\subseteq T_*^\ell$, and hence for every $f_0\in R_0$ one has $\{f_0\}\times R_1\times [m]^{[n]\setminus I}$, i.e. $\{f_0\}\times R_1$ is consistent with $T_*^\ell$.
\end{proof}

The lemma below is an important tool that we will use in the actual proof of Lemma~\ref{lm:mu-ell-j}. The lemma bounds the size of a subset of the terminal subcube under the predecessor map:
\begin{lemma}\label{lm:rect-nu-j}
For every $\ell \in [L]$, every $j=0,\ldots, \ell$, every fixing $f$ of coordinates in $\Sp(\B^\ell)$ consistent with $T_*^\ell$ (as per Definition~\ref{def:consistent}), every rectangle $R$ in $\Sp(\B^{>\ell})$  the rectangle 
$$
F:=\{f\}\times R\times [m]^{[n]\setminus \Sp(\B^{\geq \ell})},
$$ 
satisfies
$$
(\ln 2-C/K)^j |F|\leq |\nu_{\ell, j}(F)|\leq (\ln 2)^j |F|
$$
for an absolute constant $C>0$.
\end{lemma}
\begin{proof}
The proof is by induction on $j$. The inductive claim is that
for every $\ell \in [L]$, every fixing $f$ of coordinates in $\Sp(\B^\ell)$ consistent with $T_*^\ell$, every rectangle $Y$ in $\Sp(\B^{>\ell})$ the rectangle 
$$
F=\{f\}\times Y\times [m]^{[n]\setminus \Sp(\B^{\geq \ell})},
$$ satisfies

$$
(\ln 2-C/K)^j |F|\leq |\nu_{\ell, j}(F)|\leq (\ln 2)^j |F|
$$
for an absolute constant $C>0$.

\noindent{\bf Base: $j=0$.} We have $|\nu_{\ell, j}(F)|=|\nu_{\ell, 0}(F)|=|F|$, as required.

\noindent{\bf Inductive step: $j \to j+1$.} Fix $\ell\in [L]$ and fix $k\in [K/2]$. We write $T:=T^\ell, T_*:=T_*^\ell$, as well as $\tau:=\tau^\ell$, $\downset:=\downset^\ell$ to simplify notation.
Let $J:=J^{\ell-1}$, let $J':=J^\ell$. Let $q_k:=q^\ell_k$ denote the $k$-th compression index in $\B^\ell$, and let $r'=r^\ell$ and $r=r^{\ell-1}$ denote the compression indices for $\B^\ell$ and $\B^{\ell-1}$ respectively.  Note that $\ell>0$, since otherwise we must have $j=0$.

Let $\rho_k$ be the $(K-k, q_k)$-compressing map as per Definition~\ref{def:rho}.  Since $f$ is consistent with $T_*$, we have $F\subseteq T_k$. 
Furthermore, since $q_k\not \in \Sp(\B^{\geq \ell})$ (see Definition~\ref{def:cext} and Property~\ref{prop:q-k}),  we have that the rectangle $F$ does not depend on coordinate $q_k$ (as per Definition~\ref{def:u-dep}). This means that by Lemma~\ref{lm:densification}
 the map $\rho_k$ maps
$$
\downset_k(F)\delequal \left\{x\in F: \weight(x) \in \left[0, \frac1{K-k}\right)\cdot W \pmod{ W} \right\}
$$ 
bijectively to 
\begin{equation}\label{eq:fwt-def}
\left\{x\in F: x_{q_k}/m\in \left[0, \frac1{K-k}\right)\right\},
\end{equation}
which in particular implies
\begin{equation}\label{eq:fk-downset}
|\rho(\downset_k(F))|=\frac1{K-k}|F|.
\end{equation}

Let $f_0$ denote the restriction of $f$ to $J'_{<k}\subseteq \Sp(\B^\ell)$ and let $f_1$ denote the restriction of $f$ to $\Sp(\B^\ell)\setminus J'_{<k}=J'_{\geq k}\cup \{r'\}$. Recall the definitions of the index set $I'$ (see~\eqref{eq:i-prime-def})
$$
I'=J'_{<k}\cup \Ext_k\cup \{q_k\}
$$
and index set $I$ (see~\eqref{eq:i-def}) 
$$
I=J\cup \{r\}.
$$

We let $H=\Ext_k\cup \{q_k\}$ for convenience, and define for $a\in [m]^H$
\begin{equation}\label{eq:93yty9uFGHsnXN}
F(a):=\{(f_0, a)\}\times \{f_1\}\times Y\times [m]^{[n]\setminus (\Sp(\B^{\geq \ell})\cup H)}.
\end{equation}
We note that $(f_0, a)\in [m]^{I'}$ -- this property makes it convenient to reason about the image of $F(a)$ under $\tau_k$, as we show below.
Also note that rectangles $F(a)$ defined above are disjoint for distinct choices of $a$ and 
\begin{equation}\label{eq:union234fgzdF}
\rho_k(\downset_k(F))=\bigcup_{a\in Q} F(a),
\end{equation}
where $Q=[m]^{\Ext_k}\times \left\{0, 1,\ldots, \frac{m}{K-k}-1\right\}$ by~\eqref{eq:fwt-def}.  We now apply Lemma~\ref{lm:pi-prop} to rectangle $F(a)$ for $a\in Q$. We invoke Lemma~\ref{lm:pi-prop} with $x=(f_0, a)\in [m]^{I'}$, rectangle $R=\{f_1\} \times Y$ and 
$$
\Lambda=\Sp(\B^{>\ell})\cup J'_{\geq k}\cup \{r'\}.
$$
We note that $[n]\setminus (\Lambda \cup I')=[n]\setminus (\Sp(\B^{\geq \ell})\cup H)$, which is consistent with~\eqref{eq:93yty9uFGHsnXN}. Also note that $\Lambda\subset \Sp(\B^{\geq \ell})$. By Lemma~\ref{lm:pi-prop} we get
$$
\Pi_k\left(\{x\}\times R\times [m]^{[n]\setminus (\Lambda\cup I')}\right)=\left\{M(x)\right\}\times R\times [m]^{[n]\setminus (\Lambda\cup I)},
$$
where $M(x)\in [m]^I$ is consistent with the terminal subcube $T_*$ by definition of $M$ (see~\eqref{eq:m-def}). Substituting the setting of $x$ and $R$, we get $\Pi_k(F(a))=\wh{F}(a)$, where 
$$
\wh{F}(a)=\left\{M(x)\right\}\times \{f_1\}\times Y\times [m]^{[n]\setminus (\Lambda\cup I)}.
$$
This together with~\eqref{eq:union234fgzdF} implies
\begin{equation}\label{eq:ihigg43gujgdjg}
\Pi_k(\rho_k(\downset_k(F)))=\bigcup_{a\in Q} \wh{F}(a).
\end{equation}

We now apply the inductive hypothesis to $\wh{F}(a)$ with fixing $M(x)\in [m]^{\Sp(\B^{\ell-1})}$ of coordinates and rectangle $R= \{f_1\}\times Y\in [m]^\Lambda$. The preconditions of the lemma are satisfied since $\Lambda\subseteq \Sp(\B^{\geq \ell})$. The inductive hypothesis gives
\begin{equation}\label{eq:ind-uwg8tg}
(\ln 2-C/K)^{j-1} |\wh{F}(a)|\leq |\nu_{\ell-1, j-1}(\wh{F}(a))|\leq (\ln 2)^{j-1} |\wh{F}(a)|.
\end{equation}
Applying the function $\nu_{\ell-1, j-1}$ to both sides of~\eqref{eq:ihigg43gujgdjg}, and using~\eqref{eq:ind-uwg8tg}, the fact that $\nu_{\ell-1, j-1}$ is injective as well as the fact that $\wh{F}(a)$ are disjoint for different $a\in [m]^H$ we have
\begin{equation}\label{eq:union234fgzdF-lb}
\begin{split}
|\nu_{\ell-1, j-1}(\Pi_k(\rho_k(\downset_k(F)))|&=\left|\bigcup_{a\in Q} \nu_{\ell-1, j-1}(\wh{F}(a))\right|\\
&=\sum_{a\in Q} \left|\nu_{\ell-1, j-1}(\wh{F}(a))\right|\\
&\geq \sum_{a\in Q} (\ln 2-C/K)^{j-1}  \left|\wh{F}(a)\right|\\
&=(\ln 2-C/K)^{j-1}  \sum_{a\in Q} \left|\wh{F}(a)\right|\\
&=(\ln 2-C/K)^{j-1}  \sum_{a\in Q} \left|F(a)\right|\\
&=(\ln 2-C/K)^{j-1} \left|\bigcup_{a\in Q} F(a)\right|\\
&=(\ln 2-C/K)^{j-1} \left|\rho_k(\downset_k(F))\right|\\
&=(\ln 2-C/K)^{j-1} \frac1{K-k}\left|F\right|.
\end{split}
\end{equation}
In the fifth transition we used the fact that $|F(a)|=|\wh{F}(a)|$, which follows by Lemma~\ref{lm:pi-prop} together with the fact that $\Pi_k$ is injective. In the seventh transition we used~\eqref{eq:union234fgzdF}. The final transition uses~\eqref{eq:fk-downset}.

For the upper bound we similarly have,  applying the function $\nu_{\ell-1, j-1}$ to both sides of~\eqref{eq:ihigg43gujgdjg}, and using~\eqref{eq:ind-uwg8tg}, the fact that $\nu_{\ell-1, j-1}$ is injective as well as the fact that $\wh{F}(a)$ are disjoint for different $a\in [m]^H$ we have
\begin{equation}\label{eq:union234fgzdF-ub}
\begin{split}
|\nu_{\ell-1, j-1}(\Pi_k(\rho_k(\downset_k(F)))|&=\left|\bigcup_{a\in Q} \nu_{\ell-1, j-1}(\wh{F}(a))\right|\\
&=\sum_{a\in Q} \left|\nu_{\ell-1, j-1}(\wh{F}(a))\right|\\
&\leq \sum_{a\in Q} (\ln 2)^{j-1}  \left|\wh{F}(a)\right|\\
&=(\ln 2)^{j-1}  \sum_{a\in Q} \left|\wh{F}(a)\right|\\
&=(\ln 2)^{j-1}  \sum_{a\in Q} \left|F(a)\right|\\
&=(\ln 2)^{j-1} \left|\bigcup_{a\in Q} F(a)\right|\\
&=(\ln 2)^{j-1} \left|\rho_k(\downset_k(F))\right|\\
&=(\ln 2)^{j-1} \frac1{K-k}\left|F\right|.
\end{split}
\end{equation}
In the fifth transition we used the fact that $|F(a)|=|\wh{F}(a)|$, which follows by Lemma~\ref{lm:pi-prop} together with the fact that $\Pi_k$ is injective. In the seventh transition we used~\eqref{eq:union234fgzdF}. The final transition uses~\eqref{eq:fk-downset}.

We now get, summing the above over $k\in [K/2]$ 
\begin{equation}\label{eq:092hg92hgc23rfhFUHF-lb}
|\nu_{\ell, j}(F)|=\sum_{k\in [K/2]} |\nu_{\ell-1, j-1}(\tau(\downset_k(F)))|\geq \left(\sum_{k\in [K/2]} \frac1{K-k}\right)\cdot (\ln 2-C/K)^{j-1} |F|
\end{equation}
and 
\begin{equation}\label{eq:092hg92hgc23rfhFUHF-ub}
|\nu_{\ell, j}(F)|=\sum_{k\in [K/2]} |\nu_{\ell-1, j-1}(\tau(\downset_k(F)))|\leq \left(\sum_{k\in [K/2]} \frac1{K-k}\right)\cdot (\ln 2)^{j-1} |F|
\end{equation}

At the same time one has by Claim~\ref{cl:sum-int} 
$$
\ln 2-1/K\leq \sum_{k\in [K/2]} \frac1{K-k}\leq \ln 2
$$
Putting this together with~\eqref{eq:092hg92hgc23rfhFUHF-lb} and~\eqref{eq:092hg92hgc23rfhFUHF-ub} completes the proof of the inductive step (we assume that $C\geq 1$), and completes the proof of the lemma.\end{proof}

\begin{corollary}\label{cor:nu-star-ub}
For every $\ell \in [L]$, every $j=0,\ldots, \ell$, every rectangle  $R$ in $\Sp(\B^{\geq \ell})$ consistent with the terminal subcube $T_*^\ell$ (as per Definition~\ref{def:consistent}) the extended rectangle 
$$
F=R\times [m]^{[n]\setminus \Sp(\B^{\geq \ell})}
$$
satisfies
$$
(\ln 2-C/K)^j |F|\leq |\nu_{\ell, j}(F)|\leq (\ln 2)^j |F|
$$
for an absolute constant $C>0$. 
\end{corollary}
\begin{proof}
Write $R=R_0\times R_1$, where $R_0$ is a rectangle in $\Sp(\B^\ell)$ and $R_1$ is a rectangle in $\Sp(\B^{>\ell})$ (this is possible by Definition~\ref{def:rectangle} of a rectangle). We have 
\begin{equation}\label{eq:923ytyxSyHDH}
F=\bigcup_{a\in R_0} F(a),
\end{equation}
where 
$$
F(a)=\{a\}\times R_1\times [m]^{[n]\setminus \Sp(\B^{\geq \ell})}.
$$
Note that by Claim~\ref{cl:fixing} every $f\in R_0$ is consistent with the terminal subcube since $R$ is consistent with the terminal subcube by assumption. Thus, the preconditions of Lemma~\ref{lm:rect-nu-j} are satisfied, and we have 
\begin{equation}\label{eq:rguugUGGF}
(\ln 2-C/K)^j |F(a)|\leq |\nu_{\ell, j}(F(a))|\leq (\ln 2)^j |F(a)|.
\end{equation}
Applying $\nu_{\ell, j}$ to~\eqref{eq:923ytyxSyHDH}, combining with~\eqref{eq:rguugUGGF} and using the fact that $\nu_{\ell, j}$ is injective by Lemma~\ref{lm:nu-prop}, we get
\begin{equation*}
\begin{split}
|\nu_{\ell, j}(F)|&=\sum_{a\in R_0} |\nu_{\ell, j}(F(a))|\\
&\leq (\ln 2)^j \sum_{a\in R_0} |F(a)|\\
&= (\ln 2)^j |F|.
\end{split}
\end{equation*}
Similarly, we get
\begin{equation*}
\begin{split}
|\nu_{\ell, j}(F)|&=\sum_{a\in R_0} |\nu_{\ell, j}(F(a))|\\
&\geq (\ln 2-C/K)^j \sum_{a\in R_0} |F(a)|\\
&= (\ln 2-C/K)^j |F|.
\end{split}
\end{equation*}
\end{proof}

We now give 

\begin{proofof}{Lemma~\ref{lm:mu-ell-j}}
We write $T:=T^\ell, T_*:=T_*^\ell$, as well as $\tau:=\tau^\ell$, $\downset:=\downset^\ell$ to simplify notation. We start by writing 
\begin{equation*}
\mu_{\ell, j}(T\setminus T_*)=\mu_{\ell, j}(T\setminus T_{K/2})=\bigcup_{k\in [K/2]} \mu_{\ell, j}(T_k\setminus T_{k+1}).
\end{equation*}
Since $\mu_{\ell, j}$ is injective by Lemma~\ref{lm:nu-prop}, {\bf (1)}, one has $\mu_{\ell, j}(T_k\setminus T_{k+1}) \cap \mu_{\ell, j}(T_{k'}\setminus T_{k'+1})=\emptyset$ for distinct $k, k'\in [K/2]$ (indeed, as the sets $T_k$ are nested, $T_k\setminus T_{k+1}$ are disjoint for distinct $k$). Thus, 
\begin{equation}\label{eq:923h92ht12Wgugf}
|\mu_{\ell, j}(T\setminus T_*)|=\sum_{k\in [K/2]} |\mu_{\ell, j}(T_k\setminus T_{k+1})|,
\end{equation}
and in order to bound $|\mu_{\ell, j}(T\setminus T_*)|$ it suffices to bound $|\mu_{\ell, j}(T_k\setminus T_{k+1})|$ for every $k\in [K/2]$.  Fix $k\in [K/2]$. We have 
\begin{equation*}
\begin{split}
\mu_{\ell, j}(T_k\setminus T_{k+1})&=\mu_{\ell-1, j-1}(\tau(\downset(T_k\setminus T_{k+1})))\\
&=\bigcup_{s=0}^k \mu_{\ell-1, j-1}(\tau(\downset_s(T_k\setminus T_{k+1}))),\\
\end{split}
\end{equation*}
where the first transition uses Claim~\ref{cl:nu-prop-basic}, {\bf (2)}, and the second transition is by Definition~\ref{def:downset} and Remark~\ref{rm:truncated-downset}.
We bound $|\mu_{\ell, j}(T_k\setminus T_{k+1})|$ by bounding the size of individual terms on the rhs of the equation above. This suffices since $\mu_{\ell-1, j-1}(\tau(\downset_s(T_k\setminus T_{k+1})))$ are disjoint for different $s$ -- this follows by noting that $\downset_s(T_k\setminus T_{k+1}))$ are disjoint for different $s$ by construction, $\tau$ is bijective by Lemma~\ref{lm:tau-prop} and $\mu_{\ell-1, j-1}$ is injective by Lemma~\ref{lm:nu-prop}, {\bf (1)}. Formally,
\begin{equation*}
\begin{split}
|\mu_{\ell, j}(T_k\setminus T_{k+1})|=\sum_{s=0}^k |\mu_{\ell-1, j-1}(\tau(\downset_s(T_k\setminus T_{k+1})))|.
\end{split}
\end{equation*}
Furthermore, since for every set $U\subseteq T^{\ell-1}$ one has 
$$
|\mu_{\ell-1, j-1}(U)|=|\nu_{\ell-1, j}(U)|,
$$
by Claim~\ref{cl:mu-prop}, 
 we have
\begin{equation}\label{eq:823gt8g8GDGNKNFx}
\begin{split}
|\mu_{\ell, j}(T_k\setminus T_{k+1})|&=\sum_{s=0}^k |\nu_{\ell, j-1}(\tau(\downset_s(T_k\setminus T_{k+1})))|.
\end{split}
\end{equation}

\paragraph{Bounding the rhs of~\eqref{eq:823gt8g8GDGNKNFx}.} We now bound the terms on the rhs of~\eqref{eq:823gt8g8GDGNKNFx}. Let $J=J^\ell$ and $r=r^\ell$ to simplify notation. For $s\in \{0, 1,\ldots, k\}$ let $q_s:=q^\ell_s$ and let $\rho_s$ be the $(K-s, q_s)$-densifying map as per Definition~\ref{def:rho}.  Define
$$
I':=J_{<s}\cup \Ext_s\cup \{q_s\}.
$$

Since $T_k\setminus T_{k+1}$ does not depend on $q_s$ (by Property~\ref{prop:q-k}; see also Definition~\ref{def:u-dep}), we have by Lemma~\ref{lm:densification} 
\begin{equation}
\begin{split}
\rho_s(\downset_s(T_k\setminus T_{k+1}))&\delequal \left\{x\in [m]^n: x_{J_t}/m\in \left[0, 1-\frac1{K-t}\right)\text{~for all~}t=0,\ldots, k-1\right.\\
&\text{~~~~~~~~~~~~~~~~~~~~~~~~~~~and~}\\
&\text{~~~~~~~~~~~~~~~~~~~~~~~~~~}x_{J_k}/m\in \left(1-\frac1{K-k}, 1\right]\\
&\text{~~~~~~~~~~~~~~~~~~~~~~~~~~~and~}\\
&\text{~~~~~~~~~~~~~~~~~~~~~~~~~~}\left.x_{q_s}/m\in \left[0, \frac1{K-s}\right)\right\}.
\end{split}
\end{equation}

Define 

\begin{equation}
\begin{split}
Q&=\left\{x\in [m]^{I'}: x_{J_t}/m\in \left[0, 1-\frac1{K-t}\right)\text{~for all~}t=0,\ldots, s-1\right\}\\
&\text{~~~~~~~~~~~~~~~~~~~~~~~~~~~and~}\\
&\text{~~~~~~~~~~~~~~~~~~~~~~~~~~}\left.x_{q_s}/m\in \left[0, \frac1{K-s}\right)\right\}
\end{split}
\end{equation}
and, letting  $\Lambda:=\Sp(\B^{\geq \ell})\setminus J_{<s}$,
\begin{equation}
\begin{split}
R&=\left\{x\in [m]^{\Lambda}: x_{J_t}/m\in \left[0, 1-\frac1{K-t}\right)\text{~for all~}t=s,\ldots, k-1\right.\\
&\text{~~~~~~~~~~~~~~~~~~~~~~~~~~~and~}\\
&\text{~~~~~~~~~~~~~~~~~~~~~~~~~~}\left.x_{J_k}/m\in \left(1-\frac1{K-k}, 1\right]\right\}.
\end{split}
\end{equation}
so that 
\begin{equation}\label{eq:rho-tk-diff}
\begin{split}
\rho_s(\downset_s(T_k\setminus T_{k+1}))&=Q\times R\times [m]^{[n]\setminus (\Lambda\cup I')}=\bigcup_{a\in Q} F(a).\\
\end{split}
\end{equation} Further, for $a\in Q\subseteq [m]^{I'}$ let
$$
F(a):=\{a\} \times R\times [m]^{[n]\setminus (\Lambda\cup I')}.
$$
We note that $F(a)\cap F(a')=\emptyset$ for $a\neq a'$. By Lemma~\ref{lm:pi-prop} we have 
\begin{equation}\label{eq:93g7GFGSHkjLMXasS}
\begin{split}
\Pi_s(F(a))&=\Pi_s\left(\{a\}\times R\times [m]^{[n]\setminus (\Lambda\cup I')}\right)\\
&=\left\{M(a)\right\}\times R\times [m]^{[n]\setminus (\Lambda\cup I)}:=\wh{F}(a).
\end{split}
\end{equation}
Since $M(a)\in [m]^I$ (see~\eqref{eq:i-def} and~\eqref{eq:m-def}) is consistent with the terminal subcube $T_*^\ell$, we get by Lemma~\ref{lm:rect-nu-j} 
\begin{equation}\label{eq:823gtgCGGFK}
(\ln 2-C/K)^j |\wh{F}(a)|\leq |\nu_{\ell-1, j}(\wh{F}(a))|\leq (\ln 2)^j |\wh{F}(a)|.
\end{equation}
We now apply $\nu_{\ell-1, j}(\Pi_s(\cdot))$ to both sides of~\eqref{eq:rho-tk-diff}, obtaining 
\begin{equation}\label{eq:iehfihefh9efKPCO}
\begin{split}
|\nu_{\ell-1, j}(\Pi_s(\rho_s(\downset_s(T_k\setminus T_{k+1}))))|&=\sum_{a\in Q} |\nu_{\ell-1, j}(\Pi_s(F(a)))|\\
&=\sum_{a\in Q} |\nu_{\ell-1, j}(\wh{F}(a))|,
\end{split}
\end{equation}
where the last transition uses the definition of $\wh{F}(a)$ in~\eqref{eq:93g7GFGSHkjLMXasS}. At the same time we have by~\eqref{eq:823gtgCGGFK}
\begin{equation*}
\begin{split}
\sum_{a\in Q} \left|\nu_{\ell-1, j}(\wh{F}(a))\right|&\geq \sum_{a\in Q} (\ln 2-C/K)^j\left|\wh{F}(a)\right|\\
&=(\ln 2-C/K)^j \sum_{a\in Q}\left|\wh{F}(a)\right|\\
&=(\ln 2-C/K)^j \sum_{a\in Q}\left|F(a)\right|\\
&=(\ln 2-C/K)^j \left|\rho_s(\downset_s(T\setminus T_*))\right|\\
&= (\ln 2-C/K)^j \frac1{K-s}|T_k\setminus T_{k+1}|.
\end{split}
\end{equation*}
and 
\begin{equation*}
\begin{split}
\sum_{a\in Q} \left|\nu_{\ell-1, j}(\wh{F}(a))\right|&\leq \sum_{a\in Q} (\ln 2)^j\left|\wh{F}(a)\right|\\
&=(\ln 2)^j \sum_{a\in Q}\left|\wh{F}(a)\right|\\
&=(\ln 2)^j \sum_{a\in Q}\left|F(a)\right|\\
&=(\ln 2)^j \left|\rho_s(\downset_s(T_k\setminus T_{k+1}))\right|\\
&= (\ln 2)^j \frac1{K-s}|T_k\setminus T_{k+1}|.
\end{split}
\end{equation*}
In both cases above the last transition uses the fact that 
$$
\left|\rho_s(\downset_s(T_k\setminus T_{k+1}))\right|=\frac1{K-s}|T_k\setminus T_{k+1}|,
$$
which follows by noting that $T_k\setminus T_{k+1}$ does not depend on $q_s$ (by Property~\ref{prop:q-k}) and using Lemma~\ref{lm:densification}.  Putting the above bounds together with~\eqref{eq:iehfihefh9efKPCO} gives
\begin{equation*}
(\ln 2-C/K)^j \frac1{K-s}|T_k\setminus T_{k+1}| \leq |\nu_{\ell-1, j}(\tau_s(\downset_s(T_k\setminus T_{k+1})))| \leq (\ln 2)^j \frac1{K-s}|T_k\setminus T_{k+1}|
\end{equation*}

We now get by~\eqref{eq:823gt8g8GDGNKNFx}
\begin{equation*}
\begin{split}
|\mu_{\ell, j}(T_k\setminus T_{k+1})|&=\sum_{s=0}^k |\nu_{\ell-1, j}(\tau_s(\downset_s(T_k\setminus T_{k+1})))|\\
&\geq (\ln 2-C/K)^j \left(\sum_{s=0}^k \frac1{K-s}\right)|T_k\setminus T_{k+1}|
\end{split}
\end{equation*}
and 
\begin{equation*}
\begin{split}
|\mu_{\ell, j}(T_k\setminus T_{k+1})|&=\sum_{s=0}^k |\nu_{\ell-1, j}(\tau_s(\downset_s(T_k\setminus T_{k+1})))|\\
&\leq (\ln 2)^j \left(\sum_{s=0}^k \frac1{K-s}\right)|T_k\setminus T_{k+1}|
\end{split}
\end{equation*}
Now using~\eqref{eq:923h92ht12Wgugf} and the fact that
$$
|T_k\setminus T_{k+1}|=\left(1-\frac{k}{K}\right)\cdot |T_0^\ell|-\left(1-\frac{k+1}{K}\right)\cdot |T_0^\ell|=\frac1{K}|T_0^\ell|
$$
 for every $k\in [K/2]$ by Lemma~\ref{lm:size-bounds}, {\bf (1)},  we get
\begin{equation}\label{eq:92g8g8nnfifiwebfninICN}
\begin{split}
(\ln 2-C/K)^j \gamma |T_0| \leq |\mu_{\ell, j}(T\setminus T_*)|\leq (\ln 2)^j \gamma |T_0|
\end{split}
\end{equation}
for 
$$
\gamma=\frac1{K} \sum_{k\in [K/2]} \left(\sum_{s=0}^k \frac1{K-s}\right).
$$

Finally, we note that
\begin{equation*}
\begin{split}
\frac1{K}\sum_{k\in [K/2]} \left(\sum_{s=0}^k \frac1{K-s}\right)&=\sum_{s=0}^{K/2-1} \sum_{k=s}^{K/2-1} \frac1{K-s}\\
&=\frac1{K}\sum_{s=0}^{K/2-1} \frac{K/2-s}{K-s}\\
&=\frac1{K}\sum_{s=0}^{K/2-1} \left(1-\frac{K/2}{K-s}\right)\\
&=\frac1{2}-\frac1{2}\sum_{s=0}^{K/2-1} \frac1{K-s}, 
\end{split}
\end{equation*}
and thus by Claim~\ref{cl:sum-int}
$$
\frac1{2} (1-\ln 2)\leq \gamma\leq \frac1{2}(1-\ln 2+1/K).
$$
Combining this with~\eqref{eq:92g8g8nnfifiwebfninICN} gives 
$$
(\ln 2-C/K)^j \frac1{2}(1-\ln 2) |T_0| \leq |\mu_{\ell, j}(T\setminus T_*)|\leq (\ln 2+C/K)^j \frac1{2}(1-\ln 2) |T_0|
$$
as required.
\end{proofof}

\subsubsection{Proof of key structural property (Lemma~\ref{lm:cut-structure})}\label{sec:cut-structure}
We now present 

\begin{proofof}{Lemma~\ref{lm:cut-structure}} Our proof is by induction on $j$. The inductive claim is

\vspace{0.05in}

\fbox{
\begin{minipage}{0.99\textwidth}
For every $\ell \in [L]$, for every $x\in \nu_{\ell+j, j}(T^{\ell+j}\setminus T_*^{\ell+j})\subseteq T^\ell$, every $y\in T^\ell$, if $y_i=x_i$ for all $i\in \Gamma$ (see Definition~\ref{def:basic-coordinates}),  then  $y\in \nu_{\ell+j, j}(T^{\ell+j}\setminus T_*^{\ell+j})$.
\end{minipage}}

\vspace{0.05in}

\noindent{\bf Base: $j=0$.} Recall that $\nu_{\ell, 0}$ is the identity map. Letting $J:=J^\ell$, we have
\begin{equation*}
T^\ell\setminus T_*^\ell=\left\{z\in T^\ell: z_{J_k}/m\in \left[1-\frac1{K-k}, 1\right)\text{~for some~}k\in [K/2+1]\right\}.
\end{equation*}
Let $k\in [K/2+1]$ be such that $x_{J_k}/m\in \left[1-\frac1{K-k}, 1\right)$. Since $y_i=x_i$ for all $i\in \Gamma$, and in particular for $i\in \Sp(\B^\ell)$ (which includes $J_0,\ldots, J_{K/2}$ and in particular $J_{k}$), we get $y_{J_k}/m\in \left[1-\frac1{K-k}, 1\right)$ and therefore $y\in T^\ell\setminus T_*^\ell=\nu_{\ell, 0}(T^{\ell}\setminus T_*^{\ell})$ as required.

\noindent{\bf Inductive step: $j-1\to j$.} By Lemma~\ref{lm:tau-inverse-prop}  there exists $a\in [K/2]$ as well as $u, v\in S^{\ell+1}_a$ such that $x= \tau^{\ell+1}(u)$, $y=\tau^{\ell+1}(v)$ and $u_\Gamma=v_\Gamma$ (the set of basic coordinates as per Definition~\ref{def:basic-coordinates}). Let $x'\in T^{\ell+1}, y'\in T^{\ell+1}$ be such that $x'\delequal u$ and $y'\delequal v$, and note that $x'_\Gamma=y'_\Gamma$. Now recall that by Definition~\ref{def:nu}
$$
\nu_{\ell+j, j}(T^{\ell+j}\setminus T_*^{\ell+j})=\tau^{\ell+1}(\downset^{\ell+1}(\nu_{(\ell+1)+(j-1), j-1}(T^{(\ell+1)+(j-1)}\setminus T_*^{(\ell+1)+(j-1)}))).
$$
Since $x\in \nu_{\ell+j, j}(T^{\ell+j}\setminus T_*^{\ell+j})$ by assumption, we get that 
$$
x'\in \nu_{(\ell+1)+(j-1), j-1}(T^{(\ell+1)+(j-1)}\setminus T_*^{(\ell+1)+(j-1)}),
$$
and therefore by the inductive hypothesis, using the fact that $x'_\Gamma=y'_\Gamma$, we get 
$$
y'\in \nu_{(\ell+1)+(j-1), j-1}(T^{(\ell+1)+(j-1)}\setminus T_*^{(\ell+1)+(j-1)}).
$$
As a consequence $y\in \tau^{\ell+1}(\downset^{\ell+1}(\{y'\}))\subseteq \nu_{\ell+j, j}(T^{\ell+j}\setminus T_*^{\ell+j})$, as required.
\end{proofof}

We also give

\begin{proofof}{Corollary~\ref{cor:cut-structure}}
Let $y'\in T^\ell$ be such that $y\delequal y'$ -- such a $y'$ exists by definition of $S^\ell$, and note that $y'_\Gamma=x_\Gamma$ since $y_\Gamma=x_\Gamma$ by assumption of the corollary. We have 
$$
y'\in \nu_{\ell+j, j}(T^{\ell+j}\setminus T_*^{\ell+j})
$$
by Lemma~\ref{lm:cut-structure}. Since $\mu_{\ell+j, j}(T^{\ell+j}\setminus T_*^{\ell+j})=\downset^\ell(\nu_{\ell+j, j}(T^{\ell+j}\setminus T_*^{\ell+j}))$, we get $y\in \mu_{\ell+j, j}(T^{\ell+j}\setminus T_*^{\ell+j})$, as required.
\end{proofof}

\subsection{Proof of Theorem~\ref{thm:main-simple}}\label{sec:simple-lb}
We now define the hard input distribution $\mathcal{D}$ on graphs $\wh{G}=(P, Q, \wh{E})$. A graph $\wh{G} \sim \mathcal{D}$ is sampled as follows. First, for every round $\ell\in [L]$ and phase $k\in [K/2]$ one {\bf arbitrarily} selects 
\begin{enumerate}
\item the extension indices $\Ext^\ell_k$ from $\B^\ell_k$;
\item a compression index $q^\ell_k$ in $\B^\ell_k\setminus \Ext^\ell_k$.
\end{enumerate}
One also selects $r^\ell\in \B^\ell_{K/2}$ arbitrarily.
Recall that for $k\in [K/2]$ we let (see Definition~\ref{def:cext})
$$
\ac{\B}^\ell_k=\B^\ell_k\setminus (\Ext^\ell_k\cup \{q^\ell_k\})
$$
and $\ac{\B}^\ell_{K/2}=\B^\ell_{K/2}\setminus \{r^\ell\}$.
Finally, one selects, for every $\ell\in [L]$ and $k\in [K/2]$,
$$
J^\ell_k\sim UNIF(\ac{\B}^\ell_k)
$$
independently.

\paragraph{Edge set of $\wh{G}=(P, Q, \wh{E})$.} We first define
\begin{equation}\label{eq:tau-star-def}
\tau_*(x)=\left\lbrace
\begin{array}{ll}
\tau^\ell(x)&\text{~if~}x\in S^\ell\text{~for~}\ell>0\\
x&\text{o.w.}
\end{array}
\right.
\end{equation}
and define for every edge $e=(u, v)\in E^\ell, u\in S^\ell, v\in T^\ell,\ell\in [L]$
\begin{equation}\label{eq:tau-of-edge}
\tau_*(e)=(\tau_*(u), v).
\end{equation}

We now let 
\begin{equation}\label{eq:g-hat-edges}
\wh{E}=\bigcup_{\ell\in [L]} \wh{E}^\ell,
\end{equation}
where 
\begin{equation}\label{eq:g-hat-edges-ell}
\wh{E}^\ell=\bigcup_{k\in [K/2]}\bigcup_{j\in \ac{\B}^\ell_k} \tau_*(E^\ell_{k, j}),
\end{equation}
and $E^\ell_{k, j}$ is defined by~\eqref{eq:edges-ei-def}.

\paragraph{Ordering of edges of $\wh{G}$ in the stream.} The graph $G^\ell=(S^\ell, T^\ell, E^\ell)$ is presented in the stream over $L$ {\em rounds} and $K/2$ {\em phases} as follows. For every $\ell\in \{1, \ldots, L-1\}$, for every $k\in [K/2]$, the edges in $\tau^\ell(E^\ell_k)$ are presented in the stream; the ordering within $\tau^\ell(E^\ell_k)$ is arbitrary.

We have 
\begin{lemma}\label{lm:large-matching-ghat}
The graph $\wh{G}=(P, Q, \wh{E})$ contains a matching of size $(1-O(1/L))|P|$.
\end{lemma}
\begin{proof}
By Lemma~\ref{lm:matching} for every $\ell\in [L],\ell>0$ there exists a matching $M^\ell$ in $E^\ell$ that matches a $(1-O(1/K))$ fraction of $S^\ell$ to $T^\ell\setminus T_*^\ell$.
Since $\tau^\ell$ is injective by Lemma~\ref{lm:tau-prop}, we have that $\tau^\ell(M^\ell)$ is also a matching. Furthermore, since $\tau^\ell$ maps $S^\ell$ to $T_*^{\ell-1}$, avoiding vertices in $T^{\ell-1}\setminus T_*^{\ell-1}$, which may be matched by $\tau^{\ell-1}(M^{\ell-1}
)$, we have that the union of edges
$$
\bigcup_{\ell\in [L], \ell>0} \tau^\ell(M^\ell)
$$
forms a matching. For every $\ell$ we have $|M^\ell|=(1-O(1/K))|S^\ell|$, and by Lemma~\ref{lm:size-bounds}, {\bf (2)}, one has $|S^\ell|=\sum_{k\in [K/2]} |S^\ell_k|=\frac1{2}|T^\ell|=\frac1{2}N$. Since by Lemma~\ref{lm:size-bounds}, {\bf (1)}, with $k=K/2$ one has $|T^\ell_k|=\frac1{2}|T^\ell|$, we have by~\eqref{eq:p-def} 
\begin{equation*}
\begin{split}
|P|&=\left|\left(\bigcup_{\text{even~}\ell\in [L]} T^\ell\right)\right|\\
&=\sum_{\text{even~}\ell\in [L]} |T^\ell|\\
&=(L/2)\cdot N\\
&=L\cdot N/2.
\end{split}
\end{equation*}
This means that $\bigcup_{\ell\in [L], \ell>0} \tau^\ell(M^\ell)$ is a matching of size $(L-1)\cdot (1-O(1/K))\cdot N/2=(1-O(1/L)) |P|$, since $L\leq K$ by~\ref{p2}.
\end{proof}

\paragraph{Upper bounding size of matching constructed by a low space algorithm.}
The following sets of vertices are hard to match well, as we show below: 
\begin{equation}\label{eq:apm-def}
\begin{split}
A_P&=\bigcup_{\substack{\ell\in [L]\\ \ell \text{~even}}} \nu_{\ell, *}(T^\ell\setminus T_*^\ell)\\
A_Q&=\bigcup_{\substack{\ell\in [L]\\ \ell\text{~odd}}} \nu_{\ell, *}(T^\ell\setminus T_*^\ell).
\end{split}
\end{equation}
To show that $A_P$ and $A_Q$ are hard to match well, we show that the subset of edges of $G$ retained by a small space generalized online algorithm typically admits a small vertex cover that avoids $A_P$ and $A_Q$. The two sets below (and some other vertices that contribute lower order terms to the size of the vertex cover) will be included:
\begin{equation}\label{eq:bpm-def}
\begin{split}
B_Q&=\bigcup_{\substack{\ell\in [L]\\ \ell \text{~even}}} \tau_*(\mu_{\ell, *}(T^\ell\setminus T_*^\ell))\\
B_P&=\bigcup_{\substack{\ell\in [L]\\ \ell\text{~odd}}} \tau_*(\mu_{\ell, *}(T^\ell\setminus T_*^\ell)).
\end{split}
\end{equation}

We have 
\begin{claim}\label{cl:ab-disjoint}
$A_P\cap B_P=\emptyset$ and $A_Q\cap B_Q=\emptyset$.
\end{claim}
\begin{proof}
We prove the first claim (the proof of the second is analogous).
One has by ~\eqref{eq:apm-def}
\begin{equation}\label{eq:9239ggdsgpffdf}
\begin{split}
A_P&=\bigcup_{\substack{\ell\in [L]\\ \ell \text{~even}}} \nu_{\ell, *}(T^\ell\setminus T_*^\ell)\\
&=\bigcup_{\substack{\ell\in [L]\\ \ell \text{~even}}} \bigcup_{\substack{j=0\\j\text{~even}}}^\ell \nu_{\ell, j}(T^\ell\setminus T_*^\ell)
\end{split}
\end{equation}
and by~\eqref{eq:bpm-def}
\begin{equation}\label{eq:9239ggdsgpffdfsf}
\begin{split}
B_P&=\bigcup_{\substack{\ell\in [L]\\ \ell\text{~odd}}} \tau_*(\mu_{\ell, *}(T^\ell\setminus T_*^\ell))\\
&=\bigcup_{\substack{\ell\in [L]\\ \ell\text{~odd}}} \bigcup_{\substack{j=0\\j\text{~even}}}^\ell\tau_*(\mu_{\ell, j}(T^\ell\setminus T_*^\ell))\\
&=\bigcup_{\substack{\ell\in [L]\\ \ell\text{~odd}}} \bigcup_{\substack{j=0\\j\text{~even}}}^\ell\tau^{\ell-j}(\mu_{\ell, j}(T^\ell\setminus T_*^\ell))\\
&=\bigcup_{\substack{\ell\in [L]\\ \ell\text{~odd}}} \bigcup_{\substack{j=0\\j\text{~even}}}^\ell\nu_{\ell, j+1}(T^\ell\setminus T_*^\ell),\\
\end{split}
\end{equation}
where we used the definition of $\tau_*$ (see~\eqref{eq:tau-star-def}) in the third transition and Definition~\ref{def:nu} in the forth transition. Disjointness now follows by Lemma~\ref{lm:nu-prop}, {\bf (2)}, since the range of $(\ell, j)$ pairs in~\eqref{eq:9239ggdsgpffdf} is disjoint from the range of $(\ell, j+1)$ pairs in~\eqref{eq:9239ggdsgpffdfsf}.
\end{proof}

Before exhibiting the vertex cover, we show that $A_P\cup B_P$ is almost all of $P$, and $A_Q\cup B_Q$ is almost all of $Q$:
\begin{lemma}[Almost partition of $P$ and $Q$]\label{lm:partition}
One has $|P\setminus (A_P\cup B_P)|=O(N)$ and  $|Q\setminus (A_Q\cup B_Q)|=O(N)$ for sets $A_P, A_Q, B_P, B_Q$ defined in~\eqref{eq:apm-def} and~\eqref{eq:bpm-def}.
\end{lemma}
\begin{proof} 
Recall that by~\eqref{eq:p-def} and~\eqref{eq:q-def}  $P\cup Q=S^0\cup \bigcup_{\ell\geq 0} T^\ell$. We have by Lemma~\ref{lm:t-star-ell-rec}
\begin{equation*}
T_*^\ell=\nu_{L-1, L-1-\ell}(T_*^{L-1})\cup \bigcup_{j=1}^{L-1-\ell} \nu_{\ell+j, j}(T^{\ell+j}\setminus T_*^{\ell+j}).
\end{equation*}

Putting these two equalities together, and letting $D=\bigcup_{j=0}^{L-1} \nu_{L-1, L-1-\ell}(T_*^{L-1})$  to simplify notation, we get
\begin{equation}\label{eq:49h9g4g}
\begin{split}
P\cup Q&=S^0\cup \bigcup_{\ell\geq 0} T^\ell\\
&=S^0\cup D\cup \left( \bigcup_{\ell\geq 0} \bigcup_{\substack{j\geq 0}} \nu_{\ell+j, j}(T^{\ell+j}\setminus T_*^{\ell+j})\right)\\
&=S^0\cup D\cup \bigcup_{\ell=0}^{L-1} \bigcup_{\substack{j=0}}^\ell \nu_{\ell, j}(T^\ell\setminus T_*^\ell)\\
&=S^0\cup D\cup  \left(\bigcup_{\substack{\ell\geq 0\\ \ell\text{~even}}} \bigcup_{\substack{j=0}}^\ell \nu_{\ell, j}(T^\ell\setminus T_*^\ell)\right)\cup \left(\bigcup_{\substack{\ell\geq 0\\ \ell\text{~odd}}} \bigcup_{\substack{j=0}}^\ell \nu_{\ell, j}(T^\ell\setminus T_*^\ell)\right)\\
\end{split}
\end{equation}
Note that it follows from Corollary~\ref{cor:nu-star-ub} that $|D|=O(N)$. Indeed, 
\begin{equation*}
\begin{split}
|D|&=\left|\bigcup_{j=0}^{L-1} \nu_{L-1, L-1-\ell}(T_*^{L-1})\right|\\
&\leq \sum_{j\geq 0} |\nu_{\ell, j}(T_*^{L-1})|\\
&\leq \sum_{j\geq 0} (\ln 2)^j |T_*^{L-1}|\\
&=\frac1{1-\ln 2}|T_*^{L-1}|\\
&=\frac1{2(1-\ln 2)}N\\
&=O(N).
\end{split}
\end{equation*}
Thus, since $|S^0|=\sum_{k\in [K/2]} |S^0_k|=N/2$ by Lemma~\ref{lm:size-bounds}, {\bf (2)},  it suffices to show that the union of the third and forth terms above equals $A_P\cup A_Q\cup B_P\cup B_Q$. To that effect we note that for every $\ell=0,\ldots, L-1$ and $j=0,\ldots, \ell$
$$
\nu_{\ell, j+1}(T^\ell\setminus T_*^\ell)=\tau^{\ell-j}(\downset^{\ell-j}(\nu_{\ell, j}(T^\ell\setminus T_*^\ell))).
$$

This means that the third term on the last line of~\eqref{eq:49h9g4g} can be rewritten as
\begin{equation*}
\begin{split}
&\bigcup_{\substack{\ell\geq 0\\ \ell\text{~even}}} \bigcup_{\substack{j=0}}^\ell \nu_{\ell, j}(T^\ell\setminus T_*^\ell)\\
&= \bigcup_{\substack{\ell\geq 0\\ \ell\text{~even}}} \bigcup_{\substack{j=0\\j~\text{even}}}^\ell \left(\nu_{\ell, j}(T^\ell\setminus T_*^\ell)\cup \tau^{\ell-j}(\downset^{\ell-j}(\nu_{\ell, j}(T^\ell\setminus T_*^\ell)))\right)\\
&= \bigcup_{\substack{\ell\geq 0\\ \ell\text{~even}}} \bigcup_{\substack{j=0\\j~\text{even}}}^\ell \left(\nu_{\ell, j}(T^\ell\setminus T_*^\ell)\cup \tau_*(\downset^{\ell-j}(\nu_{\ell, j}(T^\ell\setminus T_*^\ell)))\right)\\
&= \left(\bigcup_{\substack{\ell\geq 0\\ \ell\text{~even}}} \nu_{\ell, *}(T^\ell\setminus T_*^\ell)\right)\cup \left(\bigcup_{\substack{\ell\geq 0\\ \ell\text{~even}}} \tau_*(\mu_{\ell, *}(T^\ell\setminus T_*^\ell))\right)\\
&=A_P\cup B_Q,\\
\end{split}
\end{equation*}
where $\tau_*$ is as defined in~\eqref{eq:tau-star-def}, and we let $\tau^0(U)=\emptyset$ for every $U$ for convenience to simplify notation.
Similarly, we get for the forth term on the last line of~\eqref{eq:49h9g4g}
\begin{equation*}
\begin{split}
&\bigcup_{\substack{\ell\geq 0\\ \ell\text{~odd}}} \bigcup_{\substack{j=0}}^\ell \nu_{\ell, j}(T^\ell\setminus T_*^\ell)\\
&= \left(\bigcup_{\substack{\ell\geq 0\\ \ell\text{~odd}}} \bigcup_{\substack{j=0\\j~\text{even}}}^\ell \nu_{\ell, j}(T^\ell\setminus T_*^\ell)\right)\cup \left(\bigcup_{\substack{\ell\geq 0\\ \ell\text{~odd}}} \bigcup_{\substack{j=0\\j~\text{even}}}^\ell \tau^{\ell-j}(\downset^{\ell-j}(\nu_{\ell, j}(T^\ell\setminus T_*^\ell)))\right)\\
&=\left(\bigcup_{\substack{\ell\geq 0\\ \ell\text{~odd}}} \nu_{\ell, *}(T^\ell\setminus T_*^\ell)\right)\cup \left(\bigcup_{\substack{\ell\geq 0\\ \ell\text{~odd}}} \tau_*(\mu_{\ell, *}(T^\ell\setminus T_*^\ell))\right)\\
&=\left(\bigcup_{\substack{\ell\geq 0\\ \ell\text{~odd}}} \nu_{\ell, *}(T^\ell\setminus T_*^\ell)\right)\cup \left(\bigcup_{\substack{\ell\geq 0\\ \ell\text{~odd}}} \tau_*(\mu_{\ell, *}(T^\ell\setminus T_*^\ell))\right)\\
&=A_Q\cup B_P,
\end{split}
\end{equation*}
as required. 
\end{proof}

The next lemma upper bounds the cardinality of $B_P$ and $B_Q$, which later leads to our upper bound on the size of the constructed vertex cover.
\begin{lemma}\label{lm:sizeof-ab}
One has 
\begin{equation*}
\begin{split}
|B_P|&\leq (1+O(1/L))\cdot \frac{L}{2}\cdot \frac{N}{2}\cdot\frac1{1+\ln 2}\\
&\text{and}\\
|B_Q|&\leq (1+O(1/L))\cdot \frac{L}{2}\cdot \frac{N}{2}\cdot\frac1{1+\ln 2}.
\end{split}
\end{equation*}
\end{lemma}
\begin{proof} We prove the bound for $B_Q$ (the bound for $B_P$ is analogous). Using~\eqref{eq:bpm-def} we get
\begin{equation*}
\begin{split}
|B_Q|&=\left|\bigcup_{\substack{\ell\in [L]\\ \ell \text{~even}}} \tau_*(\mu_{\ell, *}(T^\ell\setminus T_*^\ell))\right|\\
&\leq \left|\bigcup_{\substack{\ell\in [L]\\ \ell \text{~even}}} \mu_{\ell, *}(T^\ell\setminus T_*^\ell)\right|\\
&\leq \sum_{\substack{\ell\in [L]\\ \ell \text{~even}}} |\mu_{\ell, *}(T^\ell\setminus T_*^\ell)|,\\
\end{split}
\end{equation*}
so it suffices to upper bound the summands above. For every $\ell \in [L]$ by Definition~\ref{def:nu}
\begin{equation}\label{eq:923yt8hihgfygFYGSFG}
\left|\mu_{\ell, *}(T^\ell\setminus T_*^\ell)\right|=\left|\bigcup_{\substack{0\leq j\leq \ell\\j\text~{even}}} \mu_{\ell, j}(T^\ell\setminus T_*^\ell)\right|\leq \bigcup_{\substack{0\leq j \leq \ell\\j\text~{even}}} \left|\mu_{\ell, j}(T^\ell\setminus T_*^\ell)\right|
\end{equation}
and by Lemma~\ref{lm:mu-ell-j} we have for an absolute constant $C>0$
$$
|\mu_{\ell, j}(T^\ell\setminus T_*^\ell)|\leq  \frac1{2}(\ln 2+C/K)^j (1-\ln 2)|T^\ell|.
$$
Summing over all even $j$ as per~\eqref{eq:923yt8hihgfygFYGSFG}, we get 
\begin{equation*}
\begin{split}
\left|\mu_{\ell, *}(T^\ell\setminus T_*^\ell)\right|&\leq \sum_{\substack{0\leq j\leq \ell\\j\text~{even}}} \frac1{2}(\ln 2+C/K)^j (1-\ln 2)|T^\ell|\\
&\leq \frac1{2}(1-\ln 2)|T^\ell| \sum_{\substack{j\geq 0}} (\ln 2+C/K)^{2j}\\
&=\frac1{2}(1-\ln 2)|T^\ell| \frac1{1-(\ln 2+C/K)^2}\\
&\leq (1+O(1/K))\frac1{2}(1-\ln 2)|T^\ell| \frac1{1-(\ln 2)^2}\\
&= (1+O(1/K))\frac1{2}\frac1{1+\ln 2}|T^\ell|.
\end{split}
\end{equation*}

Summing the above over all even $\ell\in [L]$ and recalling that $|T^\ell|=N$ and using the fact that $L\leq K$ by~\ref{p2} gives the required bound.
\end{proof}

\begin{lemma}\label{lm:vertex-cover}
For every matching $M$ in $G$ one has 
$$
|M|\leq |M\cap (A_P\times (Q\setminus B_Q))|+\frac1{1+\ln 2}|P|+O(|P|/L).
$$
\end{lemma}
\begin{proof}
We exhibit a vertex cover of appropriate size for $M$. Specifically, we add to the vertex cover one endpoint of every edge in
$$
M\cap (A_P\times (Q\setminus B_Q)),
$$
as well as all vertices in $P\setminus A_P\approx B_P$ and $B_Q$. Note that this is indeed a vertex cover: $A_P\cap B_P=\emptyset$ and $A_Q\cap B_Q=\emptyset$ by Claim~\ref{cl:ab-disjoint}, so every edge of $M$ either has an endpoint in $P\setminus A_P$, or belongs to $A_P\times (Q\setminus B_Q)$, or belongs to $A_P\times B_Q$, in which case it has an endpoint in $B_Q$.

The size of the vertex cover is 
\begin{equation}\label{eq:23t77Gy8uGBBYVF}
\begin{split}
&|M\cap (A_P\times (Q\setminus B_Q))|+|P\setminus A_P|+|B_Q|\\
\leq &|M\cap (A_P\times (Q\setminus B_Q))|+|B_P|+|B_Q|+O(N),
\end{split}
\end{equation}
where we used Lemma~\ref{lm:partition} to conclude that 
$$
|P\setminus A_P|\leq |B_P|+|P\setminus (A_P\cup B_P)|=|B_P|+O(N).
$$ 

By Lemma~\ref{lm:sizeof-ab} we have 
\begin{equation*}
\begin{split}
|B_P|&\leq \frac{L}{2}\cdot \frac{N}{2}\frac1{1+\ln 2}(1+O(1/L))\\
&\text{and}\\
|B_Q|&\leq \frac{L}{2}\cdot \frac{N}{2}\frac1{1+\ln 2}(1+O(1/L)).
\end{split}
\end{equation*}
Putting the above together with~\eqref{eq:23t77Gy8uGBBYVF} and recalling that by~\eqref{eq:p-def}
$$
|P|=\left|\bigcup_{\text{even~}\ell\in [L]} T^\ell\right|=L\cdot N/2
$$
gives the result.
\end{proof}

We now prove 
\begin{lemma}\label{lm:special-edges}
For every matching $M\subseteq \wh{E}$ one has 
$$
M\cap (A_P\times (Q\setminus B_Q))\subseteq \bigcup_{\ell\in [L], k\in [K/2]} \tau^\ell(E^\ell_{k, J^\ell_k}).
$$
\end{lemma}

\begin{proofof}{Lemma~\ref{lm:special-edges}} Consider an edge $(u, v)\in\wh{E}$ such that
$u\in A_P, v\in Q\setminus B_Q$. Let $\ell\in [L]$ be an even integer such that  $u\in T^\ell$. Such an $\ell$ exists because by~\eqref{eq:apm-def} one has
\begin{equation*}
A_P=\bigcup_{\substack{\ell\in [L]\\ \ell \text{~even}}} \nu_{\ell, *}(T^\ell\setminus T_*^\ell)\\
\end{equation*}
and by Definition~\ref{def:nu} one has 
$$
\nu_{\ell, *}(T^\ell\setminus T_*^\ell):=\bigcup_{\substack{i=0\\i \text{~even}}}^\ell \nu_{\ell, i}(T^\ell\setminus T_*^\ell),
$$
so that 
$$
\nu_{\ell, *}(T^\ell\setminus T_*^\ell)\subseteq \bigcup_{\text{even~}\ell\in [L]} T_*^\ell.
$$
To summarize, we have
\begin{equation}\label{eq:u-apm}
u\in \nu_{\ell+i, i}(T^{\ell+i}\setminus T_*^{\ell+i})
\end{equation}
for a unique choice of even
 $\ell\in [L]$ and even $i$ (uniqueness follows by Lemma~\ref{lm:nu-prop}, {\bf (2)}). We now consider two cases: depending on whether $v\in T^{\ell-1}$ ({\bf case 1}) or $v\in T^{\ell+1}$ ({\bf case 2}).

\paragraph{\bf Case 1.} In this case there exists a unique $y\in S^\ell$ such that $\tau^\ell(y)=v$. Indeed, otherwise the edge $(u, v)$ would not be in the graph $\wh{G}$ as per~\eqref{eq:g-hat-edges}.  Let $x=u$ for convenience. We now show using Corollary~\ref{cor:cut-structure}  that 
$$
y\in \mu_{\ell+i, i}(T^{\ell+i}\setminus T_*^{\ell+i}),
$$
which implies, by~\eqref{eq:bpm-def} together with the definition of $\mu_{\ell, *}$ (Definition~\ref{def:nu}), that $v=\tau^\ell(y)\in B_Q$. We first verify that preconditions of Corollary~\ref{cor:cut-structure} are satisfied. Let $k\in [K/2]$ be the unique index such that both $x\in T^\ell_k$ and $y\in S^\ell_k$ (uniqueness follows since $(x, y)\in E^\ell_k$ due to $(u, v)\in \wh{E}$, and the edge sets in~\eqref{eq:edges-ei-def} are disjoint by Lemma~\ref{lm:disjoint}). We have $(u, v)\in \wh{E}$ by assumption, which means that $(x, y)\in E^\ell_k$, and therefore $y_i=x_i$ for all $i\in [n], i\neq j$ for some $j\in \ac{\B}^\ell_k$ by~\eqref{eq:line-def} and~\eqref{eq:edges-ei-def}. We assume towards a contradiction that $j\neq J^\ell_k$. Since 
$$
\ac{\B}^\ell_k\cap \Gamma=\{J^\ell_k\},
$$
we thus get that $x_\Gamma=y_\Gamma$, and preconditions of Corollary~\ref{cor:cut-structure} are indeed satisfied. We thus get that
$x\in \nu_{\ell+i, i}(T^{\ell+i}\setminus T_*^{\ell+i})$,
implies
$y\in \mu_{\ell+i, i}(T^{\ell+i}\setminus T_*^{\ell+i}).$ 
At the same time by Definition~\ref{def:nu} for every $\ell\in [L]$ 
$$
\mu_{\ell, *}(T^\ell\setminus T_*^\ell):=\bigcup_{\substack{i=0\\i \text{~even}}}^\ell \mu_{\ell, i}(T^\ell\setminus T_*^\ell),
$$
which means that $y\in \mu_{\ell, *}(T^\ell\setminus T_*^\ell)$ (recall that $i$ is even) and thus $v=\tau^\ell(y)\in \tau_*(\mu_{\ell, *}(T^\ell\setminus T_*^\ell)) \subseteq B_Q$, as required.

\paragraph{Case 2.} In this case there exists a unique $x'\in S^{\ell+1}$ such that $\tau^{\ell+1}(x')=u$. Indeed, otherwise the edge $(u, v)$ would not be in the graph $\wh{G}$ as per~\eqref{eq:g-hat-edges}; uniqueness follows by injectivity of $\tau^\ell$ (by Lemma~\ref{lm:tau-prop}).  Let $y=v$. Let $k\in [K/2]$ be the unique index such that both $x'\in S^{\ell+1}_k$ and $y\in T^{\ell+1}_k$ (uniqueness follows since $(x', y)\in E^{\ell+1}_k$ due to $(u, v)\in \wh{E}$, and the edge sets in~\eqref{eq:edges-ei-def} are disjoint by Lemma~\ref{lm:disjoint}).  Let $x\in T^{\ell+1}$ be such that $x\delequal x'$ -- such a vertex exists by definition of $S^{\ell+1}$ -- see~\eqref{eq:def-s}.

We have $(u, v)\in \wh{E}$ by assumption, which means that $(x, y)\in E^{\ell+1}$, and therefore $y_i=x_i$ for all $i\in [n], i\neq j$ for some $j\in \ac{\B}^{\ell+1}_k$ by~\eqref{eq:line-def} and~\eqref{eq:edges-ei-def}.  We assume towards a contradiction that $j\neq J^{\ell+1}_k$. Since 
$$
\ac{\B}^{\ell+1}_k\cap \Gamma=\{J^{\ell+1}_k\},
$$
we thus get that $x_\Gamma=y_\Gamma$, and we can apply Lemma~\ref{lm:cut-structure} to $x$ and $y$. By~\eqref{eq:u-apm} we have
\begin{equation*}
\begin{split}
u&=\tau^{\ell+1}(x')\\
&\in \nu_{\ell+i, i}(T^{\ell+i}\setminus T_*^{\ell+i})\\
&\subseteq \tau^{\ell+1}\left(\mu_{\ell+i, i-1}(T^{\ell+i}\setminus T_*^{\ell+i})\right),
\end{split}
\end{equation*}
and therefore 
$$
x'\in \mu_{\ell+i, i-1}(T^{\ell+i}\setminus T_*^{\ell+i}).
$$
Since 
$$
\mu_{\ell+i, i-1}(T^{\ell+i}\setminus T_*^{\ell+i})=\downset^{\ell+1}(\nu_{\ell+i, i-1}(T^{\ell+i}\setminus T_*^{\ell+i})),
$$
we have 
\begin{equation*}
\begin{split}
x&\in \nu_{\ell+i, i-1}(T^{\ell+i}\setminus T_*^{\ell+i})\\
&= \nu_{(\ell+1)+(i-1), i-1}(T^{(\ell+1)+(i-1)}\setminus T_*^{(\ell+1)+(i-1)})\\
\end{split}
\end{equation*}
By Lemma~\ref{lm:cut-structure}\footnote{Note that we are applying the lemma with $\ell+1$ as opposed to $\ell$ here, since $x, y\in T^{\ell+1}$.} we thus have \footnote{When $\ell+1=L-1$, we have $\ell+2=L$, which does not technically correspond to a gadget in our input graph. However, we think of artifically adding such a gadget here to handle this corner case for simplicity.}
\begin{equation*}
\begin{split}
y&\in \nu_{(\ell+1)+(i-1), i-1}(T^{(\ell+1)+(i-1)}\setminus T_*^{(\ell+1)+(i-1)})\\
&=\tau^{\ell+2}(\mu_{\ell+i, i-2}(T^{\ell+i}\setminus T_*^{\ell+i})).
\end{split}
\end{equation*}
At the same time by Definition~\ref{def:nu} for every $\ell\in [L]$ 
$$
\mu_{\ell, *}(T^\ell\setminus T_*^\ell):=\bigcup_{\substack{j=0\\j \text{~even}}}^\ell \mu_{\ell, j}(T^\ell\setminus T_*^\ell),
$$
which means that $y=v\in B_Q$, as required.
\end{proofof}

\begin{definition}[Ordering on $(\ell, k)$ pairs]
We write $(\ell', k')<(\ell, k)$ iff $\ell'<\ell$ or $\ell'=\ell$ but $k'<k$. We write $(\ell', k')\leq (\ell, k)$ iff $\ell'<\ell$ or $\ell'=\ell$ but $k'<k$. 
\end{definition}

\begin{definition}
For $\ell\in [L]$ and $k\in [K/2]$ we write 
$$
\wh{G}_{(\ell, k)}=(P, Q, \wh{E}^\ell_k),
$$ 
and write 
$$
\wh{G}_{\leq (\ell, k)}=\left(P, Q, \bigcup_{\substack{\ell'\in [L], k'\in [K/2]\\ (\ell', k')\leq (\ell, k)}} \wh{E}^{\ell'}_{k'}\right).
$$
\end{definition}

\begin{definition}
For every $\ell\in [L], k\in [K/2+1]$ define $\Lambda_{\ell, k}:=(J^\ell_k)$. We write $\Lambda_{<(\ell, k)}=\left(\Lambda_{\ell', k'}\right)_{(\ell', k')<(\ell, k)}$.
\end{definition}
Note that $\wh{G}_{\leq (\ell, k)}$ is fully determined by $\Lambda_{<(\ell, k)}$. Here it is important to note that the restriction of the map $\tau^\ell$ onto $S^\ell_{\leq k}$ is indeed determined by $\Lambda_{<(\ell, k)}$ -- see Remark~\ref{rm:tau-prefix}.

We prove 
\begin{theorem}\label{thm:main-simple-prime}
For any sufficiently large constant $K$, any generalized online algorithm ALG with space budget $s=o(|P|\log |P|)$ cannot output a matching $M_{ALG}$ satisfying 
$$
|M_{ALG}|\geq \left(\frac1{1+\ln 2}+O(1/K)\right) |M_{OPT}|
$$
with probability more than $1/10$.
\end{theorem}
\begin{proof} 
Since we are evaluating the performance of the algorithm with respect to a distribution, by Yao's minimax principle we may assume that ALG is deterministic. 

We have by Lemma~\ref{lm:vertex-cover} that the size of the maximum matching $M_{ALG}$ in $G$ is upper bounded by  
\begin{equation}\label{eq:m-alg-bound}
|M_{ALG}\cap A_P\times (Q\times B_Q)|+\left(\frac1{1+\ln 2}+O(1/K)\right)|P|,
\end{equation}
where we used the fact that $L=K$ as per~\ref{p2}.

Recall that in every round $\ell\in [L]$ and every phase $k\in [K/2]$ of round $\ell$ the algorithm is presented with edges in 
$$
\wh{E}^\ell_k=\bigcup_{j\in \ac{\B}^\ell_k} \wh{E}^\ell_{k, j},
$$
as per~\eqref{eq:g-hat-edges} and~\eqref{eq:g-hat-edges-ell}.  Let $\text{ALG}^\ell_k\subseteq \wh{E}^\ell_k$  denote the subset of $\wh{E}^\ell_k$ remembered by ALG (recall the definition of the generalized online model -- see Definition~\ref{def:gen-online-intro}). Note that since we are assuming that ALG is deterministic, the set $\text{ALG}^\ell_k$ is fully determined by $\Lambda_{<(\ell, k)}$ (which determines $\wh{G}_{\leq (\ell, k)}$). At the same time, recall that conditioned on $\Lambda_{<(\ell, k)}$, the index $J^\ell_k$ is uniformly random in $\ac{\B}^\ell_k$:
$$
J_k^\ell\sim UNIF(\ac{\B}^\ell_k).
$$
Thus, one has, for any $\Lambda_{<(\ell, k)}$,  
\begin{equation}\label{eq:expect-lk}
\begin{split}
\expect_{\wh{G}\sim \mathcal{D}}\left[|\text{ALG}^\ell_k\cap \wh{E}^\ell_{k, J^\ell_k}| |\Lambda_{<(\ell, k)}\right]&=\sum_{j\in \ac{\B}^\ell_k} |\text{ALG}^\ell_k\cap \wh{E}^\ell_{k, j}|\cdot \prob[J^\ell_k=j|\Lambda_{<(\ell, k)}]\\
&=\frac1{|\ac{\B}^\ell_k|}\sum_{j\in \ac{\B}^\ell_k} |\text{ALG}^\ell_k\cap \wh{E}^\ell_{k, j}|\\
&=\frac1{|\ac{\B}^\ell_k|}|\text{ALG}^\ell_k|\\
&\leq \frac1{n/(2KL)}|\text{ALG}^\ell_k|\\
&\leq \frac1{n/(2KL)}s\\
\end{split}
\end{equation}
In the third transition we used the fact that $E^\ell_{k, j}$ are disjoint for different $j\in \ac{\B}^\ell_k$ by Lemma~\ref{lm:disjoint}, and therefore $\wh{E}^\ell_{k, j}$ are also disjoint for different $j$ since $\tau_*$ is injective (in turn, because individual maps $\tau^\ell$ are injective by Lemma~\ref{lm:tau-prop} and have disjoint ranges). In the forth transition we used the fact that 
$$
|\ac{\B}^\ell_k|\geq |\B^\ell_k|-|\Ext^\ell_k\cup \{q^\ell_k\}|\geq n/(KL)-K\geq n/(2KL)
$$
since $n$ is sufficiently large as a function of $K$ and $L$. In the forth transition we used the assumption that the total number of edges remembered by ALG is bounded by $s$. Now by Lemma~\ref{lm:special-edges} one has
$$
M_{ALG}\cap (A_P\times (Q\setminus B_Q))\subseteq \bigcup_{\ell\in [L], k\in [K/2]} \tau^\ell(E^\ell_{k, J^\ell_k}),
$$
and therefore
\begin{equation*}
\begin{split}
|M_{ALG}\cap (A_P\times (Q\setminus B_Q))|&\leq \sum_{\ell\in [L], k\in [K/2]} |\text{ALG}^\ell_k\cap \tau^\ell(E^\ell_{k, J^\ell_k})|\\
&=\sum_{\ell\in [L], k\in [K/2]} |\text{ALG}^\ell_k\cap \wh{E}^\ell_{k, J^\ell_k}|,
\end{split}
\end{equation*}
since ALG can only output edges that it remembered as per model definition (Definition~\ref{def:gen-online-intro}). Taking expectations of both sides and using~\eqref{eq:expect-lk}, we get
\begin{equation*}
\begin{split}
\expect_{\wh{G}\sim \mathcal{D}} \left[|M_{ALG}\cap (A_P\times (Q\setminus B_Q))|\right]&\leq \sum_{\ell\in [L], k\in [K/2]} \expect_{\wh{G}\sim \mathcal{D}} \left[|\text{ALG}^\ell_k\cap \wh{E}^\ell_{k, J^\ell_k}|\right]\\
&\leq \sum_{\ell\in [L], k\in [K/2]} \expect_{\wh{G}\sim \mathcal{D}} \left[|\text{ALG}^\ell_k\cap \wh{E}^\ell_{k, J^\ell_k}|\right]\\
&\leq LK\cdot \frac1{n/(2KL)}s\\
&\leq 2L^2K^2\cdot \frac{s}{n}\\
&=O\left(\frac{s}{\log |P|}\right),\\
\end{split}
\end{equation*}
where the last transition uses the fact that $n=\log_m N=\Omega(\log n)=\Omega(\log |P|)$. Since $s=o(|P|\log |P|)$ by assumption, we get 
$$
\expect_{\wh{G}\sim \mathcal{D}} \left[|M_{ALG}\cap (A_P\times (Q\setminus B_Q))|\right]=o(|P|),
$$
and therefore by Markov's inequality 
$$
\prob_{\wh{G}\sim \mathcal{D}} \left[|M_{ALG}\cap (A_P\times (Q\setminus B_Q))|>(1/L)|P|\right]=o(1).
$$
Finally, we note that the graph $\wh{G}$ contains a matching $M_{OPT}$ satisfying  $|M_{OPT}|\geq (1-O(1/L))|P|$ by Lemma~\ref{lm:large-matching-ghat}.  Combining the above bounds with~\eqref{eq:m-alg-bound}, we get
$$
\prob_{\wh{G}\sim \mathcal{D}} \left[|M_{ALG}|>\left(\frac1{1+\ln 2}+O(1/L)\right)|M_{OPT}|\right]=o(1),
$$
 as required.
\end{proof}

\begin{proofof}{Theorem~\ref{thm:main-simple}}
Follows directly by Theorem~\ref{thm:main-simple-prime} by setting $K$ to be a sufficiently large constant.
\end{proofof}


\newpage

\renewcommand{\duallabel}[1]{\label{#1-full}}
\renewcommand{\dualref}[1]{\ref{#1-full}}
\renewcommand{\dualeqref}[1]{\eqref{#1-full}}
\renewcommand{\F}{{\mathcal{F}}}
\renewcommand{\u}{{\mathbf{u}}}
\renewcommand{\v}{{\mathbf{v}}}
\renewcommand{\r}{{\mathbf{r}}}

\newcommand{\block}{\text{block}}
\newcommand{\subspace}{\text{subspace}}
\renewcommand{\Int}{\text{Int}}
\renewcommand{\I}{\mathbf{I}}
\newcommand{\Q}{\mathbf{Q}}
\renewcommand{\i}{\mathbf{i}}
\renewcommand{\j}{\mathbf{j}}
\renewcommand{\k}{\mathbf{k}}
\newcommand{\h}{\mathbf{h}}
\newcommand{\J}{\mathbf{J}}
\renewcommand{\a}{\mathbf{a}}
\renewcommand{\d}{\mathbf{d}}
\newcommand{\f}{\mathbf{f}}
\renewcommand{\l}{\mathbf{l}}
\renewcommand{\H}{\mathbf{H}}
\renewcommand{\u}{\mathbf{u}}

\renewcommand{\c}{\mathbf{c}}
\renewcommand{\d}{\mathbf{d}}

\newcommand{\deltagrid}{{\Delta\cdot \mathbb{Z}\cap [0, 1]}}
\newcommand{\deltagridInt}{{\Delta\cdot \mathbb{Z}\cap [0, 1)}}
\newcommand{\Dom}{\text{Dom}}

\newcommand{\rect}{\textsc{Rect}}

\section{Main result}\duallabel{sec:main-result}
In the rest of the paper we prove our main result, i.e. Theorem~\ref{thm:main}.
We define the individual instances $G^\ell$, establish their main properties and define the glueing map $\tau^\ell$ in Section~\dualref{sec:main-construction}. We then define the predecessor map $\nu$ and establish its main properties in Section~\dualref{sec:nu}. We then give the proof of the lower bound in Section~\dualref{sec:main-theorem}.

\section{Basic gadgets and the glueing map $\tau$}\duallabel{sec:main-construction}

 The input graph $G=(P, Q, E)$ is a edge disjoint (but not vertex disjoint) union of graphs $G^\ell=(S^\ell, T^\ell, E^\ell), \ell\in [L],$ that we define below. For every $\ell\in [L]$ we have $|T^\ell|=N=m^n$, and have $|S^\ell|\approx N/2$. The instances $G^\ell$ are then tied together via carefully designed maps $\tau^\ell$:
$$
\tau^\ell: S^\ell\to T_*^{\ell-1},
$$
where $T_*^{\ell-1}$ is a special subset of $T^{\ell-1}$ that we refer to as the {\em terminal subcube} of $T^{\ell-1}$. The maps $\tau^\ell$ are injective, but not defined on the entirety of $S^\ell$: a small fraction of vertices are left unmapped, and contribute to various error terms in our analyisis. Overall, this mapping ensures that the bipartition $P\cup Q$ of the graph $G$ satisfies
\begin{equation*}
\begin{split}
P&\approx \bigcup_{\text{even~}\ell\in [L/2]} T^\ell\\
\end{split}
\end{equation*}
and 
$$
Q\approx S^0\cup \bigcup_{\text{odd~}\ell\in [L/2]} T^\ell.
$$
The $\approx$ sign in the equations above reflects a small fraction of vertices in $S^\ell, \ell\in [L], \ell>0,$ that the corresponding map $\tau^\ell$ is not defined on -- see~\dualeqref{eq:p-def} and~\dualeqref{eq:q-def} in Section~\dualref{sec:main-theorem} below.

 \subsection{Basic definitions and notation}\duallabel{sec:defs}
 
Throughout the paper we use the notation $[a]=\{0, 1,\ldots, a-1\}$ for a positive integer $a$.

\paragraph{Associating vertices with points in the hypercube $[m]^n$.} Every vertex in $P$ and $Q$ is equipped with a label from $[m]^n$ which we denote by 
$$
\lbl: P\cup Q\to [m]^n.
$$
For a pair of vertices $x\in P$ and $y\in Q$ we write $x\delequal y$ if $\lbl(x)=\lbl(y)$. For every $\ell\in [L]$ vertices in $T^\ell$ have distinct labels, and for every $k\in [K/2]$ vertices in $S^\ell_k$ also have distinct labels (their labels are a subset of the labels of $T^\ell$).  Thus, we will often think of vertices in $G^\ell$ as points in the hypercube when we think of vertices in $T^\ell$, or vertices in $S^\ell_k$ and $k$ is fixed.

\begin{definition}[Weight of a vertex (or point in the hypercube)]
For every $x\in [m]^n$ we define 
\begin{equation*}
\weight(x)=\sum_{j\in [n]} x_j.
\end{equation*}
We will routinely apply the weight function to vertices of $G$. For a vertex $x$ of $G$ we write $\weight(x)$ to denote $\weight(\lbl(x))$.
\end{definition}

\begin{definition}[Boundary points]\duallabel{def:boundary}
We define the set $B\subset [m]^n$ of boundary points by 
$$
B=\{x\in [m]^n: x_i<n^2\text{~or~}x_i>m-n^2\text{~for some~}i\in [n]\}.
$$
\end{definition}

We have
\begin{claim}\duallabel{cl:boundary-points}
The fraction of boundary points in $[m]^n$ is bounded by $1/n^{10}$  long as $m\geq n^{20}$ and $n>2$, which we assume throughout the paper.
\end{claim}
\begin{proof}
This follows by a union bound. Pick a point $x\in [m]^n$ uniformly at random. The probability that a a fixed coordinate is smaller than $n^2$ of larger than $m-n^2$ is at most $2n^2/m$. Thus, the probability that at least one coordinate of $x$ is at most $n^2$ or at least $m-n^2$ is bounded by $2n^3/m$ by a union bound. Since $m\geq n^{20}$ by assumption, the result follows.
\end{proof}

We note that the assumption of Claim~\dualref{cl:boundary-points} above is satisfied by property~\dualref{p0} of parameter setting that we ensure throughout this section.

\paragraph{Family of fixed weight vectors $\mathcal{F}$ with small pairwise dot products.} We let $\mathcal{F}$ be a family of vectors in $\{0, 1\}^n$ of Hamming weight $w=(\e/2) n$ such that for every $\u, \v\in \F$ one has 
$$
\langle \u, \v\rangle \leq \e w.
$$
Fix such a family $\F$ with $|\F|=2^{\Omega(\e^2 n)}$. The existence of such a family can be established by the probabilistic method -- we include the proof in Appendix~\dualref{app:F-construction} for completeness.
We partition $\mathcal{F}$ into disjoint subsets  of equal size, letting 
$$\mathcal{F}=\B^0\cup \B^1\cup \ldots \cup \B^{L-1},
$$ 
where $\B^i\cap \B^j=\emptyset$ if $i\neq j$.
For every $\ell\in [L]$ the set of vectors $\B^\ell$ will be used to define a corresponding graph $G^\ell$, and these graphs will be presented to the algorithm in the stream sequentially for $\ell\in [L]$. Every set $\B^\ell$ is partitioned as 
\begin{equation}\duallabel{eq:b-def}
\B^\ell=\B^\ell_0\cup\ldots\cup \B^\ell_{K/2},
\end{equation}
where $|\B^\ell_k|=\frac1{L(K/2+1)}\cdot |\mathcal{F}|$ for $k\in [K/2+1]$, and $\B^\ell_k\cap \B^\ell_{k'}=\emptyset$ for $k\neq k'$. The $\ell$-th graph $G^\ell$ is  mainly parameterized by a sequence 
\begin{equation}\duallabel{eq:j-def}
\J^\ell\in \B^\ell_0\times \ldots\times \B^\ell_{K/2},
\end{equation}
i.e., $\J^\ell_k\in \B^\ell_k$ for $k\in [K/2+1]$, as well as a vector $\r^\ell\in \B^\ell_{K/2}$ that we refer to as the $\ell$-th compression vector (see Definition~\dualref{def:ext} below).

\begin{definition}[Special vectors of the $\ell$-th instance]\duallabel{def:special-vectors}
We refer to $\J^\ell$ and $\r^\ell$ as the {\em special vectors} of instance $G^\ell$, and let
$$
\Sp(\B^\ell):=(\J^\ell, \r^\ell).
$$
\end{definition}

We also define the extended special coordinates
\begin{equation}\duallabel{eq:sp-ext}
\Spt(\B^\ell):=\J^\ell\cup \{\r^\ell\}\cup \bigcup_{k\in [K]} (\Ext^\ell_k\cup \{\q^\ell_k\}).
\end{equation}

For every $\ell\in [L], \ell>0,$ the map $\tau^\ell$ is parameterized by vector $\r^\ell\in \B^\ell$, referred to as the {\em compression vector} for the terminal subcube $T_*^\ell$, as well as a collection of {\em extension vectors} for every $k\in [K/2]$.
\begin{definition}[Compression vectors and extension vectors] \duallabel{def:ext}
 For every $k\in [K/2]$\footnote{Note that we only define the extension vectors $\Ext^\ell_k$ and the compression vector $\q^\ell_k$ for $k\in [K/2]=\{0, 1,2,\ldots, K/2-1\}$, even though the sequence $\J^\ell$ is of length $K/2+1$. This is for convenience in defining the glueing map $\tau^\ell$ -- see Section~\dualref{sec:tau-ell} for more details.} let 
$$
\Ext^\ell_k\subseteq \B^\ell_k
$$ 
denote a set of $K/2+1-k$ vectors referred to as the {\em extension vectors} and 
$$
\q^\ell_k\in \B^\ell_k\setminus \Ext^\ell_k
$$ 
denote the {\em compression vector} for the $k$-th phase of the graph $G^\ell$. Let $\r^\ell\in \B^\ell_{K/2}$ denote the {\em $\ell$-th compression vector.}
\end{definition}

Define for $k\in [K/2]$
\begin{equation}\duallabel{eq:int-b}
\ac{\B}^\ell_k=\B^\ell_k\setminus (\{\q^\ell_k\}\cup \Ext^\ell_k)
\end{equation}
and let
\begin{equation}\duallabel{eq:int-b-k2}
\ac{\B}^\ell_{K/2}=\B^\ell_{K/2}\setminus \{\r^\ell\}.
\end{equation}

For every $\ell\in [L]$ and $k\in [K/2+1]$ select 
$$
\J^\ell_k\in \ac{\B}^\ell_k.
$$

\subsection{Parameter setting}\duallabel{sec:params}
We choose parameters $\e, \delta$, $M, W$, $K$ and $L$ so that $\e, \delta, L$ only depend on $K$ and the following properties are satisfied:
\begin{description}
\item[(p0)\duallabel{p0}] $m=n^{20}$ 
\item[(p1)\duallabel{p1}] $W/w=\text{lcm}(K, K-1,\ldots, 2, 1)$
\item[(p2)\duallabel{p2}] $\delta^{-1} \cdot \text{lcm}(K, K-1,\ldots, 2, 1)\cdot W/w \mid M/w$.
\item[(p3)\duallabel{p3}] $\Delta=\frac1{\text{lcm}(K, K-1, K-2,\ldots, 2, 1)}$; note that $\Delta\leq 1/K$ and that $\Delta\cdot M/w$ is an integer by \dualref{p2}.
\item[(p4)\duallabel{p4}] $L=\sqrt{K}$
\item[(p5)\duallabel{p5}] $\delta\leq \Delta^{100K^2}$ 
\item[(p6)\duallabel{p6}] $\e \leq  \delta^2$
\item[(p7)\duallabel{p7}] $\e \geq w/M$
\end{description}
In the above we write $\text{lcm}(a_1,a_2,\ldots, a_s)$ to denote the least common multiple of $a_1,a_2,\ldots, a_s$. For $a>0, b>0$ we write $a \mid b$ if $b/a$ is an integer.

\begin{lemma}\duallabel{lm:params}
For every constant $K$ there exists a setting of $\e$ as a function of $K$ and a setting of parameters $M, W, \Delta, L, \delta$ and $m=\text{poly}(n)$ that satisfies~\dualref{p0}-\dualref{p7}.
\end{lemma}
\begin{proof}
For any $\e\in (0, 1)$ such that $(1/\e)^{1/3}$ is an integer, let $w=(\e/2) n$, as required by the construction of the set $\F$, and let 
$$
\delta=\e^{1/2},
$$
ensuring that~\dualref{p6} holds with equality. Let 
$$
M=w/\e=n/2,
$$ 
so that~\dualref{p7} is satisfied with equality. Let $W=\text{lcm}(K, K-1,\ldots, 2, 1)\cdot w$, as per~\dualref{p1}. Note that in order to satisfy~\dualref{p2}, it suffices to ensure that
\begin{equation*}
\begin{split}
\frac{M/w}{\delta^{-1}\cdot \text{lcm}(K, K-1,\ldots, 2, 1)\cdot W/w}
&=\frac{1/\e}{\e^{-1/2}\cdot (\text{lcm}(K, K-1,K-2,\ldots, 2, 1))^2}\\
&=\frac{(1/\e)^{1/2}}{(\text{lcm}(K, K-1,K-2,\ldots, 2, 1))^2}\\
\end{split}
\end{equation*}
is an integer. We let 
$$
1/\e=(\text{lcm}(K, K-1,K-2,\ldots, 2, 1))^{400K^2}, 
$$
ensuring that 
\begin{equation*}
\begin{split}
\frac{(1/\e)^{1/2}}{(\text{lcm}(K, K-1,K-2,\ldots, 2, 1))^2}&=\frac{\left((\text{lcm}(K, K-1,K-2,\ldots, 2, 1))^{400K^2}\right)^{1/2}}{(\text{lcm}(K, K-1,K-2,\ldots, 2, 1))^2}\\
&=(\text{lcm}(K, K-1,K-2,\ldots, 2, 1))^{200K^2-2},
\end{split}
\end{equation*}
ensuring that~\dualref{p2} holds.

We set $\Delta=\frac1{\text{lcm}(K, K-1, K-2,\ldots, 2, 1)}$ as per~\dualref{p3}, and verify that 
\begin{equation*}
\begin{split}
\delta&=(1/\e)^{1/2}\\
&=\left((\text{lcm}(K, K-1,K-2,\ldots, 2, 1))^{-400K^2}\right)^{1/2}\\
&=\text{lcm}(K, K-1,K-2,\ldots, 2, 1))^{-200K^2}\\
&\leq \Delta^{100K^2}, 
\end{split}
\end{equation*}
so~\dualref{p5} is satisfied.

Finally, we let $L=\sqrt{K}$, satisfying~\dualref{p4}. Letting $K$ be a sufficiently large constant and $n$ sufficiently large as a function of $K$ and letting $m=n^{20}$ to satisfy~\dualref{p0} completes the setting of parameters.
\end{proof}

\subsection{Basic gadgets $G^\ell$: vertex set and main definitions}\duallabel{sec:g-ell}
In this section we define our gadgets $G^\ell=(S^\ell, T^\ell, E^\ell)$. Since $\ell$ is fixed throughout this section, we omit the superscript and let $S=S^\ell, T=T^\ell, E=E^\ell$.

\paragraph{Vertices of $G$ and their labels.}  Let $m\geq 1$ be a sufficiently large integer. We have $|T|=m^n$, and label vertices in $T$ with points in the hypercube $[m]^n$, where $[m]=\{0, 1, 2,\ldots, m-1\}.$  The labelling defines a bijective mapping from the vertex set $T$ to $[m]^n$, and we hence sometimes refer to vertices in $T$ as simply points in $[m]^n$. The vertices on the $S$ side of the bipartition will also be labelled with points on the hypercube $[m]^n$, as defined below. The average degree of a vertex in our construction will be $2^{\Omega_\e(n)}$, which translates to average degree $N^{\Omega_\e(1/\log \log N)}$ when $m=\text{poly}(n)$ (this is how we set $m$ as per~Lemma~\dualref{lm:params}).

  The set $S$ of vertices is partitioned into disjoint subsets $S=S_0\uplus S_1\uplus \ldots \uplus S_{K/2-1}$ whose vertices are also labeled with elements of $[m]^n$. We now define $S_k$ for $k\in [K/2]$. Let $\B:=\B^\ell$ as per~\dualeqref{eq:b-def}, so that $\B=\B_0 \cup \ldots \cup \B_{K/2}$, and let $\J:=\J^\ell$ as per~\dualeqref{eq:j-def}.  For every vertex $y\in S\cup T$ and vector $\u\in \F$ we use the notation
  $$
  \langle y, \u\rangle=\sum_{s\in [n]} y_s\cdot \u_s,
  $$
  where $y_s$ stands for the $s$-th coordinate of the label of $y$.  In what follows we often write, for two vertices $x, y\in S\cup T$ and a vector $\u\in \mathbb{Z}^n$
  $$
  x=y+\u
  $$
if the label of $x$ can be obtained by adding $\u$ to the label of $y$, i.e. $x_s=y_s+\u_s$ for every $s\in [n]$. Similarly, we often write $x=y+\u$ when $x\in S\cup T$ and $y\in [m]^n$ if the label of $x$ is the sum of $y$ and $\u$. In other words, we treat vertices of $G$  and points in $[m]^n$  where this does not lead to confusion (see Remark~\dualref{rm:sk-vertex-identification}).

\paragraph{Nested sequence $T=T_0\supset T_1\supset \ldots\supset T_{K/2}$ and downsets $S_0,S_1,\ldots, S_{K/2-1}$.} We let $T_0=T$, i.e. every $x\in T_0$ is labeled with an element of $[m]^n$. For every $k\in [K/2]$ let 
\begin{equation}\duallabel{eq:def-tk}
\begin{split}
T_{k+1}:=\left\{y\in T_k: \langle y, \j_k\rangle \pmod{M}\in \left[0, 1-\frac{1}{K-k}\right)\cdot M\right\}.\\
\end{split}
\end{equation}
Note that $T_0\supset T_1\supset\ldots\supset T_{K/2}$ form a nested sequence. Also note that for every $k\in [K/2]$ one has 
\begin{equation}\duallabel{eq:def-tk-all-constraints}
\begin{split}
T_k:=\left\{y\in T_0: \langle y, \j_s\rangle \pmod{M}\in \left[0, 1-\frac{1}{K-s}\right)\cdot M\text{~for all~}s\in \{0, 1,\ldots, k-1\}\right\}.\\
\end{split}
\end{equation}
The innermost set in this sequence is a central object of our construction:
\begin{definition}[Terminal subcube]\duallabel{def:terminal-subcube}
We refer to $T_*:=T_{K/2}$ as the {\em terminal subcube}.
\end{definition}

Recall that for a pair of vertices $x, y\in S\cup T$ the relation $x\delequal y$ stands for `the label of $x$ equals the label of $y$'. We extend this relation to sets in the natural way, writing $A\delequal B$ for $A, B\subseteq S\cup T$ if there exists a bijective map $\pi: A\to B$ such that for every $x\in A$ one has $x\delequal \pi(x)$.  With this notation we define
\begin{equation}\duallabel{eq:def-sk}
\begin{split}
S_k&:\delequal \left\{x\in T_k: \weight(x) \in \left[0, \frac{1}{K-k}\right)\cdot W \pmod{ W} \right\},
\end{split}
\end{equation}
The above stands for $S_k$ being a set of vertices such that $S_k\delequal \wt{T}_k$, where
\begin{equation*}
\wt{T}_k:=\left\{x\in T_k: \weight(x) \in \left[0, \frac{1}{K-k}\right)\cdot W \pmod{ W} \right\}
\end{equation*}
is the set of vertices in $T_k$ whose weight modulo $W$ belongs to a certain range. We stress here that unlike the collection of sets $T_k$, the sets $S_k$ are disjoint. 

\begin{remark}\duallabel{rm:sk-vertex-identification}
The labels of vertices in $S_k$ for any $k\in [K/2]$ are distinct, the labels of vertices in $S_k$ are a subset of the labels of vertices in $S_l$ for $k>l$. Thus, while a vertex in $T$ is uniquely identified by its label, a vertex in $S$ is not. However, a vertex in $S$ is uniquely identified by its label together with the index $k\in [K/2]$ of the set $S_k$ that it belongs to.
\end{remark}

We also let, for every $k\in [K/2]$ and $\j\in \B_k$ 
\begin{equation}\duallabel{eq:skj}
\begin{split}
T_k^\j&=\left\{y\in T_k: \langle y, \j\rangle \pmod{M}\in \left[0, 1-\frac{1}{K-k}\right)\cdot M\right\}\\
S_k^\j&=\left\{x\in S_k: \langle x, \j\rangle \pmod{M}\in \left[0, 1-\frac{1}{K-k}\right)\cdot M\right\}.\\
\end{split}
\end{equation}
We gather basic bounds on the sizes of the sets $T_k, S_k$ in 
\begin{lemma}\duallabel{lm:size-bounds}
One has
\begin{itemize}
\item[{\bf (1)}] For every $k\in [K/2+1]$ one has $|T_k|=(1\pm \sqrt{\e})\cdot |T_0|(1-k/K)$;
\item[{\bf (2)}] For every $k\in [K/2]$ one has  $|S_k|=(1\pm \sqrt{\e}) \cdot |T_0|/K$;
\item[{\bf (3)}] For every $k\in [K/2]$, every $\j\in \B_k$ one has $|S^\j_k|=(1\pm \sqrt{\e})(1-\frac1{K-k})|T_0|/K$.
\item[{\bf (4)}] For every $k\in [K/2]$, every $\j\in \B_k$ one has $|T^\j_k|=(1\pm \sqrt{\e})(1-\frac{k+1}{K})|T_0|$.
\end{itemize}
\end{lemma}

\begin{remark}
Note that the sets $S_k$ are defined for $k\in [K/2]$, whereas $T_k$ is defined for $k\in [K/2+1]$ -- this is to ensure that the number of vertices in the terminal subcube $T_*$ can be made arbitrarily close to the total size of $\biguplus_{k\in [K/2]} S_k$ for any fixed $K$ by choosing $\e$ sufficiently small, simplifying the definition and analysis of the glueing maps $\tau^\ell$ (see Section~\dualref{sec:tau-ell}) that map sets of the latter type to sets of the former type.
\end{remark}

Since per~\dualeqref{eq:def-sk} for every $k\in [K/2]$ the set $S_k$ is essentially a subsampling of the corresponding set $T_k$, for every $U\subseteq T_k$ we define the projection of $U$ to $S_k$, denoted by $\downset_k(U)$, as the set of vertices in $S_k$ whose labels match the labels of vertices in $U$:
\begin{definition}[Downset of a subset of $T$]\duallabel{def:downset}
For every $U\subseteq T$ and $k\in [K/2]$ we define the {\em downset of $U$ in $S_k$} by 
$$
\downset_k(U)=\{x\in S_k: \exists y\in U \text{~s.t.~} x\delequal y\}.
$$
We define
$$
\downset(U)=\bigcup_{k\in [K/2]} \downset_k(U).
$$
\end{definition}

\begin{remark} \duallabel{rm:downset} We note that $\downset$ is defined as a map from subsets of $T$ to subsets of $S$. This certainly defines a natural mapping from elements of $T$: element $x\in T$ is mapped to $\downset(\{x\})$, i.e. the downset of the singleton set containing $x$. However, this map is not one to one: $\downset(\{x\})$ may be a set of size up to $K/2$ (note, however, that for every $k\in [K/2]$ one has $|\downset_k(\{x\})|\leq 1$).
\end{remark}

\begin{remark}\duallabel{rm:truncated-downset}
Note that if $U\subset T_k\setminus T_{k+1}$ for some $k\in [K/2]$, then $\downset_s(U)=\emptyset$ for all $s\in \{k+1, \ldots, K/2-1\}$. Thus, in that case we have
$$
\downset(U)=\bigcup_{s=0}^{k} \downset_s(U).
$$
\end{remark}

\subsection{Edges of $G$}
Similarly to our construction in Section~\ref{sec:toy-construction}, we define the edge set of $G$ to be a union of constant size complete bipartite subgraphs, where for every $k\in [K/2]$ and every direction $\j\in \B_k$ the edge set $E_k\subseteq T_k\times S_k$ consists of a disjoint union of small bipartite subgraphs for every `line' in direction $\j$. Unlike the construction of Section~\ref{sec:toy-construction}, it takes more care to define lines appropriately when $\j$ is not just a coordinate direction, but rather a general binary vector in $\F$, and different directions are not necessarily orthogonal, but rather just have small dot products. For that we first need

\begin{definition}[Block of $x$ with respect to a vector $\j$]
For $\j\in \F$ we define $\block_\j(x):=\lfloor \langle x, \j\rangle/M\rfloor$. 
\end{definition}

We can now define
 \begin{definition}[Line through $x$ in direction $\j$]\duallabel{def:line}
 For each  $\j\in \B_k$ for each $x\in [m]^n$ we denote the line in direction $\j$ going through $x$ by 
$$
\text{line}_\j(x)=\left\{x'\in [m]^n: x'=x+\lambda\cdot \j\text{~for~}\lambda\in \mathbb{Z}\text{~s.t.~}\block_\j(x')=\block_\j(x)\right\}.
$$
\end{definition}

Some basic properties of lines are given in 
\begin{claim}[Basic bounds on lines]\duallabel{cl:line-size}
For every $\j\in \F$:
\begin{description}
\item[(1)] for every $x\in [m]^n$ one has  $|\text{line}_\j(x)|\leq M/w$ and for every $x\in [m]^n\setminus B$ one has $|\text{line}_\j(x)|=M/w$; furthermore, for every $x\in [m]^n$ and $y\in \text{line}_\j(x)$ one has $y=x+\lambda \cdot \j$ for some integer $\lambda$ satisfying $|\lambda|\leq 2M/w$.
\item[(2)] for every $x\in [m]^n\setminus B$, for every $c\in \Delta\cdot \mathbb{Z}\cap [0, 1)$ one has 
$$
\left|\{y\in \text{line}_\j(x): \langle y, \j\rangle \pmod{M}\in [c, c+\delta)\cdot M\}\right|=\delta\cdot M/w.
$$
\item[(3)] for every $x\in [m]^n\setminus B$, for every $c\in \Delta\cdot \mathbb{Z}\cap (0, 1]$ one has
$$
\left|\{y\in \text{line}_\j(x): \langle y, \j\rangle \pmod{M}\in [c-\delta, c)\cdot M\}\right|=\delta\cdot M/w.
$$
\end{description}
\end{claim}
\begin{proof} We start by proving some useful basic facts, and the proceed to prove {\bf (1), (2)} and {\bf (3)}.
First note that for $x\in [m]^n$ and $x'=x+\lambda\cdot \j\in \mathbb{Z}^n$ (but not necessarily in $[m]^n$) one has
\begin{equation}\duallabel{eq:923hg98h9fqrwh9ggygfdygd}
\langle x', \j\rangle=\langle x+\lambda\cdot \j, \j\rangle=\langle x, \j\rangle+\lambda\cdot w.
\end{equation}
This means that $|\lambda|\leq 2M/w$ for all such $x'\in \text{line}_\j(x)$, as otherwise $\block_\j(x')\neq \block_\j(x)$. By Definition~\dualref{def:boundary} we have $n^2\leq x_i \leq m-n^2$ for all $x\in [m]^n\setminus B$ and $i\in [n]$. Thus, for all such $x$ we have, since $\j\in \{0, 1\}^n$,
$$
n^2-2M/w\leq x_i+\lambda\cdot \j_i\leq m-n^2+2M/w
$$
for all $i\in [n]$. By~\dualref{p1} and \dualref{p2} together with the fact that $n$ is sufficiently large as a function of $M/w, W/w, K, L, \Delta$, and $\delta$, we get
$$
0\leq x_i+\lambda\cdot \j_i\leq m-1
$$
for all $i\in [n]$. Thus, 
\begin{equation}\duallabel{eq:934ht9hGADBF}
x'=x+\lambda \cdot \j\in [m]^n. 
\end{equation}

We now prove {\bf (1)}. Let $q=\langle x, \j\rangle \pmod{M}$ to simplify notation, so that  $\langle x, \j\rangle= \block_\j(x)\cdot M+q$. Further, let $a=\lfloor \frac1{w}q\rfloor$ and $b=q \pmod{w}$. With this notation in place we have
\begin{equation}\duallabel{eq:923hg98h9fd}
\langle x', \j\rangle=\langle x+\lambda\cdot \j, \j\rangle=\langle x, \j\rangle+\lambda\cdot w=\block_\j(x)\cdot M+(a+\lambda)\cdot w+b.
\end{equation}
We thus have
$$
\block_\j(x')=\lfloor (\block_\j(x)\cdot M+(a+\lambda)\cdot w+b)/M\rfloor=\block_\j(x)+\lfloor((a+\lambda)\cdot w+b)/M\rfloor,
$$
and hence $\block_\j(x')=\block_\j(x)$ if and only if $((a-\lambda)\cdot w+b)/M\in [0, 1)$. On the other hand, since $b\in \{0, 1,\ldots, w-1\}$, we have
$$
\left\{\lambda\in \mathbb{Z}: ((a+\lambda)\cdot w+b)/M\in [0, 1)\right\}=\{-a,\ldots, -a+M/w-1\},
$$
which is a set of size $M/w$ since $w \mid M$ by \dualref{p2}. This proves the upper bound in {\bf (1)}. For the lower bound we note that if $x\in [m]^n\setminus B$, then every $x'=x+\lambda\cdot \j$ such that $\block_\j(x')=\block_\j(x)$ one has $|\lambda|\leq 2M/w$, and therefore  $x'\in [m]^n$ by~\dualeqref{eq:934ht9hGADBF} (see argument above for more details).  This implies the lower bound, and hence the equality in {\bf (1)}. In particular, we get for $x\in [m]^n\setminus B$
\begin{equation}\duallabel{eq:823gf8g823fggiv}
\begin{split}
\text{line}_\j(x)=\{x+\lambda\cdot \j: \lambda\in \{-a,\ldots, -a+M/w-1\}.
\end{split}
\end{equation}

We now prove {\bf (2)}. First note that by~\dualeqref{eq:823gf8g823fggiv} we have 
\begin{equation*}
\begin{split}
&\left|\{y\in \text{line}_\j(x): \langle y, \j\rangle \pmod{M}\in [c, c+\delta)\cdot M\}\right|\\
&=\left|\{\lambda\in \{-a,\ldots, -a+M/w-1\}: \langle x+\lambda\cdot \j, \j\rangle \pmod{M}\in [c, c+\delta)\cdot M\}\right|\\
&=\left|\{\lambda\in \{-a,\ldots, -a+M/w-1\}: (a+\lambda)\cdot w+b\in [c, c+\delta)\cdot M\}\right|.
\end{split}
\end{equation*}
Since $ \delta^{-1}  \mid M/w$ by \dualref{p2}, $c\in \Delta\cdot \mathbb{Z}\cap [0, 1)$ by assumption of the claim and $\Delta \mid M/w$ by \dualref{p3}, we can write $c\cdot M=\alpha\cdot w$ and $c+\delta=\beta\cdot w$ for integers $\alpha, \beta\in \{0, 1,\ldots, M/w-1\}, \alpha<\beta$ (here we used the fact that $\delta<\Delta$ by~\dualref{p5}). The last line of the equation above can thus be rewritten as 
\begin{equation*}
\begin{split}
&\left|\{\lambda\in \{-a,\ldots, -a+M/w-1\}: (a+\lambda)\cdot w+b\in [c, c+\delta)\cdot M\}\right|\\
&=\left|\{\lambda\in \{-a,\ldots, -a+M/w-1\}: \alpha\cdot w\leq (a+\lambda)\cdot w+b<\beta\cdot w\}\right|\\
&=\beta-\alpha\\
&=(\beta\cdot w-\alpha\cdot w)/w\\
&=((c+\delta)\cdot M-c\cdot M)/w\\
&=\delta\cdot M/w,
\end{split}
\end{equation*}
where the second equality holds because $b\in \{0,\ldots, w-1\}$ and the fourth equality is by definition of $\alpha$ and $\beta$. This proves {\bf (2)}.
The proof of {\bf (3)} is analogous.
\end{proof}

\begin{lemma}[Lines form a partition]\duallabel{lm:line-partition}
For every $\j\in \F$, every $x, x'\in [m]^n$  one has either $\text{line}_\j(x)=\text{line}_\j(x')$ or $\text{line}_\j(x)\cap \text{line}_\j(x')=\emptyset$.
\end{lemma}

The proof of the lemma follows from a more general statement about subspaces (see Claim~\dualref{cl:line-vs-subspace} and Lemma~\dualref{lm:subspace-partition}) and its proof is given in Section~\dualref{sec:subspaces}.
\begin{definition}[Minimal $\j$-line cover]\duallabel{def:line-cover}
We say that a set  $C\subset [m]^n$ is a minimal $\j$-line cover if 
 $\bigcup_{x\in C} \text{line}_\j(x)=[m]^n$ and $\text{line}_\j(x)\cap \text{line}_\j(x')=\emptyset$ for $x, x'\in C$, $x\neq x'$. 
\end{definition}

We now define the edges of $G$ incident on $S_k$ for every $k\in [K/2]$. For every $\j\in \ac{\B}_k$ 
let 
\begin{equation}\duallabel{eq:c-j-def}
C_\j\subset [m]^n
\end{equation}
be a minimal $\j$-line cover as per Definition~\dualref{def:line-cover}.
For every $y\in C$, we include a complete bipartite graph between $\text{line}_\j(y)\cap \Int_\delta(S_k^\j)$ and $\text{line}_\j(y)\cap (T_k\setminus T_k^\j)$: let $E=\bigcup_{k\in [K/2]} E_k$, where
\begin{equation}\duallabel{eq:edges-ei-def}
E_k=\bigcup_{\j\in \ac{\B}_k} E_{k, \j}
\end{equation}
and
\begin{equation}\duallabel{eq:edges-eij-def}
E_{k, \j}=\bigcup_{y\in C_\j} (\text{line}_\j(y)\cap \Int_\delta(S_k^\j)) \times (\text{line}_\j(y) \cap (T_k\setminus T_k^\j)).
\end{equation}
In the equation above we use the notation $\Int_\delta(S_k^\j)$ for the {\em $\delta$-interior} of the set $S_k^\j$, which we now define. First recall that by~\dualeqref{eq:def-tk-all-constraints} and definition of $S_k$ in~\dualeqref{eq:def-sk} we have
\begin{equation*}
\begin{split}
S_k\delequal &\left\{y\in [m]^n: \langle y, \j_s\rangle \pmod{M}\in \left[0, 1-\frac{1}{K-s}\right)\cdot M\text{~for all~}s\in \{0, 1,\ldots, k-1\}\right.\\
&\text{~and~}\\
&\left.\weight(x) \in \left[0, \frac{1}{K-k}\right)\cdot W \pmod{ W}\right\}\\
\end{split}
\end{equation*}
and 
\begin{equation*}
\begin{split}
S_k^\j\delequal &\left\{y\in [m]^n: \langle y, \j_s\rangle \pmod{M}\in \left[0, 1-\frac{1}{K-s}\right)\cdot M\text{~for all~}s\in \{0, 1,\ldots, k-1\}\right.\\
&\text{~and~}\\
&\langle y, \j\rangle \pmod{M}\in \left[0, 1-\frac{1}{K-k}\right)\cdot M\\
&\text{~and~}\\
&\left.\weight(x) \in \left[0, \frac{1}{K-k}\right)\cdot W \pmod{ W}\right\}\\
\end{split}
\end{equation*}

The interior of $S_k^\j$, denoted by $\Int_\delta(S_k^\j)$, is simply the set of points in $S_k^\j$ that satisfy all the constraints above (except the subsampling constraint) with a margin of $\delta$:
\begin{equation}\duallabel{eq:283g8gausbfuagf}
\begin{split}
\Int_\delta(S_k^\j)\delequal &\left\{y\in [m]^n: \langle y, \j_s\rangle \pmod{M}\in \left[\delta, 1-\frac{1}{K-s}-\delta\right)\cdot M\text{~for all~}s\in \{0, 1,\ldots, k-1\}\right.\\
&\text{~and~}\\
&\langle y, \j\rangle \pmod{M}\in \left[\delta, 1-\frac{1}{K-k}-\delta\right)\cdot M\\
&\text{~and~}\\
&\left.\weight(x) \in \left[0, \frac{1}{K-k}\right)\cdot W \pmod{ W}\right\}\\
\end{split}
\end{equation}

\begin{remark}
We note that our definition of the interior $\Int_\delta(S_k^\j)$ of $S_k^\j$ is a special case of Definition~\dualref{def:int} below. We prefer to present it here first before presenting the more general version to alleviate notation in the definition of the basic gadgets $G^\ell$.
\end{remark}

\begin{remark}
Note that the edge set $E_k$ is fully defined by the prefix $\J_{<k}$ (note that we consider the compression indices and extension indices fixed and $\J$ variable; this is useful since in the actual hard input distribution we will fix the compression and extension indices arbitrarily, and select $\J$ uniformly at random from $\ac{\B}_0\times \ac{\B}_1\times \ldots\times \ac{\B}_{K/2}$ -- see Section~\dualref{sec:main-theorem}).
\end{remark}

\begin{remark}
We note that the edge set defined in~\dualeqref{eq:edges-eij-def} does not depend on the specific choice of a cover $C_\j$ used, i.e. any minimal $\j$-line cover produces the same edge set as per~\dualeqref{eq:edges-eij-def}.
\end{remark}

The following lemma shows that the complete bipartite graphs defined above are disjoint (this will be useful for analyzing a subsampling of the gadgets $G^\ell$ later in Section~\dualref{sec:main-theorem})
\begin{lemma}\duallabel{lm:disjoint}
For every $k\in [K/2]$, every $\i, \j\in \ac{\B}_k, \i\neq \j$, every $x\in C_\i, y\in C_\j$, where $C_\i$ and $C_\j$ are minimal $\i$- and $\j$-line covers respectively, the edge sets
$$
(\text{line}_\i(x)\cap \Int_\delta(S_k^\i)) \times (\text{line}_\i(x) \cap (T_k\setminus T_k^\i))
$$
and 
$$
(\text{line}_\j(y)\cap \Int_\delta(S_k^\j)) \times (\text{line}_\j(y) \cap (T_k\setminus T_k^\j))
$$
are disjoint.
\end{lemma}
\begin{proof}
We argue by contradiction. Note that the edge sets above intersect if and only if there exist $a, b$ such that
\begin{equation}\duallabel{eq:0293yt88gfdSF}
a\in (\text{line}_\i(x)\cap \Int_\delta(S_k^\i))\cap (\text{line}_\j(y)\cap \Int_\delta(S_k^\j))
\end{equation}
and 
\begin{equation}\duallabel{eq:92yt99gfef}
b\in (\text{line}_\i(x) \cap (T_k\setminus T_k^\i))\cap (\text{line}_\j(y) \cap (T_k\setminus T_k^\j)).
\end{equation}
Since $a, b\in \text{line}_\j(y)$, we have by Claim~\dualref{cl:line-size}, {\bf (1)}, that 
\begin{equation*}
b=a+\lambda\cdot \j
\end{equation*}
for some integer $\lambda$ with $|\lambda|\leq 2M/w$. This in particular means that 
\begin{equation}\duallabel{eq:ba-proj}
|\langle b, \i\rangle-\langle a, \i\rangle |=|\lambda| \langle \j, \i\rangle\leq |\lambda|\cdot \e w\leq 2\e M.
\end{equation}

On the other hand, since $a\in \Int_\delta(S_k^\i)$, we have by~\dualeqref{eq:skj} together with~\dualeqref{eq:283g8gausbfuagf} (see also Definition~\dualref{def:int})
\begin{equation*}
\langle a, \i\rangle \pmod{M}\in \left[\delta, 1-\frac{1}{K-k}-\delta\right)\cdot M.
\end{equation*}

Putting this together with~\dualeqref{eq:ba-proj} yields
\begin{equation*}
(\delta -2\e) M\leq \langle b, \i\rangle \pmod{M}\leq (1-\frac{1}{K-k}-\delta+2\e)\cdot M,
\end{equation*}
and thus since $\e\leq \delta^2<2\delta$ by~\dualref{p6}, we get
\begin{equation}
\langle b, \i\rangle \pmod{M}\in \left[0, 1-\frac{1}{K-k}\right)\cdot M,
\end{equation}
a contradiction with the assumption that $b\in T_k\setminus T_k^\i$ by ~\dualeqref{eq:92yt99gfef}.
\end{proof}

\subsection{Rectangles and their properties}\duallabel{sec:subspaces}
Our construction in this section is at a high level quite similar to the construction from Section~\ref{sec:toy-construction}. Unfortunately, however, it is more complicated, mainly due to the fact that we cannot rely on clean product structure of naturally defined rectangles (see Definition~\ref{def:rectangle}). However, our analysis is still based on a concept of a rectangle, which we define below -- see Definition~\dualref{def:rect}. While this is no longer a product set since our vectors in $\F$ are not orthogonal, but merely have small dot product, rectangles as per Definition~\dualref{def:rect} still behave is rather similar way to product sets. This section is devoted to proving some basic properties of rectangles that facilitate later analysis.

For two vectors $\a, \b$ of the same dimension we use the notation $\a<\b$ for $\a$ being coordinate-wise smaller than $\b$. We often index coordinates of a vector by elements of some set. For example, $\a\in [0, 1)^\I$ stands for $\a$ being a vector of length $|\I|$ whose entries are $\a_\i, \i\in \I$, and for a subset $\H\subset \I$ we write $\a_\H$ to denote the restriction of $\a$ to elements of $\H$.

\begin{definition}[Rectangles]\duallabel{def:rect}
For every $\I\subseteq \F$, every $\c, \d\in [0, 1]^\I, \c< \d$ the set 
$$
\rect(\I, \c, \d):=\{y\in [m]^n: \langle y, \i\rangle \pmod{M}\in \left[\c_\i, \d_\i\right)\cdot M \text{~for all~}\i \in \I\}
$$
is called a rectangle. 
\end{definition}

It is useful to introduce a more lightweight intermediate definition of rectangles with all side lengths equal to a parameter $\Delta$ -- see Definition~\dualref{def:cubes} below. This definition is useful since we can express every rectangle with coordinates divisible by $\Delta$ as a disjoint union of cubes, and at the same time cubes are somewhat more compact to represent, and will serve as our basic building blocks in what follows.

\begin{definition}[Cubes]\duallabel{def:cubes}
For every $\I\subseteq \F$, every $\a\in \deltagridInt^\I$ we let
$$
\rect(\I, \a)=\{y\in [m]^n: \langle y, \i\rangle \pmod{M}\in  [\a_\i, \a_\i+\Delta)\cdot M \text{~for all~}\i \in \I\}
$$
denote a rectangle with all side lengths equal to $\Delta$. 
\end{definition}

\begin{claim}[Decomposition into subcubes]\duallabel{cl:subcubes-decomp}
For every $\I, \H\subseteq \F, \I\cap \H=\emptyset$, every $\a, \b\in (\deltagrid)^{\I\cup \H}, \a<\b,$  the rectangle $F=\rect(\I\cup \H, \a, \b)$ satisfies
$$
F=\bigcup_{\f\in Q} \rect(\I\cup \H, (\f, \a_\H), (\f+\Delta\cdot \mathbf{1}_\I, \b_\H)),
$$
where 
$$
Q=\{0, \Delta,2 \Delta, \ldots, 1-\Delta\}^\I\cap \prod_{\i\in \I} [\a_\i, \b_\i).
$$ In particular, $|Q|=\Delta^{-|\I|}\prod_{\i\in \I} (\b_\I-\a_\I)$.
\end{claim}
\begin{proof}
Recall that by Definition~\dualref{def:rect} one has 
$$
\rect(\I\cup \H, \a,  \b)=\{y\in [m]^n: \langle y, \i\rangle \pmod{M}\in \left[\a_\i, \b_\i\right)\cdot M \text{~for all~}\i \in \I\cup \H\},
$$
which means that 
\begin{equation*}
\begin{split}
\{y\in [m]^n:& \langle y, \i\rangle \pmod{M}\in \left[\a_\i, \b_\i\right)\cdot M \text{~for all~}\i \in \I\cup \H\}\\
=\{y\in [m]^n:& \langle y, \i\rangle \pmod{M}\in \left[\a_\i, \b_\i\right)\cdot M \text{~for all~}\i \in \I\\
&\text{and}\\
& \langle y, \i\rangle \pmod{M}\in \left[\a_\i, \b_\i\right)\cdot M \text{~for all~}\i \in \H\}\\
=\bigcup_{\f\in Q} \{y\in [m]^n:& \langle y, \i\rangle \pmod{M}\in \left[\f_\i, \f_\i+\Delta\right)\cdot M \text{~for all~}\i \in \I\\
&\text{and}\\
& \langle y, \i\rangle \pmod{M}\in \left[\a_\i, \b_\i\right)\cdot M \text{~for all~}\i \in \H\},\\
\end{split}
\end{equation*}
where 
$$
Q=\{0, \Delta,2 \Delta, \ldots, 1-\Delta\}^\I\cap \prod_{\i\in \I} [\a_\i, \b_\i).
$$ It remains to note that for every $\f\in Q$ one has
\begin{equation*}
\begin{split}
\{y\in [m]^n:& \langle y, \i\rangle \pmod{M}\in \left[\f_\i, \f_\i+\Delta\right)\cdot M \text{~for all~}\i \in \I\\
&\text{and}\\
& \langle y, \i\rangle \pmod{M}\in \left[\a_\i, \b_\i\right)\cdot M \text{~for all~}\i \in \H\}\\
&=\rect(\I\cup \H, (\f, \a_\H), (\f+\Delta\cdot \mathbf{1}_\I, \b_\H)).
\end{split}
\end{equation*}
\end{proof}

As mentioned below, cubes will serve as our basic building blocks. For example, the local permutation map $\Pi_{R'\to R}$ (see Definition~\dualref{def:pi} in Section~\dualref{sec:loc-pi} below) is defined on individual cubes and then extended to a global map $\Pi^*$ (see Definition~\dualref{def:pi-global}), ultimately letting us define the glueing map $\tau$ (see Definition~\dualref{def:tau} below).

\begin{lemma}[Bounds on sizes of rectangles]\duallabel{lm:rect-size}
For every $\I\subseteq \F$ such that $|\I|\leq K^2$, for every $\c, \d\in (\deltagrid)^\I, \c< \d$, 
$$
\gamma=\prod_{\i\in \I} (\d_\i-\c_\i)
$$
and 
$$
R=\rect(\I, \c, \d), 
$$
the following conditions hold:

\begin{description}
\item[(1)] the cardinality of $R$ is bounded as
$$
(1-\sqrt{\e})\gamma \leq |R|/m^n\leq (1+\sqrt{\e})\cdot \gamma
$$
\item[(2)] for every positive integer $\lambda\leq K$, if 
$$
R'=\{x\in R: \weight(x)\pmod{W}\in [0, 1/\lambda)\cdot W\},
$$ 
then the cardinality of $R'$ is bounded as
$$
\frac1{\lambda}\cdot (1-\sqrt{\e})\gamma \leq \left|R'\right|/m^n\leq \frac1{\lambda}\cdot (1+\sqrt{\e}) \gamma.
$$
\end{description}
\end{lemma}

We now prove Lemma~\dualref{lm:size-bounds}, restated here for convenience of the reader:

\noindent{{\em {\bf Lemma~\dualref{lm:size-bounds}} (Restated)
One has
\begin{itemize}
\item[{\bf (1)}] For every $k\in [K/2+1]$ one has $|T_k|=(1\pm \sqrt{\e})\cdot |T_0|(1-k/K)$;
\item[{\bf (2)}] For every $k\in [K/2]$ one has  $|S_k|=(1\pm \sqrt{\e}) \cdot |T_0|/K$;
\item[{\bf (3)}] For every $k\in [K/2]$, every $\j\in \B_k$ one has $|S^\j_k|=(1\pm \sqrt{\e})(1-\frac1{K-k})|T_0|/K$.
\item[{\bf (4)}] For every $k\in [K/2]$, every $\j\in \B_k$ one has $|T^\j_k|=(1\pm \sqrt{\e})(1-\frac{k+1}{K})|T_0|$.
\end{itemize}
}
\begin{proof}
We start with {\bf (1)}. Let $R=\rect(\J, \c, \d)$, where $\J=\J_{< k}$ and for every $s=0,\ldots, k-1$ one has $\c_{\j_s}=0$ and $\d_{\j_s}=1-\frac1{K-s}$, and note that $R=T_k$ by~\dualeqref{eq:def-tk-all-constraints} . By Lemma~\dualref{lm:rect-size}, {\bf (1)}, one has 
$$
(1-\sqrt{\e})\gamma \leq |\rect(\I, \c, \d)|/m^n\leq (1+\sqrt{\e})\cdot \gamma,
$$
where 
\begin{equation*}
\begin{split}
\gamma&=\prod_{\i\in \I} (\d_\i-\c_\i)\\
&=\prod_{s=0}^{k-1} \left(1-\frac1{K-s}\right)\\
&=\prod_{s=0}^{k-1} \frac{K-s-1}{K-s}\\
&=\frac{K-(k-1)-1}{K}\\
&=1-k/K,
\end{split}
\end{equation*}
as required.  The proof of {\bf (4)} is analogous.

We now prove {\bf (2).} Let $R=T_k=\rect(\I, \c, \d)$, where $\I=\J_{< k}$ and for every $s=0,\ldots, k-1$ one has $\c_{\j_s}=0$ and $\d_{\j_s}=1-\frac1{K-s}$.  Let
$$
R':=\{x\in R: \weight(x)\pmod{W}\in [0, 1/(K-k))\cdot W\},
$$
and note that $R'=S_k$ by~\dualeqref{eq:def-sk}.  Then by Lemma~\dualref{lm:rect-size}, {\bf (2)}, with $\lambda=K-k$ and 
$$
\gamma=\prod_{s=0}^{k-1} (\d_{\j_s}-\c_{\j_s})=\prod_{s=0}^{k-1} \left(1-\frac1{K-s}\right)=1-\frac{k}{K}
$$
we have
\begin{equation*}
\frac1{K-k}\cdot (1-\sqrt{\e})\left(1-\frac{k}{K}\right)\leq \left|R'\right|/m^n\leq \frac1{K-k}\cdot (1+\sqrt{\e}) \left(1-\frac{k}{K}\right).
\end{equation*}
Simplifying, we get
\begin{equation*}
(1-\sqrt{\e})\frac1{K}\leq \left|R'\right|/m^n\leq  (1+\sqrt{\e}) \frac1{K},
\end{equation*}
as required.

We now prove {\bf (3)}. Similarly to {\bf (2)}, let $R=T_k=\rect(\I, \c, \d)$, where $\I=\J_{< k}\cup \{\j\}$. For every $s=0,\ldots, k-1$ one has $\c_{\j_s}=0$ and $\d_{\j_s}=1-\frac1{K-s}$.  Also let $\c_\j=0$  and $\d_\j=1-\frac1{K-k}$. Let
$$
R':=\{x\in R: \weight(x)\pmod{W}\in [0, 1/(K-k))\cdot W\},
$$
and note that $R'=S_k^\j$ by~\dualeqref{eq:skj}.  Then by Lemma~\dualref{lm:rect-size}, {\bf (2)}, with $\lambda=K-k$ and 
$$
\gamma=(\d_\j-\c_\j)\cdot \prod_{s=0}^{k-1} (\d_{\j_s}-\c_{\j_s})=\left(1-\frac{1}{K-k}\right)\prod_{s=0}^{k-1} \left(1-\frac1{K-s}\right)=\left(1-\frac{1}{K-k}\right)\left(1-\frac{k}{K}\right)
$$
we have
\begin{equation*}
\frac1{K-k}\cdot (1-\sqrt{\e})\left(1-\frac{1}{K-k}\right)\left(1-\frac{k}{K}\right)\leq \left|R'\right|/m^n\leq \frac1{K-k}\cdot (1+\sqrt{\e}) \left(1-\frac{1}{K-k}\right)\left(1-\frac{k}{K}\right).
\end{equation*}
Simplifying, we get
\begin{equation*}
(1-\sqrt{\e})\left(1-\frac{1}{K-k}\right)\frac1{K}\leq \left|R'\right|/m^n\leq  (1+\sqrt{\e}) \left(1-\frac{1}{K-k}\right)\frac1{K},
\end{equation*}
as required.

\end{proof}

\subsection{Interior and exterior of a rectangle}
The main difference between our main construction in this section and the toy construction from Section~\ref{sec:toy-construction} is the fact that vectors in $\F$ are not orthogonal, but merely have small dot products. As a consequence, we generally need to introduce some `padding' to our construction to obtain the same induced properties as we did in the original construction. For example, note that for the basic Lemma~\ref{lm:disjoint} that shows that edge sets $E_{k, j}$ defined in~\eqref{eq:edges-ei-def} are disjoint for distinct $j\in \B_k$ it was sufficient to ensure that we have introduce a complete bipartite graph  between $\text{line}_j(y)\cap S_k^j$ and $\text{line}_j(y) \cap (T_k\setminus T_k^j)$ -- the fact that (the downset of) $T_k^j$ is subtracted in the second set was enough to guarantee disjointness. To ensure similar property with nearly orthogonal vectors, however, one must include some `margin of error' in the construction -- this is why the corresponding definition in our main construction (see~\dualeqref{eq:edges-eij-def}) uses the interior $\Int_\delta(S_k^\j)$ as opposed to just $S_k^\j$. We define the interior now.

\begin{definition}[$\delta$-interior of (a downset of) a rectangle]\duallabel{def:int}
For $\I\subseteq \F$, $\c, \d \in (\deltagrid)^\I, \c<\d$, the $\delta$-interior $\Int_\delta(F)$ of the rectangle $F=\rect(\I, \c, \d)$ is defined as
$$
\Int_\delta(F)=\{y\in [m]^n\setminus B: \langle y, \i\rangle \pmod{M}\in \left[\c_\i+\delta, \d_\i-\delta\right)\cdot M \text{~for all~}\i\in \I\},
$$
where the set $B$ of boundary points is as in Definition~\dualref{def:boundary}.
For every $k\in [K/2]$ we define 
\begin{equation*}
\begin{split}
\Int_\delta(\downset_k(F))=\left\{y\in \Int_\delta(F): \weight(x) \in \left[0, \frac{1}{K-k}\right)\cdot W \pmod{ W}\right\}.
\end{split}
\end{equation*}

\end{definition}

The following simple claim is the rationale behind our definition of the interior of a rectangle:
\begin{lemma}[Vertex neighborhood of $\Int_\delta(R)$ is contained in $R$]\duallabel{lm:shift-int}
If $\e<\delta$, for every rectangle $R\subseteq [m]^n$, $R=(\I, \a, \b)$, $\a, \b\in (\deltagrid)^\I, \a<\b,$ every $\r\in \F\setminus \I$  for every integer $\lambda$ such that $|\lambda|\leq M/w$, for every $x\in \Int_\delta(R)$ one has $x+\lambda \r\in R$.
\end{lemma}
\begin{proof}
For every $\i\in \I$ one has 
$$
\left|\langle x+\lambda\cdot \r, \i\rangle-\langle x, \i\rangle\right|=|\lambda|\cdot \langle \r, \i\rangle\leq (M/w)\cdot \e \cdot w\leq \e M<\delta M
$$
since $\r\in \F\setminus \I$ by assumption of the lemma. Since 
$$
x\in \Int_\delta(F)=\{y\in [m]^n\setminus B: \langle y, \i\rangle \pmod{M}\in \left[\c_\i+\delta, \d_\i-\delta\right)\cdot M \text{~for all~}\i\in \I\}
$$
by assumption, we get that 
$$
x+\lambda \cdot \r\in \{y\in [m]^n: \langle y, \i\rangle \pmod{M}\in \left[\c_\i, \d_\i\right)\cdot M \text{~for all~}\i\in \I\}=F,
$$
as required. Note that the assumption that $x\in \Int_\delta(F)\subseteq [m]^n\setminus B$ is used to ensure that for every $j\in [n]$ one has $0\leq (x+\lambda \cdot \r)_j <m$, and therefore $x+\lambda \cdot \r\in [m]^n$. Indeed, we have
$$|(x+\lambda \cdot \r)_j-x_j|\leq |\lambda|\leq M/w,$$
and therefore since $n^2\leq x_j\leq m-n^2$ by assumption that $x\not \in B$, together with the fact that $M/w$ is a constant depending on $K$ (by~\dualref{p0}, \dualref{p1} and~\dualref{p2}) and $n$ is sufficiently large, we get that $0\leq (x+\lambda \cdot \r)_j <m$.
\end{proof}

We also define
\begin{definition}[$\delta$-exterior of a rectangle]\duallabel{def:exterior}
For $\I\subseteq \F$, $\c, \d\in (\deltagrid)^\I, \c<\d$, the $\delta$-exterior $\Ext_\delta(F)$ of the rectangle $F=\rect(\I, \c, \d)$ is defined as follows. 

If $\c_\i>\delta$, then 
$$
\Ext_\delta(F)=\{y\in [m]^n: \langle y, \i\rangle \pmod{M}\in \left[\c_\i-\delta, \d_\i+\delta\right)\cdot M \},
$$
and
$$
\Ext_\delta(F)=\{y\in [m]^n: \langle y, \i\rangle \pmod{M}\in \left[0, \d_\i+\delta\right)\cdot M\cup \left[1-\delta+\c_\i, 1\right)\cdot M \}
$$
otherwise.
\end{definition}

The interior (resp. exterior) of a rectangle is quite close to the rectangle itself in terms of size, i.e. there are few points on the boundary (under appropriate conditions):
\begin{lemma}\duallabel{lm:rect-int-size}
For every $\I\subseteq \F, |\I|\leq K^2$, for every $\c, \d\in (\deltagrid)^\I, \c<\d$, if $R=\rect(\I, \c, \d)$, one has
$$
|R\setminus \Int_\delta(R)|\leq \sqrt{\delta} |R|
$$
and 
$$
|\Ext_\delta(R)\setminus R|\leq \sqrt{\delta}|R|
$$
\end{lemma}
\begin{proof}
We start by proving {\bf (1)}. We have
\begin{equation}\duallabel{eq:9hgh2ttDFE}
\begin{split}
|R\setminus \Int_\delta(R)|&=\left|\left\{x\in R: \langle x, \i\rangle \pmod{M}\in \left(\left[\c_\i, \c_\i+\delta) \cup [\d_\i-\delta, \d_\i\right)\right)\cdot M\text{~for some~}\i\in \I\right\}\right|\\
&\leq \sum_{\i\in \I} \left|\left\{x\in R: \langle x, \i\rangle \pmod{M}\in \left(\left[\c_\i, \c_\i+\delta) \cup [\d_\i-\delta, \d_\i\right)\right)\cdot M\right\}\right|\\
&\leq \sum_{\i\in \I} |R_\i\setminus \Int_\delta(R_\i)|,
\end{split}
\end{equation}
where we let $R_\i:=\rect(\{\i\}, \c_\i, \d_\i)$ to simplify notation.

We now fix $\i\in \I$ and upper bound $|R_\i\setminus \Int_\delta(R_\i)|$.
Let $C\subset [m]^n$ be a minimal $\{\i\}$-subspace cover (see Definition~\dualref{def:subspace-cover}). Fix $x\in C$. Recall that
$$
\text{line}_\i(x)=\left\{x'\in [m]^n: x'=x+\lambda\cdot \i\text{~for~some~integer~}\lambda\text{~s.t.~}\left\lfloor\langle x', \i\rangle/M\right\rfloor =\left\lfloor \langle x, \i\rangle/M \right\rfloor \right\}.
$$
By Claim~\dualref{cl:line-size}, {\bf (2)} and {\bf (3)}, we have for $x\in C\setminus B$
\begin{equation}
\begin{split}
|\text{line}_\i(x)\cap (R_\i\setminus \Int_\delta(R_\i))|&=\left|\{y\in \text{line}_\i(x): \langle y, \i\rangle \pmod{M}\in [\c_\i, \c_\i+\delta)\cdot M\}\right|\\
&+\left|\{y\in \text{line}_\i(x): \langle y, \i\rangle \pmod{M}\in [\d_\i-\delta, \d_\i)\cdot M\}\right|\\
&=2\delta\cdot M/w,\\
\end{split}
\end{equation}
where we used the fact that $\text{line}_\i(x)\subset [m]^n$ for all $x\in [m]^n\setminus B$.

Summing over all $x\in C$, we thus get
\begin{equation}\duallabel{eq:923hg9hg9hHfsf}
\begin{split}
|R_\i\setminus \Int_\delta(R_\i)|&=\sum_{x\in C} |\text{line}_\i(x)\cap  (R_\i\setminus \Int_\delta(R_\i))|\\
&=\sum_{x\in C\setminus B} |\text{line}_\i(x)\cap  (R_\i\setminus \Int_\delta(R_\i))|+\sum_{x\in B} |\text{line}_\i(x)\cap  (R_\i\setminus \Int_\delta(R_\i))|\\
&\leq 2\delta (M/w)\cdot |C\setminus B|+\sum_{x\in B} |\text{line}_\i(x)\cap  (R_\i\setminus \Int_\delta(R_\i))|.
\end{split}
\end{equation}

We now note that since $|\text{line}_\i(x)|=M/w$ for every $x\in [m]^n \setminus B$  by Claim~\dualref{cl:line-size}, {\bf (1)}, we have
\begin{equation*}
\begin{split}
|C\setminus B|&=(M/w)^{-1} \sum_{x\in C\setminus B} |\text{line}_\i(x)|\\
&\leq (M/w)^{-1} \sum_{x\in C} |\text{line}_\i(x)|\\
&\leq (M/w)^{-1} m^n.
\end{split}
\end{equation*}

Substituting this into~\dualeqref{eq:923hg9hg9hHfsf}, we get
\begin{equation*}
\begin{split}
|R_\i\setminus \Int_\delta(R_\i)|&\leq 2\delta (M/w)\cdot |C\setminus B|+\sum_{x\in B} |\text{line}_\i(x)\cap  (R_\i\setminus \Int_\delta(R_\i))|\\
&\leq 2\delta (M/w)\cdot |C\setminus B|+\sum_{x\in B} |\text{line}_\i(x)|\\
&\leq 2\delta m^n+(M/w)\cdot \frac1{n^{10}} \cdot m^n\\
&\leq 3\delta m^n,
\end{split}
\end{equation*}
where the third transition uses the fact that 
$|\text{line}_\i(x)|\leq M/w$ for every $x\in [m]^n \setminus B$ by Claim~\dualref{cl:line-size}, {\bf (1)}. 

Combining the above with~\dualeqref{eq:9hgh2ttDFE}, we get
\begin{equation*}
\begin{split}
|R\setminus \Int_\delta(R)|&\leq \sum_{\i\in \I} |R_\i\setminus \Int_\delta(R_\i)|\\
&\leq 3\delta \cdot |\I|\cdot m^n\\
&\leq 3\delta \cdot |\I|\cdot 2\Delta^{-|\I|} |R|\\
&\leq \sqrt{\delta} |R|,\\
\end{split}
\end{equation*}
as required. The third transition use the fact that by Lemma~\dualref{lm:rect-size}, {\bf (1)} one has
$(1-\sqrt{\e}) \Delta^{|\I|} \leq |R|/m^n\leq (1+\sqrt{\e}) \Delta^{|\I|}$
as well as the assumption that $\e$ is smaller than an absolute constant (smaller than $1/4$ suffices here). The forth transition uses the assumption that $|\I|\leq K^2$ together with the assumption that $\delta<\Delta^{100K^2}$ by \dualref{p5}.

The proof of {\bf (2)} is similar and we omit the details.
\end{proof}

\subsection{Subspaces and their properties}

We now introduce the notion of subspaces, our main tool in defining the local permutation map $\Pi$, and ultimately the map $\tau$ glueing together two basic gadgets (see Section~\dualref{sec:loc-pi} and Section~\dualref{sec:tau} below). We first introduce

\begin{definition}[Block of $x$ with respect to a sequence of vectors $\J$]
For a subset $\J\subset \F$ we let $\block_\J(x):=(\lfloor \langle x, \j\rangle/M\rfloor)_{\j\in \J}$.
\end{definition}

\begin{definition}[Subspace of $x$]\duallabel{def:subspace}
For every subset $\I\subseteq \F$ for every $x\in [m]^n$ define
\begin{equation*}
\begin{split}
\subspace_\I(x)&:=\left\{x'\in [m]^n: x'=x+\sum_{\i\in \I} t_\i \cdot \i\text{~~for~}t\in \mathbb{Z}^\I\right.\\
&\left.\text{~~~~~~s.t.~}\block_\I(x)=\block_\I(x')\text{~and~}||t||_\infty\leq 2M/w\right\}.
\end{split}
\end{equation*}
\end{definition}

The more lightweight definition of lines used in Section~\dualref{sec:g-ell} to define the edge set $E^\ell$ of our basic gadget $G$ in fact coincides with a one-dimensional subspace as per Definition~\dualref{def:subspace}, as we show below. This lets us reuse claims about subspaces:
\begin{claim}\duallabel{cl:line-vs-subspace}
For every $\j\in \F$, then for every $x\in [m]^n$ one has $\text{line}_\j(x)=\subspace_{\{\j\}}(x)$, where $\text{line}_\j(x)$ is as per Definition~\dualref{def:line}.
\end{claim}
\begin{proof}
We have by Definition~\dualref{def:line}
$$
\text{line}_\j(x)=\left\{x'\in [m]^n: x'=x+\lambda\cdot \j\text{~for~some~integer~}\lambda\text{~s.t.~}\left\lfloor\langle x', \j\rangle/M\right\rfloor =\left\lfloor \langle x, \j\rangle /M\right\rfloor \right\}.
$$
and by Definition~\dualref{def:subspace}
\begin{equation*}
\begin{split}
\subspace_{\{\j\}}(x)&:=\left\{x'\in [m]^n: x'=x+ \lambda \cdot \j\text{~~for~}\lambda\in \mathbb{Z}\right.\\
&\left.\text{~~~~~~s.t.~}\block_{\{\j\}}(x)=\block_{\{\j\}}(x')\text{~and~}|\lambda|\leq 2M/w\right\}.
\end{split}
\end{equation*}

At the same time if $x'=x+\lambda\cdot \j$ for an integer $\lambda$, one has 
\begin{equation*}
\begin{split}
\langle x', \j\rangle=\langle x+\lambda\cdot \j, \j\rangle=\langle x, \j\rangle+\lambda\cdot w, 
\end{split}
\end{equation*}
so if $|\lambda|>2M/w$ (for example, when $\lambda>2M/w$; the other case is similar), one has 
\begin{equation*}
\begin{split}
\block_{\{\j\}}(x')&=\lfloor\langle x', \j\rangle/M\rfloor=\lfloor (\langle x, \j\rangle+2M)/M\rfloor=\lfloor \langle x, \j\rangle/M+2\rfloor\geq \lfloor \langle x, \j\rangle/M\rfloor+1=\block_{\{\j\}}(x)+1.
\end{split}
\end{equation*}
Thus, the constraint $|\lambda|\leq 2M/w$ is implied by the constraint $\block_{\{\j\}}(x')=\block_{\{\j\}}(x)$, and thus $\text{line}_\j(x)=\block_{\{\j\}}(x)$, as required.\end{proof}
\begin{remark}
We note that while Claim~\dualref{cl:line-vs-subspace} shows that the $\ell_\infty$ constraint in Definition~\dualref{def:subspace} is redundant when $|\I|=1$, it is not redundant for general $\I$, since the vectors in $\F$ are only nearly orthogonal.
\end{remark}

We show that subspaces partition $[m]^n$. This fact is key, and lets us define various maps (e.g., the local permutation map $\Pi$, see~Section~\dualref{sec:loc-pi}), locally on subspaces, and then naturally extend them to the full space.

\begin{lemma}[Subspaces form a partition]\duallabel{lm:subspace-partition}
For every $\I\subset \F$, every $\e\in (0, 1/(10 |\I|))$,  every $x, x'\in [m]^n$  one has either $\subspace_\I(x)=\subspace_\I(x')$ or $\subspace_\I(x)\cap \subspace_\I(x')=\emptyset$.
\end{lemma}
\begin{proof}
Consider an element $y\in \subspace_\I(x)\cap \subspace_\I(x')$. There exist integer coefficients $(t_\i)_{\i\in \I}$ and $(t'_\i)_{\i\in \I}$ such that
$$
x+\sum_{\i\in \I} t_\i \cdot \i=y=x'+\sum_{\i\in \I} t'_\i \cdot \i,
$$
so that 
$$
x'-x=\sum_{\i\in \I} (t'_\i-t_\i) \cdot \i.
$$
At the same time for every $z\in \subspace_\I(x)$ one has $\block_\I(z)=\block_\I(x)$, and there exists integer coefficients $(s_\i)_{\i\in \I}$ such that
$z=x+\sum_{\i\in \I} s_\i \cdot \i$. Combining this with the equation above, we get
$$
z=x+\sum_{\i\in \I} s_\i \cdot \i=x'+\sum_{\i\in \I} (s_\i+t'_\i-t_\i) \cdot \i.
$$ 
The existence of $y\in \subspace_\I(x)\cap \subspace_\I(x')$ also implies that $\block_\I(x)=\block_\I(y)=\block_\I(x')$, and hence $\block_\I(z)=\block_\I(x)=\block_\I(x')$. Thus, in order to show that $z\in \subspace_\I(x')$, it suffices to prove that $|s_\i+t'_\i-t_\i|\leq 2M/w$ for all $\i\in \I$, i.e. $||s+t'-t||_\infty\leq 2M/w$. Suppose not, and let $\j\in \I$ be such that $|s_\j+t'_\j-t_\j|>2M/w$. Then we have, recalling that $\langle \i, \i\rangle=w$ for all $\i\in \F$ and $\langle \i, \i'\rangle\leq \e w$ for $\i, \i'\in \F$, $\i\neq \i'$,
\begin{equation*}
\begin{split}
\langle z, \j\rangle=\langle x, \j\rangle+(s_\j+t'_\j-t_\j)\cdot w+\sum_{\i\in \I\setminus \{\j\}} (s_\i+t'_\i-t_\i)\cdot \langle \i, \j\rangle,
\end{split}
\end{equation*}
so 
\begin{equation*}
\begin{split}
\left|\langle z, \j\rangle-\langle x, \j\rangle-(s_\j+t'_\j-t_\j)\cdot w\right|&\leq \e |\I|\cdot ||s+t'-t||_\infty\cdot w\\
&\leq \e |\I| (||s||_\infty+||t'||_\infty+||t||_\infty)\cdot w\\
&\leq 6\e |\I|\cdot M,
\end{split}
\end{equation*}
where in the last transition we used the fact that $||s||_\infty\leq 2M/w$, $||t||_\infty\leq 2M/w$ and $||t'||_\infty\leq 2M/w$. We thus have, since $\e<1/(10 |\I|)$ by assumption of the lemma,  
\begin{equation*}
\begin{split}
\left|\langle z, \j\rangle-\langle x, \j\rangle\right|\geq |s_\j+t'_\j-t_\j|\cdot w-6\e |\I|\cdot M>(2M/w) \cdot w-6\e |\I|\cdot M> M.
\end{split}
\end{equation*}
This means that $\lfloor\langle z, \j\rangle/M\rfloor\neq \lfloor\langle x, \j\rangle/M\rfloor$, and hence $\block_\I(z)\neq \block_\I(x)$, which is a contradiction. We thus get that $||s+t'-t||_\infty\leq 2M/w$, and hence $z\in \subspace_\I(x')$, as required.

\end{proof}

Since subspaces partition $[m]^n$, we often select a minimal number of representative points subspaces through which cover the entire space, and define, e.g., the local permutation map $\Pi$ (see~Section~\dualref{sec:loc-pi}), on subspaces through these representative points.

\begin{definition}[Minimal $\I$-subspace cover]\duallabel{def:subspace-cover}
We say that a set  $C\subset [m]^n$ is a minimal $\I$-subspace cover if 
$$
\bigcup_{x\in C} \subspace_\I(x)=[m]^n
$$ and $\subspace_\I(x)\cap \subspace_\I(x')=\emptyset$ for $x, x'\in C$, $x\neq x'$. 
\end{definition}

It follows from Lemma~\dualref{lm:subspace-partition} that for every $\I\subseteq \F$ there exists a minimal $\I$-subspace cover $C$: start with $C$ being the empty set and iteratively add $x\in [m]^n$ to $C$ if $\subspace_\I(x)\cap \subspace_\I(x')=\emptyset$ for every $x'\in C$.

\begin{lemma}[Intersection of a rectangle with a subspace]\duallabel{lm:rect-subspace-size}
For every $\I, \J\subset \F, |\I|, |\J|\leq K^2$, every $\a, \b\in \deltagrid^\J, \a<\b$, if 
$$
\gamma=\prod_{\i\in \I\cap \J} (\b_\i-\a_\i)
$$
and 
$$
R=\rect(\J, \a, \b),
$$
 the following conditions hold.

\begin{description}
\item[(1)] For every $x\in [m]^n\setminus B$ one has
$$
(1-\e^{2/3}) \cdot \gamma\cdot G \leq \left|\subspace_\I(x)\cap R\right|\leq (1+\e^{2/3})\cdot \gamma\cdot G,
$$
where $G=(M/w)^{|\I|}$. 

\item[(2)] For every positive integer $\lambda\leq K$ such that $\lambda \mid W/w$, if 
$$
R'=\{x\in R: \weight(x)\pmod{W}\in [0, 1/\lambda)\cdot W\},
$$
one has for every $x\in [m]^n\setminus B$
$$
(1-\e^{2/3})\cdot \frac1{\lambda}\cdot \gamma\cdot G \leq \left|\subspace_\I(x)\cap R'\right|\leq (1+\e^{2/3})\cdot \frac1{\lambda}\cdot \gamma\cdot G,
$$
where $G=(M/w)^{|\I|}$. 
\end{description}
\end{lemma}

\subsection{Large matchings in individual gadgets}

We prove that the basic gadget $G=(S, T)$ contains a matching of most of $S$ to $T\setminus T_*$:
\begin{lemma}\duallabel{lm:large-matching}
There exists a matching of a $(1-O(1/K))$ fraction of vertices in $S$ to $T\setminus T_*$.
\end{lemma}
\begin{proof}
The proof proceeds in two steps. In {\bf step 1} we show that for every $k\in [K/2]$, every $x\in [m]^n\setminus B$ one has 
$$
\left|\text{line}_\j(x)\cap S_k^\j\right|=(1\pm O(1/K)) \left|\text{line}_\j(x) \cap (T_k\setminus T_k^\j)\right|,
$$
which in particular implies that a complete bipartite graph between these two sets of vertices contains a matching of required size. In {\bf step 2} we use this fact to conclude the result of the lemma, in particular taking care of the fact that the actual edge set of $G^\ell$ only contains a complete graph between $\text{line}_\j(x)\cap \Int_\delta(S_k^\j)$ and $\text{line}_\j(x) \cap (T_k\setminus T_k^\j)$.

\paragraph{Step 1: defining the matching on lines.} Fix $k\in [K/2]$. Let $\j=\J_k$, and recall that for every $x\in [m]^n$ one has $\text{line}_\j(x)=\subspace_{\{\j\}}(x)$  by Claim~\dualref{cl:line-vs-subspace}.  Let $R=T_k=\rect(\J, \c, \d)$, where $\J=\J_{<k}$ and for every $s=0,\ldots, k-1$ one has $\c_{\j_s}=0$ and $\d_{\j_s}=1-\frac1{K-s}$. For every $x\in [m]^n\setminus B$ by Lemma~\dualref{lm:rect-subspace-size}, {\bf (1)},  one has 
\begin{equation}\duallabel{eq:line-tk}
(1-\sqrt{\e})\cdot (M/w)  \leq \left|\text{line}_\j(x)\cap T_k\right|\leq (1+\sqrt{\e})\cdot (M/w),
\end{equation}
where $G=M/w$. Note that the error term in the lemma is $\e^{2/3}<\sqrt{\e}$ since $\e\in (0, 1)$. Also note that in the application of the lemma we have $\gamma=1$, since $\j\not \in \J_{<k}$.

Now let $R=T_k^\j=T_{k+1}=\rect(\J, \c, \d)$(since $\j=\J_k$), where $\J=\J_{\leq k}$ and for every $s=0,\ldots, k$ one has $\c_{\j_s}=0$ and $\d_{\j_s}=1-\frac1{K-s}$ . For every $x\in [m]^n\setminus B$ by Lemma~\dualref{lm:rect-subspace-size}, {\bf (1)}, one has 
\begin{equation}\duallabel{eq:line-tkj}
(1-\sqrt{\e})\cdot \left(1-\frac1{K-k}\right)\cdot (M/w)  \leq \left|\text{line}_\j(x)\cap T_k^\j\right|\leq (1+\sqrt{\e})\cdot \left(1-\frac1{K-k}\right)\cdot (M/w)
\end{equation}
Note that in the application of the lemma we have $\gamma=\d_\j-\c_\j=1-\frac1{K-k}$, since $\{\j\} \cap \J_{\leq k}=\{\j\}$. Putting~\dualeqref{eq:line-tk},~\dualeqref{eq:line-tkj} together, we get 
\begin{equation}\duallabel{eq:9y39h9jgsgsadf}
|\text{line}_\j(x) \cap (T_k\setminus T_k^\j)|=(1+O(K\sqrt{\e}))\cdot \frac1{K-k}\cdot (M/w)
\end{equation}

We now bound $|\text{line}_\j(x)\cap S_k^\j|$. To that effect let 
$$
R':=\left\{x\in R: \weight(x)\pmod{W}\in \left[0, \frac1{K-k}\right)\cdot W\right\},
$$
and note that $R'=S_k^\j$ by~\dualeqref{eq:skj}. For every $x\in [m]^n\setminus B$ by Lemma~\dualref{lm:rect-subspace-size}, {\bf (2}), one has 
\begin{equation}\duallabel{eq:line-skj}
(1-\sqrt{\e}) \frac1{K-k}\left(1-\frac1{K-k}\right)\cdot (M/w)\leq |\text{line}_\j(x)\cap S_k^\j|\leq (1+\sqrt{\e}) \frac1{K-k}\left(1-\frac1{K-k}\right)\cdot (M/w).
\end{equation}

Now recall that by~\dualeqref{eq:edges-ei-def} for every $\j\in \B_k$ and every $y\in C_\j$ (for a minimal $\j$-line cover $C_\j$) the edge set $E_k$ contains all edges in the set
\begin{equation}\duallabel{eq:923y523tSFIHF}
(\text{line}_\j(x)\cap \Int_\delta(S_k^\j)) \times (\text{line}_\j(x) \cap (T_k\setminus T_k^\j)).
\end{equation}

Putting ~\dualeqref{eq:9y39h9jgsgsadf} together with ~\dualeqref{eq:line-skj}, using the fact that $O(K\sqrt{\e})=O(1/K)$ by~\dualref{p3},\dualref{p5} and \dualref{p6}, and recalling that $0\leq k\leq K/2-1$, we get that for every $x\in [m]^n\setminus B$ there exists a matching of a $1-O(1/K)$ fraction of $\text{line}_\j(x)\cap S_k^\j$ to $\text{line}_\j(x) \cap (T_k\setminus T_k^\j)$ using edges in 
\begin{equation*}
(\text{line}_\j(x)\cap S_k^\j) \times (\text{line}_\j(x) \cap (T_k\setminus T_k^\j)).
\end{equation*}
We show in {\bf step 2} below that taking the union of these matchings over $y\in C_\j$ and restricting the resulting matching to edges that do not touch $S_k^\j\setminus \Int_\delta(S_k^\j)$ reduces the size of the matching only slightly, and ensures that the matching uses only the edges that are present in the graph, i.e. edges in~\dualeqref{eq:923y523tSFIHF}, as required.

\paragraph{Step 2: defining the global matching.} Let $C_\j\subseteq [m]^n$ denote a minimal $\j$-line cover (one can think of this cover as the one used to defined the corresponding edge set of $G$; however, one notes that the actual edge set does not depend on the specific choice of a cover). In {\bf step 1} we showed the existence of a matching of a $1-O(1/K)$ fraction of $\text{line}_\j(x)\cap S_k^\j$ to $\text{line}_\j(x) \cap (T_k\setminus T_k^\j)$ for every $k\in [K/2]$ and every $x\in [m]^n\setminus B$ using edges in~\dualeqref{eq:923y523tSFIHF}.

We now note that for every $k\in [K/2]$
\begin{equation}\duallabel{eq:9h234g92h4g}
\begin{split}
\left|S_k\setminus \bigcup_{x\in C\cap B} (\text{line}_\j(x)\cap S_k^\j)\right|&\leq \left|S_k^\j\setminus \bigcup_{x\in C\cap B} (\text{line}_\j(x)\cap S_k^\j)\right|+|S_k\setminus S^\j_k|\\
&\leq |B|\cdot (M/w)+|S_k\setminus S_k^\j|,\\
\end{split}
\end{equation}
where we used the fact that $|\text{line}_\j(x)|\leq M/w$ by Claim~\dualref{cl:line-size}, {\bf (1)}, for all $x\in [m]^n$ and all $\j\in \F$.

We now bound the second term in~\dualeqref{eq:9h234g92h4g}. By Lemma~\dualref{lm:size-bounds}, {\bf (2)} and Lemma~\dualref{lm:size-bounds}, {\bf (3)}, one has
$$
|S_k|/m^n=(1\pm \sqrt{\e})  \cdot |T_0|/K
$$
and
$$
|S_k^\j|/m^n=(1\pm \sqrt{\e})  \left(1-\frac1{K-k}\right) \cdot |T_0|/K.
$$

This means that the second term in~\dualeqref{eq:9h234g92h4g} is upper bounded by 
$$
\frac1{K}\left(2\sqrt{\e}+\frac1{K-k}\right)|T_0|=O(1/K^2)\cdot |T_0|,
$$
where we used the fact that  
\begin{equation*}
\begin{split}
\sqrt{\e}&\leq \delta\text{~~~~~~~~~~~~~~~~~~~~(by~\dualref{p6})}\\
&\leq \Delta^{100K^2}\text{~~~~~~~~~~~(by~\dualref{p5})}\\
&\leq K^{-100K^2}\text{~~~~~~~~(by~\dualref{p3})}\\
&\leq K^{-4}.
\end{split}
\end{equation*}
We now bound the first term in~\dualeqref{eq:9h234g92h4g} by noting that by Claim~\dualref{cl:boundary-points}
$$
|B|\cdot (M/w)\leq n^{-9}\cdot m^n=O(1/K^2) |T_0|=O(1/K)|S_k|
$$ 
for every $k\in [K/2]$.

Putting the above bounds together, we get that for every $k\in [K/2]$ there exists a matching of all but a $O(1/K)$ fraction of $S_k$ to $T_k\setminus T_k^\j$, where $\j=\J_k$, using edges in~\dualeqref{eq:923y523tSFIHF}. It remains to remove from this matching edges incident on vertices in $S_k^\j\setminus \Int_\delta(S_k^\j)$. The matching is reduced by at most 
\begin{equation*}
\begin{split}
|S_k^\j\setminus \Int_\delta(S_k^\j)|&\leq |T_k^\j\setminus \Int_\delta(T_k^\j)|\\
&\leq \sqrt{\delta} |T_k^\j|\\
&\leq 2K\sqrt{\delta} |S_k^\j|\\
&=O(1/K) |S_k^\j|.
\end{split}
\end{equation*}
The first transition above is by definition of $S_k^\j$ and $S_k$ (see~\dualref{eq:skj} and~\dualref{eq:def-sk}). The second transition is by Lemma~\dualref{lm:rect-int-size}. The third transition is due to the fact that by Lemma~\dualref{lm:size-bounds}, {\bf (3)} and {\bf (4)}, one has $|S_k^\j|\geq (1/K) |T_k^\j|$. The forth transition is by~\dualref{p3} and~\dualref{p5}.

In other words, for every $k\in [K/2]$ there exists a matching of all but $O(1/K)$ fraction of $S_k$ to $T_k\setminus T_{k+1}$. Since the sets $T_k$ form a nested sequence, the sets $T_k\setminus T_{k+1}$ are disjoint, similarly to the sets $S_k$. Thus, the matchings extend to a matching of a $1-O(1/K)$ fraction of
$$
S=S_0\uplus S_1\uplus \ldots\uplus S_{K/2-1}
$$ 
to 
$$
\bigcup_{k\in [K/2]} T_k\setminus T_{k+1}=T_0\setminus T_{K/2}=T\setminus T_*.
$$
Since $\sum_{k\in [K/2]} |S_k|=\frac1{2}(1+O(1/K))\cdot |T_0|=(1+O(1/K))|T\setminus T_*|$ by Lemma~\dualref{lm:size-bounds}, {\bf (1)} and {\bf (2}) together with the choice of $\e$ (as per~\dualref{p3}, \dualref{p5} and~\dualref{p6}), the result of the lemma follows.
\end{proof}

\subsection{$1-e^{-1}$ hardness using basic gadgets}
We show how the $1-e^{-1}$ hardness from~\cite{Kapralov13} follows using our basic gadgets in Appendix~\ref{app:D}.

\subsection{Maps $\tau^\ell$ identifying the basic gadgets}\duallabel{sec:tau-ell}
The main result of this section is the definition of maps 
$$
\tau^\ell: S^\ell\to T_*^{\ell-1}
$$
mapping the $S$ side of the bipartition (the `arriving vertices') of the $\ell$-th gadget $G^\ell$ to the terminal subcube $T_*^{\ell-1}$ of the previous gadget $G^{\ell-1}$.

Fix $\ell\in [L], \ell>0$. To simplify notation, let $\B=\B^{\ell-1}, \B'=\B^\ell$, and recall that both sets are partitioned into $K/2$ disjoint equal size sets 
\begin{equation*}
\begin{split}
\B&=\B_0\cup \B_1\cup \ldots \cup \B_{K/2}\\
\B'&=\B'_0\cup \B'_1\cup \ldots \B'_{K/2}.
\end{split}
\end{equation*} 
Let $G=(S, T, E)=G^{\ell-1}, G'=(S', T', E')=G^\ell$. Let $\J=\J^{\ell-1}$,  $\J'=\J^\ell$, $\r=\r^{\ell-1}, \r'=\r^\ell$, and recall that 
\begin{equation*}
\begin{split}
\J&\in \B_0\times \B_1\times \ldots\times \B_{K/2}\\
\J'&\in \B'_0\times \B'_1\times \ldots \times \B'_{K/2}.
\end{split}
\end{equation*}

With this notation in place, we will define the map
$$
\tau: S'\to T^*\cup \{\bot\},
$$
where for a vertex $x\in S'$ we write $\tau(x)=\bot$ to denote the fact that $\tau$ is not defined on $x$. Thus, in essence $\tau$ is a partial map.
We later use $\tau$ to identify basic gadgets $G^\ell, \ell\in [L],$ arriving in the stream.  We start by defining an auxiliary map $\rho$ that we refer to as the {\em densifying map} (see Section~\dualref{sec:rho} below). The map $\rho$ maps a subsampled rectangle such as a set $S_k, k\in [K/2]$, to a regular rectangle. The map $\tau$ is then defined by composing $\rho$ with another auxiliary transformation that we refer to as the {\em local permutation map} defined in Section~\dualref{sec:loc-pi}. The map $\tau$ is then defined in Section~\dualref{sec:tau}.

\subsubsection{Densifying map $\rho$}\duallabel{sec:rho}

The densifying map is defined as follows:
\begin{definition}[$(\alpha, \r)$-densifying map]\duallabel{def:rho}
For a positive integer $\alpha$ and $\r\in \F$ the $(\alpha, \r)$-densifying map $\rho:[m]^n\setminus B\to [m]^n$ is defined as follows. For $x\in [m]^n$ and $\r\in \F$ we first let
$$
\langle x, \r\rangle \pmod{M}=aW+b (W/\alpha)+c,
$$
where $a\in [M/W]$,  $b\in [\alpha]$ and $c\in [W/\alpha]$. Then define 
\begin{equation*}
\begin{split}
\rho(x):=x-\r\cdot \left(\frac{W}{w}(1-1/\alpha)\cdot a+\frac{W}{\alpha w} b\right).
\end{split}
\end{equation*}
\end{definition}

We note that the map $\rho$ is well defined since for every $i\in [n]$ one has $(\rho(x))_i\leq x_i<m$ and 
\begin{equation*}
\begin{split}
(\rho(x))_i&=x_i-\r_i\cdot \left(\frac{W}{w}(1-1/\alpha)\cdot a+\frac{W}{\alpha w} b\right)\\
&\geq x_i-\left(\frac{W}{w}(1-1/\alpha)\cdot (M/W-1)+\frac{W}{\alpha w} (\alpha-1)\right)\\
&= x_i-\left(\frac{W}{w}(1-1/\alpha)\cdot (M/W-1)+\frac{W}{w} (1-1/\alpha)\right)\\
&= x_i-\left(\frac{W}{w}\cdot (M/W-1)+\frac{W}{w}\right)(1-1/\alpha)\\
&= x_i-\frac{M}{w}(1-1/\alpha)\\
&\geq 0
\end{split}
\end{equation*}
for all $x\in [m]^n\setminus B$ since $n$ is sufficiently large as a function of $M/w$, $W/w$, $K, \Delta, \delta,$ and $L$, and in particular $n>M/w$.

The next lemma summarizes the relevant properties of the map $\rho$:
\begin{lemma}[Densification of a subsampled set]\duallabel{lm:densification}
For every integer $\alpha\geq 2$, every $\r\in \F$, every rectangle $U\subseteq [m]^n$, $U=(\I, \a, \b)$, $\a, \b\in (\deltagrid)^\I, \a<\b,$ such that $\r\not \in \I$, the following conditions hold for the $(\alpha, \r)$-densifying map $\rho$ (see Definition~\dualref{def:rho}):
\begin{description}
\item[(1)] $\rho$ is injective;
\item[(2)] $\rho$ maps 
$$
\left\{x\in \Int_\delta(U): \weight(x) \in \left[0, 1/\alpha\right)\cdot W \pmod{ W} \right\}
$$
to 
$$
\left\{x\in U: \langle x, \r\rangle \pmod{M}\in \left[0, 1/\alpha\right)\cdot M\right\}
$$
\item[(3)] for every $x\in [m]^n$ one has $\rho(x)=x+\lambda\cdot \r$ for an integer $\lambda$ satisfying $|\lambda|\leq M/w$.
\end{description}
\end{lemma}
\begin{proof} 
We start by proving the {\bf (3)}. One has by Definition~\dualref{def:rho} 
$\rho(x)=x+\lambda\cdot \r$, where $\lambda=-\left(\frac{W}{w}(1-1/\alpha)\cdot a+\frac{W}{\alpha w} b\right)$ for $a\in [M/W]$ and $b\in [\alpha]$. We thus have
\begin{equation*}
\begin{split}
|\lambda|&=\left|\frac{W}{w}(1-1/\alpha)\cdot a+\frac{W}{\alpha w} b\right|\\
&\leq \left|\frac{W}{w}(1-1/\alpha)\cdot (M/W-1)+\frac{W}{\alpha w} (\alpha-1)\right|\\
&\leq \left|\frac{W}{w}(1-1/\alpha)\cdot (M/W-1)+\frac{W}{w}\cdot (1-1/\alpha)\right|\\
&\leq \left|(1-1/\alpha)\left(\frac{W}{w}\cdot (M/W-1)+\frac{W}{w}\right)\right|\\
&\leq \left|(1-1/\alpha)\cdot M/w\right|\\
&\leq M/w\\
\end{split}
\end{equation*}
as required. 

We now prove {\bf (2)}. By Definition~\dualref{def:rho} one has, letting 
$$
\langle x, \r\rangle \pmod{M}=aW+b (W/\alpha)+c,
$$
where $a\in [M/W]$,  $b\in [\alpha]$ and $c\in [W/\alpha]$,
\begin{equation}\duallabel{eq:93h4g9gewfdsfD}
\begin{split}
\rho(x):=x-\r\cdot \left(\frac{W}{w}(1-1/\alpha)\cdot a+\frac{W}{\alpha w} b\right).
\end{split}
\end{equation}
We have by~\dualeqref{eq:93h4g9gewfdsfD},
\begin{equation}\duallabel{eq:8g8ewgfBJBD}
\begin{split}
\langle \rho(x), \r\rangle \pmod{M}&=\left[\langle x, \r\rangle \pmod{M}-\langle \r, \r\rangle \cdot \left(\frac{W}{w}(1-1/\alpha)\cdot a+\frac{W}{\alpha w} b\right)\right]\pmod{M}\\
&=\left[(a W+b (W/\alpha)+c)-\langle \r, \r\rangle\cdot \left(\frac{W}{w}(1-1/\alpha)\cdot a+\frac{W}{\alpha w} b\right)\right]\pmod{M}\\
&=\left[(a W+b (W/\alpha)+c)-W(1-1/\alpha)\cdot a-\frac{W}{\alpha} b\right]\pmod{M}\\
&=(W/\alpha)a+c\in [M/\alpha],
\end{split}
\end{equation}
as required. In the last transition we used the fact that $a\in [M/W]$ and $c\in [W/\alpha]$ by definition of $a$ and $b$.

We now argue injectivity, i.e., prove {\bf (1)}. Suppose that $\rho(x)=\rho(y)$ for some $y\neq x$. Specifically, let 
\begin{equation*}
\begin{split}
\langle x, \r\rangle \pmod{M}&=aW+b (W/\alpha)+c\\
\langle y, \r\rangle \pmod{M}&=a'W+b' (W/\alpha)+c'
\end{split}
\end{equation*}
with $a, a'\in [M/W]$, $b, b'\in [\alpha]$ and $c, c'\in [W/\alpha]$. Then $\rho(x)=\rho(y)$ means that
\begin{equation}\duallabel{eq:823ghewfwef-wewe}
x-\r\cdot \left(\frac{W}{w}(1-1/\alpha)\cdot a+\frac{W}{\alpha w} b\right)=y-\r\cdot \left(\frac{W}{w}(1-1/\alpha)\cdot a'+\frac{W}{\alpha w} b'\right).
\end{equation}
First note that that by~\dualeqref{eq:8g8ewgfBJBD}
\begin{equation*}
\begin{split}
\left\langle x-\r\cdot \left(\frac{W}{w}(1-1/\alpha)\cdot a+\frac{W}{\alpha w} b\right), \r\right\rangle \pmod{M}&=a(W/\alpha)+c\\
\end{split}
\end{equation*}
and similarly 
\begin{equation*}
\begin{split}
\left\langle y-\r\cdot \left(\frac{W}{w}(1-1/\alpha)\cdot a'+\frac{W}{\alpha w} b'\right), \r\right\rangle \pmod{M}&=a'(W/\alpha)+c'.\\
\end{split}
\end{equation*}
Combining the two equations above with~\dualeqref{eq:823ghewfwef-wewe}, we get $a=a'$ and $c=c'$, and it remains to show that $b=b'$. To that effect recall that 
\begin{equation*}
\begin{split}
\weight(x)&=\sum_{i\in [n]} x_i\in [0, 1/\alpha)\cdot W \pmod{W}\\
\weight(y)&=\sum_{i\in [n]} y_i\in [0, 1/\alpha)\cdot W \pmod{W}.
\end{split}
\end{equation*}
Applying the $\weight(\cdot)$ function to both sides of ~\dualeqref{eq:823ghewfwef-wewe}, using the fact that $|\r|=w$ and rearranging terms, we get
\begin{equation}\duallabel{eq:023gjrgGDUGFefeeff2}
\begin{split}
\weight(x)-\weight(y)&=\left(W(1-1/\alpha)\cdot a+\frac{W}{\alpha} b\right)-\left(W(1-1/\alpha)\cdot a'+\frac{W}{\alpha} b'\right)\\
&=\frac{W}{\alpha} (b-b'),
\end{split}
\end{equation}
where in the last transition we used the fact that $a=a'$, as established above. Now recalling that $\weight(x) \pmod{W}\in [0, 1/\alpha)\cdot W$ and $\weight(y) \pmod{W}\in [0, 1/\alpha)\cdot W$ by assumption, we get that 
$$
(\weight(x) \pmod{W})-(\weight(y)\pmod{W})\in (-1/\alpha, 1/\alpha)\cdot W,
$$
and hence $b=b'$, which implies that $x=y$. This establishes injectivity of $\rho$, proving {\bf (1)}.

We now prove {\bf (2)}. For every $x\in \Int_\delta(U)$ we have by~\dualeqref{eq:93h4g9gewfdsfD} that $\rho(x):=x-\lambda\cdot \r$, 
where $\lambda$ is an integer satisfying $|\lambda|\leq M/w$, as established above. We thus have $\rho(x)\in U$ by Lemma~\dualref{lm:shift-int}. 
\end{proof}

\subsubsection{Local permutation map $\Pi$}\duallabel{sec:loc-pi}
We now define our local permutation map $\Pi$. 

\begin{definition}[Local permutation map $\Pi$]\duallabel{def:pi}
For two cubes $R=(\I, \a)$, $R'=(\I', \a')$ such that $\I, \I'\subset \F$, $\I\cap \I'=\emptyset$, the (partial) map 
$$
\Pi_{R'\to R}:[m]^n\to [m]^n
$$
 is defined as follows. Let $C\subseteq [m]^n$ denote a minimal $\I\cup \I'$-subspace cover (Definition~\dualeqref{def:subspace-cover}). 

For every $x\in C$ we define the mapping as follows. Let 
\begin{equation*}
\begin{split}
s&=|\subspace_{\I\cup \I'}(x)\cap R|\\
s'&=|\subspace_{\I\cup \I'}(x)\cap R'|.
\end{split}
\end{equation*}
Define  $\Pi_{R'\to R}$ on $\subspace_{\I\cup \I'}(x)$ as an arbitrary bijective mapping from 
a subset of $\subspace_{\I\cup \I'}(x)\cap R'$ of size $\min\{s, s'\}$ to a subset of $\subspace_{\I\cup \I'}(x)\cap R$ of size $\min\{s, s'\}$. 

\end{definition}

\begin{remark}
We show later (see Lemma~\dualref{lm:tau-of-rect} below)  that $s$ is quite close to $s'$ for $x\in C\setminus B$. Thus the map $\Pi$ is defined on almost all of $|\subspace_{\I\cup \I'}(x)\cap R|$ and almost all of $|\subspace_{\I\cup \I'}(x)\cap R'|$ for most choices of $x\in C$.
\end{remark}

The next lemma shows that the permutation map $\Pi_{R'\to R}$ performs sparse bounded shifts, i.e. that $\Pi_{R'\to R}(x)$ can be expresses as the sum of $x$ with a small number of vectors in $\F$, each with rather small coefficients:
\begin{lemma}[Local permutation map performs sparse bounded shifts]\duallabel{lm:t-infty}
For two cubes $R=(\I, \a)$, $R'=(\I', \a')$, $\a\in (\deltagrid)^\I, \a'\in (\deltagrid)^{\I'},$ such that $\I, \I'\subset \F$, $\I\cap \I'=\emptyset$, the following is true for the (partial) map 
$$
\Pi_{R'\to R}:[m]^n\to [m]^n.
$$
For every $z\in [m]^n$ such that $\Pi:=\Pi_{R'\to R}$ is defined on $z$ one has $\Pi(z)=z+\sum_{\i\in \I\cup \I'} t_\i \cdot \i$ with $||t||_\infty\leq 4M/w$.
\end{lemma}
\begin{proof}
This follows by Definition~\dualref{def:subspace} and Definition~\dualref{def:pi}. Indeed, recall that for a minimal $\I\cup \I'$-subspace cover $C$ and $x\in C$
the map $\Pi$ maps points $a\in \subspace_{\I\cup \I'}(x)\cap R'$ to points $b\in \subspace_{\I\cup \I'}(x)\cap R$. By definition of $\subspace_{\I\cup \I'}(x)$ (Definition~\dualref{def:subspace}) there exist coefficients $\{t^a_\i\}_{\i\in \I\cup \I'}$ and $\{t^b_\i\}_{\i\in \I\cup \I'}$ such that 
$$
a=x+\sum_{\i\in \I\cup \I'} t^a_\i \cdot \i\text{~~~and~~~}b=x+\sum_{\i\in \I\cup \I'} t^b_\i \cdot \i
$$
with $||t^a||_\infty\leq 2M/w$ and $||t^b||_\infty\leq 2M/w$. Putting the above bounds together, we get 
$$
b=a+\sum_{\i\in \I\cup \I'} (t^b_\i-t^a_\i) \cdot \i
$$
with $||t^b-t^a||_\infty\leq 4M/w$ for every $b\in \subspace_{\I\cup \I'}(x)\cap R$, as required.
\end{proof}

 While $\Pi_{R'\to R}$ is defined with respect to two cubes $R'$ and $R$, we often need to know where $\Pi$ maps an extended rectangle, namely a rectangle that beyond constraints imposed by $R'$ has further constraints -- see $R'_{ext}$ below. We show that if the additional constraints inherent in $R'_{ext}$ are nearly orthogonal (which they are since all our vectors come from the family $\F$), then at least the interior of an extended rectangle $R'_{ext}$ is mapped to an appropriate extended rectangle $R_{ext}$:
 \begin{lemma}[Action of permutation map on extended rectangles]\duallabel{lm:tau-of-rect}
For every pair of cubes $R=(\I, \a)$, $R'=(\I', \a')$, $\a\in (\deltagrid)^\I, \a'\in (\deltagrid)^{\I'},$ such that $\I, \I'\subset \F$, $|\I|=|\I'|$, $\I\cap \I'=\emptyset$, $|\I\cup \I'|\leq K^2$, if $\e<\delta/(4|\I\cup \I'|)$, the following conditions hold for the corresponding (partial) map $\Pi:=\Pi_{R'\to R}:[m]^n\setminus B\to [m]^n$ (see Definition~\dualref{def:pi}).

For every $\J\subset \F\setminus (\I\cup \I')$ and every $\mathbf{c}, \mathbf{d}\in (\deltagrid)^\J, \c<\d,$ if 
$$
R_{ext}=(\I\cup \J, (\a, \c), (\a+\Delta\cdot \mathbf{1}, \d))\text{~and~} R'_{ext}=(\I'\cup \J, (\a', \c), (\a'+\Delta\cdot \mathbf{1}, \d)),
$$
then 
\begin{description}
\item[(1)] $\Pi$ maps the interior of $R'_{ext}$ to $R_{ext}$, i.e. 
 $$
 \Pi(\Int_\delta(R'_{ext}))\subseteq R_{ext}.
 $$
 \item[(2)] the number of points in $R'$ that $\Pi$ is not defined on is bounded by $8\sqrt{\e}|R'|$.
 \end{description}
\end{lemma}
\begin{proof}
We start by proving {\bf (1)}. 
Pick $x\in \Int_\delta(R'_{ext})$ such that $\Pi(x)$ is defined. We need to verify that {\bf (a)} for every $\i\in \I$ one has $\langle \Pi(x), \i\rangle \pmod{M} \in [\a_\i, \a_\i+\Delta\cdot \mathbf{1})\cdot M$ and {\bf (b)} for every $\k\in \J$ one has $\langle \Pi(x), \k\rangle \pmod{M} \in [\c_\k, \d_\k)\cdot M$.

Condition {\bf (a)} is satisfied by construction of $\Pi$ since $\Pi$ maps points in $R'=(\I, \a')$ to points in $R=(\I, \a)$ and $R'_{ext}\subseteq R'$. We now establish {\bf (b)}. By Lemma~\dualref{lm:t-infty} one has $\Pi(x)=x+\sum_{\i\in \I\cup \I'} t_\i\cdot \i$, where $||t||_\infty\leq 4M/w$.   We thus have that for every $\k\in \J$
\begin{equation}\duallabel{eq:xp-dot}
\begin{split}
\left|\langle \Pi(x), \k\rangle-\langle x, \k\rangle\right|&\leq \left|\sum_{\substack{\i\in \I\cup \I'}} t_\i\cdot \langle \i,\k\rangle\right|\leq 4\e|\I\cup \I'|\cdot M< \delta M,
\end{split}
\end{equation}
where we used the fact that $\langle \i, \k\rangle \leq \e w$ for all $\i\in \I\cup \I'$ as $\J\subset \F\setminus (\I\cup \I')$ by assumption, as well as the fact that 
\begin{equation*}
\begin{split}
\e&<\delta^2\text{~~~~~~~~~~~~~~~~~~~~~~~~~~~~(by~\dualref{p6})}\\
 &\leq \delta \cdot \Delta^{200 K^2}\text{~~~~~~~~~~~~~~(by~\dualref{p5})}\\
 &\leq \delta \cdot K^{-200 K^2}\text{~~~~~~~~~~~(by~\dualref{p3})}\\
 &\leq \delta/(4|\I\cup \I'),
 \end{split}
 \end{equation*}
 where the last transition is due to the fact that $|\I\cup \I'|\leq K^2$ by assumption, and $K$ is larger than an absolute constant. Since $x\in \Int_\delta(R_{ext})$ by assumption, we have 
$$
\langle x, \k\rangle \pmod{M} \in [\c_\k+\delta, \d_\k-\delta)\cdot M
$$
for every $\k\in \J$. Putting this together with~\dualeqref{eq:xp-dot} gives
$$
\langle \Pi(x), \k\rangle \pmod{M} \in [\c_\k, \d_\k)\cdot M,
$$
as required.

We now prove {\bf (2)}. Let $C\subseteq [m]^n$ be the minimal $\I\cup \I'$-subspace cover used in the definition of $\Pi$.  Recall that for every $x\in C\setminus B$ one has by  Lemma~\dualref{lm:cube-subspace-size}, {\bf (1)},
$$
(1-\sqrt{\e}) \Delta^{|\I|}\cdot G \leq \left|\subspace_{\I\cup \I'}(x)\cap R\right|\leq (1+\sqrt{\e})\Delta^{|\I|}\cdot G,
$$
where $G=(M/w)^{|\I|}$. Similarly, one has  
$$
(1-\sqrt{\e}) \Delta^{|\I'|}\cdot G \leq \left|\subspace_{\I\cup \I'}(x)\cap R'\right|\leq (1+\sqrt{\e})\Delta^{|\I'|}\cdot G,
$$
since $|\I|=|\I'|$.  We thus get for every $x\in [m]^n\setminus B$, letting $s=\left|\subspace_{\I\cup \I'}(x)\cap R\right|$ and $s'=\left|\subspace_{\I\cup \I'}(x)\cap R'\right|$,
$$
\max\{s, s'\}-\min\{s, s'\}\leq 4\sqrt{\e} \cdot s'
$$
as long as $\e$ is smaller than a constant. Thus, the number of points in $\subspace_{\I\cup \I'}(x)\cap R'$ that $\Pi$ is not defined on is bounded by $4\sqrt{\e} |\subspace_{\I\cup \I'}(x)\cap R'|$.
The number of points that $\Pi$ is not defined on is bounded by 
\begin{equation*}
\begin{split}
&4\sqrt{\e}\cdot \sum_{x\in C\setminus B}  |\subspace_{\I\cup \I'}(x)\cap R'|+|B|\cdot (5M/w)^{|\I\cup \I'|}\\
&\leq 4\sqrt{\e}\cdot \sum_{x\in C\setminus B}  |\subspace_{\I\cup \I'}(x)\cap R'|+\frac1{n}|R'|\cdot (5M/w)^{|\I\cup \I'|}\\
&\leq 8\sqrt{\e}\cdot \sum_{x\in C}  |\subspace_{\I\cup \I'}(x)\cap R'|\\
&=8\sqrt{\e}\cdot |R'|,\\
\end{split}
\end{equation*}
where the first transition uses the fact that for every $x\in [m]^n$ one has $|\subspace_{\I\cup \I'}(x)|\leq (5M/w)^{|\I\cup \I'|}$ (since coordinates of $t$ are bounded by $2M/w$ in absolute value in Definition~\dualref{def:subspace}) and the second transition uses the fact that $|R'|\geq \Delta^{|\I\cup \I'|}\geq \Delta^{K^2}$, and the third transition uses the assumption that $n$ is sufficiently large as a function of $M/w, W/w, K, L, \Delta, \delta$. 
\end{proof}

 \subsubsection{Defining the glueing map $\tau$}\duallabel{sec:tau}
We define the glueing map $\tau$ in this section. To do that, first for every $k\in [K/2]$ we define a map 
$$
\tau_k: S'_k\to T_*\cup \{\bot\},
$$ 
where for a vertex $x\in S'_k$ we write $\tau(x)=\bot$ to denote the fact that $\tau$ is not defined on $x$. Thus, in essence $\tau$ is a partial map.
We ensure that
\begin{enumerate}
\item $\tau_k$ is injective on elements of $S'_k$ that it does not map to $\bot$, i.e.,  if $\tau_k(x)\neq \bot$ and $\tau_k(y)\neq \bot$, then $\tau_k(x)\neq \tau_k(y)$ for $x\neq y$.
\item the images of $\tau_k$ are disjoint for different $k$, i.e. these maps extend naturally to an injective partial map from the union of $S'_k$ over all $k\in [K/2]$ to $T_*$ that is defined on almost all of $S'$.
\end{enumerate}
Then the map $\tau$ is defined as mapping an element in $x\in S_k$ to $\tau_k(x)$ for every $k\in [K/2]$.

The map $\tau_k$ is parameterized by the compression vector $\r\in \B_{K/2}, \r\not \in \J_k,$ for the terminal subcube $T_*$, as well as the extension and compression vectors for every $k\in [K/2]$ (see Definition~\dualref{def:ext})
$\Ext_k\subseteq \B'_k$ and $\q_k\in \B'_k$. Define sets 
\begin{equation}\duallabel{eq:i}
\I=\J\cup \{\r\}\subset \F
\end{equation}
and 
\begin{equation}\duallabel{eq:i-prime}
\I'_k=\J'_{<k}\cup \Ext_k\cup \{\q_k\} \subset \F
\end{equation}
We sometimes write $\I'$ when $k$ is fixed and clear from context. Let  $\rho_k$ be the $(K-k, \q_k)$-densifying map as per Definition~\dualref{def:rho}.  By Lemma~\dualref{lm:densification} we have
\begin{equation}\duallabel{eq:image-rho-sk}
\rho_k(\Int_\delta(S'_k))\subseteq \left\{x\in T'_k: \langle x, \q_k\rangle \pmod{M} \in \left[0, \frac1{K-k}\right)\cdot M \right\},
\end{equation}
Indeed, we invoke the lemma with $U=T'_k$, since $S'_k=\downset_k(T'_k)$ so that 
\begin{equation*}
\begin{split}
\Int_\delta(S'_k)&=\downset_k(\Int_\delta(T'_k))\\
&=\left\{x\in \Int_\delta(T'_k): \weight(x) \in \left[0, \frac{1}{K-k}\right)\cdot W \pmod{ W} \right\},
\end{split}
\end{equation*}
 by definition of a $\delta$-interior (see Definition~\dualref{def:int}). Recall that  $T'_k$ is indeed a rectangle, as required by Lemma~\dualref{lm:densification}, since $T'_k=\rect(\J'_{<k}, \c, \d)$ with $\c_{\j'_s}=0$,  $\d_{\j'_s}=1-\frac1{K-s}$ for all $s\in [k]$. Note that the preconditions of Lemma~\dualref{lm:densification} are satisfied since $T'_k$ is indeed a rectangle (see~\dualeqref{eq:def-tk-all-constraints}) and $\q_k\not \in \J'_{<k}$ (note that $\Delta \mid \frac1{K-s}$ for all $s\in [k]$ by~\dualref{p3}, so rectangle boundaries are indeed in $\deltagrid$, as required by Lemma~\dualref{lm:densification}).

\begin{definition}\duallabel{def:dk}
For $k\in [K/2]$ let $\mathbb{D}_k\subseteq (\deltagrid)^{\I'_k}$ be such that 
\begin{equation*}
\begin{split}
\left\{x\in T'_k: \langle x, \q_k\rangle \pmod{M} \in \left[0, \frac1{K-k}\right)\cdot M \right\}&=\bigcup_{\d\in \mathbb{D}_k} \rect(\I'_k, \d).
\end{split}
\end{equation*}
\end{definition}

Note that such a set $\mathbb{D}_k$ exists since $T'_k$ is a rectangle in $\I'_k$. Indeed, let $\c_{\j'_s}=0, \d_{\j'_s}=1-\frac1{K-s}$ for $s\in [k]$, let $\c_{\q_k}=0, \d_{\q_k}=\frac1{K-k},$ and $\c_\i=0, \d_\i=1$ for $\i\in \Ext_k$. Then 
$$
\left\{x\in T'_k: \langle x, \q_k\rangle \pmod{M} \in \left[0, \frac1{K-k}\right)\cdot M \right\}=\rect(\I'_k, \c, \d),
$$
and by Claim~\dualref{cl:subcubes-decomp}, we get that the set $\mathbb{D}_k$ from Definition~\dualref{def:dk} exists and satisfies
\begin{equation}\duallabel{eq:dk-def}
\mathbb{D}_k=(\deltagridInt)^{\I'_k}\cap \prod_{\i \in \I'_k} [\c_\i, \d_\i).
\end{equation}
Combining the definition above with~\dualeqref{eq:image-rho-sk}, we get
\begin{equation}\duallabel{eq:92h3t9h2gt9h3gzffdfdfFDF}
\rho_k(\Int_\delta(S'_k))\subseteq \bigcup_{\d\in \mathbb{D}_k} \rect(\I'_k, \d).
\end{equation}
Similarly let $\mathbb{A}\subseteq \deltagrid^\I$ be such that 
\begin{equation}\duallabel{eq:m-a-def}
T_*=\bigcup_{\a\in \mathbb{A}} \rect(\I, \a).
\end{equation}
Note that such $\mathbb{A}$ exists by Claim~\dualref{cl:subcubes-decomp} since $T^*_k$ is a rectangle in $\I$. The latter holds because $\J\subseteq \I$,  $T^*_k=\rect(\J, \c, \d)$ with $\c_{\j_s}=0$,  $\d_{\j_s}=1-\frac1{K-s}$ for all $s\in [K/2]$, and $\Delta \mid  \frac1{K-s}$ for all $s\in [K/2]$.  Now let
\begin{equation}\duallabel{eq:m-def}
\mathsf{M}: \bigcup_{k\in [K/2]} \mathbb{D}_k \to \mathbb{A}.
\end{equation}
be a bijective map. Such a map exists since $\left|\sum_{k\in [K/2]} \mathbb{D}_k\right|=\sum_{k\in [K/2]} |\mathbb{D}_k|=|\mathbb{A}|$. Indeed, by Definition~\dualref{def:dk} one has for every $k\in [K/2]$
\begin{equation*}
\begin{split}
|\mathbb{D}_k|=\frac1{\Delta^{|\I'_k|}}\cdot \frac1{K-k}\cdot \prod_{s=0}^{k-1} \left(1-\frac1{K-s}\right)=\frac1{\Delta^{|\I'_k|}}\cdot \frac1{K}
\end{split}
\end{equation*}
and by~\dualeqref{eq:m-a-def} one has
$$
|\mathbb{A}|=\frac1{\Delta^{|\I|}}\cdot\prod_{s=0}^{K/2} \left(1-\frac1{K-s}\right)=\frac1{\Delta^{|\I|}}\cdot\frac1{2}, 
$$ 
and therefore 
$$
\sum_{k\in [K/2]} |\mathbb{D}_k|=\frac1{\Delta^{|\I'_k|}}\cdot \frac1{K}\cdot (K/2)=\frac1{\Delta^{|\I'_k|}}\cdot\frac1{2}=|\mathbb{A}|,
$$ 
as required (since $|\I|=|\I'_k|$ for every $k\in [K/2]$).

\begin{remark}\duallabel{rm:tau-prefix}
Note that for every $k\in [K/2]$ the set $\mathbb{D}_k$ is determined by $\I$ and $\I'_k=\J'_{<k}\cup \Ext_k\cup \{\q_k\}$, and $\mathbb{A}$ is determined by $\I$. Thus, we can construct the map $\mathsf{M}$ incrementally, by fixing  $\mathsf{M}|_{\mathbb{D}_k}: \mathbb{D}_k \to \mathbb{A}$ as soon as $\I'_k$ becomes known. The latter in fact amounts to knowing $\J'_{<k}$, since we fix $\Ext_k$ and $\q_k$ for our hard input distribution.
\end{remark}

\begin{remark}We note that while the terminal subcube is defined as $T_*=T_{K/2}$, the parameter $k$ ranges over $[K/2]=\{0, 1,2\ldots, K/2-1\}$, i.e. not including $k=K/2$. This is exactly in order to ensure that $\sum_{k\in [K/2]} |S'_k|$ equals $|T_*|$ up to lower order terms that can be made small as a function of $\e$, and in particular can be made arbitrarily smaller than $K^K$ -- this allows us to control the number of vertices left out by the glueing map $\tau$ in Lemma~\dualref{lm:t-star-minus-tau-sp}.
\end{remark}

For convenience of notation, we first define a map $\Pi^*_k$  for each $k\in [K/2]$ that pieces together local permutation maps $\Pi_{R'\to R}$. We refer to these maps as {\em global permutation maps}:
\begin{definition}[Global permutation maps $\Pi^*_k$]\duallabel{def:pi-global}
For every $k\in [K/2]$ and every 
$$
z\in \left\{x\in T'_k: \langle x, \q_k\rangle \pmod{M} \in \left[0, \frac1{K-k}\right)\cdot M \right\}
$$ 
we let $\d\in \mathbb{D}_k$ be such that $z\in \Int_\delta(R')$, where $R'=\rect(\I'_k, \d)$, if such $\d$ exists (otherwise leave $\Pi^*_k$ undefined on $z$). Let $R=\rect(\I, \mathsf{M}(\d))$, where $\mathsf{M}$ is as per~\dualeqref{eq:m-def}. We then define 
$$
\Pi^*_k(z):=\Pi_{R'\to R}(z)
$$ 
if $\Pi_{R'\to R}(z)$ is defined (otherwise leave $\Pi^*_k$ undefined on $z$).
\end{definition}

Finally, we define

\begin{definition}[Glueing map $\tau$]\duallabel{def:tau}
For every $x\in S$, if $k\in [K/2]$ is such that $x\in S'_k$, we let 
\begin{equation}\duallabel{eq:tau-def}
\tau(x):=\Pi^*_k(\rho_k(x)) \in T_*
\end{equation}
if $\Pi^*_k(\rho_k(x))$ is defined, and leave $\tau(x)$ undefined otherwise.

For a subset $U\subseteq S'$ we define
$$
\tau(U)=\bigcup_{x\in U} \{\tau(x)\},
$$
where we think of $\{\tau(x)\}$ as the empty set if $\tau(x)$ is not defined.
\end{definition}

We gather some basic properties of the global permutation maps in 
\begin{claim}[Injectivity of $\Pi^*_k$ and $\tau$]\duallabel{cl:pi-star-injective}
For every $k\in [K/2]$ the global permutation map $\Pi^*_k$ is injective, and the ranges of $\Pi^*_k$ are disjoint for $k\in [K/2]$. Furthermore, the map $\tau$ is injective.
\end{claim}
\begin{proof}
Fix $k\in [K/2]$, let $\I=\I_k$ and let
\begin{equation*}
Q'_k=\left\{x\in T'_k: \langle x, \q_k\rangle \pmod{M} \in \left[0, \frac1{K-k}\right)\cdot M \right\}
\end{equation*}
for convenience.  We show that for every $z_0, z_1\in Q'_k$ such that $\Pi^*$ is defined on both one has $\Pi^*(z_0)\neq \Pi^*(z_1)$. For $z_0\in Q'_k$ we let $\d_0\in \mathbb{D}_k$ be such that $z_0\in \Int_\delta(R'_0)$, where $R'_0=\rect(\I', \d_0)$, if such $\d_0$ exists (otherwise there is nothing to prove since $\Pi^*_k$ undefined on $z_0$).  Let $R_0=\rect(\I, \mathsf{M}(\d_0))$. For $z_1\in Q'_k$ we let $\d_1\in \mathbb{D}_k$ be such that $z_1\in \Int_\delta(R'_1)$, where $R'_1=\rect(\I', \d_1)$, if such $\d_1$ exists (otherwise there is nothing to prove since $\Pi^*_k$ undefined on $z_1$).  Let $R_1=\rect(\I, \mathsf{M}(\d_1))$. 

Recall that $\Pi^*_k(z_0)=\Pi_{R'_0\to R_0}(z_0)$ and $\Pi^*_k(z_1)=\Pi_{R'_1\to R_1}(z_1)$. If either of these maps is undefined on $z_0$ and $z_1$ respectively, there is nothing to prove. Now suppose that both of them are defined.
By definition of $\Pi^*_k$ one has 
$$
\Pi^*_k(z_0)\in R_0\text{~~and~~}\Pi^*_k(z_1)\in R_1.
$$
We thus get that if $\d_0\neq \d_1$, then $\Pi^*_k(z_0)\neq \Pi^*_k(z_1)$ since $R_0\cap R_1=\rect(\I, \mathsf{M}(\d_0))\cap \rect(\I, \mathsf{M}(\d_1))=\emptyset$ when $\d_0\neq \d_1$. On the other hand, if $\d_0=\d_1$, then $\Pi^*_k(z_0)\neq \Pi^*_k(z_1)$ because the map $\Pi_{R'_0\to R_0}=\Pi_{R'_1\to R_1}$ is injective by construction. This proves that $\Pi^*_k$ is injective. Injectivity of $\Pi^*$ follows from the fact that the map $\mathsf{M}$ (see~\dualeqref{eq:m-def}) is injective, as well as the fact that for every $\a_0, \a_1\in \mathbb{A}$ one has $\rect(\I, \a_0)\cap \rect(\I, \a_1)=\emptyset$ when $\a_0\neq \a_1$. 

Finally, we prove injectivity of $\tau$. Pick two distinct vertices $x, y\in S'$. Let $a, b\in [K/2]$ be such that $x\in S'_a$ and $y\in S'_b$. If $a\neq b$, then $\tau(x)\neq \tau(y)$ since the images of $\Pi^*_k$ are disjoint by definition of $\mathsf{M}$ (see~\dualeqref{eq:m-def}), and the fact that for every $\a_0, \a_1\in \mathbb{A}$ one has $\rect(\I, \a_0)\cap \rect(\I, \a_1)=\emptyset$ when $\a_0\neq \a_1$. If $a=b$, then $\tau(x)=\Pi^*_a(\rho_a(x))$ and $\tau(y)=\Pi^*_a(\rho_a(y))$, where $\rho_a$ is a $(K-a, \r)$-densifying map, so the result follows by injectivity of $\Pi^*_k$, as well as the fact that $\rho$ is injective by Lemma~\dualref{lm:densification}.
\end{proof}


Similarly to the local (and therefore also global) permutation maps, $\tau$ performs sparse bounded shifts:

\begin{lemma}[Glueing map $\tau$ performs sparse bounded shifts]\duallabel{lm:tau-sparse}
For every $x\in S'_k$ if $y=\tau(x)$, then there exist integer coefficients $\{t_\i\}_{\i\in \I\cup \I'_k}$ such that 
$$
y=x+\sum_{\i\in \I\cup \I'_k} t_\i\cdot \i
$$
such that $\|t\|_\infty\leq 5M/w$.
\end{lemma}
\begin{proof}
Let $z=\rho_k(x)$, and note that 
$$
z=x+\lambda\cdot \q_k
$$
for some integer $\lambda$ satisfying $|\lambda|\leq M/w$ by Lemma~\dualref{lm:densification}. Let $\d\in (\deltagrid)^{\I'_k}$ be such that $z\in R'$ with $R'=\rect(\I'_k, \d)$. Let $R:=\rect(\I, \mathsf{M}(\d))$. Note that $\Pi^*_k(z)=\Pi_{R'\to R}$, and one hence by Lemma~\dualref{lm:t-infty}  one has 
$$
\Pi^*_k(z)=z+\sum_{\i\in \I\cup \I'_k} s_\i\cdot \i,
$$
where $s$ satisfies $\| s\|_\infty\leq 4M/w$. These two facts give the result.
\end{proof}

Unlike the map $\tau$ defined in our toy construction from Section~\ref{sec:toy-construction}, the map $\tau$ is not quite a bijection. However, the range of $\tau$ covers almost all of $T_*$:
\begin{lemma}[$\tau$ maps almost all of $S'$ onto terminal subcube $T_*$]\duallabel{lm:t-star-minus-tau-sp}
We have $|T_*\setminus \tau(S')|\leq \delta^{1/4}\cdot |T_0|$. Furthermore, the map $\tau$ is defined on all but $\delta^{1/4} |S'|$ vertices in $S'$.
\end{lemma}
\begin{proof}
We have
\begin{equation}\duallabel{eq:h4g90h349gh}
\begin{split}
|T_*\setminus \tau(S')|&\leq |T_*|-|\tau(S')|\text{~~~~~~~~~~~~~~~~~~~~~~~~~~~~~~~~~(since $\tau(S')\subseteq T_*$)}\\
&\leq |T_*|-\sum_{k\in [K/2]} |\tau_k(S'_k)|\text{~~~~~(since the images of $\tau_k$ are disjoint)}\\
&\leq |T_*|-\sum_{k\in [K/2]} |\tau_k(\Int_\delta(S'_k))|.
\end{split}
\end{equation}
The second transition used the fact that images of $\tau_k$ are disjoint for different $k$. This follows from the fact that $M:\bigcup_{k\in [K/2]} \mathbb{D}_k\to \mathbb{A}$ is a bijective mapping, together with the fact that by ~\dualeqref{eq:tau-def} for every $x\in S'_k$, if $\d\in \mathbb{D}_k$ is such that $\rho_k(x)\in R':=\rect(\I'_k, \d)$, let $R:=\rect(\I, \mathsf{M}(\d))$, then either $\tau(x)=\bot$ or $\tau(x)\in R$.

We let 
\begin{equation*}
Q'_k=\left\{x\in T'_k: \langle x, \q_k\rangle \pmod{M} \in \left[0, \frac1{K-k}\right)\cdot M \right\}
\end{equation*}
for convenience, and note that $Q'_k$ is exactly the rhs of~\dualeqref{eq:image-rho-sk}. We have,  using~\dualeqref{eq:image-rho-sk}
\begin{equation}\duallabel{eq:09239h9gh}
\begin{split}
|\tau_k(\Int_\delta(S'_k))|&\geq |\Pi^*_k(Q'_k)|-|Q'_k\setminus \Int_\delta(S'_k)|\\
&\geq |\Pi^*_k(Q'_k)|-|Q'_k|+|\Int_\delta(S'_k)|\\
&= |\Pi^*_k(Q'_k)|-|Q'_k|+|S'_k|-|S'_k\setminus \Int_\delta(S'_k)|\\
&\geq |\Pi^*_k(Q'_k)|-|Q'_k|+(1-\sqrt{\delta})|S'_k|,
\end{split}
\end{equation}
where the second transition uses the fact that $\rho_k$ is injective, and the last transition uses Lemma~\dualref{lm:rect-int-size}. We now lower bound the first term above. We have, letting $R'_\d:=\rect(\I'_k, \d)$ and $R_\d:=\rect(\I, \mathsf{M}(\d))$ for every $\d$ to simplify notation and recalling that 
$$
Q'_k=\bigcup_{\d\in \mathbb{D}_k} \rect(\I'_k, \d)=\bigcup_{\d\in \mathbb{D}_k} R'_\d
$$
 by Definition~\dualref{def:dk},
\begin{equation}\duallabel{eq:0932t8gGFYEGF}
\begin{split}
|\Pi^*_k(Q'_k)|&\geq \sum_{\substack{\d \in \mathbb{D}_k}} |\Pi^*_k(R'_\d)|\geq (1-8\sqrt{\e})\sum_{\substack{\d \in \mathbb{D}_k}} |R'_\d|\geq (1-8\sqrt{\e})|Q'_k|
\end{split}
\end{equation}
In the equation above we used the fact that by Lemma~\dualref{lm:tau-of-rect}, {\bf (2)},  $\Pi^*_k(R'_\d)$ is defined on all but a $8\sqrt{\e}\cdot |R'_\d|$ vertices of $R'_\d$.  This lower bounds the first term in~\dualeqref{eq:09239h9gh}. To upper bound the second term on the rhs of~\dualeqref{eq:09239h9gh}, we first note that by Lemma~\dualref{lm:rect-size}, {\bf (1)}, with $\gamma=\frac1{K-k}\cdot \prod_{i=0}^{k-1}(1-\frac{i}{K})=\frac1{K-K}(1-\frac{k}{K})=\frac1{K}$ one has 
$$
(1-\sqrt{\e})\frac1{K}\leq |Q'_k|/m^n\leq (1+\sqrt{\e})\frac1{K}
$$
and by Lemma~\dualref{lm:rect-size}, {\bf (2)}, one has 
\begin{equation}\duallabel{eq:082g38g832g8gZNZN}
(1-\sqrt{\e})\frac1{K}\leq |S'_k|/m^n\leq (1+\sqrt{\e})\frac1{K},
\end{equation}
so that 
\begin{equation}\duallabel{eq:0923t238gGFYGE}
(1-3\sqrt{\e})|Q'_k|\leq |S'_k|\leq (1+3\sqrt{\e})|Q'_k|.
\end{equation}

Substituting the above into~\dualeqref{eq:09239h9gh}, we get
\begin{equation}\duallabel{eq:93h42gh903hg8yuwg832g8GFX}
\begin{split}
|\tau_k(\Int_\delta(S'_k))|&\geq |\Pi^*_k(Q'_k)|-|Q'_k|+(1-\sqrt{\delta})|S'_k|\\
&\geq |\Pi^*_k(Q'_k)|-|Q'_k|+(1-\sqrt{\delta})(1-3\sqrt{\e})|Q'_k|\text{~~~~~(by~\dualeqref{eq:0923t238gGFYGE})}\\
&\geq |\Pi^*_k(Q'_k)|-(\sqrt{\delta}+3\sqrt{\e})|Q'_k|\\
&\geq (1-8\sqrt{\e})|Q'_k|-(\sqrt{\delta}+3\sqrt{\e})|Q'_k|\text{~~~~~~~~~~~~~~~~~~~~(by~\dualeqref{eq:0932t8gGFYEGF})}\\
&\geq (1-11\sqrt{\e}-\sqrt{\delta})|Q'_k|\\
&\geq (1-14\sqrt{\e}-\sqrt{\delta})|S'_k|\\
&\geq (1-2\sqrt{\delta})|S'_k|.
\end{split}
\end{equation}
In the last two transitions we used \dualref{p6} and the assumption that $K$ is larger than an absolute constant (so that $\e$ is smaller than an absolute constant).
.

\paragraph{Putting it together.} Noting that 
$$
(1-\sqrt{\e})\frac1{2}\leq |T_*|/m^n\leq (1+\sqrt{\e})\frac1{2}
$$
by Lemma~\dualref{lm:rect-size}, {\bf (1)}, and substituting~\dualeqref{eq:93h42gh903hg8yuwg832g8GFX} into~\dualeqref{eq:h4g90h349gh}, we get
\begin{equation*}
\begin{split}
|T_*\setminus \tau(S')|&\leq |T_*|-\sum_{k\in [K/2]} |\tau_k(\Int_\delta(S'_k))|\\
&\leq (1+\sqrt{\e})\frac1{2}|T_0|-\sum_{k\in [K/2]} (1-2\sqrt{\delta})|S'_k|\text{~~~~~~~~~~~(by~\dualeqref{eq:93h42gh903hg8yuwg832g8GFX})}\\
&\leq (1+\sqrt{\e})\frac1{2}|T_0|-(1-2\sqrt{\delta})\sum_{k\in [K/2]} |S'_k|\\
&\leq (1+\sqrt{\e})\frac1{2}|T_0|-(1-4\sqrt{\delta})(1-\sqrt{\e})\sum_{k\in [K/2]} \frac1{K}|T_0|\text{~~~~(by~\dualeqref{eq:082g38g832g8gZNZN})}\\
&\leq (1+\sqrt{\e})\frac1{2}|T_0|-(1-4\sqrt{\delta}-\sqrt{\e})\frac1{2}|T_0|\\
&\leq (\sqrt{\e}+4\sqrt{\delta}+\sqrt{\e})\frac1{2}|T_0|\\
&\leq \delta^{1/4} |T_0|\\
\end{split}
\end{equation*}
In the last two transitions we used \dualref{p6} and the assumption that $K$ is larger than an absolute constant (so that $\e$ and $\delta$ are smaller than an absolute constant). The second bound follows similarly.
\end{proof}

The next lemma shows that the inverse of $\tau$ maps two points from the same cube to the same set $S'_k$ for some $k\in [K/2]$.
\begin{lemma}[Inverse of $\tau$ on a cube]\duallabel{lm:same-rect}
For every $x, y\in T_*$ such that $x\in \tau(S')$ and $y\in \tau(S')$, if there exists $\d\in (\deltagridInt)^\I$ such that $x, y\in \rect(\I, \d)$, the following conditions hold: {\bf (1)} there exists $k\in [K/2]$,  $\wt{x}, \wt{y}\in S'_k$ such that $x\in \tau(\wt{x})$ and $y\in \tau(\wt{y})$, {\bf (2)}  there exists $\d'\in \I'=\J'_{<k}\cup \Ext_k\cup \{\q_k\}$ such that $\rho_k(\wt{x})\in \rect(\I', \d')$ and $\rho_k(\wt{y})\in \rect(\I', \d')$.
\end{lemma}
\begin{proof}
Let $\wt{x}, \wt{y}\in S'$ be such that $x\in \tau(\wt{x})$ and $y\in \tau(\wt{y})$ (such $\wt{x}$ and $\wt{y}$ exist by assumption of the lemma). Let $k_x, k_y\in [K/2]$ by such that $\wt{x}\in S'_{k_x}$ and $\wt{y}\in S'_{k_y}$. We will show that $k_x=k_y$.
Recall that by Definition~\dualref{def:pi-global} 
we let $\d_x\in \mathbb{D}_{k_x}$ be such that $\rho_{k_x}(\wt{x})\in \Int_\delta(\rect(\I'_{k_x}, \d_x))$, and  let $\d_y\in \mathbb{D}_{k_y}$ be such that $\rho_{k_y}(\wt{y})\in \Int_\delta(\rect(\I'_{k_y}, \d_y))$, where $\rho_{k_x}$ and $\rho_{k_y}$ are corresponding densification maps as per Definition~\dualref{def:rho} (along directions $\q_{k_x}$ and $\q_{k_y}$ respectively).  Then 
$$
x=\Pi^*_{k_x}(\rho_{k_x}(\wt{x})))\in \rect(\I, \mathsf{M}(\d_x))
$$
 and 
 $$
y=\Pi^*_{k_y}(\rho_{k_y}(\wt{y})))\in \rect(\I, \mathsf{M}(\d_y)).
$$

Since $x, y\in \rect(\I, \d)$ by assumption of the lemma, we get that $\mathsf{M}(\d_x)=\mathsf{M}(\d_y)=\d$. Since $\mathsf{M}$ is injective, this in particular implies that $k_x=k_y=k$ and $\d_x=\d_y$, as required. Letting $\I'=\J'_{<k}\cup \Ext_k\cup \{\q_k\}$, we thus get $\rho_k(\wt{x})\in \rect(\I', \d')$ and $\rho_k(\wt{y})\in \rect(\I', \d')$.
\end{proof}

\newcommand{\out}{\Xi}
\section{Predecessor map $\nu$ and its properties}\duallabel{sec:nu}
In this section we define the predecessor map $\nu$, which lets us define a good upper bound on the size of the maximum matching constructed by a low space algorithm later in Section~\dualref{sec:main-theorem} (specifically, see definition of the sets $A_P, A_Q, B_P, B_Q$ in~\dualeqref{eq:apm-def} and~\dualeqref{eq:bpm-def}). Intuitively, the predecessor map $\nu_{\ell, j}$ maps a subset of $T^\ell$ for some $\ell\in [L]$ through $j$ repeated applications of the glueing map $\tau^\ell$ interleaved with applications of the $\downset$ map. This is a natural object, since our construction is motivated by the fact that for appropriately defined `nice' subsets $U\subseteq T^\ell$, namely for appropriately defined rectangles (see Lemma~\ref{lm:key-intro} in Section~\ref{sec:tech-overview}), the edge boundary of the set $U\cup \downset^\ell(U)$ is very sparse, which is the basis of our hard input instance.

\begin{definition}[Predecessor map $\nu$]\duallabel{def:nu}
We define the map $\nu_{\ell, j}: 2^{T^\ell}\to 2^{T^{\ell-j}}$ (mapping subsets of $T^\ell$ to subsets of $T^{\ell-j}$) as follows. For $U\subseteq T^\ell$ we let
\begin{equation*}
\begin{split}
\nu_{\ell, 0}(U)&:=U\\
&\text{and}\\
\nu_{\ell, j}(U)&:=\tau^{\ell-(j-1)}(\downset^{\ell-(j-1)}(\nu_{\ell, j-1}(U)))\\
\end{split}
\end{equation*}
for $j=1,\ldots, \ell$. We define the {\em closure} map $\nu_{\ell, *}$ by
$$
\nu_{\ell, *}(U):=\bigcup_{\substack{j=0\\j \text{~even}}}^\ell \nu_{\ell, j}(U).
$$
Similarly, we define
$$
\mu_{\ell, j}(U):=\downset^{\ell-j}(\nu_{\ell, j}(U))
$$
for $j=0,\ldots, \ell$, and let
$$
\mu_{\ell, *}(U):=\bigcup_{\substack{j=0\\j \text{~even}}}^\ell \mu_{\ell, j}(U).
$$
\end{definition}
\begin{remark}
We note that $\nu_{\ell, j}$ can be viewed as mapping elements of $T^\ell$ to $T^{\ell-j}$: for $x\in T^\ell$ the image of $x$ under $\nu_{\ell, j}$ is naturally defined as $\nu_{\ell, j}(\{x\})$, i.e. the image of a singleton set containing $x$. This map, however, is not a one-to-one map because $\downset$ is not (see Definition~\dualref{def:downset} and Remark~\dualref{rm:downset}). This in particular is the reason why we prefer to define $\nu_{\ell, j}$ as mapping sets to sets in Definition~\dualref{def:nu}. On the other hand, $\nu_{\ell, j}$ maps every vertex to at most $K^j$ vertices, since $\downset$ maps every vertex to at most $K/2$ vertices, and $\tau$ is a one to (at most) one map.
\end{remark}

The main results of this section are the following two lemmas.   The first lemma bounds the size of the image of the non-terminal part of $T^\ell$, namely $T^\ell\setminus T_*^\ell$, under the predecessor map $\mu$:
\begin{lemma}\duallabel{lm:mu-ell-j}
For every $\ell \in [L]$, every $j=0,\ldots, \ell$, one has 
$$
(\ln 2-C/K)^j \frac1{2}(1-\ln 2)|T_0| \leq |\mu_{\ell, j}(T^\ell\setminus T_*^\ell)|\leq (\ln 2+C/K)^j \frac1{2}(1-\ln 2)|T_0|
$$
for an absolute constant $C>0$.
\end{lemma}

The second lemma is Lemma~\dualref{lm:cut-structure}, which proves a key property (equivalent to Lemma~\ref{lm:key-intro} in Section~\ref{sec:tech-overview}) allowing us reason about the structure of the upper bounding vertex cover in Lemma~\dualref{lm:special-edges} of Section~\dualref{sec:main-theorem}.

\subsection{Basic properties of $\nu$ and $\mu$}

\begin{definition}[Injectivity for maps defined on sets]
A map $\nu: 2^{A}\to 2^{B}$ (mapping subsets of $A$ to subsets of $B$) is called injective if for every $x, y\in A, x\neq y$ one has $\nu(\{x\})\cap \nu(\{y\})=\emptyset$.
\end{definition}

The following properties of $\nu$ and $\mu$ will be useful:
\begin{claim}\duallabel{cl:nu-prop-basic}
For every $\ell\in [L]$, every $j\in 1,\ldots, \ell$ and every $U\subseteq T^\ell$ one has {\bf (1)} $\nu_{\ell, j}(U)=\nu_{\ell-1, j-1}(\tau^\ell(\downset^\ell(U))$ and {\bf (2)}  $\mu_{\ell, j}(U)=\mu_{\ell-1, j-1}(\tau^\ell(\downset^\ell(U)))$.
\end{claim}
\begin{proof}
For {\bf (1)} we have by Definition~\dualref{def:nu} 
\begin{equation*}
\begin{split}
\nu_{\ell, j}(U)&=\tau^{\ell-(j-1)}(\downset^{\ell-(j-1)}(\nu_{\ell, j-1}(U)))\\
&=\tau^{\ell-(j-1)}(\downset^{\ell-(j-1)}(\tau^{\ell-(j-2)}(\downset^{\ell-(j-2)}(\nu_{\ell, j-2}(U)))))\\
&=\tau^{\ell-(j-1)}(\downset^{\ell-(j-1)}(\tau^{\ell-(j-2)}(\downset^{\ell-(j-2)}(\ldots \tau^{\ell}(\downset^\ell(U))\ldots))))\\
&=\nu_{\ell-1, j-1}(\tau^{\ell}(\downset^\ell(U)))\\
\end{split}
\end{equation*}

For {\bf (2)} we have by Definition~\dualref{def:nu} and using {\bf (1)}
\begin{equation*}
\begin{split}
\mu_{\ell, j}(U)&=\downset^{\ell-j}(\nu_{\ell, j}(U))\\
&=\downset^{\ell-j}(\nu_{\ell-1, j-1}(\tau^{\ell}(\downset^\ell(U))))\\
&=\mu_{\ell-1, j-1}(\tau^{\ell}(\downset^\ell(U))).
\end{split}
\end{equation*}
\end{proof}

We also need
\begin{lemma}[Basic properties of the maps $\nu_{\ell, j}$ and $\mu_{\ell, j}$]\duallabel{lm:nu-prop}
The following conditions hold for the maps $\nu$ and $\mu$ defined above:
\begin{description}
\item[(1)] for every $\ell\in [L]$ and every $0\leq j\leq \ell$ the maps $\nu_{\ell, j}$ and $\mu_{\ell, j}$ are injective;
\item[(2)] every $\ell, \ell'\in [L]$ every $0\leq j\leq \ell$, $0\leq j'\leq \ell'$, one has 
$$
\nu_{\ell, j}(T^\ell\setminus T_*^\ell)\cap \nu_{\ell', j'}(T^{\ell'}\setminus T_*^{\ell'})=\emptyset
$$ unless $\ell=\ell'$ and $j=j'$.
\item[(3)] every $\ell, \ell'\in [L]$ every $0\leq j\leq \ell$, $0\leq j'\leq \ell'$, one has 
$$
\mu_{\ell, j}(T^\ell\setminus T_*^\ell)\cap \mu_{\ell', j'}(T^{\ell'}\setminus T_*^{\ell'})=\emptyset
$$ unless $\ell=\ell'$ and $j=j'$.
\end{description}
\end{lemma}
\begin{proof}
\noindent {\bf (1)}  follows since $\tau$ is injective by Claim~\dualref{cl:pi-star-injective} and $\downset$ is injective by construction (Definition~\dualref{def:downset}).

We now show {\bf (2)}. First note that  $\nu_{\ell, j}(T^\ell\setminus T_*^\ell)\subseteq T^{\ell-j}$ and $\nu_{\ell', j'}(T^{\ell'}\setminus T_*^{\ell'})\subseteq T^{\ell'-j'}$, and hence the two sets are disjoint if $\ell-j\neq \ell'-j'$. Now suppose that $\ell-j=\ell'-j'$ and assume without loss of generality that $\ell\leq \ell'$. Furthermore, we can assume that $\ell<\ell'$, since if $\ell=\ell'$, one must have $j=j'$ as otherwise the sets are disjoint by the previous argument. Now note that 
\begin{equation*}
\begin{split}
\nu_{\ell', \ell'-\ell}(T^{\ell'}\setminus T_*^{\ell'})&=\tau^{\ell+1}(\downset^{\ell+1}(\nu_{\ell', \ell'-\ell-1}(T^{\ell'}\setminus T_*^{\ell'}))\subseteq T_*^\ell,
\end{split}
\end{equation*}
since the range of $\tau^{\ell+1}$ is $T_*^\ell$ (see~Definition~\dualref{def:tau}). 
This means that 
\begin{equation*}
\begin{split}
\nu_{\ell', j'}(T^{\ell'}\setminus T_*^{\ell'})&=\nu_{\ell, j}(\nu_{\ell', \ell'-\ell}(T^{\ell'}\setminus T_*^{\ell'}))\\
&\subseteq \nu_{\ell, j}(T_*^\ell),
\end{split}
\end{equation*}
and we get that
$$
\nu_{\ell, j}(T^\ell\setminus T_*^\ell)\cap  \nu_{\ell', j'}(T^{\ell'}\setminus T_*^{\ell'})\subseteq \nu_{\ell, j}(T^\ell\setminus T_*^\ell)\cap \nu_{\ell, j}(T_*^\ell)=\emptyset
$$
since $\nu_{\ell, j}$ is injective by {\bf (1)}.

We now prove {\bf (3)}. First note that by Definition~\dualref{def:nu}
$$
\mu_{\ell, j}(T^\ell\setminus T_*^\ell)\subseteq S^{\ell-j}
$$
and 
$$
\mu_{\ell', j'}(T^{\ell'}\setminus T_*^{\ell'})\subseteq S^{\ell'-j'},
$$
and hence similarly to above the two sets are disjoint unless $\ell-j=\ell'-j'$. {\bf (3)} now follows by noting that, again using Definition~\dualref{def:nu}, we get, since $\ell-j=\ell'-j'$ and $\downset$ is injective,
\begin{equation*}
\begin{split}
\mu_{\ell, j}(T^\ell\setminus T_*^\ell)\cap \mu_{\ell', j'}(T^{\ell'}\setminus T_*^{\ell'})&=\downset^{\ell-j}(\nu_{\ell, j}(T^\ell\setminus T_*^\ell))\cap \downset^{\ell'-j'}(\nu_{\ell', j'}(T^{\ell'}\setminus T_*^{\ell'}))\\
&=\downset^{\ell-j}(\nu_{\ell, j}(T^\ell\setminus T_*^\ell) \cap \nu_{\ell', j'}(T^{\ell'}\setminus T_*^{\ell'}))\\
&=\emptyset,
\end{split}
\end{equation*}
where we used {\bf (2)} in the last transition.
\end{proof}

While for a given $\ell$ the terminal subcube $T_*^\ell$ is almost entirely covered by the range of $\tau^{\ell+1}$, it will be useful to know that almost all of $T_*^\ell$ can be covered by the image of the non-terminal part of $T^{\ell+j}$ under  $\nu_{\ell+j, j}$:
\begin{lemma}\duallabel{lm:t-star-ell-rec}
For every $\ell\in [L]$ there exists $Z^\ell\subset T_*^\ell$ such that 
\begin{equation*}
T_*^\ell= Z^\ell\cup \left(\nu_{L-1, L-1-\ell}(T_*^{L-1})\cup \bigcup_{j=1}^{L-1-\ell} \nu_{\ell+j, j}(T^{\ell+j}\setminus T_*^{\ell+j})\right)
\end{equation*}
and $|Z^\ell|\leq K^{L}\delta^{1/4} \cdot |P|$.
\end{lemma}
\begin{proof}
We prove by induction on $\ell=L-1,\ldots, 0$ that there exists sets $Z^\ell\subset T_*^\ell$ such that 
\begin{equation}\duallabel{eq:t-star-2ug2438}
T_*^\ell=Z^\ell\cup \left(\nu_{L-1, L-1-\ell}(T_*^{L-1})\cup \bigcup_{j=1}^{L-1-\ell} \nu_{\ell+j, j}(T^{\ell+j}\setminus T_*^{\ell+j})\right)
\end{equation}
and $|Z^\ell|\leq K^{L-\ell} \delta^{1/4} |P|$.

\noindent{\bf Base: $\ell=L-1$.} One has $T_*^\ell=\nu_{L-1, 0}(T_*^{L-1})$, since $\nu_{L-1, 0}$ is the identity map by definition (see Definition~\dualref{def:nu}). We let $Z^{L-1}=\emptyset$, so that $T_*^{L-1}=Z^\ell\cup \nu_{L-1, 0}(T_*^{L-1})$.

\noindent{\bf Inductive step: $\ell\to \ell-1$.}   Let
$$
Z^\ell=T_*^\ell\setminus \left(\nu_{L-1, L-1-\ell}(T_*^{L-1})\cup\bigcup_{\substack{j\geq 1}} \nu_{\ell+j, j}(T^{\ell+j}\setminus T_*^{\ell+j})\right),
$$
and note that 
\begin{equation}\duallabel{eq:92DUIHE24}
T_*^\ell=Z^\ell\cup  \left(\nu_{L-1, L-1-\ell}(T_*^{L-1})\cup\bigcup_{\substack{j\geq 1}} \nu_{\ell+j, j}(T^{\ell+j}\setminus T_*^{\ell+j})\right).
\end{equation}

Applying $\tau^\ell(\downset^\ell(\cdot))$ to both sides of~\dualeqref{eq:92DUIHE24}, we get, letting $Q=\nu_{L-1, L-1-\ell}(T_*^{L-1})$ and 
\begin{equation}\duallabel{eq:z-prime}
Z'=\tau^\ell(\downset^\ell(Z))
\end{equation}
to simplify notation,
\begin{equation}\duallabel{eq:o4329hg394hg}
\begin{split}
\tau^\ell(\downset^\ell(T_*^\ell))&=Z'\cup \tau^\ell\left(\downset^\ell\left(Q\cup \bigcup_{\substack{j\geq 1}} \nu_{\ell+j, j}(T^{\ell+j}\setminus T_*^{\ell+j})\right)\right)\\
&= Z'\cup\tau^\ell\left(\downset^\ell\left(Q\right)\right)\cup \bigcup_{\substack{j\geq 1}} \tau^\ell\left(\downset^\ell\left(\nu_{\ell+j, j}(T^{\ell+j}\setminus T_*^{\ell+j})\right)\right)\\
&= Z'\cup\nu_{L-1, L-\ell}(T_*^{L-1})\cup \bigcup_{\substack{j\geq 1}} \nu_{\ell+j, j+1}(T^{\ell+j}\setminus T_*^{\ell+j})\\
&= Z'\cup\nu_{L-1, L-\ell}(T_*^{L-1})\cup \bigcup_{\substack{j\geq 1}} \nu_{(\ell-1)+(j+1), j+1}(T^{(\ell-1)+(j+1)}\setminus T_*^{(\ell-1)+(j+1)})\\
&= Z'\cup\nu_{L-1, L-1-(\ell-1)}(T_*^{L-1})\cup \bigcup_{\substack{j\geq 2}} \nu_{(\ell-1)+j, j}(T^{(\ell-1)+j}\setminus T_*^{(\ell-1)+j}),
\end{split}
\end{equation}
where the third transition is by Claim~\dualref{cl:nu-prop-basic}, {\bf (1)}.
At the same time we also have
\begin{equation}\duallabel{eq:9h429gh23g}
\tau^\ell(\downset^\ell(T^\ell\setminus T_*^\ell))=\nu_{\ell, 1}(T^\ell\setminus T_*^\ell)=\nu_{(\ell-1)+1, 1}(T^\ell\setminus T_*^\ell).
\end{equation}

Let $Z''=T_*^{\ell-1}\setminus \tau^\ell(S^\ell)$. We have
\begin{equation}\duallabel{eq:239hg92hg}
\begin{split}
T_*^{\ell-1}&=Z''\cup \tau^\ell(S^\ell)\\
&=Z''\cup \tau^\ell(\downset^\ell(T^\ell))\\
&=Z''\cup \tau^\ell(\downset^\ell(T^\ell\setminus T_*^\ell))\cup \tau^\ell(\downset^\ell(T_*^\ell))\\
\end{split}
\end{equation}

Substituting~\dualeqref{eq:o4329hg394hg} and~\dualeqref{eq:9h429gh23g} into~\dualeqref{eq:239hg92hg}, we get
\begin{equation}\duallabel{eq:92h393yh9h9HUFGDUG}
T_*^{\ell-1}=Z^{\ell-1}\cup \nu_{L-1, L-\ell}(T_*^{L-1})\cup\bigcup_{\substack{j\geq 1}} \nu_{\ell-1+j, j}(T^{\ell-1+j}\setminus T_*^{\ell-1+j}),
\end{equation}
where we let $Z^{\ell-1}:=Z'\cup Z''$, so that by~\dualeqref{eq:z-prime}
$$
Z^{\ell-1}=Z'\cup Z''=\tau^\ell(\downset^\ell(Z^\ell))\cup (T_*^{\ell-1}\setminus \tau^\ell(S^\ell)).
$$
Thus, in order to complete the proof of the inductive claim, we need to show that $|Z^{\ell-1}|\leq K^{L-1-(\ell-1)}\delta^{1/4}\cdot |P|$. 

We have $|Z''|\leq \delta^{1/4} |T_0|$ by Lemma~\dualref{lm:t-star-minus-tau-sp}, and hence
\begin{equation}
\begin{split}
|Z^\ell|&\leq |\tau^\ell(\downset^\ell(Z^\ell))|+|Z''|\\
&\leq (K/2)|Z^\ell|+|Z''|\\
&\leq (K/2)\cdot K^{L-1-\ell} \delta^{1/4} |P|+\delta^{1/4} |P|\\
&\leq  K^{L-1-(\ell-1)} \delta^{1/4} |P|,\\
\end{split}
\end{equation}
The second inequality is due to the fact that $\downset^\ell$ maps every vertex to at most $K/2$ vertices, and $\tau^\ell$ is a one-to-one map. The third inequality uses the inductive hypothesis and the bound $|Z''|\leq \delta^{1/4} |T_0|\leq \delta^{1/4} |P|$. This completes the proof of the inductive step.

Finally, to obtain the result of the lemma, we extend the union on the right hand side of~\dualeqref{eq:92h393yh9h9HUFGDUG} to include $j=0$, getting
$$
T^\ell=Z^{\ell-1}\cup \nu_{L-1, L-1-\ell}(T_*^{L-1})\cup\bigcup_{\substack{j\geq 0}} \nu_{\ell+j, j}(T^{\ell+j}\setminus T_*^{\ell+j}),
$$
as required. This completes the proof of the inductive step.
\end{proof}

\subsection{Image of non-terminal subsets $T^\ell\setminus T_*^\ell$ under $\mu$ (proof of Lemma~\dualref{lm:mu-ell-j})} 
In this section we prove Lemma~\dualref{lm:mu-ell-j}.  We start with two auxiliary lemmas, and a definition:

\begin{definition}[Vector consistent with the terminal subcube]\duallabel{def:cons}
For every $\ell\in [L]$, if $\J=\Sp(\B^\ell)$, we say that a vector $\f\in \deltagrid^{\J}$ is {\em consistent with the terminal subcube $T_*^\ell$} if for every $k\in [K/2]$ one has 
$\f_{\j_k}\in [0, 1-\frac1{K-k})$.
\end{definition}

The first lemma bounds the size of the image of a rectangle $F$ consistent with the terminal subcube under the predecessor map $\nu$:
\begin{lemma}\duallabel{lm:rect-nu-j}
There exists an absolute constant $C>0$ such that for every $\ell \in [L]$, every $j=0,\ldots, \ell$,  if $\I=\Sp(\B^\ell)$ and $\H\subseteq \Sp(\B^{>\ell})$, the following conditions hold. 

For every $\f\in \deltagrid^{\I}$ consistent with the terminal subcube $T_*^\ell$ (as per Definition~\dualref{def:cons}) and every $\c, \d\in \deltagrid^{\H}, \c<\d$, 
if $F=\rect(\I\cup \H, (\f, \c), (\f+\Delta\cdot \mathbf{1}, \d))$, then
$$
(\ln 2-C/K)^j |F|\leq |\nu_{\ell, j}(F)|\leq (\ln 2+C/K)^j |F|
$$
for an absolute constant $C>0$.
\end{lemma}
\begin{proof}
The proof is by induction on $j$. The inductive claim is that
for every $\ell \in [L]$,  if $\I=\Sp(\B^\ell)$ and $\H\subset \Sp(\B^{\geq \ell+1})$, the following conditions hold. For every $\f\in \deltagrid^{\I}$ and every $\c,  \d\in \deltagrid^\H$, if $F=(\I\cup \H, (\f, \c), (\f+\Delta\cdot \mathbf{1}, \d))$, then 
$$
(\ln 2-C/K)^j |F|\leq |\nu_{\ell, j}(F)|\leq (\ln 2+C/K)^j |F|
$$
where $C>0$ is the absolute constant.

\noindent{\bf Base: $j=0$.} We have $|\nu_{\ell, j}(F)|=|F|$ since $\nu_{\ell, 0}(F)=F$ for every $F$.

\noindent{\bf Inductive step: $j \to j+1$.} Fix $k\in [K/2]$, and let $\rho_k$ be the  $(K-k, \q_k)$-compressing map as per Definition~\dualref{def:rho}, where $\q_k=\q_k^\ell\in \F$ is the compression vector for the $k$-th phase of the graph $G^\ell$ (see Definition~\dualref{def:ext}). We let $\J':=\J^\ell$, $\J:=\J^{\ell-1}$,  $\B':=\B^\ell$, $\B:=\B^{\ell-1}$ to simplify notation. Similarly define $T':=T^\ell$, $T:=T^{\ell-1}$ and $S':=S^\ell$, $S:=S^{\ell-1}$ to simplify notation. Let $\r':=\r^\ell$ and $\r:=\r^{\ell-1}$ denote the $\ell$-th and the $(\ell-1)$-th compression index respectively. We define 
$$
Z_k:=\left\{x\in [m]^n: \weight(x) \in \left[0, \frac{1}{K-k}\right)\cdot W \pmod{ W}\right\},
$$
and let 
\begin{equation}\duallabel{eq:q-ext-k}
Q^{ext}_k:=\left\{x\in F: \langle x, \q_k\rangle \pmod{M} \in \left[0, \frac1{K-k}\right)\cdot M \right\}
\end{equation}
to simplify notation.  The proof of the inductive step proceeds in several steps. In {\bf Step 1} we show that for every $k\in [K/2]$ the image of the $k$-th downset of $F$ under the compression map $\rho_k$ is essentially the entire set $Q^{ext}_k$ (this is formally stated in~\dualeqref{eq:923hgh239gfdfdfF} below). Then in {\bf Step 2} we bound $|\nu_{\ell-1, j-1}(\Pi^*_k(Q^{ext}_k))|$ for $k\in [K/2]$ using the inductive hypothesis. Finally, in {\bf Step 3} we put our bounds together to obtain the result of the lemma.

We start by establishing some basic bounds relating $|Q^{ext}_k|$ to $|F|$ that will be useful throughout the proof.
 We let $\gamma=\Delta^{|\I|}\cdot \prod_{\i\in \H} (\d_\i-\c_\i)$ and invoke Lemma~\dualref{lm:rect-size}. By Lemma~\dualref{lm:rect-size}, {\bf (1)}, we get
\begin{equation}\duallabel{eq:923ty3yHGFGFDFF}
(1-\sqrt{\e})\gamma\leq \left|F\right|/m^n\leq (1+\sqrt{\e})\gamma,
\end{equation}
and by Lemma~\dualref{lm:rect-size}, {\bf (2)},
\begin{equation}\duallabel{eq:092ht98gt8g8GFDEdZ}
(1-\sqrt{\e})\frac1{K-k}\cdot \gamma\leq \left|F\cap Z_k\right|/m^n\leq (1+\sqrt{\e})\frac1{K-k}\cdot \gamma.
\end{equation}
Similarly, by Lemma~\dualref{lm:rect-size}, {\bf (1)}, $\gamma'=\frac1{K-k}\cdot \Delta^{|\I|}\cdot \prod_{\i\in \H} (\d_\i-\c_\i)=\frac1{K-k}\cdot \gamma$  we get
$$
(1-\sqrt{\e})\frac1{K-k}\cdot \gamma\leq \left|Q^{ext}_k\right|/m^n\leq (1+\sqrt{\e})\frac1{K-k}\cdot \gamma.
$$

The bounds above imply 
\begin{equation}\duallabel{eq:f-vs-qext}
(1-3\sqrt{\e}) \frac1{K-k}|F|\leq |Q^{ext}_k|\leq (1+3\sqrt{\e}) \frac1{K-k}|F|
\end{equation}
for every $k\in [K/2]$.  Recalling that 
$$
\downset_k(F)\delequal F\cap Z_k
$$ and putting the above bounds together, we get $\rho_k\left(\downset_k(\Int_\delta(F))\right)\subseteq Q^{ext}_k$ and
\begin{equation*}
\begin{split}
\left|\rho_k\left(\downset_k(F)\right)\cap Q^{ext}_k\right| &\geq \left|\rho_k\left(\downset_k(\Int_\delta(F))\right)\cap Q^{ext}_k\right|\\
&\geq \left|\downset_k(\Int_\delta(F))\right|\\
&\geq \left|\downset_k(F)\right|-|F\setminus \Int_\delta(F)|\\
&= \left|F\cap Z_k\right|-|F\setminus \Int_\delta(F)|\\
&\geq (1-\sqrt{\e})\frac1{K-k}\cdot \gamma\cdot m^n-\sqrt{\delta} |F|\\
&\geq (1-2K\sqrt{\delta}-3\sqrt{\e})|Q^{ext}_k|.\\
\end{split}
\end{equation*}
The first transition uses the fact that by Lemma~\dualref{lm:densification} one has that 
$$
\rho_k\left(\Int_\delta(F)\cap Z_k\right)\subseteq Q^{ext}_k
$$
and $\rho_k$ is injective on $\Int_\delta(F)\cap Z_k$.
The transition from line~4 to line~5 is by Lemma~\dualref{lm:rect-int-size} and~\dualeqref{eq:092ht98gt8g8GFDEdZ}. The transition from line~5 to line~6 is by~\dualeqref{eq:f-vs-qext}. Similarly, we get
\begin{equation*}
\begin{split}
\left|\rho_k\left(\downset_k(F)\right)\right| &=\left|\rho_k\left(\downset_k(F)\right)\cap Q^{ext}_k\right|+\left|\rho_k\left(\downset_k(F)\right)\setminus Q^{ext}_k\right|\\
&\leq \left|Q^{ext}_k\right|+\left|\rho_k\left(\downset_k(F)\right)\setminus \rho_k\left(\downset_k(\Int_\delta(F))\right)\right|\\
&\leq \left|Q^{ext}_k\right|+\left|F\setminus \Int_\delta(F)\right|\\
&\leq \left|Q^{ext}_k\right|+\sqrt{\delta}\left|F\right|\\
&\leq (1-2K\sqrt{\delta})|Q^{ext}_k|.\\
\end{split}
\end{equation*}
The second transition uses the fact that by Lemma~\dualref{lm:densification}, {\bf (2)}, one has 
$\rho_k\left(\Int_\delta(F)\cap Z_k\right)\subseteq Q^{ext}_k$
and $\rho_k$ is injective by Lemma~\dualref{lm:densification}, {\bf (1)}.
The transition from line~3 to line~4 is by Lemma~\dualref{lm:rect-int-size}. The transition from line~4 to line~5 is by~\dualeqref{eq:f-vs-qext}.

Noting that $2K\sqrt{\delta}+2\sqrt{\e}\leq \delta^{1/4}$ by ~\dualref{p4} and~\dualref{p5}, we record the above in the simpler form 
\begin{equation}\duallabel{eq:923hgh239gfdfdfF}
(1-\delta^{1/4})|Q^{ext}_k|\leq \left|\rho_k\left(\downset_k(F)\right)\cap Q^{ext}_k\right|\leq (1+\delta^{1/4})|Q^{ext}_k|. 
\end{equation}

\noindent{\bf Step 1.}  By Claim~\dualref{cl:nu-prop-basic}, {\bf (1)}, we have
\begin{equation*}
\begin{split}
\nu_{\ell, j}(F)&=\nu_{\ell-1, j-1}(\tau^\ell(\downset^{\ell}(F))\\
&=\bigcup_{k\in [K/2]} \nu_{\ell-1, j-1}(\tau^\ell_k(\downset^{\ell}_k(F))\\
&=\bigcup_{k\in [K/2]} \nu_{\ell-1, j-1}(\tau_k(\downset_k(F)),
\end{split}
\end{equation*}
where we dropped the superscipt $\ell$ in the last line to simplify notation.
For every $k\in [K/2]$ we have using~\dualeqref{eq:923hgh239gfdfdfF} and the fact that $\tau_k=\Pi^*_k\circ \rho_k$
\begin{equation}\duallabel{eq:92ythe8ghwe8hg8hf8Hffd}
\begin{split}
|\nu_{\ell-1, j-1}(\tau_k(\downset_k(F)))|&=|\nu_{\ell-1, j-1}(\Pi^*_k(\rho_k(F\cap Z_k)))|\\
&\geq \left|\nu_{\ell-1, j-1}\left(\Pi^*_k(Q^{ext}_k)\right)\right|-\max_{\substack{S\subseteq Q^{ext}_k\\|S|\leq \delta^{1/4} |Q^{ext}_k|}}\left|\nu_{\ell-1, j-1}\left(\Pi^*_k(S)\right)\right|\\
&\geq \left|\nu_{\ell-1, j-1}\left(\Pi^*_k(Q^{ext}_k)\right)\right|- K^{j-1}\delta^{1/4}\cdot |Q^{ext}_k|,
\end{split}
\end{equation}
since for every $S$ one has $|\nu_{\ell-1, j-1}(S)|\leq K^{j-1} |S|$.  For the upper bound we have for every $k\in [K/2]$ 
\begin{equation}\duallabel{eq:92ythe8ghwe8hg8hf8Hffd-upper}
\begin{split}
|\nu_{\ell-1, j-1}(\tau_k(\downset_k(F)))|&=|\nu_{\ell-1, j-1}(\Pi^*_k(\rho_k(F\cap Z_k)))|\\
&\leq \left|\nu_{\ell-1, j-1}\left(\Pi^*_k(Q^{ext}_k)\right)\right|+\left|\nu_{\ell-1, j-1}\left(\Pi^*_k(\rho_k(F\cap Z_k)\setminus Q^{ext}_k)\right)\right|\\
&\leq \left|\nu_{\ell-1, j-1}\left(\Pi^*_k(Q^{ext}_k)\right)\right|+ K^{j-1}\delta^{1/4}\cdot |Q^{ext}_k|.
\end{split}
\end{equation}
We used the second inequality in~\dualeqref{eq:923hgh239gfdfdfF} in the last transition, together with the fact that $\Pi^*_k$ maps every vertex to at most one vertex.

\noindent{\bf Step 2.} We now apply the inductive hypothesis to bound the first term in~\dualeqref{eq:92ythe8ghwe8hg8hf8Hffd} (which coincides with the first term on the rhs of~\dualeqref{eq:92ythe8ghwe8hg8hf8Hffd-upper}).
Define
\begin{equation}\duallabel{eq:i0-prime}
\begin{split}
\I_0'&=\J'_{< k}\cup \Ext_k\cup \{\q_k\}\\
\I_1'&=\J'_{\geq k}\cup \{\r'\} \cup \H,
\end{split}
\end{equation}
where we let $\Ext_k:=\Ext_k^\ell$ to simplify notation. Let $\f_0$ denote the restriction of $\f$ to coordinates in $\J'_{<k}$, and let $\f_1$ denote the restriction of $\f$ to coordinates in $\Sp(\B^\ell)\setminus \J'_{<k}=\J'_{\geq k}\cup \{\r'\}$. We also let 
\begin{equation}\duallabel{eq:i0}
\begin{split}
\I_0&=\J\cup \{\r\}\\
\I_1&=\I_1'.
\end{split}
\end{equation}

Now note that the definition of $\I_0'$ in~\dualeqref{eq:i0-prime} coincided with the definition of $\I'_k$ in~\dualeqref{eq:i-prime}, and the definition of $\I_0$ in~\dualeqref{eq:i0} coincides with the definition of $\I$ in~\dualeqref{eq:i}. Thus, by~Definition~\dualref{def:tau} the map $\tau_k:S'_k\to T_*$ is defined by letting 
$$
\tau_k(x)=\Pi^*_k(\rho_k(x)),
$$  
where $\Pi^*_k$ is defined as follows. For $z=\rho_k(x)\in [m]^n$ (leave $\tau_k$ undefined if $\rho_k(x)$ is not defined) one lets $\a_0\in \mathbb{D}_k$, $R'(\a_0):=(\I_0', \a_0, \a_0+\Delta\cdot \mathbf{1})$ be such that 
\begin{equation}\duallabel{eq:def-rd}
z\in \Int_\delta(R'(\a_0))
\end{equation}
if such an $\a_0\in \mathbb{D}_k$ exists (otherwise $\Pi^*(z)$ is left undefined). Then one lets
\begin{equation}\duallabel{eq:rp-d}
R(\a_0):=(\I_0, \mathsf{M}(\a_0), \mathsf{M}(\a_0)+\Delta\cdot \mathbf{1}),
\end{equation}
where $\mathsf{M}$ is as in~\dualeqref{eq:m-def}, and sets, as per Definition~\dualref{def:pi-global},
\begin{equation}\duallabel{eq:wigh0934h-thgdd}
\Pi^*_k(z):=\Pi_{R'(\a_0)\to R(\a_0)}(z). 
\end{equation}

We now show that $\tau_k(Q^{ext}_k)$ can be approximated by a union of rectangles consistent with the terminal subcube $T_*$, to which we can apply the inductive hypothesis. Recall that by~\dualeqref{eq:q-ext-k} one has
\begin{equation*}
Q^{ext}_k:=\left\{x\in F: \langle x, \q_k\rangle \pmod{M} \in \left[0, \frac1{K-k}\right)\cdot M \right\}
\end{equation*}
where 
$$
F=\rect(\I\cup \H, (\f, \c), (\f+\Delta\cdot \mathbf{1}, \d)).
$$

For $\a_0\in (\deltagrid)^{\I_0'}$ we define the extended rectangles 
$$
R_{ext}'(\a_0)=(\I_0'\cup \I_1', (\a_0, (\f_1, \c)), (\a_0+\Delta\cdot \mathbf{1}, (\f_1+\Delta\cdot \mathbf{1}, \d)))
$$
and 
$$
R_{ext}(\a_0)=(\I_0\cup \I_1, (\mathsf{M}(\a_0), (\f_1, \c)), (\mathsf{M}(\a_0)+\Delta\cdot \mathbf{1}, (\f_1+\Delta \cdot \mathbf{1}, \d))).
$$

We now recall the definition of $\mathbb{D}_k$ (see Definition~\dualref{def:dk}). Indeed, let $\u_{\j'_s}=0, \v_{\j'_s}=1-\frac1{K-s}$ for $s\in [k]$, let $\u_{\q_k}=0, \v_{\q_k}=\frac1{K-k},$ and $\u_\i=0, \v_\i=1$ for $\i\in \Ext_k$. Then, noting that $\I_0'$ as per~\dualeqref{eq:i0-prime} is equal to $\I'$ as per~\dualeqref{eq:i-prime}
one has as per~\dualeqref{eq:dk-def}
\begin{equation}\duallabel{eq:dk-def-tt}
\mathbb{D}_k=(\deltagridInt)^{\I'_0}\cap \prod_{\i\in \I'_0} [\u_\i, \v_\i).
\end{equation}
We start by noting  that $Q^{ext}_k$ is a rectangle in $\I_0'\cup \I'_1$.  
Indeed, let 
\begin{equation*}
\begin{split}
&\u^0_{\j'_s}=\f_{\j_s}, \v^0_{\j'_s}=\f_{\j_s}+\Delta\text{~~~~~ for~} s\in [k]\text{~~~~~~~~~~~~~~~~~~~~~~~~~~~~~~~~~~~~~}\\
&\u^0_{\q_k}=0, \v^0_{\q_k}=\frac1{K-k}\\
&\text{and}\\
&\u^0_\i=0, \v^0_\i=1\text{~~~~~for~} \i\in \Ext_k,
\end{split}
\end{equation*}
so that $\u^0, \v^0\in (\deltagrid)^{\I'_0}$ (note that $\u^0=\f_0$ and $\v^0=\f_0+\Delta\cdot \mathbf{1}$). Also let 
\begin{equation*}
\begin{split}
&\u^1_{\j'_s}=\f_{\j_s}, \v^1_{\j'_s}=\f_{\j_s}+\Delta\text{~~~~~ for~} s\in \{k, k+1,\ldots, K/2-1\},\\
&\u^1_{\r'_k}=0, \v^1_{\r'_k}=1,\\
&\text{and}\\
&\u^1_\i=\c_\i, \v^1_\i=\d_\i\text{~~~~~for~}\i\in \H,
\end{split}
\end{equation*}
so that $\u^1, \v^1\in (\deltagrid)^{\I'_1}$. We have $Q^{ext}_k=\rect(\I_0'\cup \I_1', (\u^0, \u^1), (\v^0, \v^1))$ by Claim~\dualref{cl:subcubes-decomp}
\begin{equation}\duallabel{eq:0923yt923y9hDHFISX}
Q^{ext}_k=\bigcup_{\a_0\in \mathbb{D}^{ext}_k} R'_{ext}(\a_0),
\end{equation}
where 
\begin{equation*}
\begin{split}
\mathbb{D}^{ext}_k&=(\deltagridInt)^{\I'_0}\cap \prod_{\i\in \I'_0} [\u^0_\i, \v^0_\i)\\
&\subseteq (\deltagridInt)^{\I'_0}\cap \prod_{\i\in \I'_0} [\u_\i, \v_\i)\\
&=\mathbb{D}_k,
\end{split}
\end{equation*}
as required. The first transition is by Claim~\dualeqref{cl:subcubes-decomp}. The second transition is due to the fact that $Q^{ext}_k\subseteq T'_*$ since $F$ is consistent with $T'_*$ by assumption, and hence $Q^{ext}_k\subseteq T'_k$. The last transition is by definition of $\mathbb{D}_k$.

We let $\gamma=\Delta^{|\I'_0|}\cdot \prod_{\i\in \I'_1} (\d_\i-\c_\i)=\Delta^{|\I_0|}\cdot \prod_{\i\in \I_1} (\d_\i-\c_\i)$ and invoke Lemma~\dualref{lm:rect-size}. By Lemma~\dualref{lm:rect-size}, {\bf (1)}, we get
\begin{equation}\duallabel{eq:r-ext-prime-size}
(1-\sqrt{\e})\gamma\leq \left|R'_{ext}(\a_0)\right|/m^n\leq (1+\sqrt{\e})\gamma,
\end{equation}
and 
\begin{equation}\duallabel{eq:r-ext-size}
(1-\sqrt{\e})\gamma\leq \left|R_{ext}(\a_0)\right|/m^n\leq (1+\sqrt{\e})\gamma.
\end{equation}

Fix some $\a_0\in \mathbb{D}_k$. We write $R'_{ext}$ and $R_{ext}$ to denote $R'_{ext}(\a_0)$ and $R_{ext}(\a_0)$, and write $R'$ and $R$ to denote $R'(\a_0)$ and $R(\a_0)$, omitting the dependence on $\a_0$ to simplify notation, when $\a_0$ is fixed. By Lemma~\dualref{lm:tau-of-rect}, {\bf (1)} we have 

\begin{equation}\duallabel{eq:HDOSJF}
\Pi^*_k(\Int_\delta(R'_{ext}))\subseteq R_{ext}.
\end{equation}
 At the same time by Lemma~\dualref{lm:rect-int-size} we have 
 \begin{equation}\duallabel{eq:923yt23ht23t}
 |\Int_\delta(R'_{ext})|\geq (1-\sqrt{\delta}) |R'_{ext}|,
 \end{equation}
  and by Lemma~\dualref{lm:rect-size} one has $|R_{ext}|\leq (1+3\sqrt{\e}) |R'_{ext}|$. 
Putting these bounds together with~\dualeqref{eq:HDOSJF} and using the fact that $\e$ is smaller than $\delta$ by a large enough absolute constant by ~\dualref{p6}, we get, writing $\text{Dom}(\Pi^*_k)$ to denote the domain of $\Pi^*_k$,

\begin{equation}\duallabel{eq:ihgSGF139h9HdxXV}
\begin{split}
|R_{ext}\setminus \Pi^*_k(\Int_\delta(R'_{ext}))|&\leq  |R_{ext}|-|R'_{ext}|+|R'_{ext}\setminus \Int_\delta(R'_{ext})|+|R'_{ext}\setminus \text{Dom}(\Pi^*_k)|\\
&\leq  3\sqrt{\e} |R'_{ext}|+|R'_{ext}\setminus \Int_\delta(R'_{ext})|+|R'_{ext}\setminus \text{Dom}(\Pi^*_k)|\\
&\leq  (3\sqrt{\e}+\sqrt{\delta}) |R'_{ext}|+|R'_{ext}\setminus \text{Dom}(\Pi^*_k)|.\\
\end{split}
\end{equation}
We now bound $|R'_{ext}\setminus \text{Dom}(\Pi^*_k)|$. By Lemma~\dualref{lm:tau-of-rect}, {\bf (2)} we have 
\begin{equation*}
\begin{split}
|R'_{ext}\setminus \text{Dom}(\Pi^*_k)|&\leq |R'\setminus \text{Dom}(\Pi^*_k)|\\
&\leq 8\sqrt{\e} |R'|\\
&\leq 8\sqrt{\e} \Delta^{-2K^2}|R'_{ext}|,\\
\end{split}
\end{equation*}
where we used the fact that $|R'_{ext}|\geq \Delta^{2K^2}\cdot |R'|$ by Lemma~\dualref{lm:rect-size}, {\bf (1)}, together with the fact that $|\Sp(\B^\ell)|\leq KL\leq K^2$. Substituting this into~\dualeqref{eq:ihgSGF139h9HdxXV}, we get
\begin{equation}\duallabel{eq:9h30ytar9hglkfS32gf-lb}
\begin{split}
|R_{ext}\setminus \Pi^*_k(\Int_\delta(R'_{ext}))|&\leq  (3\sqrt{\e}+\sqrt{\delta}+8\sqrt{\e} \Delta^{-2K^2}) |R'_{ext}|\\
&\leq  2\sqrt{\delta} |R'_{ext}|\\
\end{split}
\end{equation}
by \dualref{p5} and~\dualref{p6}.  At the same time, we have by Lemma~\dualref{lm:rect-int-size} and~\dualeqref{eq:HDOSJF}
\begin{equation}\duallabel{eq:9h30ytar9hglkfS32gf-ub}
\begin{split}
|\Pi^*_k(R'_{ext})\setminus R_{ext}|&\leq  |R'_{ext}\setminus \Int_\delta(R'_{ext})|\leq  \sqrt{\delta} |R'_{ext}|.\\
\end{split}
\end{equation}

We now note that $\I_0=\Sp(\B^\ell)$ as per~\dualeqref{eq:i0} and $\mathsf{M}(\a_0)$ is consistent with $T_*$ by definition of the map $\mathsf{M}$ (see~\dualeqref{eq:m-def} and~\dualeqref{eq:m-a-def}). Thus, the inductive hypothesis applies to the rectangle $R_{ext}$ (the rhs of ~\dualeqref{eq:HDOSJF}), and we get 
\begin{equation}\duallabel{eq:28h2f8hfdsfSDFF}
(\ln 2-C/K)^{j-1} |R_{ext}|\leq |\nu_{\ell-1, j-1}(R_{ext})|\leq (\ln 2+C/K)^{j-1} |R_{ext}|,
\end{equation}
and hence, using the first inequality above together with~\dualeqref{eq:9h30ytar9hglkfS32gf-lb},
\begin{equation}\duallabel{eq:3y823UGFUG-lb}
\begin{split}
|\nu_{\ell-1, j-1}(\Pi^*_k(R'_{ext}))|&\geq |\nu_{\ell-1, j-1}(R_{ext})|-|\nu_{\ell-1, j-1}(R_{ext}\setminus \Pi^*_k(\Int_\delta(R'_{ext})))|\\
&\geq |\nu_{\ell-1, j-1}(R_{ext})|-K^{j-1}\cdot 2\sqrt{\delta} \cdot |R'_{ext}|\\
&\geq (\ln 2-C/K)^{j-1}\cdot |R_{ext}|-\delta^{1/4} \cdot |R'_{ext}|\\
&\geq ((\ln 2-C/K)^{j-1} (1-3\sqrt{\e})-\delta^{1/4}) \cdot |R'_{ext}|,\\
\end{split}
\end{equation}
where the transition from the first line to the second is because for every $\ell'$ the map $\tau^{\ell'}$ maps no vertex in $S'$ to more than $K/2$ vertices in $T_*$, and in particular $\nu_{\ell-1, j-1}$ maps no vertex in $S^{\ell-1}$ to more than $(K/2)^{j-1}$ vertices in $T_*^{\ell-j}$ (note that we are using the looser bound of $K^j$ on the product of these two bounds to simplify notation). The transition to the second to last line uses the fact that $K^{j-1}2\sqrt{\delta}\leq \delta^{1/4}$ by ~\dualref{p5}.  The transition to the last line uses~\dualeqref{eq:r-ext-prime-size} and~\dualeqref{eq:r-ext-size}.  Using the second inequality in~\dualeqref{eq:28h2f8hfdsfSDFF} together with~\dualeqref{eq:9h30ytar9hglkfS32gf-ub}, we similarly get 
\begin{equation}\duallabel{eq:3y823UGFUG-ub}
\begin{split}
|\nu_{\ell-1, j-1}(\Pi^*_k(R'_{ext}))|&\leq |\nu_{\ell-1, j-1}(R_{ext})|+|\nu_{\ell-1, j-1}(\Pi^*_k(R'_{ext})\setminus R_{ext})|\\
&\leq |\nu_{\ell-1, j-1}(R_{ext})|+K^{j-1}\cdot 2\sqrt{\delta} \cdot |R'_{ext}|,\\
&\leq (\ln 2+C/K)^{j-1}\cdot |R_{ext}|+\delta^{1/4} \cdot |R'_{ext}|\\
&\leq ((\ln 2+C/K)^{j-1}\cdot (1+3\sqrt{\e})+\delta^{1/4}) \cdot |R'_{ext}|,\\
\end{split}
\end{equation}
where the transition from the first line to the second is because for any $\ell'$ the map $\tau^{\ell'}$ maps no vertex in $S'$ to more than $K/2$ vertices in $T_*$, and in particular $\nu_{\ell-1, j-1}$ maps no vertex in $S^{\ell-1}$ to more than $(K/2)^{j-1}$ vertices in $T_*^{\ell-j}$, as well as the fact that $\Pi^*_k(R'_{ext})\setminus R_{ext}\subseteq \Pi^*_k(R'_{ext}\setminus \Int_\delta(R'_{ext}))$ by~\dualeqref{eq:HDOSJF}. The transition to the second to last line uses the fact that $K^{j-1}2\sqrt{\delta}\leq \delta^{1/4}$ by~\dualref{p5}. The transition to the last line uses~\dualeqref{eq:r-ext-prime-size} and~\dualeqref{eq:r-ext-size}.


Putting~\dualeqref{eq:3y823UGFUG-lb} together with~\dualeqref{eq:0923yt923y9hDHFISX} and~\dualeqref{eq:923hgh239gfdfdfF}, and recalling that $R'_{ext}=R'_{ext}(\a_0)$, we get for the lower bound 
\begin{equation}\duallabel{eq:82g8t8GFYEGFY}
\begin{split}
|\nu_{\ell, j}(F))|&=\sum_{k\in [K/2]} |\nu_{\ell-1, j-1}(\tau(\downset_k(F)))|\\
&=\sum_{k\in [K/2]} |\nu_{\ell-1, j-1}(\Pi^*_k(\rho_k(\downset_k(F)))|\\
&\geq (1-\delta^{1/4}) \sum_{k\in [K/2]} |\nu_{\ell-1, j-1}(\Pi^*_k(Q^{ext}_k))|\\
&= (1-\delta^{1/4}) \sum_{k\in [K/2]} \sum_{\substack{\a_0\in \mathbb{D}^{ext}_k}} |\nu_{\ell-1, j-1}(\Pi^*_k(R'_{ext}(\a_0)))|\\
&\geq \sum_{k\in [K/2]} ((\ln 2-C/K)^{j-1}(1-3\sqrt{\e})-\delta^{1/4})\sum_{\substack{\a_0\in \mathbb{D}^{ext}_k}} |R'_{ext}(\a_0)|\\
&\geq ((\ln 2-C/K)^{j-1}(1-3\sqrt{\e})-3\delta^{1/4}) \sum_{k\in [K/2]} |Q^{ext}_k|.
\end{split}
\end{equation}
The first transition uses the fact that  $\tau$ is injective by Claim~\dualref{cl:pi-star-injective}, $\nu_{\ell-1, j-1}$ is injective by Lemma~\dualref{lm:nu-prop}, {\bf (1)}, and $\downset_k(F)\cap \downset_{k'}(F)=\emptyset$ for $k\neq k'$.  The second transition uses the definition of $\tau$ (see Definition~\dualref{def:tau}).
The third transition uses~\dualeqref{eq:923hgh239gfdfdfF} and the last transition uses ~\dualeqref{eq:0923yt923y9hDHFISX}. For the upper bound we get using~\dualeqref{eq:3y823UGFUG-ub}
\begin{equation}\duallabel{eq:82g8t8GFYEGFY-ub}
\begin{split}
|\nu_{\ell, j}(F)|&= \sum_{k\in [K/2]} |\nu_{\ell-1, j-1}(\tau(\downset_k(F)))|\\
&=\sum_{k\in [K/2]} |\nu_{\ell-1, j-1}(\Pi^*_k(\rho_k(\downset_k(F)))|\\
&\leq (1+\delta^{1/4}) \sum_{k\in [K/2]} |\nu_{\ell-1, j-1}(\Pi^*_k(Q^{ext}_k))|\\
&= (1+\delta^{1/4}) \sum_{k\in [K/2]} \sum_{\substack{\a_0\in \mathbb{D}^{ext}_k}} |\nu_{\ell-1, j-1}(\Pi^*_k(R'_{ext}(\a_0)))|\\
&\leq ((\ln 2+C/K)^{j-1}(1+3\sqrt{\e})+3\delta^{1/4})\sum_{k\in [K/2]}\sum_{\substack{\a_0\in \mathbb{D}^{ext}_k}} |R'_{ext}(\a_0)|\\
&= ((\ln 2+C/K)^{j-1}(1+3\sqrt{\e})+3\delta^{1/4})\sum_{k\in [K/2]}|Q^{ext}_k|,\\
\end{split}
\end{equation}
where the third transition uses~\dualeqref{eq:923hgh239gfdfdfF} and the last transition uses ~\dualeqref{eq:0923yt923y9hDHFISX}.

\paragraph{Step 3: putting it together.} 

For the lower bound we have by~\dualeqref{eq:82g8t8GFYEGFY} and the fact that $|Q^{ext}_k|\geq (1-3\sqrt{\e})\frac1{K-k}|F|$ by ~\dualeqref{eq:f-vs-qext}
\begin{equation}\duallabel{eq:92y8g8GFYDG}
\begin{split}
|\nu_{\ell, j}(F))|&\geq ((\ln 2-C/K)^{j-1}(1-3\sqrt{\e})-\delta^{1/4}) \sum_{k\in [K/2]} |Q^{ext}_k|\\
&\geq  ((\ln 2-C/K)^{j-1}(1-3\sqrt{\e})-\delta^{1/4})\cdot \sum_{k\in [K/2]} (1-3\sqrt{\e}) \frac1{K-k}\cdot |F|\\
&=((\ln 2-C/K)^{j-1}(1-3\sqrt{\e})-\delta^{1/4})\cdot (1-3\sqrt{\e})\cdot |F| \cdot \sum_{k\in [K/2]} \frac1{K-k}\\
&\geq ((\ln 2-C/K)^{j-1}(1-3\sqrt{\e})-\delta^{1/4})\cdot (1-3\sqrt{\e})(\ln 2-1/K)|F|\\
&\geq (\ln 2-C/K)^j|F|.
\end{split}
\end{equation}
The second to last transition is by Claim~\ref{cl:sum-int}. The last transition used the fact that $\sum_{k\in [K/2]} \frac1{K-k}\geq \ln 2-1/K$ by Claim~\ref{cl:sum-int} and our choice of $C$ as a sufficiently large absolute constant, as well as the fact that $\delta<K^{-10}$ and $\e<K^{-10}$ by~\dualref{p4} and~\dualref{p5} together with the fact that $\Delta\leq 1/K$ by~\dualref{p3} and the fact that $K$ is larger than an absolute constant.

For the upper bound we get using~\dualeqref{eq:3y823UGFUG-ub} and the fact that $|Q^{ext}_k|\leq (1+3\sqrt{\e})\frac1{K-k}|F|$ by ~\dualeqref{eq:f-vs-qext}
\begin{equation*}
\begin{split}
|\nu_{\ell, j}(F)| &\leq ((\ln 2+C/K)^{j-1}(1+3\sqrt{\e})+\delta^{1/4})\sum_{k\in [K/2]}|Q^{ext}_k|\\
&\leq ((\ln 2+C/K)^{j-1}(1+3\sqrt{\e})+\delta^{1/4})\sum_{k\in [K/2]}(1+3\sqrt{\e}) \frac1{K-k}\cdot |F|\\
&= ((\ln 2+C/K)^{j-1}(1+3\sqrt{\e})+\delta^{1/4})(1+3\sqrt{\e})\cdot |F|\cdot \sum_{k\in [K/2]}\frac{1}{K-k}\\
&\leq ((\ln 2+C/K)^{j-1}(1+3\sqrt{\e})+\delta^{1/4})(1+3\sqrt{\e})(\ln 2) \cdot |F|\\
&= (\ln 2+C/K)^j\cdot |F|\\
\end{split}
\end{equation*}
The second to last transition is by Claim~\ref{cl:sum-int}. The last transition used the fact that $\sum_{k\in [K/2]} \frac1{K-k}\leq \ln 2$ by Claim~\ref{cl:sum-int} and our choice of $C$, as well as the fact that $\delta<K^{-10}$ and $\e<K^{-10}$ by~\dualref{p4} and~\dualref{p5} together with the fact that $\Delta\leq 1/K$ by~\dualref{p3} and the fact that $K$ is larger than an absolute constant. 
This completes the proof of the inductive step, and establishes the claim of the lemma.
\end{proof}

\begin{corollary}\duallabel{cor:rect-nu-j}
There exists an absolute constant $C>0$ such that for every $\ell \in [L]$, every $j=0,\ldots, \ell$,  if $\I=\Sp(\B^\ell)$ and $\H\subseteq \Sp(\B^{>\ell})$, the following conditions hold. 

For every $\c, \d\in \deltagrid^{\I\cup \H}, \c<\d,$ such that $\rect(\I, \c_\I, \d_\I)\subseteq T_*^\ell$ the rectangle $F=\rect(\I\cup \H, \c, \d)$ satisfies
$$
(\ln 2-C/K)^j |F|\leq |\nu_{\ell, j}(F)|\leq (\ln 2+C/K)^j |F|
$$
for an absolute constant $C>0$.
\end{corollary}
\begin{proof}
One has by Claim~\dualref{cl:subcubes-decomp}
\begin{equation}\duallabel{eq:93h98DUSGDUHx}
F=\rect(\I\cup \H, \c, \d)=\bigcup_{\a\in Q} F(\a),
\end{equation}
where $Q= (\deltagrid)^\I\cap \rect(\I, \c_\I, \d_\I)$ and 
$$
F(\a)=\rect(\I\cup \H, (\a, \c_\H), (\a+\Delta\cdot \mathbf{1}, \d_\H)).
$$ 
Since $\rect(\I, \c_\I, \d_\I)\subseteq T_*^\ell$ by assumption, we get that every $\f\in Q$ is consistent with $T_*^\ell$. Thus, Lemma~\dualref{lm:rect-nu-j} applies, and we get

\begin{equation*}
\begin{split}
(\ln 2-C/K)^j |F(\a)|\leq |\nu_{\ell, j}(F(\a))|\leq (\ln 2+C/K)^j |F(\a)|.
\end{split}
\end{equation*}
Substituting the above into~\dualeqref{eq:93h98DUSGDUHx}, using the fact that $F(\a)\cap F(\a')=\emptyset$ for $\a\neq \a'$ as well as the fact that $\nu_{\ell, j}$ is injective by Lemma~\ref{lm:nu-prop}, {\bf (1)}, gives the result.
\end{proof}

Finally, we give

\begin{proofof}{Lemma~\dualref{lm:mu-ell-j}}
We start by noting that by Definition~\dualref{def:nu}
\begin{equation*}
\begin{split}
\nu_{\ell, j+1}(T^\ell\setminus T_*^\ell)&=\tau^{\ell-j}(\downset^{\ell-j}(\nu_{\ell, j}(T^\ell\setminus T_*^\ell)))=\tau^{\ell-j}(\mu_{\ell, j}(T^\ell\setminus T_*^\ell)).
\end{split}
\end{equation*}
This means, using injectivity of $\tau^{\ell-j}$, that
\begin{equation*}
\begin{split}
\left||\mu_{\ell, j}(T^\ell\setminus T_*^\ell)|-|\nu_{\ell, j+1}(T^\ell\setminus T_*^\ell)|\right|&\leq |\{y\in S^{\ell-j}: \tau^{\ell-j}(y)=\bot\}|\\
&\leq \delta^{1/4} |T|\\
\end{split}
\end{equation*}
by Lemma~\dualref{lm:t-star-minus-tau-sp}\footnote{Note that for convenience of notation in the corner case $j=\ell$ we imagine adding a pair of sets $(T^{-1}, S^{-1})$ and a corresponding map $\tau^0: S^0\to T^{-1}_*$ so that we can talk about $\nu_{\ell, j+1}$ for all $j=0,\ldots, \ell$.}.
We have $\delta\leq K^{-100 K^2}$ by~\dualref{p3} and~\dualref{p5}, and therefore $\delta^{1/4}\leq K^{-25K^2}\leq K^{-100}\cdot (\ln 2)^L$, as $L\leq \sqrt{K}$ by~\dualref{p4}. This means that the above contributes a low order term to the final bound, and it suffices to prove that for every $\ell \in [L]$, every $j=1,\ldots, \ell+1$, one has 
\begin{equation}\duallabel{eq:923yth9ihFHUFHHF}
\frac1{2} (1-\ln 2) (\ln 2-C/K)^{j-1}|T_0| \leq |\nu_{\ell, j}(T^\ell\setminus T_*^\ell)|\leq \frac1{2}(1-\ln 2)(\ln 2+C/K)^{j-1}|T_0|
\end{equation}
for an absolute constant $C>0$. Note the power of $j-1$ as opposed to $j$ (this comes from the fact that we are using $\nu_{\ell, j}$ as a proxy for $\mu_{\ell, j-1}$, as per the argument above).

We let $\J':=\J^\ell$, $\J:=\J^{\ell-1}$,  $\B':=\B^\ell$, $\B:=\B^{\ell-1}$ to simplify notation. Similarly define $T':=T^\ell$, $T:=T^{\ell-1}$ and $S':=S^\ell$, $S:=S^{\ell-1}$ to simplify notation. Let $\r':=\r^\ell, \r:=\r^{\ell-1}$ denote the $\ell$-th and the $(\ell-1)$-th compression indices respectively. We write 
$$
T'\setminus T'_*=\bigcup_{k=0}^{K/2-1}  T'_k\setminus T'_{k+1},
$$
and note that 
$$
\nu_{\ell, j}(T'\setminus T'_*)=\bigcup_{k=0}^{K/2-1}  \nu_{\ell, j}(T'_k\setminus T'_{k+1}).
$$
We now fix $k\in [K/2]$  and note that 
\begin{equation}\duallabel{eq:9w9gh9gh9HFfojef3FS}
\begin{split}
\nu_{\ell, j}(T_k\setminus T_{k+1})&=\bigcup_{s=0}^{k} \nu_{\ell-1, j-1}(\tau^\ell(\downset^\ell_s(T_k\setminus T_{k+1}))),
\end{split}
\end{equation}
where we used Definition~\dualref{def:downset} and Remark~\dualref{rm:truncated-downset}. Also note that the sets on the rhs are disjoint since $\nu_{\ell-1, j-1}$ is injective by Lemma~\dualref{lm:nu-prop}, {\bf (1)}, $\tau^\ell$ is injective by Claim~\dualref{cl:pi-star-injective} and $\downset^\ell_s$ is injective by construction (Definition~\dualref{def:downset}). In what follows we bound the cardinality of 
$$
\nu_{\ell-1, j-1}(\tau^\ell(\downset^\ell_s(T\setminus T_*))
$$ 
for fixed $k\in [K/2]$ and $s\in \{0, 1,\ldots, k\}$, and then put these bounds together to achieve the final result of the lemma.

\paragraph{Step 1.} Define 
\begin{equation}\duallabel{eq:z-def}
Z_s:=\left\{x\in [m]^n: \weight(x) \in \left[0, \frac{1}{K-s}\right)\cdot W \pmod{ W}\right\}.
\end{equation}
Also recall that by~\dualeqref{eq:def-tk}
\begin{equation}\duallabel{eq:tk-8ht8gtFUGF}
\begin{split}
T'_k\setminus T'_{k+1}&=\left\{x\in T'_k: \langle x, \j'_k\rangle \pmod{M} \in \left[1-\frac1{K-k}, 1\right)\cdot M\right\}\\
&=\left\{x\in [m]^n: \langle x, \j'_i\rangle \pmod{M} \in \left[0, 1-\frac1{K-i}\right)\cdot M \text{~for all~}i=0,\ldots, k-1\right.\\
&\left.\text{~~~~~~~~~~~~~~~~~~~and~}\langle x, \j'_k\rangle \pmod{M} \in \left[1-\frac1{K-k}, 1\right)\cdot M\right\}
\end{split}
\end{equation}

We let 
\begin{equation}\duallabel{eq:qk-def-23865gugfs}
Q_k:=T'_k\setminus T'_{k+1}
\end{equation}
and, writing $\q'_s:=\q_s^\ell$ for $s=0,1,\ldots, k,$ to simplify notation,  let
\begin{equation}\duallabel{eq:qks-def-23865gugfs}
Q_{k, s}:=\left\{x\in Q_k: \langle x, \q'_s\rangle \pmod{M} \in \left[0, \frac1{K-s}\right)\cdot M\right\}.
\end{equation}
Let $\rho_s$ be the $(K-s, \q'_s)$-densifying map as per Definition~\dualref{def:rho}.  Now by Lemma~\dualref{lm:densification} we have 
\begin{equation}\duallabel{eq:densification-29h9hzxsd8gDJD}
\rho_s\left(\Int_\delta\left(Q_k\cap Z_s\right)\right)\subseteq Q_{k, s}.
\end{equation}

We start by noting that 
\begin{equation}\duallabel{eq:092y3tyh943ty9034yt-lb}
\begin{split}
|\nu_{\ell-1, j-1}(\tau(\downset_s(Q_k)))|&=|\nu_{\ell-1, j-1}(\tau(Q_k\cap Z_s))|\\
&=|\nu_{\ell-1, j-1}(\Pi^*_s(\rho_s(Q_k\cap Z_s)))|\\
&\geq |\nu_{\ell-1, j-1}(\Pi^*_s(Q_{k, s}))|-|\nu_{\ell-1, j-1}(\Pi^*_s(Q_{k, s}\setminus \rho_s(Q_k\cap Z_s)))|\\
&\geq |\nu_{\ell-1, j-1}(\Pi^*_s(Q_{k, s}))|-K^{j-1}|Q_{k, s}\setminus \rho_s(Q_k\cap Z_s)|.\\
\end{split}
\end{equation}
The first transition above is by Definition~\dualref{def:downset}, the second transition is by Definition~\dualref{def:tau}, and the forth transition uses the fact that $\Pi_s^*$ maps every vertex to at most one vertex, as well as the fact that $\nu_{\ell-1, j-1}$ maps every vertex to at most $K^{j-1}$ vertices. Similarly,
\begin{equation}\duallabel{eq:092y3tyh943ty9034yt-ub}
\begin{split}
|\nu_{\ell-1, j-1}(\tau(\downset_s(Q_k)))|&=|\nu_{\ell-1, j-1}(\tau(Q_k\cap Z_s))|\\
&=|\nu_{\ell-1, j-1}(\Pi^*_s(\rho_s(Q_k\cap Z_s)))|\\
&\leq |\nu_{\ell-1, j-1}(\Pi^*_s(Q_{k, s}))|+|\nu_{\ell-1, j-1}(\Pi^*_s(\rho_s(Q_k\cap Z_s)\setminus Q_{k, s}))|\\
&\leq |\nu_{\ell-1, j-1}(\Pi^*_s(Q_{k, s}))|+K^{j-1}|\rho_s(Q_k\cap Z_s)\setminus Q_{k, s}|.\\
\end{split}
\end{equation}
The first transition above is by Definition~\dualref{def:downset}, the second transition is by Definition~\dualref{def:tau}, and the forth transition uses the fact that $\Pi_s^*$ maps every vertex to at most one vertex, as well as the fact that $\nu_{\ell-1, j-1}$ maps every vertex to at most $K^{j-1}$ vertices.

\noindent{\bf Step 1.} We now upper bound $Q_{k, s}\setminus \rho_s(Q_k\cap Z_s)$ and $\rho_s(Q_k\cap Z_s)\setminus Q_{k, s}$, which allows us to upper bound the error terms in~\dualeqref{eq:092y3tyh943ty9034yt-lb} and~\dualeqref{eq:092y3tyh943ty9034yt-ub} respectively. We first apply Lemma~\dualref{lm:rect-size}, {\bf (1)}, to $Q_k$ and $Q_{k, s}$, and Lemma~\dualref{lm:rect-size}, {\bf (2)}, to $Q_k\cap Z_k$ with 
$$
\gamma:=\left(\prod_{i=0}^{k-1} \left(1-\frac1{K-i}\right)\right)\frac1{K-k}=\frac1{K}.
$$
The resulting bounds are
\begin{equation}\duallabel{eq:24yt2yt-size-bounds}
\begin{split}
|Q_k|/m^n&=(1\pm \sqrt{\e})\gamma=(1\pm \sqrt{\e})\frac1{K}\\
|Q_{k, s}|/m^n&=(1\pm \sqrt{\e})\frac1{K-s}\gamma=(1\pm \sqrt{\e})\frac1{K-s}\cdot \frac1{K}\\
|Q_k\cap Z_s|/m^n&=(1\pm \sqrt{\e}) \frac1{K-s}\cdot \gamma=(1\pm \sqrt{\e})\frac1{K-s}\cdot \frac1{K}.
\end{split}
\end{equation}

We thus get, using~\dualeqref{eq:densification-29h9hzxsd8gDJD} together with the fact that $\rho_s$ is injective,
\begin{equation}
\begin{split}
|Q_{k, s}\setminus \rho_s(Q_k\cap Z_s)|&\leq |Q_{k, s}|-|\Int_\delta(Q_k)\cap Z_s|\\
&\leq |Q_{k, s}|-|Q_k\cap Z_s|+|Q_k\setminus \Int_\delta(Q_k)|\\
&\leq (1+3\sqrt{\e}) \frac1{K-s}|Q_k|-(1-\sqrt{\e}) \frac1{K-s} |Q_k|+\sqrt{\delta} |Q_k|\\
&\leq (4\sqrt{\e}+K\sqrt{\delta})|Q_{k, s}|\\
\end{split}
\end{equation}

Similarly, since $\rho_s$ is injective,
\begin{equation}
\begin{split}
|\rho_s(Q_k\cap Z_s)\setminus Q_{k, s}|&\leq |\rho_s(Q_k\setminus \Int_\delta(Q_k))|\\
&\leq \sqrt{\delta} |Q_k|\\
&\leq K\sqrt{\delta} |Q_{k, s}|,\\
\end{split}
\end{equation}
where we used Lemma~\dualref{lm:rect-int-size} in the second transition.

Substituting the two bounds above into~\dualeqref{eq:092y3tyh943ty9034yt-lb} and~\dualeqref{eq:092y3tyh943ty9034yt-ub} respectively and noting that $K^j\cdot \sqrt{\delta}\leq \delta^{1/4}$ by~\dualref{p5} together with~\dualref{p3} and the fact that $K$ is larger than an absolute constant, we get
\begin{equation}\duallabel{eq:3yt3yt2yDF-lb}
|\nu_{\ell-1, j-1}(\tau(\downset_s(Q_k))|\geq |\nu_{\ell-1, j-1}(\Pi^*_s(Q_{k, s}))|-\delta^{1/4} |Q_{k, s}|
\end{equation}
and 
\begin{equation}\duallabel{eq:3yt3yt2yDF-ub}
|\nu_{\ell-1, j-1}(\tau(\downset_s(Q_k))|\leq |\nu_{\ell-1, j-1}(\Pi^*_s(Q_{k, s}))|+\delta^{1/4} |Q_{k, s}|.
\end{equation}

\paragraph{Step 2.}
Define
\begin{equation*}
\begin{split}
\I_0'&=\J'_{< k} \cup \Ext'_k\cup \{\q'_k\}\\
\I_1'&=\J'_{\geq k}\cup \{\r'\}
\end{split}
\end{equation*}
and $\I'=\I_0'\cup \I_1'$, as well as
\begin{equation}\duallabel{eq:i0-sm}
\begin{split}
 \I_0&=\J\cup \{\r\}\\
 \I_1&=\I_1'
 \end{split}
 \end{equation}
 and $\I=\I_0\cup \I_1$. Recall that by Definition~\dualref{def:dk} together with \dualeqref{eq:i-prime} and~Definition~\dualref{def:dk} for $k\in [K/2]$ the set $\mathbb{D}_k$ is a subset of $(\deltagrid)^{\I'_0}$ such that 
\begin{equation*}
\begin{split}
\left\{x\in T'_k: \langle x, \q_k\rangle \pmod{M} \in \left[0, \frac1{K-k}\right)\cdot M \right\}&=\bigcup_{\d\in \mathbb{D}_k} \rect(\I'_0, \d).
\end{split}
\end{equation*}

We now recall the definition of $\mathbb{D}_k$ (see Definition~\dualref{def:dk}). Indeed, let $\u_{\j'_s}=0, \v_{\j'_s}=1-\frac1{K-s}$ for $s\in [k]$, let $\u_{\q_k}=0, \v_{\q_k}=\frac1{K-k},$ and $\u_\i=0, \v_\i=1$ for $\i\in \Ext_k$. Then, noting that $\I_0'$ as per~\dualeqref{eq:i0-prime} is equal to $\I'$ as per~\dualeqref{eq:i-prime}
one has as per~\dualeqref{eq:dk-def}
\begin{equation}\duallabel{eq:dk-def-tt-3258y8y8g}
\mathbb{D}_k=(\deltagridInt)^{\I'_0}\cap \prod_{\i\in \I'_0} [\u_\i, \v_\i).
\end{equation}
We now note  that $Q_{k, s}$ is a rectangle in $\I_0'\cup \I'_1$.  
Indeed, let 
\begin{equation*}
\begin{split}
&\u^0_{\j'_i}=0, \v^0_{\j'_i}=1-\frac1{K-i}\text{~~~~~ for~} i\in [s]\text{~~~~~~~~~~~~~~~~~~~~~~~~~~~~~~~~~~~~~}\\
&\u^0_{\q_s}=0, \v^0_{\q_s}=\frac1{K-k}\\
&\text{and}\\
&\u^0_\i=0, \v^0_\i=1\text{~~~~~for~} \i\in \Ext'_s,
\end{split}
\end{equation*}
so that $\u^0, \v^0\in (\deltagrid)^{\I'_0}$. Also let 
\begin{equation*}
\begin{split}
&\u^1_{\j'_i}=0, \v^1_{\j'_i}=1-\frac1{K-i}\text{~~~~~ for~} i\in \{s, s+1,\ldots, k-1\},\\
&\u^1_{\j'_k}=1-\frac1{K-k}, \v^1_{\j'_k}=1,\\
&\u^1_{\j'_i}=0, \v^1_{\j'_i}=1\text{~~~~~ for~} i\in \{k, k+1,\ldots, K/2\}\\
&\u^1_{\r'_k}=0, \v^1_{\r'_k}=1,\\
\end{split}
\end{equation*}
so that $\u^1, \v^1\in (\deltagrid)^{\I'_1}$. We have $Q_{k, s}=\rect(\I_0'\cup \I_1', (\u^0, \u^1), (\v^0, \v^1))$, and hence using Claim~\dualref{cl:subcubes-decomp} we get
\begin{equation}\duallabel{eq:0u23tggddxv}
Q_{k, s}=\bigcup_{\substack{\a=(\a_0, \a_1)\\ \a_0\in \mathbb{D}_{k, s}, \a_1\in \mathbb{D}^{ext}_{k, s}}} R'_{ext}(\a).
\end{equation}
We have
\begin{equation}\duallabel{eq:293ht8gbdgHFHS}
\begin{split}
\mathbb{D}_{k, s}&=(\deltagridInt)^{\I'_0}\cap \prod_{\i\in \I'_0} [\u^0_\i, \v^0_\i)\\
&\subseteq (\deltagridInt)^{\I'_0}\cap \prod_{\i\in \I'_0} [\u_\i, \v_\i)\\
&=\mathbb{D}_k,
\end{split}
\end{equation}
and 
\begin{equation*}
\begin{split}
\mathbb{D}^{ext}_{k, s}&=(\deltagridInt)^{\I'_0}\cap \prod_{\i\in \I'_1} [\u^1_\i, \v^1_\i).
\end{split}
\end{equation*}

The first transition in~\dualeqref{eq:293ht8gbdgHFHS} is by Claim~\dualeqref{cl:subcubes-decomp}. The second transition is due to the fact that $Q_{k, s}\subseteq T'_k$ by~\dualeqref{eq:qk-def-23865gugfs} and~\dualeqref{eq:qks-def-23865gugfs}. The last transition is by definition of $\mathbb{D}_k$.

For $\a\in Q_k$ we write $\a_0$ to denote the restriction of $\a$ to $\I_0'$, $\a_1$ to denote the restriction of $\a$ to $\I_1'$. Let $M: \mathbb{D}_k\to \mathbb{A}$ denote the map  that defines $\tau$ (see Definition~\dualref{def:tau}, Definition~\dualref{def:pi-global} and~\dualeqref{eq:m-def}).
For $\a_0\in \deltagrid^{\I'_0}$ let
\begin{equation*}
\begin{split}
R'(\a_0)&=\rect(\I'_0, \a_0, \a_0+\Delta\cdot \mathbf{1})\\
&\text{and}\\
R(\a_0)&=\rect(\I_0, \mathsf{M}(\a_0), \mathsf{M}(\a_0)+\Delta\cdot \mathbf{1}).
\end{split}
\end{equation*}

We also define extended rectangles by letting for $\a=(\a_0, \a_1)\in (\deltagrid)^{\I'}=(\deltagrid)^{\I_0'\cup \I_1'}$
\begin{equation}\duallabel{eq:rext-def-iwbg903g}
\begin{split}
R'_{ext}(\a)&=\rect(\I_0'\cup \I_1', ((\a_0, \a_1), (\a_0+\Delta\cdot \mathbf{1}, \a_1+\Delta\cdot \mathbf{1})))\\
&\text{and}\\
R_{ext}(\a)&=\rect(\I_0\cup \I_1, ((\mathsf{M}(\a_0), \a_1), (\mathsf{M}(\a_0)+\Delta\cdot \mathbf{1}, \a_1+\Delta\cdot \mathbf{1}))).
\end{split}
\end{equation}

By Lemma~\dualref{lm:tau-of-rect}, {\bf (1)} we have, omitting the dependence on $\a$ to simplify notation while $\a$ is fixed, 
\begin{equation}\duallabel{eq:HDOSJFdfwg}
\Pi^*_s(\Int_\delta(R'_{ext}))\subseteq R_{ext}.
\end{equation}
 At the same time by Lemma~\dualref{lm:rect-int-size} we have 
 \begin{equation}\duallabel{eq:923yt23ht23tdfwg}
 |\Int_\delta(R'_{ext})|\geq (1-\sqrt{\delta}) |R'_{ext}|,
 \end{equation}
and by Lemma~\dualref{lm:rect-size} one has $|R_{ext}|\leq (1+3\sqrt{\e}) |R'_{ext}|$. 
  
\begin{equation}\duallabel{eq:ihgSGF139h9HdxXV-iwhrg}
\begin{split}
|R_{ext}\setminus \Pi^*_s(\Int_\delta(R'_{ext}))|&\leq  |R_{ext}|-|R'_{ext}|+|R'_{ext}\setminus \Int_\delta(R'_{ext})|+|R'_{ext}\setminus \text{Dom}(\Pi^*_s)|\\
&\leq  3\sqrt{\e} |R'_{ext}|+|R'_{ext}\setminus \Int_\delta(R'_{ext})|+|R'_{ext}\setminus \text{Dom}(\Pi^*_s)|\\
&\leq  (3\sqrt{\e}+\sqrt{\delta}) |R'_{ext}|+|R'_{ext}\setminus \text{Dom}(\Pi^*_s)|\\
\end{split}
\end{equation}
We now bound $|R'_{ext}\setminus \text{Dom}(\Pi^*_s)|$. By Lemma~\dualref{lm:tau-of-rect}, {\bf (2)} we have 
\begin{equation*}
\begin{split}
|R'_{ext}\setminus \text{Dom}(\Pi^*_k)|&\leq |R'\setminus \text{Dom}(\Pi^*_s)|\\
&\leq 8\sqrt{\e} |R'|\\
&\leq 8\sqrt{\e} \Delta^{-2K^2}|R'_{ext}|,\\
\end{split}
\end{equation*}
where we used the fact that $|R'_{ext}|\geq \Delta^{2K^2}\cdot |R'|$ by Lemma~\dualref{lm:rect-size}, {\bf (1)}. Substituting this into~\dualeqref{eq:ihgSGF139h9HdxXV-iwhrg}, we get
\begin{equation}\duallabel{eq:9h30ytar9hglkfS32gf-iwhrg-lb}
\begin{split}
|R_{ext}\setminus \Pi^*_s(\Int_\delta(R'_{ext}))|&\leq  (3\sqrt{\e}+\sqrt{\delta}+8\sqrt{\e} \Delta^{-2K^2}) |R'_{ext}|\leq  2\sqrt{\delta} |R'_{ext}|,
\end{split}
\end{equation}
where we used the fact that 
\begin{equation*}
\begin{split}
8\sqrt{\e} \Delta^{-2K^2}&\leq 8\delta \Delta^{-2K^2}\text{~~~~~~~~~(by~\dualref{p6})}\\
&\leq 8\delta \Delta^{-2K^2}\text{~~~~~~~~~(by~\dualref{p5})}\\
&\leq 8\delta^{98/100}\text{~~~~~~~~~~(by~\dualref{p5})},
\end{split}
\end{equation*}
and therefore, since $\sqrt{\e}\leq \delta$ by~\dualref{p6},
$$
3\sqrt{\e}+\sqrt{\delta}+8\sqrt{\e} \Delta^{-2K^2}\leq 3\delta+\sqrt{\delta}+8\delta^{98/100}\leq 2\sqrt{\delta}
$$
since $K$ is larger than an absolute constant and $\delta<\Delta^{100K^2}\leq K^{100K^2}$ by~\dualref{p3} together with~\dualref{p4}.

At the same time, we have by Lemma~\dualref{lm:rect-int-size}
\begin{equation}\duallabel{eq:9h30ytar9hglkfS32gf-iwhrg-ub}
\begin{split}
|\Pi^*_s(R'_{ext})\setminus R_{ext}|&\leq  |R'_{ext}\setminus \Int_\delta(R'_{ext})|\leq  \sqrt{\delta} |R'_{ext}|.\\
\end{split}
\end{equation}

We summarize these bounds in
\begin{equation}\duallabel{eq:9h30ytar9hglkfS32gfdfwg}
\begin{split}
|R_{ext}\setminus \Pi^*_s(R'_{ext})|&\leq  2\sqrt{\delta} |R'_{ext}|\\
\text{~~~and~}\\
|\Pi^*_s(R'_{ext})\setminus R_{ext}|&\leq  2\sqrt{\delta} |R'_{ext}|\\
\end{split}
\end{equation}
We also note that 
\begin{equation}\duallabel{eq:934hh34t9gHUHUDHFXX}
(1-3\sqrt{\e})|R'_{ext}|\leq R_{ext}\leq (1+3\sqrt{\e})|R'_{ext}|.
\end{equation}
Indeed, to obtain the bound above we apply Lemma~\dualref{lm:rect-size}, {\bf (1)} to $R_{ext}$ and $R'_{ext}$. This gives
\begin{equation}\duallabel{eq:24yt2yt-size-bounds}
\begin{split}
|R'_{ext}|/m^n&=(1\pm \sqrt{\e})\Delta^{|\I'|}\\
|R_{ext}|/m^n&=(1\pm \sqrt{\e})\Delta^{|\I|}.
\end{split}
\end{equation}
Taking the ratio of the two bounds above and using the fact that $|\I'|=|\I|$  yields~\dualeqref{eq:934hh34t9gHUHUDHFXX}, as required.

We now note that $\I_0=\Sp(\B)$ as per~\dualeqref{eq:i0-sm} and $\mathsf{M}(\a_0)$ is consistent with $T_*$ by definition of the map $\mathsf{M}$ (see~\dualeqref{eq:m-def} and~\dualeqref{eq:m-a-def}). Thus, by Lemma~\dualref{lm:rect-nu-j} applied to the rectangle $R_{ext}$ from~\dualeqref{eq:rext-def-iwbg903g} we get 
\begin{equation}\duallabel{eq:28h2f8hfdsfSDFF12sf}
(\ln 2-C/K)^{j-1} |R_{ext}|\leq |\nu_{\ell-1, j-1}(R_{ext})|\leq (\ln 2+C/K)^{j-1} |R_{ext}|.
\end{equation}
Using the first inequality above together with the first bound in~\dualeqref{eq:9h30ytar9hglkfS32gfdfwg}, we get
\begin{equation}\duallabel{eq:3y823UGFUG3t34-lb}
\begin{split}
|\nu_{\ell-1, j-1}(\Pi^*_s(R'_{ext}))|&\geq |\nu_{\ell-1, j-1}(R_{ext})|-|\nu_{\ell-1, j-1}(R_{ext}\setminus \Pi^*_s(\Int_\delta(R'_{ext})))|\\
&\geq |\nu_{\ell-1, j-1}(R_{ext})|-K^{j-1}\cdot 2\sqrt{\delta} \cdot |R'_{ext}|\\
&\geq (\ln 2-C/K)^{j-1}\cdot |R_{ext}|-\delta^{1/4} \cdot |R'_{ext}|\\
&\geq ((\ln 2-C/K)^{j-1}(1-3\sqrt{\e})-\delta^{1/4})\cdot |R'_{ext}|,\\
\end{split}
\end{equation}
where the transition from the first line to the second is because for every $\ell'$ the map $\tau^{\ell'}$ maps no vertex in $S'$ to more than $K/2$ vertices in $T_*$, and in particular $\nu_{\ell-1, j-1}$ maps no vertex in $S^{\ell-1}$ to more than $(K/2)^{j-1}$ vertices in $T_*^{\ell-j}$ (note that we are using the looser bound of $K^j$ on the product of these two bounds to simplify notation). The penultimate transition uses the fact that $K^{j-1}2\sqrt{\delta}\leq \delta^{1/4}$ by ~\dualref{p5}, and the transition to the last line uses~\dualeqref{eq:934hh34t9gHUHUDHFXX}.  Using the second inequality in~\dualeqref{eq:28h2f8hfdsfSDFF12sf} together with the second bound in~\dualeqref{eq:9h30ytar9hglkfS32gfdfwg}, we similarly get 
\begin{equation}\duallabel{eq:3y823UGFUG3t34-ub}
\begin{split}
|\nu_{\ell-1, j-1}(\Pi^*_s(R'_{ext}))|&\leq |\nu_{\ell-1, j-1}(R_{ext})|+|\nu_{\ell-1, j-1}(\Pi^*_s(R'_{ext})\setminus R_{ext})|\\
&\leq |\nu_{\ell-1, j-1}(R_{ext})|+K^{j-1}\cdot 2\sqrt{\delta} \cdot |R'_{ext}|,\\
&\leq (\ln 2+C/K)^{j-1}\cdot |R_{ext}|+\delta^{1/4} \cdot |R'_{ext}|\\
&\leq ((\ln 2+C/K)^{j-1}(1+3\sqrt{\e})+\delta^{1/4})\cdot |R'_{ext}|
\end{split}
\end{equation}
where the transition from the first line to the second is because for any $\ell'$ the map $\tau^{\ell'}$ maps no vertex in $S'$ to more than $K/2$ vertices in $T_*$, and in particular $\nu_{\ell-1, j-1}$ maps no vertex in $S^{\ell-1}$ to more than $(K/2)^{j-1}$ vertices in $T_*^{\ell-j}$, as well as the fact that $\Pi^*_s(R'_{ext})\setminus R_{ext}\subseteq \Pi^*_s(R'_{ext}\setminus \Int_\delta(R'_{ext}))$ by~\dualeqref{eq:HDOSJF}. The transition to the last line uses the fact that $K^{j-1}2\sqrt{\delta}\leq \delta^{1/4}$ by~\dualref{p5}.

Substituting~\dualeqref{eq:3y823UGFUG3t34-lb} and~\dualeqref{eq:3y823UGFUG3t34-ub} respectively into~\dualeqref{eq:0u23tggddxv}, we get

\begin{equation}\duallabel{eq:0u23tggddxv-size-lb}
\begin{split}
|\nu_{\ell-1, j-1}(\Pi^*_s(Q_{k, s}))|&=\sum_{\substack{\a=(\a_0, \a_1)\\ \a_0\in \mathbb{D}_{k, s}, \a_1\in \mathbb{D}^{ext}_{k, s}}} |\nu_{\ell-1, j-1}(\Pi^*_s(R'_{ext}(\a)))|\\
&\geq \sum_{\substack{\a=(\a_0, \a_1)\\ \a_0\in \mathbb{D}_{k, s}, \a_1\in \mathbb{D}^{ext}_{k, s}}} ((\ln 2-C/K)^{j-1}(1-3\sqrt{\e})-\delta^{1/4})\cdot |R'_{ext}(\a)|\\
&= ((\ln 2-C/K)^{j-1}-2\delta^{1/4}) \sum_{\substack{\a=(\a_0, \a_1)\\ \a_0\in \mathbb{D}_{k, s}, \a_1\in \mathbb{D}^{ext}_{k, s}}} |R'_{ext}(\a)|\\
&= ((\ln 2-C/K)^{j-1}-2\delta^{1/4}) |Q_{k, s}|.
\end{split}
\end{equation}

and 

\begin{equation}\duallabel{eq:0u23tggddxv-size-ub}
\begin{split}
|\nu_{\ell-1, j-1}(\Pi^*_s(Q_{k, s}))|&=\sum_{\substack{\a=(\a_0, \a_1)\\ \a_0\in \mathbb{D}_{k, s}, \a_1\in \mathbb{D}^{ext}_{k, s}}} |R'_{ext}(\a)|\\
&\leq \sum_{\substack{\a=(\a_0, \a_1)\\ \a_0\in \mathbb{D}_{k, s}, \a_1\in \mathbb{D}^{ext}_{k, s}}} ((\ln 2-C/K)^{j-1}(1+3\sqrt{\e})+\delta^{1/4})\cdot |R'_{ext}(\a)|\\
&= ((\ln 2+C/K)^{j-1}+2\delta^{1/4}) \sum_{\substack{\a=(\a_0, \a_1)\\ \a_0\in \mathbb{D}_{k, s}, \a_1\in \mathbb{D}^{ext}_{k, s}}} |R'_{ext}(\a)|\\
&= ((\ln 2+C/K)^{j-1}+2\delta^{1/4}) |Q_{k, s}|.
\end{split}
\end{equation}

\paragraph{Step 3.} Substituting~\dualeqref{eq:0u23tggddxv-size-lb} and~\dualeqref{eq:0u23tggddxv-size-ub}  into~\dualeqref{eq:3yt3yt2yDF-lb} and~\dualeqref{eq:3yt3yt2yDF-ub} respectively, and using~\dualeqref{eq:24yt2yt-size-bounds}, we get
\begin{equation*}
\begin{split}
|\nu_{\ell-1, j-1}(\tau(\downset_s(Q_k))|&\geq |\nu_{\ell-1, j-1}(\Pi^*_s(Q_{k, s}))|-\delta^{1/4} |Q_{k, s}|\\
&\geq ((\ln 2-C/K)^{j-1}-2\delta^{1/4}) |Q_{k, s}|\\
&\geq ((\ln 2-C/K)^{j-1}-2\delta^{1/4}-2\sqrt{\e}) \frac1{K-s} |Q_k|\\
&\geq ((\ln 2-C/K)^{j-1}-4\delta^{1/4}) \frac1{K-s} |Q_k|\\
\end{split}
\end{equation*}
and 
\begin{equation}\duallabel{eq:3yt3yt2yDF-ub}
\begin{split}
|\nu_{\ell-1, j-1}(\tau(\downset_s(Q_k))|&\leq |\nu_{\ell-1, j-1}(\Pi^*_s(Q_{k, s}))|+\delta^{1/4} |Q_{k, s}|\\
&\leq ((\ln 2+C/K)^{j-1}+2\delta^{1/4}) |Q_{k, s}|\\
&\leq ((\ln 2+C/K)^{j-1}+2\delta^{1/4}+2\sqrt{\e}) \frac1{K-s}|Q_k|\\
&\leq ((\ln 2+C/K)^{j-1}+4\delta^{1/4}) \frac1{K-s}|Q_k|.
\end{split}
\end{equation}

We now get by~\dualeqref{eq:9w9gh9gh9HFfojef3FS}
\begin{equation}\duallabel{eq:9h3438FHEFH-lb}
\begin{split}
|\nu_{\ell, j}(T_k\setminus T_{k+1})|&=|\nu_{\ell, j}(Q_k)|\\
&\geq \sum_{s=0}^{k} |\nu_{\ell-1, j-1}(\tau^\ell(\downset^\ell_s(Q_k)))|\\
&\geq \sum_{s=0}^{k} ((\ln 2-C/K)^{j-1}-4\delta^{1/4}) \frac1{K-s} |Q_k|\\
&\geq ((\ln 2-C/K)^{j-1}-4\delta^{1/4}) |Q_k| \left(\sum_{s=0}^{k}  \frac1{K-s}\right)
\end{split}
\end{equation}
and 
\begin{equation}\duallabel{eq:9h3438FHEFH-ub}
\begin{split}
|\nu_{\ell, j}(T_k\setminus T_{k+1})|&=|\nu_{\ell, j}(Q_k)|\\
&\leq \sum_{s=0}^{k} |\nu_{\ell-1, j-1}(\tau^\ell(\downset^\ell_s(Q_k)))|\\
&\leq \sum_{s=0}^{k} ((\ln 2+C/K)^{j-1}+4\delta^{1/4}) \frac1{K-s} |Q_k|\\
&\leq ((\ln 2+C/K)^{j-1}+4\delta^{1/4}) |Q_k| \left(\sum_{s=0}^{k}  \frac1{K-s}\right).
\end{split}
\end{equation}

We now recall that by~\dualeqref{eq:24yt2yt-size-bounds} one has $|Q_k|/m^n=(1\pm \sqrt{\e})\frac1{K}$. At the same time
\begin{equation*}
\begin{split}
\frac1{K}\sum_{k\in [K/2]} \sum_{s=0}^{k}  \frac1{K-s}&=\frac1{K}\sum_{k=0}^{K/2-1} \sum_{s=0}^{k}  \frac1{K-s}\\
&=\frac1{K}\sum_{s=0}^{K/2-1} \frac{K/2-s}{K-s}\\
&=\frac1{K}\sum_{s=0}^{K/2-1} \frac{-K/2+K-s}{K-s}\\
&=\frac1{K}\sum_{s=0}^{K/2-1} \left(1-\frac{-K/2}{K-s}\right)\\
&=\frac1{2}-\frac1{2}\sum_{s=0}^{K/2-1} \frac1{K-s},\\
\end{split}
\end{equation*}
and hence by Claim~\ref{cl:sum-int}
$$
\frac1{2}-\frac1{2}\ln 2\leq \frac1{K}\sum_{k\in [K/2]} \sum_{s=0}^{k}  \frac1{K-s}\leq \frac1{2}-\frac1{2}(\ln 2+1/K).
$$

Substituting these bounds into~\dualeqref{eq:9h3438FHEFH-lb} and~\dualeqref{eq:9h3438FHEFH-ub}, we get 
\begin{equation*}
\begin{split}
|\nu_{\ell, j}(T_k\setminus T_{k+1})|&\geq ((\ln 2-C/K)^{j-1}-4\delta^{1/4}) |Q_k| \left(\sum_{s=0}^{k}  \frac1{K-s}\right)\\
&\geq ((\ln 2-C/K)^{j-1}-4\delta^{1/4}) \left(\frac1{2}-\frac1{2}\ln 2\right)\cdot m^n\\
&\geq \frac1{2}(1-\ln 2)(\ln 2-C/K)^{j-1}\cdot m^n
\end{split}
\end{equation*}
and 
\begin{equation*}
\begin{split}
|\nu_{\ell, j}(T_k\setminus T_{k+1})|&\leq ((\ln 2+C/K)^{j-1}+4\delta^{1/4}) |Q_k| \left(\sum_{s=0}^{k}  \frac1{K-s}\right)\\
&\leq ((\ln 2+C/K)^{j-1}+4\delta^{1/4}) (1+3/K)\left(\frac1{2}-\frac1{2}\ln 2\right)\cdot m^n\\
&\leq \frac1{2}(1-\ln 2)(\ln 2+C/K)^{j-1}\cdot m^n
\end{split}
\end{equation*}
This establishes~\dualeqref{eq:923yth9ihFHUFHHF} and completes the proof of the lemma.
\end{proofof}

\subsection{Key lemma: insensitivity of $\nu$ and $\mu$ to bounded near orthogonal shifts}

\paragraph{Outlier vertices.} We define sets of outlier vertices recursively for every $\ell\in [L]$.  First define $\out^L=\emptyset$ for convenience. Then for every $\ell\in [L]$ we define $\out^\ell$ in terms of $\out^{\ell'}, \ell'>\ell$ as follows.  First let
\begin{equation}\duallabel{eq:def-sets}
\I^\ell=\J^\ell\cup  \{\r^\ell\}\text{~and~}\I_k^\ell=\J^\ell_{<k}\cup \Ext^\ell_k\cup \{\q^\ell_k\}
\end{equation}
for simplicity of notation. Intuitively, the outlier vertices are simply vertices that are too close to boundaries of cubes in coordinates $\I^\ell$ and $\I_k^\ell$ for all $k\in [K/2]$, or vertices in $T^\ell$ that are not in the range of $\tau^{\ell+1}$. It is convenient to define, for a  $\H\subseteq \F$ and $\d \in (\deltagridInt)^\H$, the boundary of a cube as
$$
\partial \rect(\H, \d):=\rect(\H, \d)\setminus \Int_\delta(\rect(\H, \d)).
$$

First let, denoting $\I=\I^{L-1}$ and $\I_k=\I_k^{L-1}$ for convenience, 
\begin{equation}\duallabel{eq:o-ellm1-def}
\begin{split}
\out^{L-1}:=\left(\bigcup_{\d\in (\deltagridInt)^\I} \partial \rect(\I, \d)\right)\cup \left(\bigcup_{k\in [K/2]} \bigcup_{\d\in (\deltagridInt)^{\I_k}} \partial \rect(\I_k, \d)\right),
\end{split}
\end{equation}
and then for every $\ell\in [L], \ell<L-1,$ let, denoting $\I=\I^\ell$ and $\I_k=\I_k^\ell$ for convenience, 
\begin{equation}\duallabel{eq:o-def}
\begin{split}
\out^\ell:=\nu_{\ell+1,1}(\out^{\ell+1})\cup (T_*^\ell\setminus \tau^{\ell+1}(S^{\ell+1}))&\cup \left(\bigcup_{\d\in (\deltagridInt)^\I} \partial \rect(\I, \d)\right)\\
&\cup \left(\bigcup_{k\in [K/2]} \bigcup_{\d\in (\deltagridInt)^{\I_k}} \partial \rect(\I_k, \d)\right).
\end{split}
\end{equation}
Finally, let
\begin{equation}\duallabel{eq:o-union-def}
\out=\bigcup_{\ell\in [L]} \out^\ell.
\end{equation}
\begin{remark}
Abusing notation somewhat, we will think of the set $\out$ as the set of labels in $[m]^n$, and in particular will write $x\in \out$ for a vertex $x\in T^\ell$, as well as sometimes write $y\in \out$ for a vertex $y\in S^\ell$ for some $\ell\in [L]$.
\end{remark}
This set of outlier vertices is quite small, as the following claim shows:
\begin{claim}\duallabel{cl:o-size}
$|\out|\leq \delta^{1/8} |T^0|$.
\end{claim}
\begin{proof}
By Lemma~\dualref{lm:t-star-minus-tau-sp} we have for every $\ell\in [L], \ell<L-1,$ that $\left|T_*^\ell\setminus \tau^{\ell+1}(S^{\ell+1}))\right|\leq \delta^{1/4} \left|T^\ell\right|$, and therefore the total contribution of $T_*^\ell\setminus \tau^{\ell+1}(S^{\ell+1})$ due to recursive application of $\nu_{\ell+1, 1}$ in the first line of~\dualeqref{eq:o-def} contributes a set of size at most  $\sum_{i=0}^\ell K^i \delta^{1/4} \left|T^\ell\right|$, since $\nu_{\ell', 1}$ maps every point to at most $K$ points, for every $\ell'\in [L]$.

Now note that for every $\H\subset \F$ with $|\H|\leq K^2$ one has using Lemma~\dualref{lm:rect-int-size}
\begin{equation}\duallabel{eq:38gg3894g8g89gf89ahg}
\begin{split}
\left|\bigcup_{\d\in (\deltagridInt)^\H} \partial \rect(\H, \d)\cap T^\ell\right| &=\sum_{\d\in (\deltagridInt)^\H} |\partial\rect(\H, \d)\cap T^\ell|\\
&\leq \sum_{\d\in (\deltagridInt)^\H} \sqrt{\delta} |\rect(\H, \d)\cap T^\ell|\text{~~~~~~~(by Lemma~\dualref{lm:rect-int-size})}\\
&= \sqrt{\delta} \sum_{\d\in (\deltagridInt)^\H}  |\rect(\H, \d)\cap T^\ell|\\
&= \sqrt{\delta} |T^0|.
\end{split}
\end{equation}
Applying this to~\dualeqref{eq:o-def}, and using the fact that for every $\ell$ and $j$ the map $\nu_{\ell', 1}$ does not map any point to more than $K$ points, we get that the contribution of the set in~\dualeqref{eq:38gg3894g8g89gf89ahg} after all recursive applications of $\nu_{\ell, 1}$ in the first line of~\dualeqref{eq:o-def} contributes at most $\sum_{i=0}^\ell K^i \sqrt{\delta} \left|T^\ell\right|$. Summing these contributions over $\H=\I$ and $\H=\I_k, k\in [K/2]$, we get for every $\ell\in [L]$
\begin{equation}
|\out^\ell|\leq 2K\cdot L\cdot K^L\cdot \delta^{1/4} |T^0|
\end{equation}
and $|\out|\leq \sum_{\ell\in [L]} |\out^\ell|\leq 2K\cdot L^2\cdot K^L\cdot \delta^{1/4} |T^0|$.
Finally, it remains to note that 
\begin{equation*}
\begin{split}
2K\cdot L^2\cdot K^L\cdot \delta^{1/4}&\leq (K^{2K}\cdot \delta^{1/8})\cdot \delta^{1/8}\text{~~~~~~~~~~~~~~~~~~~~~(since $L\leq K$ by~\dualref{p4})}\\
&\leq (K^{2K}\cdot K^{-(1/8)100K^2})\cdot \delta^{1/8}\text{~~~~~(since $\delta\leq K^{-100K^2}$ by~\dualref{p3} and~\dualref{p5})}\\
&\leq \delta^{1/8},
\end{split}
\end{equation*}
where the last transition uses the fact that $K$ is larger than an absolute constant.
\end{proof}

\begin{lemma}\duallabel{lm:traceback}
For every $\ell\in [L]$, every $x\in T^\ell\setminus \out^\ell$ (where $\out^\ell$ is as in~\dualeqref{eq:o-def}) such that $x\in \nu_{\ell+z, z}(T^{\ell+z}\setminus T_*^{\ell+z})$ for some $z\in \{0, 1, \ldots, L-1-\ell\}$ the following conditions hold for every $y\in T^\ell\setminus \out^\ell$ satisfying $y=x+\lambda \cdot \u$ for some $\u\in \B^\ell\setminus \Spt(\B^\ell), |\lambda|\leq 2M/w$. 

For every $j=0,\ldots, z$ there exists $k\in [K/2]$ such that
\begin{description}
\item[(1)] there exist unique
 $\wt{x}, \wt{y}\in T_k^{\ell+j}$ such that  
\begin{equation*}
\begin{split}
x\in \nu_{\ell+j, j}(\wt{x})\text{~~and~~}y\in \nu_{\ell+j, j}(\wt{y})
\end{split}
\end{equation*}
\item[(2)] there exists a collection of vectors

\begin{equation}\duallabel{eq:ilj-def}
\I^{\ell+j}\subset \left(\J^{\ell+j}_{<k}\cup \Ext_k^{\ell+j}\cup \{\q_k^{\ell+j}\}\right) \cup \bigcup_{s=0}^{j-1} \Spt(\B^{\ell+s}).
\end{equation}

with $|\I^{\ell+j}|=j(K+1)$ together with integer coefficients $t_\i, \i\in \I^{\ell+j},$ satisfying $|t_\i|\leq 20M/w$ for $\i\in \B^{<\ell+j}$ and $|t_\i|\leq 10M/w$ for $\i\in \B^{\ell+j}$ such that
\begin{equation}\duallabel{eq:02ht9g84tgtFE3r}
\wt{y}=\wt{x}+\lambda \cdot \u+\sum_{\i\in \I^{\ell+j}} t_\i\cdot \i.
\end{equation}
\item[(3)] for $\I=\J^{\ell+j}\cup \{\r^{\ell+j}\}$ there exists a rectangle $R=\rect(\I, \d)$, $\d\in (\deltagridInt)^\I,$ such that 
$$
\wt{x}\in \Int_\delta(R)\text{~~and~~}\wt{y}\in \Int_\delta(R).
$$
\end{description}
\end{lemma}
\begin{proof}
We first note that the choice of $z$ is unique by Lemma~\dualref{lm:nu-prop}, {\bf (2)}. We establish properties {\bf (1)}, {\bf (2)} and {\bf (3)} above by induction on $j\in \{0, 1,\ldots, z\}$. 

\noindent{\bf Base: $j=0$.} Note that $\nu_{\ell, 0}$ is the identity map, so we can take $\wt{x}=x, \wt{y}=y$. Property {\bf (1)} follows by construction. Property {\bf (2)} follows taking $\I^\ell=\emptyset$. Property {\bf (3}) follows since $y\in T^\ell\setminus \out^\ell$ by assumption and $\out^\ell$ includes all points that are too close to boundaries of rectangles in $\J^\ell\cup \{\r^\ell\}$ by construction (see~\dualeqref{eq:o-def}).

\noindent{\bf Inductive step: $j-1\to j$.}  We write $\J=\J^{\ell+j-1}$ and write  $\r=\r^{\ell+j-1}$ to denote the $(\ell+j-1)$-th compression index. We write $\J'=\J^{\ell+j}$, and write $\Ext_k=\Ext_k^{\ell+j}$ and $\q_k=\q_k^{\ell+j}$ to denote the extension and compression indices of phase $k$ at stage $\ell+j$. We let $\r'=\r^{\ell+j}$ denote the $(\ell+j)$-th compression index.  We define
\begin{equation}\duallabel{eq:09823y8g-i-def}
\I:=\J\cup \{\r\}
\end{equation}
to simplify notation. By the inductive hypothesis for $x, y\in T^\ell\setminus \out^\ell$ there exist $u, v\in T^{\ell+j-1}$ such that 
\begin{equation*}
\begin{split}
x&\in \nu_{\ell+j-1, j-1}(u)\\
&\text{~~~~~~and}\\
y&\in \nu_{\ell+j-1, j-1}(v)
\end{split}
\end{equation*}
together with $\d'\in (\deltagridInt)^{\I}$ such that 
\begin{equation}\duallabel{eq:tilde-rect-398rg834}
u\in \Int_\delta(\rect(\I, \d'))\text{~~and~~}v\in \Int_\delta(\rect(\I, \d')),
\end{equation}
and 
\begin{equation}\duallabel{eq:uv-inf-293y5y}
v=u+\sum_{\i\in \I^{\ell+j-1}} t_\i\cdot \i
\end{equation}
for integer coefficients $t$ with $|t_\i|\leq 20M/w$ for $\i\in \B^{<\ell+j-1}$ and $|t_\i|\leq 10M/w$ for $\i\in \B^{\ell+j-1}$. 

We assume that $j-1<z$, as otherwise there is nothing to prove.  Otherwise, since $j-1<z$, we have $\rect(\I, \d)\subset T_*^{\ell+j-1}$, as $\rect(\I, \d)$ cannot intersect both $T_*^{\ell+j-1}$ and $T^{\ell+j-1}\setminus T_*^{\ell+j-1}$ and the choice of $z$ is unique (by Lemma~\dualref{lm:nu-prop}, {\bf (2)}, as noted above). Since
$$
\out^\ell\supset \nu_{\ell+j-1, j-1}(T_*^{\ell+j-1}\setminus \tau^{\ell+j}(S^{\ell+j})),
$$
there exist $x', y'\in S^{\ell+j}$ such that $\tau^{\ell+j}(x')=u$, $\tau^{\ell+j}(y')=v$. By~\dualeqref{eq:tilde-rect-398rg834} together with Lemma~\dualref{lm:same-rect}, {\bf (1)}  and {\bf (2)}, there exists $k\in [K/2]$ such that $x', y'\in S^{\ell+j}_k$ and $\d'\in \I'$ such that 
\begin{equation}\duallabel{eq:we8g8g2asf}
\rho_k(x')\in \rect(\I', \d')\text{~~and~~}\rho_k(y')\in \rect(\I', \d').
\end{equation}
We will use the above fact shortly.

By Lemma~\dualref{lm:tau-sparse} there exist integer coefficients $t^x_\i, t^y_\i$ such that 
$$
x'=u+\sum_{\i\in \I'\cup \I} t^x_\i\cdot \i \text{~~~~with~~~} \|t^x\|_\infty\leq 5M/w
$$
and 
$$
y'=v+\sum_{\i\in \I'\cup \I} t^y_\i\cdot \i\text{~~~~with~~~} \|t^y\|_\infty\leq 5M/w,
$$
where we define 
\begin{equation}\duallabel{eq:09823y8g-ip-def}
\I'=\J'_{<k}\cup \Ext_k\cup \{\q_k\}
\end{equation}
to simplify notation. Combining this with~\dualeqref{eq:uv-inf-293y5y}, we get
\begin{equation}\duallabel{eq:029h4t9h490gh9hDNJDBF}
y'=x'+\sum_{\i\in \I^{\ell+j}} s_\i \cdot \i,
\end{equation}
where we let $\I^{\ell+j}:=\I^{\ell+j-1}\cup \I'\cup \I$ and let 
$$
s_\i=t_\i+t^y_\i-t^x_\i,
$$
extending $t_\i$ to be zero for $\i\not \in \I^{\ell+j-1}$, $t^x_\i$ to be zero for $\i \not \in \I'\cup \I$ and $t^y_\i$ to be zero for $\i\not \in \I'\cup \I$. Note that
\begin{equation*}
\begin{split}
\I^{\ell+j}&=\I^{\ell+j-1}\cup (\I'\cup \I)\\
&\subset \left(\left(\J^{\ell+j-1}_{<k}\cup \Ext_k^{\ell+j-1}\cup \{\q_k^{\ell+j-1}\}\right) \cup \bigcup_{s=0}^{j-2} \Spt(\B^{\ell+s})\right)\cup (\I'\cup \I)\text{~~~~(by the inductive hypothesis)}\\ 
&\subset \left(\J^{\ell+j}_{<k}\cup \Ext_k^{\ell+j}\cup \{\q_k^{\ell+j}\}\right) \cup \bigcup_{s=0}^{j-1} \Spt(\B^{\ell+s}).
\end{split}
\end{equation*}
The last transition uses the fact that 
$$
\I'=\J'_{<k}\cup \Ext_k\cup \{\q_k\}=\J^{\ell+j}_{<k}\cup \Ext_k^{\ell+j}\cup \{\q_k^{\ell+j}\}
$$ 
by~\dualeqref{eq:09823y8g-ip-def} and 
$$
\I=\J\cup \{\r\} \subset \Spt(\B^{\ell+j-1})
$$ 
by~\dualeqref{eq:09823y8g-i-def}.

We now upper bound the magnitude of the coefficients $s_\i$.  First, for $\I^{\ell+j}\setminus \B^{<\ell+j}\subseteq \J^{\ell+j}_{<k}\cup \Ext^{\ell+j}_k\cup \{\q_k^{\ell+j}\}$ one has 
$$
|s_\i|=|t^y_\i-t^x_\i|\leq |t^y_\i|+|t^x_\i|\leq 5M/w+5M/w\leq 10M/w
$$
as required. Now consider $\i\in \I^{\ell+j}\cap \B^{<\ell+j}$. First note that if $\i\in \B^{<\ell+j-1}$ then one has $t^x_\i=t^y_\i=0$ since $\B^{<\ell+j-1}\cap (\I\cup \I')=\emptyset$, and thus one has $|s_\i|=|t_\i|\leq 20M/w$ by the inductive hypothesis. Now consider 
$$
\i\in \I^{\ell+j}\cap \B^{\ell+j-1}
$$
In that case one has $|t_\i|\leq  10M/w$ by the inductive hypothesis, so 
$$
|s_\i|=|t_\i+t^y_\i-t^x_\i|\leq |t_\i|+|t^y_\i|+|t^x_\i|\leq 10M/w+5M/w+5M/w=20M/w
$$
as required, establishing properties {\bf (1)} and {\bf (2)}. We now turn to property {\bf (3)}. Define
\begin{equation}\duallabel{eq:09823y8g-it-def}
\wt{\I}=\J'\cup \{\r'\}
\end{equation}
to simplify notation, and  let $\a, \b\in (\deltagridInt)^{\wt{\I}}$ be such that $x'\in \rect(\wt{\I}, \a)$ and $y'\in \rect(\wt{\I}, \b)$, i.e. 
 \begin{equation}\duallabel{eq:ab-293y52y23FSDF}
\langle x', \i\rangle \pmod{M}\in [\a_\i, \a_\i+\Delta)\cdot M\text{~~and~~}\langle y', \i\rangle \pmod{M}\in [\b_\i, \b_\i+\Delta)\cdot M.
\end{equation}

We now show that in fact $\a=\b$. We start by noting that
\begin{equation}
\begin{split}
\wt{\I}\cap \I^{\ell+j}&=(\J'\cup \{\r'\})\cap \I^{\ell+j}\\
&\subseteq (\J'\cup \{\r'\})\cap (\B^{<\ell+j}\cup \J'_{<k}\cup \Ext_k\cup \{\q_k\})\\
&\subseteq  \J'_{<k}\\
\end{split}
\end{equation}
By~\dualeqref{eq:we8g8g2asf} we have for every $\i\in \wt{\I}\cap \I^{\ell+j}\subseteq \J'$
\begin{equation}\duallabel{eq:intersection-0293yty23t-0}
\langle \rho_k(x'), \i\rangle \pmod{M}\in [\d'_\i, \d'_\i+\Delta)\cdot M\text{~~and~~}\langle \rho_k(y'), \i\rangle \pmod{M}\in [\d'_\i, \d'_\i+\Delta)\cdot M.
\end{equation}
At the same time by Lemma~\dualref{lm:densification}, {\bf (3)},
$$
\rho_k(x')=x'+\lambda_x \cdot \q_k\text{~and~}\rho_k(y')=y'+\lambda_y \cdot \q_k
$$
for some integers $\lambda_x, \lambda_y$ bounded by $M/w$ in absolute value, which implies, since for every $\i\in \wt{\I}\cap \I^{\ell+j}\subseteq \J'$ one has $\langle \q_k, \i\rangle<\e \cdot w$, that for every such $\i$
\begin{equation}\duallabel{eq:intersection-0293yty23t-1}
|\langle \rho_k(x'), \i\rangle-\langle x', \i\rangle|<\e M\text{~~and~~}|\langle \rho_k(y'), \i\rangle - \langle y', \i\rangle|<\e M.
\end{equation}
Finally, since $x, y\in T^\ell\setminus \out^\ell$  by assumption, we have  $x'\not \in \partial \rect(\I', \c)$ and $y'\not \in \partial \rect(\I', \c)$  for any $\c$, ~\dualeqref{eq:intersection-0293yty23t-0} and ~\dualeqref{eq:intersection-0293yty23t-1} above, together with the fact that $\e<\delta$ by ~\dualref{p6}, imply that for every $\i\in \wt{\I}\cap \I^{\ell+j}\subseteq \J'$
\begin{equation*}
\langle x', \i\rangle \pmod{M}\in [\d'_\i, \d'_\i+\Delta)\cdot M\text{~~and~~}\langle y', \i\rangle \pmod{M}\in [\d'_\i, \d'_\i+\Delta)\cdot M.
\end{equation*}
Thus, for $\a$ and $\b$ from~\dualeqref{eq:ab-293y52y23FSDF} we have $\a_\i=\b_\i$ for all $\i\in \wt{\I}\cap \I^{\ell+j}$. At the same time for every $\i\in \wt{\I}\setminus  \I^{\ell+j}$ one has by~\dualeqref{eq:029h4t9h490gh9hDNJDBF}
\begin{equation}\duallabel{eq:ig83gr8gwf3wrdSD}
\begin{split}
\left|\langle x', \i\rangle -\langle y', \i\rangle\right|&=\left|\sum_{\h\in \I^{\ell+j}} s_\h \cdot \langle \h, \i\rangle\right|\\
&\leq |\I^{\ell+j}| \cdot (20M/w)\cdot \max_{\h\in \I^{\ell+j}} \langle \h, \i\rangle\\
&\leq |\I^{\ell+j}| \cdot (20M/w)\cdot \e\cdot w\\
&\leq 40K^3 \e\cdot M\\
&< \delta \cdot M.
\end{split}
\end{equation}
The second transition above uses the fact that $\|s\|_\infty\leq 20M/w$, The forth transition uses the fact that 
$$
\left|\I^{\ell+j}\right|\leq 2K^2\cdot L\leq 2K^3,
$$
as 
$$
\I^{\ell+j}\subseteq \bigcup_{\ell\in [L]} \left(\J^\ell\cup \{\r^\ell\}\cup \bigcup_{k\in [K]} (\Ext^\ell_k\cup \{\q^\ell_k\})\right).
$$
The fifth transition in~\dualeqref{eq:ig83gr8gwf3wrdSD} uses the fact 
\begin{equation*}
\begin{split}
40K^3 \e&\leq (40 K^3 \delta)\cdot \delta\text{~~~~~~~~~~~~~~~~~~~~~~(by~\dualref{p6})}\\
&\leq (40 K^3 \cdot K^{-100K^2})\cdot \delta\text{~~~~~(by~\dualref{p3} and~\dualref{p5})}\\
&<\delta
\end{split}
\end{equation*}
since $K$ is larger than a constant.

Now recall that $x, y\in T^\ell\setminus \out^\ell$, and in particular by~\dualeqref{eq:o-def}
$$
x, y\not \in \bigcup_{\d\in (\deltagrid)^{\wt{\I}}} \partial \rect(\wt{\I}, \d).
$$
Combining this with~\dualeqref{eq:ab-293y52y23FSDF} and~\dualeqref{eq:ig83gr8gwf3wrdSD} yields $\a_\i=\b_\i$, as required. Thus, we get
$$
x', y'\in \rect(\wt{\I}, \a),
$$
which establishes property {\bf (3)} and completes the proof of the inductive step.
\end{proof}

\begin{lemma}\duallabel{lm:cut-structure}
For every $\ell \in [L]$, every $x\in T^\ell\setminus \out^\ell$, $y\in T^\ell\setminus \out^\ell$ (where $\out^\ell$ is defined in~\dualeqref{eq:o-def}) the following conditions hold. If 
$$
x\in \nu_{\ell+j, j}(T^{\ell+j}\setminus \Ext_\delta(T_*^{\ell+j})),
$$ 
and $y=x+\lambda \cdot\u, |\lambda|\leq 2M/w,$ for some $\u\in \B^\ell\setminus \Spt(\B^\ell)$, then 
$$
y\in \nu_{\ell+j, j}(T^{\ell+j}\setminus T_*^{\ell+j}).
$$
\end{lemma}
\begin{proof} 
We invoke Lemma~\dualref{lm:traceback} with $z=j$, and let $\wt{x}$ and $\wt{y}$ denote the resulting points (we use $z=j$ in what follows), and let $k\in [K/2]$ be such that  
\begin{equation}\duallabel{eq:x-tilde-239y8ggfdfugxufGF}
\wt{x}\in (T^{\ell+z}\setminus \Ext_\delta(T^{\ell+z}_*))\cap T_k^{\ell+z}=T_k^{\ell+z}\setminus \Ext_\delta(T^{\ell+z}_*).
\end{equation}
Also recall that 
\begin{equation*}
T^{\ell+z}\setminus \Ext_\delta(T_*^{\ell+z})=\left\{ y\in [m]^n: \langle y, \j_s\rangle \pmod{M} \in \left(1-\frac{1}{K-s}+\delta, 1-\delta\right]\cdot M \text{~for~some~}s\in [K/2+1]\right\},
\end{equation*}
where we let $\J:=\J^{\ell+z}$ to simplify notation.  Since 
\begin{equation*}
T_k^{\ell+z}=\left\{ y\in [m]^n: \langle y, \j_s\rangle \pmod{M} \in \left[0, 1-\frac{1}{K-s}\right)\cdot M \text{~for~all~}s\in [k]\right\},
\end{equation*}
we have 
\begin{equation*}
\begin{split}
T_k^{\ell+z}\setminus \Ext_\delta(T_*^{\ell+z})&=\left\{ y\in T_k^{\ell+z}: \langle y, \j_s\rangle \pmod{M} \in \left(1-\frac{1}{K-s}+\delta, 1-\delta\right]\cdot M \right.\\
&\left.\text{~~~~~~~~~~~~~~~~~~~~~~~~~for~some~}s\in \{k, k+1,\ldots, K/2\}\right\}.
\end{split}
\end{equation*}
Since $\wt{x}\in T_k^{\ell+z}\setminus \Ext_\delta(T^{\ell+z}_*)$ by~\dualeqref{eq:x-tilde-239y8ggfdfugxufGF}, there exists $s\in \{k, k+1,\ldots, K/2\}$ such that 
$$
\langle \wt{x}, \j_s\rangle \pmod{M} \in \left(1-\frac{1}{K-s}+\delta, 1-\delta\right]\cdot M.
$$

At the same time using~\dualeqref{eq:ilj-def} and~\dualeqref{eq:02ht9g84tgtFE3r} we have for every $\j\in \J^{\ell+z}_{\geq k}$ and in particular for $\j=\j_s$
\begin{equation}\duallabel{eq:yprime-092yh99weggGEF}
\begin{split}
\left|\langle \wt{y}, \j\rangle-\langle \wt{x}, \j\rangle\right|&=\left|\lambda\cdot \langle \u, \j\rangle+\sum_{\i\in \I^{\ell+z}} t_\i\cdot \langle \i, \j\rangle\right|\\
&\leq \lambda\cdot \e \cdot w+\sum_{\i\in \I^{\ell+z}} t_\i\cdot \e \cdot w\\
&\leq \e \cdot M+ \e \|t\|_\infty \cdot |\I| \cdot w\\
&\leq \e (1+L\cdot K) \cdot M\\
&< \delta M.\\
\end{split}
\end{equation}
In the derivation above we first used the fact that $\u\neq \j$ since $\j\in \J^{\ell+z}\subset \Spt(\B^{\ell+z})$ and $\u\in \B^\ell\setminus \Spt(\B^\ell)$ (note that the two sets are disjoint regardless of the value of $z$). We also used the fact that $\j\not \in \I^{\ell+z}$, since  $\j\in \J^{\ell+z}_{\geq k}$ (note the crucial subindex $\geq k$) and by Lemma~\dualref{lm:traceback}, {\bf (2)} one has 
\begin{equation*}
\begin{split}
\I^{\ell+z}\cap \J^{\ell+z}_{\geq k} &\subset \left(\left(\J^{\ell+z}_{<k}\cup \Ext_k^{\ell+z}\cup \{\q_k^{\ell+z}\}\right) \cup \bigcup_{s=0}^{z-1} \Spt(\B^{\ell+s})\right)\cap \J^{\ell+z}_{\geq k}\\
&\subseteq \left(\J^{\ell+z}_{<k}\cup \Ext_k^{\ell+z}\cup \{\q_k^{\ell+z}\}\right)\cap \J^{\ell+z}_{\geq k}\\
&=\emptyset.\\
\end{split}
\end{equation*}

We then
used the fact that 
\begin{equation*}
\begin{split}
\e (1+L\cdot K)&\leq \e (1+K^2)\text{~~~~~~~~~~~~~~~~~~~~(by~\dualref{p4})}\\
&\leq \delta^2 (1+K^2)\text{~~~~~~~~~~~~~~~~~~(by~\dualref{p6})}\\
&\leq \delta \Delta^{100K^2}(1+K^2)\text{~~~~~~~(by~\dualref{p5})}\\
&\leq \delta K^{-100K^2}(1+K^2)\text{~~~~(by~\dualref{p3})}\\
&\leq \delta,
\end{split}
\end{equation*}
where the last transition is due to the fact that $K$ is larger than a constant. Then we have by~\dualeqref{eq:yprime-092yh99weggGEF} that
$$
\langle \wt{y}, \j_s\rangle \pmod{M} \in \left(1-\frac{1}{K-s}, 1\right]\cdot M.
$$
Thus, since $\wt{y}\in T_k^{\ell+z}$ by assumption, we have $\wt{y}\in T_k^{\ell+z}\setminus T_*^{\ell+z}\subseteq T^{\ell+z}\setminus T_*^{\ell+z}$ and therefore 
\begin{equation}\duallabel{eq:23ytttGYGFFkjft6246}
y=\nu_{\ell+j, j}(\wt{y})\in \nu_{\ell+j, j}(T^{\ell+z}\setminus T^{\ell+z}_*),
\end{equation}
as required.  \footnote{We note that the stronger implication that $\wt{y}\in T_k^{\ell+z}\setminus T_*^{\ell+z}$ as opposed to just $\wt{y}\in T^{\ell+z}\setminus T_*^{\ell+z}$ is not needed to conclude~\dualeqref{eq:23ytttGYGFFkjft6246}.}
\end{proof}

\begin{corollary}\duallabel{cor:cut-structure}
For every $\ell \in [L]$, every $k\in [K/2]$, $x\in T^\ell_k\setminus \out^\ell$, $y\in S^\ell_k\setminus \out^\ell$ (where $\out^\ell$ is defined in~\dualeqref{eq:o-def}) the following conditions hold. If 
$$
x\in \nu_{\ell+j, j}(T^{\ell+j}\setminus \Ext_\delta(T_*^{\ell+j})),
$$ 
and $y=x+\lambda \cdot\u, |\lambda|\leq 2M/w,$ for some $\u\in \B^\ell\setminus \Spt(\B^\ell)$, then 
$$
y\in \mu_{\ell+j, j}(T^{\ell+j}\setminus T_*^{\ell+j}).
$$
\end{corollary}
\begin{proof}
Let $y'\in T^\ell_k$ be such that $y'\delequal y$ -- such a $y'$ exists by definition of $S_k$ (see~\dualeqref{eq:def-sk}).Then by Lemma~\dualref{lm:cut-structure} we have $y'\in \nu_{\ell+j, j}(T^{\ell+j}\setminus T_*^{\ell+j})$. Since
$$
\mu_{\ell+j, j}(T^{\ell+j}\setminus T_*^{\ell+j})=\downset^\ell(\nu_{\ell+j, j}(T^{\ell+j}\setminus T_*^{\ell+j})) 
$$
by Definition~\dualref{def:nu}, we get, again using the fact that $y'\delequal y$, that $y\in \mu_{\ell+j, j}(T^{\ell+j}\setminus T_*^{\ell+j})$, as required.
\end{proof}

\section{Proof of main theorem (Theorem~\ref{thm:main})}\duallabel{sec:main-theorem}
We prove the main theorem (Theorem~\ref{thm:main}) in this section. First we define a hard distribution $\mathcal{D}$ on input graphs $\wh{G}$ in Section~\dualref{sec:distribution}. We then prove a lower bound on the size of the maximum matching in $\wh{G}$ and design a good upper bound on the size of the matching constructed by a small space algorithm in Section~\dualref{sec:malg}. Finally, we prove the main theorem in Section~\dualref{sec:main-proof}.

\subsection{Input distribution on graphs}\duallabel{sec:distribution}
We now define the hard distribution $\mathcal{D}$ on input graphs.  
First for $\ell\in [L]$ and $k\in [K/2]$ select the compression vector $\q^\ell_k$ arbitrarily from $\B^\ell_k$ and select the extension indices $\Ext^\ell_k$ arbitrarily from $\B^\ell_k$. Recall that for $k\in [K/2]$ we define
$$
\ac{\B}^\ell_k=\B^\ell_k\setminus (\{\q^\ell_k\}\cup \Ext^\ell_k)
$$
and let
$$
\ac{\B}^\ell_{K/2}=\B^\ell_{K/2}\setminus \{\r^\ell\}.
$$

\paragraph{Input distribution $\mathcal{D}$.} For every $\ell\in [L]$ and every $k\in [K/2+1]$ sample
$$
\J^\ell\sim \text{UNIF}\left(\ac{\B}^\ell_0\times \ac{\B}^\ell_1 \times \ldots\times \ac{\B}^\ell_{K/2}\right),
$$
i.e. for each $k\in [K/2+1]$ sample $\j^\ell_k$ independently and uniformly at random from $\ac{\B}^\ell_k$. Let $G^\ell=(S^\ell, T^\ell, E^\ell)$ be basic gadget graphs as defined in Section~\dualref{sec:g-ell}, and for every $\ell\in [L], \ell>0$, let 
$$
\tau^\ell: S^\ell\to T_*^{\ell-1}
$$ 
be the $\ell$-th glueing map as defined in Section~\dualref{sec:tau-ell}.

\paragraph{Subsamplings $\wt{G}^\ell$ of individual gadgets $G^\ell$.} We now fix $\ell\in [L]$ and write $S=S^\ell$ and $T=T^\ell$ to simplify notation. For every $k\in [K/2]$, $\j\in \B^\ell_k$ and $y\in S_k$ let
\begin{equation}\duallabel{eq:x-ell-def}
X^\ell_{k, \j}(y)=\text{Bernoulli}(1-1/K)
\end{equation}
denote independent Bernoulli random variables conditioned on $\sum_{y\in S_k} X^\ell_{k, \j}(y)=\lceil (1-\frac1{K})|S_k|\rceil$ for all $k$ and $\j$.  We use these variables to sample edges of the graphs $G^\ell$ as follows. Define
\begin{equation*}
\wt{E}^\ell_{k, \j}=\bigcup_{y\in C_\j} \left\{u\in \text{line}_\j(y)\cap \Int_\delta(S_k^\j): X^\ell_{k, \j}(u)=1\right\} \times (\text{line}_\j(y) \cap (T_k\setminus T_k^\j)),
\end{equation*}
where $C_\j$ is a minimal $\j$-line cover, and let 
$$
\wt{E}^\ell_k=\bigcup_{\j\in \ac{\B}^\ell_k}  \wt{E}^\ell_{k, \j}.
$$
Comparing this to the definition of the edge set of $G^\ell$ in~\dualeqref{eq:edges-ei-def},  one observes that we subsample edges of $G^\ell$ in a somewhat dependent way -- the set $\wt{E}^\ell_k$ contains, for every direction $\j\in \B^\ell_k$ and $y\in C_\j$, a complete bipartite graph between vertices $u$ in $\text{line}_\j(y)\cap \Int_\delta(S_k^\j)$ that were sampled by $X^\ell_{k, \j}(u)$ and $\text{line}_\j(y) \cap (T_k\setminus T_k^\j)$.  The fact that randomness is provided by the vertices $u\in S_k$ as opposed to edges themselves will not be a problem since we are interested in concentration of matching size in $G^\ell$ and do not need to reason about arbitrary edge sets -- see proof of Lemma~\dualref{lm:subsampling-matchings-gl} below. Let 
$$
\wt{G}^\ell=(S^\ell, T^\ell, \wt{E}^\ell).
$$ 
As discussed above, $\wt{G}^\ell$ is a slightly subsampled version of $G^\ell$. This operation has the desired effect of making it hard to store edges of $\wt{G}^\ell$ (since the algorithm intuitively must remember which edge of $G^\ell$ was included and which was not), but at the same time barely changes matching size in $G^\ell$, as we now show.  

\begin{lemma}[Large matchings in subsampled gadgets $\wt{G}^\ell$]\duallabel{lm:subsampling-matchings-gl}
With probability at least $1-1/N$ for every $\ell\in [L], \ell>0$, there exists a matching of $S^\ell$ to $T^\ell\setminus T_*^\ell$ of size at least $(1-O(1/K))|S^\ell|$.
\end{lemma}
\begin{proof}
Fix $\ell\in [L]$ (we will apply a union bound over all $\ell\in [L]$ later). We write $G=G^\ell, S=S^\ell, T=T^\ell, E=E^\ell, \wt{G}=\wt{G}^\ell, \wt{E}=\wt{E}^\ell$ to simplify notation. For every edge 
$e\in E$ define the random variable 
\begin{equation}\duallabel{eq:z-def}
Z_e=\left\{
\begin{array}{ll}
1&\text{~if~}e\in \wt{E}\\
0&\text{~o.w.}\\
\end{array}
\right.
\end{equation}
Note that for every matching $M\subseteq E$ random variables $\{Z_e\}_{e\in M}$ are negatively dependent, since a matching $M$ touches every vertex at most once.

By Lemma~\dualref{lm:large-matching}  applied to $G=(S, T, E)$ there exists a matching of a $(1-O(1/K))$ fraction of vertices in $S$ to $T\setminus T_*$ -- denote this matching by $M$.  Let 
$$
\wt{M}:=M\cap \wt{E}=\{e\in M: Z_e=1\}
$$
denote the subset of the edges of $M$ that are included in $\wt{E}$. Note that $\wt{M}$ is a matching between a subset of $S$ and a subset of $T\setminus T_*$, and we have 
$$
\expect[|\wt{M}|]=\sum_{e\in M} \prob[e\in \wt{E}]=\sum_{e\in M} \expect[Z_e]=(1-1/K) |M|
$$
by definition of $Z_e$ in~\dualeqref{eq:z-def} and the fact that every edge in $E$ is included in $\wt{E}$ with probability $1-1/K$ by~\dualeqref{eq:x-ell-def}.
Since the random variables $\{Z_e\}_{e\in M}$ are negatively dependent, we have by an application to the Chernoff bound (for negatively associated random variables)
$$
\prob[|\wt{M}|<(1-2/K)|M|]\leq \exp(-\Omega(|M|/K)).
$$
Since $M$ matches at least a constant fraction of $S$, we get that $|M|=\Omega(N/(KL))$, and therefore 
$$
\prob[|\wt{M}|<(1-2/K)|M|]\leq \exp(-\Omega(N/(KL)))\leq N^{-2},
$$
where $N$ is the number of vertices in our graph instance. Thus, for every fixed $\ell\in [L], \ell>0$, with probability at least $1-N^{-2}$ there exists a matching of at least a $1-O(1/K)$ fraction of $S^\ell$ to $T^\ell\setminus T_*^\ell$ in $\wt{G}^\ell$. The result of the lemma follows by a union bound over $\ell$.
\end{proof}

\paragraph{Defining the input graph $\wh{G}$.} We now define the graph $\wh{G}=(P, Q, \wh{E})$ arriving in the stream and specify the order of arrival. We have 
\begin{equation}\duallabel{eq:p-def}
P=\left(\bigcup_{\text{even~}\ell\in [L]} T^\ell\right)\cup \Upsilon_{odd}
\end{equation}
and 
\begin{equation}\duallabel{eq:q-def}
Q=S^0\cup \left(\bigcup_{\text{odd~}\ell\in [L]} T^\ell\right)\cup \Upsilon_{even},
\end{equation}
where 
\begin{equation}\duallabel{eq:upsilon-def}
\begin{split}
\Upsilon_{even}&=\left(\bigcup_{\text{even~}\ell\in [L], \ell>0} \{s\in S^\ell: \tau^\ell(s)\text{~is not defined}\}\right)\\
&\text{and}\\
\Upsilon_{odd}&=\left(\bigcup_{\text{odd~}\ell\in [L]} \{s\in S^\ell: \tau^\ell(s)\text{~is not defined}\}\right)
\end{split}
\end{equation}

We will show below that $|\Upsilon_{even}\cup \Upsilon_{odd}|=o(|P|)$.

\paragraph{Edge set $\wh{E}$ of $\wh{G}$.} 

Before defining the edge set $\wh{E}$, it is useful to define  a natural extension of the glueing maps $\tau^\ell$, $\ell\in [L], \ell>0$, from vertices in $S^\ell$ to edges in $E^\ell$. For an edge $(s, t)\in E^\ell, s\in S^\ell, t\in T^\ell$ we define 
\begin{equation*}
\tau^\ell(e)=\left\lbrace
\begin{array}{ll}
(\tau^\ell(s), t)&\text{~if~}\tau^\ell(s)\neq \bot\\
(s, t)&\text{~o.w.}
\end{array}
\right.
\end{equation*}
Note that $\tau^\ell$ is injective on edges since it is injective on vertices in $S^\ell$ (by Claim~\dualref{cl:pi-star-injective}). The edge set $\wh{E}$ of $\wh{G}$ is defined as 
\begin{equation}\duallabel{eq:ehat-def}
\wh{E}=\bigcup_{\ell\in [L], \ell>0} \bigcup_{k\in [K/2]} \bigcup_{\j\in \ac{\B}^\ell_k} \wh{E}^\ell_{k, \j},
\end{equation}
where 
\begin{equation}\duallabel{eq:ehat-def-local}
\wh{E}^\ell_{k, \j}:=\tau^\ell(\wt{E}^\ell_{k, \j}).
\end{equation}

In other words,  for every edge $(s, t)\in \wt{E}^\ell_{k, \j}$  where $s\in S^\ell$ and $t\in T^\ell$, $\ell>0$: 
\begin{enumerate}
\item if $\tau^\ell(s)$ is not defined, add the edge $(s, t)$ to  $\wh{E}^\ell_{k, \j}$;
\item if $\tau^\ell(s)$ is defined, then add the edge $(\tau^\ell(s), t)$ to $\wh{E}^\ell_{k, \j}$.
\end{enumerate}

Note that we do not include the edges from $S^0$ to $T^0$ for convenience (since $\tau^0$ is not defined, this would complicate notation somewhat).

\paragraph{Ordering of edges of $\wh{G}$ in the stream.} The graph $\wt{G}$ is presented in the stream over $L$ {\em rounds} and $K/2$ {\em phases} as follows. For every $\ell\in \{0, 1, \ldots, L-1\}$, for every $k\in [K/2]$, the edges in $\tau_*(E^\ell_k)$ are presented in the stream; the ordering within $\tau^\ell(E^\ell_k)$ is arbitrary.
 
\begin{definition}[Ordering on $(\ell, k)$ pairs]
For $\ell\in [L]$ and $k\in [K/2+1]$ we write $(\ell', k')<(\ell, k)$ iff $\ell'<\ell$ or $\ell'=\ell$ but $k'<k$. 
\end{definition}

\begin{definition}\duallabel{def:glk}
For $\ell\in [L]$ and $k\in [K/2]$ we write 
$$
\wh{G}_{(\ell, k)}=(P, Q, \wh{E}^\ell_k),
$$ 
and define $\wh{G}_{(\ell, K/2)}=(P, Q, \emptyset)$, for convenience. We define
$$
\wh{G}_{<(\ell, k)}=\left(P, Q, \bigcup_{\substack{\ell'\in [L], k'\in [K/2]\\ (\ell', k')<(\ell, k)}} \wh{E}^{\ell'}_{k'}\right).
$$
\end{definition}

\begin{definition}\duallabel{def:lambda}
For every $\ell\in [L], k\in [K/2]$ define 
$\Lambda_{\ell, k}$ as follows. For $k\in [K/2]$ let $\Lambda_{(\ell, k)}=(X^\ell_k, \j^\ell_k)$. For $k=K/2$ let $\Lambda_{\ell, K/2}:=(\j^\ell_{K/2})$.  We write $\Lambda_{<(\ell, k)}=\left(\Lambda_{\ell', k'}\right)_{(\ell', k')<(\ell, k)}$.
\end{definition}

\begin{remark}\duallabel{rm:lambda}
Note that $\wh{G}_{\leq (\ell, k)}$ is fully determined by $\Lambda_{<(\ell, k)}$ and $X^\ell_k$, and $\j^\ell_k$ is uniformly random in $\ac{\B}^\ell_k$ conditioned on $\Lambda_{<(\ell, k)}$ and $X^\ell_k$. It is important to note here that the restriction of the glueing map $\tau^\ell$ to $S^\ell_{\leq k}$ (which we need to fully determine $\wh{G}_{\leq (\ell, k)}$) is, crucially, determined by $\J^\ell_{<k}$ -- see Remark~\dualref{rm:tau-prefix}.
\end{remark}

\subsection{Upper and lower bounds on matchings in $\wh{G}$}\duallabel{sec:malg}
We first prove
\begin{lemma}[Large matching in $\wh{G}$]\duallabel{lm:large-matching-g-prime}
With probability at least $1-1/N$ there exists a matching in $\wh{G}$ of size at least $(1-O(1/L))|P|$.
\end{lemma}
\begin{proof}
By Lemma~\dualref{lm:subsampling-matchings-gl} with probability at least $1-N^{-1}$ for every $\ell\in [L], \ell>0$, there exists a matching of a $1-O(1/K)$ fraction of $S^\ell$ to $T^\ell\setminus T_*^\ell$ in $\wt{G}^\ell$.  We condition on this event.


Denote the corresponding matching in $\wt{G}^\ell$ by $\wt{M}^\ell$. Now recall that by construction of the graph $\wh{G}$ for every $(s, t)\in \wt{M}^\ell, s\in S^\ell, t\in T^\ell\setminus T_*^\ell$ one of the following two cases holds:
\begin{enumerate}
\item if $\tau^\ell(s)$ is not defined, and $(s, t)\in \wh{E}$;
\item if $\tau^\ell(s)$ is defined, and $(\tau^\ell(s), t)\in \wh{E}$.
\end{enumerate}
Since $\tau^\ell$ is injective, 
$\tau^\ell(\wt{M}^\ell)$ is a matching for every $\ell\in [L], \ell>0$. Now recall that $\wt{M}^\ell$ matches $S^\ell$ to $T^\ell\setminus T_*^\ell$. At the same time $\tau^\ell$ and maps $S^\ell$ to $T_*^{\ell-1}$, and whenever $\tau^\ell(s)$ is not defined, an edge $(s, t)\in S^\ell\times T^\ell$ is mapped to a separate set of vertices (see $\Upsilon_{even}$ and $\Upsilon_{odd}$ in~\dualeqref{eq:p-def} and~\dualeqref{eq:q-def}) in $\wh{G}$. so the union of these matchings still forms a matching in $\wh{G}$. 

We thus get that $\bigcup_{\ell\in [L], \ell>0} \tau^\ell(\wt{M}^\ell)$ is a matching of size at least
$$
(1-O(1/K))\sum_{\ell\in [L], \ell>0} |S^\ell|=(1-O(1/K))(L-1)\cdot (N/2)=(1-O(1/L))(L/2)\cdot N.
$$
In the above we used the fact that $|S^\ell|\geq \sum_{k\in [K/2]} (1-\sqrt{\e}) |T_0|/K\geq (1-\sqrt{\e}) N/2$ by Lemma~\dualref{lm:size-bounds}, {\bf (2)}, together with~\dualref{p5} and~\dualref{p6}, as well as the fact that $L\leq K$ by~\dualref{p4}. 

We now upper bound $|P|$. By~\dualeqref{eq:p-def}
\begin{equation*}
\begin{split}
|P|&=\left|\left(\bigcup_{\text{even~}\ell\in [L]} T^\ell\right)\cup \left(\bigcup_{\text{odd~}\ell\in [L], \ell>0} \{s\in S^\ell: \tau^\ell(s)\text{~is not defined}\}\right)\right|\\
&=\sum_{\text{even~}\ell\in [L]} |T^\ell|+\sum_{\text{odd~}\ell\in [L], \ell>0} |\{s\in S^\ell: \tau^\ell(s)\text{~is not defined}\}|\\
&\leq  (1+\e^{1/2}) (L/2) N+\sum_{\text{odd~}\ell\in [L], \ell>0} |\{s\in S^\ell: \tau^\ell(s)\text{~is not defined}\}|\\
&\leq  (1+\e^{1/2}) (L/2) N+\delta^{1/4} N\\
&\leq (1+O(1/K))(L/2) N.
\end{split}
\end{equation*}
The third transition is by Lemma~\dualref{lm:size-bounds}  the forth transition is by Lemma~\dualref{lm:t-star-minus-tau-sp} and the final transition is by~\dualref{p5} and~\dualref{p6}. Putting the two bounds together, we get  that there exists a matching of size at least
$(1-O(1/L))(L/2)\cdot N\geq (1-O(1/L)) |P|$ with probability at least $1-N^{-1}$,
as required.
\end{proof}

We now turn to upper bounding the performance of a small space streaming algorithm on our input distribution $\mathcal{D}$. Since the input is sampled from a distribution, we may assume by Yao's minimax principle that the streaming algorithm ALG is deterministic. Let ALG denote a deterministic streaming algorithm that uses $s$ bits of space and at the end of the stream outputs a matching $M_{ALG}$ in $\wh{G}$ such that 
\begin{equation*}
\prob_{\wh{G}\sim \mathcal{D}} \left[|M_{ALG}|\geq \left(\frac1{1+\ln 2}+\eta\right)|M_{OPT}|\right]\geq 3/4
\end{equation*}
for some positive $\eta\in (0, 1)$, where $M_{OPT}$ is a maximum matching in $\wh{G}$. Note that we are assuming that with probability at least $3/4$ both $M_{ALG}$ is a matching in $\wh{G}$ (i.e., in particular, the algorithm does not output edges that are not in $\wh{G}$) and the size of $M_{ALG}$ is large as above. At the same time by Lemma~\dualref{lm:large-matching-g-prime} one has
$$
\prob_{\wh{G}\sim \mathcal{D}} \left[|M_{OPT}|< (1-O(1/L))|P|\right]\leq N^{-1}.
$$ 
Putting the two bounds above together, we get
\begin{equation}\duallabel{eq:malg-large}
\begin{split}
\prob_{\wh{G}\sim \mathcal{D}} \left[|M_{ALG}|\geq \left(\frac1{1+\ln 2}+\eta-O(1/L)\right)|P|\right]\geq 1/2.
\end{split}
\end{equation}
In what follows we show that any algorithm that achieves~\dualeqref{eq:malg-large} must essentially remember, for many edges of $G^\ell, \ell\in [L]$, whether they were included in $\wt{G}^\ell$ and therefore in $\wh{G}$.

\paragraph{Upper bounding $|M_{ALG}|$.}  Let sets $\out^\ell\subset T^\ell$ of `outlier' vertices as defined in~\dualeqref{eq:o-def}, and let $\out=\bigcup_{\ell\in [L]} \out^\ell$ as in~\dualeqref{eq:o-union-def}.  
Define
\begin{equation}\duallabel{eq:apm-def}
\begin{split}
A_P&=\left(\bigcup_{\substack{\ell\in [L]\\ \ell \text{~even}}} \nu_{\ell, *}(T^\ell\setminus \Ext_\delta(T_*^\ell))\right)\setminus \out\\
A_Q&=\left(\bigcup_{\substack{\ell\in [L]\\ \ell\text{~odd}}} \nu_{\ell, *}(T^\ell\setminus \Ext_\delta(T_*^\ell))\right)\setminus \out.
\end{split}
\end{equation}
We define intermediate sets
\begin{equation}\duallabel{eq:bpm-def-prime}
\begin{split}
B'_Q&=\left(\bigcup_{\substack{\ell\in [L]\\ \ell\text{~even}}} \mu_{\ell, *}(T^\ell\setminus T_*^\ell)\right)\cup \out\\
B'_P&=\left(\bigcup_{\substack{\ell\in [L]\\ \ell \text{~odd}}} \mu_{\ell, *}(T^\ell\setminus T_*^\ell)\right)\cup \out,
\end{split}
\end{equation}
and then let
\begin{equation}\duallabel{eq:bpm-def}
\begin{split}
B_Q&=\bigcup_{\substack{\ell\in [L]\\ \ell\text{~even}}} \tau^\ell(B'_Q\cap S^\ell)\\
B_P&=\bigcup_{\substack{\ell\in [L]\\ \ell \text{~odd}}} \tau^\ell(B'_P\cap S^\ell).
\end{split}
\end{equation}

We have 
\begin{claim}\duallabel{cl:ab-disjoint}
$A_P\cap B_P=\emptyset$ and $A_Q\cap B_Q=\emptyset$.
\end{claim}
\begin{proof}
We prove the first claim (the proof of the second is analogous).
One has by ~\dualeqref{eq:apm-def}
\begin{equation}\duallabel{eq:9239ggdsgpffdf}
\begin{split}
A_P&=\left(\bigcup_{\substack{\ell\in [L]\\ \ell \text{~even}}} \nu_{\ell, *}(T^\ell\setminus T_*^\ell)\right)\setminus \out\\
&=\left(\bigcup_{\substack{\ell\in [L]\\ \ell \text{~even}}} \bigcup_{\substack{j=0\\j\text{~even}}}^\ell \nu_{\ell, j}(T^\ell\setminus T_*^\ell)\right)\setminus \out
\end{split}
\end{equation}
and by~\dualeqref{eq:bpm-def} and~\dualeqref{eq:bpm-def-prime} 
\begin{equation}\duallabel{eq:9239ggdsgpffdfsf}
\begin{split}
B_P&=\bigcup_{\substack{\ell\in [L]\\ \ell\text{~odd}}} \tau^\ell(B'_P\cap S^\ell)\\
&=\left(\bigcup_{\substack{\ell\in [L]\\ \ell\text{~odd}}} \bigcup_{\substack{j=0\\j\text{~even}}}^\ell\tau^{\ell-j}(\mu_{\ell, j}(T^\ell\setminus T_*^\ell))\right)\cup \out\\
&\subseteq \left(\bigcup_{\substack{\ell\in [L]\\ \ell\text{~odd}}} \bigcup_{\substack{j=0\\j\text{~even}}}^\ell\nu_{\ell, j+1}(T^\ell\setminus T_*^\ell)\right)\cup \out\\
\end{split}
\end{equation}
Disjointness now follows by Lemma~\dualref{lm:nu-prop}, {\bf (2)}, since the range of $(\ell, j)$ pairs in~\eqref{eq:9239ggdsgpffdf} is disjoint from the range in~\eqref{eq:9239ggdsgpffdfsf}.
\end{proof}

\begin{lemma}[Almost partition of $P$ and $Q$]\duallabel{lm:partition}
One has  
$$
|P\setminus (A_P\cup B_P)|=O(N)
$$
and 
$$
|Q\setminus (A_Q\cup B_Q)|=O(N)
$$
for sets $A_P, A_Q, B_P, B_Q$ defined in~\dualeqref{eq:apm-def} and~\dualeqref{eq:bpm-def}.
\end{lemma}
The proof of the lemma is given in Appendix~\dualref{app:partition-lemma}.

The following lemma is key to bounding the size of the vertex cover that we construct in Lemma~\dualref{lm:vertex-cover} to upper bound the size of $M_{ALG}$.
\begin{lemma}\duallabel{lm:sizeof-ab}
One has 
\begin{equation*}
\begin{split}
|B_Q|&\leq  \frac{L}{2}\cdot \frac{N}{2}\cdot\frac1{1+\ln 2}(1+O(1/K))\\
&\text{and}\\
|B_P|&\leq \frac{L}{2}\cdot \frac{N}{2}\cdot\frac1{1+\ln 2}(1+O(1/K))\\
\end{split}
\end{equation*}
\end{lemma}
\begin{proof} We prove the bound for $B_Q$ (the proof for $B_P$ is analogous). By ~\dualeqref{eq:bpm-def} and~\dualeqref{eq:bpm-def-prime} we have
\begin{equation*}
\begin{split}
|B_Q|&=\left|\bigcup_{\substack{\ell\in [L]\\ \ell\text{~even}}} \tau^\ell(B'_P\cap S^\ell)\right|\\
&\leq \left|\bigcup_{\substack{\ell\in [L]\\ \ell\text{~even}}} B'_P\cap S^\ell\right|\\
&=\left|B'_P\right|\\
&\leq \left|\left(\bigcup_{\substack{\ell\in [L]\\ \ell\text{~even}}} \mu_{\ell, *}(T^\ell\setminus T_*^\ell)\right)\cup \out\right|\\
&\leq |\out|+\sum_{\substack{\ell\in [L]\\ \ell\text{~even}}} |\mu_{\ell, *}(T^\ell\setminus T_*^\ell)|\\
&\leq \delta^{1/8} |P|+\sum_{\substack{\ell\in [L]\\ \ell\text{~even}}} |\mu_{\ell, *}(T^\ell\setminus T_*^\ell)|,
\end{split}
\end{equation*}
where the last transition is by Claim~\dualref{cl:o-size}. Thus, it suffices to upper bound $|\mu_{\ell, *}(T^\ell\setminus T_*^\ell)|$ for $\ell\in [L]$.

Using Lemma~\dualref{lm:mu-ell-j} one has for every $\ell \in [L]$
\begin{equation*}
\begin{split}
\left|\mu_{\ell, *}(T^\ell\setminus T_*^\ell)\right|&=\left|\bigcup_{\substack{0\leq j\leq \ell\\j\text~{even}}} \mu_{\ell, j}(T^\ell\setminus T_*^\ell)\right|\\
&=\bigcup_{\substack{0\leq j\leq \ell\\j\text~{even}}} \left|\mu_{\ell, j}(T^\ell\setminus T_*^\ell)\right|\text{~~~~~~~~~~~~~~(since $\mu_{\ell, j}(T^\ell\setminus T_*^\ell)$ are disjoint for different $j$ by Lemma~\dualref{cl:pi-star-injective}, {\bf (3)})}\\
&\leq \sum_{\substack{0\leq j\leq \ell\\j\text~{even}}} (\ln 2+C/K)^j \frac1{2}(1-\ln 2)|T^\ell|\text{~~~~~~~(by Lemma~\dualref{lm:mu-ell-j})}\\
&\leq \frac1{2}(1-\ln 2)|T^\ell|\sum_{j\geq 0} (\ln 2+C/K)^{2j} \\
&= \frac{1-\ln 2}{2}|T^\ell|\cdot \frac1{1-(\ln 2+C/K)^2}\\
&=\frac{1}{2(1+\ln 2)}(1+O(1/K))|T^\ell|\\
\end{split}
\end{equation*}

We now note that $|T^\ell|=N$ for all $\ell\in [L]$. Summing the above over all even $\ell$ between $0$ and $L-1$ gives the required upper bound.
\end{proof}

\begin{lemma}\duallabel{lm:vertex-cover}
For every matching $M\subseteq \wh{E}$ one has 
$$
|M|\leq |M\cap (A_P\times (Q\setminus B_Q))|+\frac1{1+\ln 2}|P|+O(|P|/L).
$$
\end{lemma}
\begin{proof}
We exhibit a vertex cover of appropriate size for $M$. Specifically, we add to the vertex cover one endpoint of every edge in
$$
M\cap (A_P\times (Q\setminus B_Q)),
$$
as well as all vertices in $P\setminus A_P\approx B_P$ and $B_Q$. Note that this is indeed a vertex cover: 
every edge of $M$ either has an endpoint in $P\setminus A_P$, or belongs to $A_P\times (Q\setminus B_Q)$, or belongs to $A_P\times B_Q$, in which case it has an endpoint in $B_Q$.

The size of the vertex cover is 
\begin{equation}\duallabel{eq:23t77Gy8uGBBYVF}
\begin{split}
&|M\cap (A_P\times (Q\setminus B_Q))|+|P\setminus A_P|+|B_Q|\\
\leq &|M\cap (A_P\times (Q\setminus B_Q))|+|B_P|+|B_Q|+O(N),
\end{split}
\end{equation}
where we used Lemma~\dualref{lm:partition} to conclude that 
$$
|P\setminus A_P|\leq |B_P|+|P\setminus (A_P\cup B_P)|=|B_P|+O(N).
$$ 

By Lemma~\dualref{lm:sizeof-ab} we have 
\begin{equation*}
\begin{split}
|B_P|&\leq \frac{L}{2}\cdot \frac{N}{2}\frac1{1+\ln 2}(1+O(1/K))\\
&\text{and}\\
|B_Q|&\leq \frac{L}{2}\cdot \frac{N}{2}\frac1{1+\ln 2}(1+O(1/K)).
\end{split}
\end{equation*}
Putting the above together with~\dualeqref{eq:23t77Gy8uGBBYVF} and recalling that $L\leq \sqrt{K}$ by~\dualref{p4} and that
$$
|P|=\left|\bigcup_{\text{even~} \ell\in [L]} T^\ell\right|=L\cdot N/2
$$
gives the result.
\end{proof}

We now prove 
\begin{lemma}\duallabel{lm:special-edges}
For every matching $M\subseteq \wh{E}$ one has 
$$
M\cap (A_P\times (Q\setminus B_Q))\subseteq \bigcup_{\ell\in [L], k\in [K/2]} \tau^\ell(E^\ell_{k, \j^\ell_k}).
$$
\end{lemma}
\begin{proof}
Suppose that $u\in A_P, v\in Q\setminus B_Q$ and $(u, v)\in M$. Let $\ell\in [L]$ be an even integer such that  $u\in T^\ell$. Such an $\ell$ exists because by~\dualeqref{eq:apm-def} one has
\begin{equation*}
A_P=\left(\bigcup_{\substack{\ell\in [L]\\ \ell \text{~even}}} \nu_{\ell, *}(T^\ell\setminus \Ext_\delta(T_*^\ell))\right)\setminus \out\\
\end{equation*}
and by Definition~\dualref{def:nu} one has 
$$
\nu_{\ell, *}(T^\ell\setminus \Ext_\delta(T_*^\ell)):=\bigcup_{\substack{j=0\\j \text{~even}}}^\ell \nu_{\ell, j}(T^\ell\setminus \Ext_\delta(T_*^\ell)),
$$
so that 
$$
\nu_{\ell, *}(T^\ell\setminus \Ext_\delta(T_*^\ell))\subseteq \bigcup_{\text{even~}\ell\in [L]} T^\ell.
$$
Uniqueness of $\ell$ follows by Lemma~\dualref{lm:nu-prop}, {\bf (2)}.  Furthermore, we get that 
\begin{equation}\duallabel{eq:u-apm}
u\in \nu_{\ell+j, j}(T^{\ell+j}\setminus \Ext_\delta(T_*^{\ell+j}))\setminus \out
\end{equation}
for some even $\ell\in [L]$ and even $j$.

We now consider two cases: depending on whether $v\in T^{\ell-1}$ ({\bf case 1}) or $v\in T^{\ell+1}$ ({\bf case 2}).

\paragraph{\bf Case 1.} In this case there exists a unique $y\in S^\ell$ such that $\tau^\ell(y)=v$. Indeed, otherwise the edge $(u, v)$ would not be in the graph $\wh{G}$ as per~\dualeqref{eq:ehat-def} and~\dualeqref{eq:ehat-def-local}; uniqueness follows from injectivity of $\tau^\ell$. Letting $x=u$, we now show using Corollary~\dualref{cor:cut-structure} that $y\in \mu_{\ell+j, j}(T^{\ell+j}\setminus T_*^{\ell+j})$, which in turn 
by~\dualeqref{eq:bpm-def} together with the definition of $\mu_{\ell, *}$ (Definition~\dualref{def:nu}) implies that $v=\tau^\ell(y)\in B_Q$, as required. We now provide the details.

Let $k\in [K/2]$ be the unique index such that $x\in T^\ell_k$ and $y\in S^\ell_k$.  Uniqueness follows since the edge sets in~\dualeqref{eq:edges-ei-def} are disjoint by Lemma~\dualref{lm:disjoint}. Note that $x\not \in \out$ since we excluded this set in~\dualeqref{eq:apm-def}. If $y\in \out$, we have $y\in B_Q$ and there is nothing to prove. Thus, it suffices to consider the case $y\not \in \out$.

We thus have $x\in T^\ell_k\setminus \out$ and $y\in S^\ell_k\setminus \out$.  Furthermore, since $(u, v)=\tau^\ell((y, x))$, the assumption that $(u, v)\in \wh{E}$ implies that $(x, y)\in E^\ell$, and therefore 
$$
y=x+\lambda \cdot\u
$$ 
for some $\u\in \ac{\B}^\ell_k$. Furthermore, it follows by Lemma~\dualref{cl:line-size} that $|\lambda|\leq 2M/w$. We assume towards a contradiction that $\u\neq \j^\ell_k$. Since 
$$
\Spt(\B^\ell)\cap \ac{\B}^\ell_k=\{\j^\ell_k\},
$$
this means that $\u\not \in \Spt(\B^\ell)$ (see~\dualeqref{eq:sp-ext} for the definition of $\Spt$). This means, since $\ac{\B}^\ell\subset \B^\ell$, that the preconditions of Corollary~\dualref{cor:cut-structure} are satisfied and we get that 
$$
x\in \nu_{\ell+j, j}(T^{\ell+j}\setminus \Ext_\delta(T_*^{\ell+j})),
$$ 
implies
$$
y\in \mu_{\ell+j, j}(T^{\ell+j}\setminus T_*^{\ell+j}).
$$ 
At the same time by Definition~\dualref{def:nu} for every $\ell\in [L]$ 
$$
\mu_{\ell, *}(T^\ell\setminus T_*^\ell):=\bigcup_{\substack{j=0\\j \text{~even}}}^\ell \mu_{\ell, j}(T^\ell\setminus T_*^\ell),
$$
which means that $y\in \mu_{\ell, *}(T^\ell\setminus T_*^\ell)\subseteq B'_Q$ as per~\dualeqref{eq:bpm-def-prime}. Therefore, $v=\tau^\ell(y)\in B_Q$, as required.

\paragraph{Case 2.} In this case there exists a unique $x'\in S^{\ell+1}$ such that $\tau^{\ell+1}(x')=u$. Indeed, otherwise the edge $(u, v)$ would not be in the graph $\wh{G}$ as per~\dualeqref{eq:ehat-def} and~\dualeqref{eq:ehat-def-local}; uniqueness follows from injectivity of $\tau^{\ell+1}$. Let $x\in T^{\ell+1}$ be such that $x\delequal x'$. Let $y=v$. Let $k\in [K/2]$ be the unique index such that $y\in T^{\ell+1}_k$ and $x'\in S^{\ell+1}_k$.  Uniqueness follows since the edge sets in~\dualeqref{eq:edges-ei-def} are disjoint by Lemma~\dualref{lm:disjoint}.  Note that $x\not \in \out$ since we excluded this set in~\dualeqref{eq:apm-def}. If $y\in \out$, we have $y\in B_Q$ and there is nothing to prove. Thus, it suffices to consider the case $y\not \in \out$. We thus have $y\in T^{\ell+1}_k\setminus \out$ and $x'\in S^{\ell+1}_k\setminus \out$.

Since $(u, v)=\tau^{\ell+1}((x', y))$, the assumption that $(u, v)\in \wh{E}$ implies that $(x', y)\in E^{\ell+1}$, and therefore, since $x\delequal x'$,
$$
y=x+\lambda \cdot\u
$$ for some $\u\in \ac{\B}^{\ell+1}_k$. Furthermore, it follows by Lemma~\dualref{cl:line-size} that $|\lambda|\leq 2M/w$. We assume towards a contradiction that $\u\neq \j^{\ell+1}_k$. Since 
$$
\Spt(\B^{\ell+1})\cap \ac{\B}^{\ell+1}_k=\{\j^{\ell+1}_k\},
$$
this means that $\u\not \in \Spt(\B^{\ell+1})$, and therefore, since $\ac{\B}^{\ell+1}\subseteq \B^{\ell+1}$, the preconditions of Lemma~\dualref{lm:cut-structure} are satisfied for $x$, $y$ and $\ell+1$. Furthermore, by~\dualeqref{eq:u-apm} we have
\begin{equation*}
\begin{split}
u&=\tau^{\ell+1}(x')\\
&\in \nu_{\ell+j, j}(T^{\ell+j}\setminus \Ext_\delta(T_*^{\ell+j}))\\
&\subseteq \tau^{\ell+1}\left(\mu_{\ell+j, j-1}(T^{\ell+j}\setminus \Ext_\delta(T_*^{\ell+j}))\right),
\end{split}
\end{equation*}
and therefore $x'\in \mu_{\ell+j, j-1}(T^{\ell+j}\setminus \Ext_\delta(T_*^{\ell+j}))$.  
Since 
$$
\mu_{\ell+j, j-1}(T^{\ell+j}\setminus \Ext_\delta(T_*^{\ell+j}))=\downset^{\ell+1}(\nu_{\ell+j, j-1}(T^{\ell+j}\setminus \Ext_\delta(T_*^{\ell+j}))),
$$
we have, since $x\delequal x'$,
$$
x\in \nu_{\ell+j, j-1}(T^{\ell+j}\setminus \Ext_\delta(T_*^{\ell+j}))=\nu_{(\ell+1)+(j-1), j-1}(T^{(\ell+1)+(j-1)}\setminus \Ext_\delta(T_*^{(\ell+1)+(j-1)})).
$$ 

This means that the preconditions of Lemma~\dualref{lm:cut-structure} are satisfied, and we have \footnote{When $\ell+1=L-1$, we have $\ell+2=L$, which does not technically correspond to a gadget in our input graph. However, we think of artifically adding such a gadget here to handle this corner case for simplicity.}
\begin{equation*}
\begin{split}
y\in \nu_{(\ell+1)+(j-1), j-1}(T^{(\ell+1)+(j-1)}\setminus T_*^{(\ell-1)+(j-1)})&=\nu_{\ell+j, j-1}(T^{\ell+j}\setminus T_*^{\ell+j})\\
&=\tau^{\ell+2}(\mu_{\ell+j, j-2}(T^{\ell+j}\setminus T_*^{\ell+j})).
\end{split}
\end{equation*}

At the same time by Definition~\dualref{def:nu} for every $\ell\in [L]$ 
$$
\mu_{\ell, *}(T^\ell\setminus T_*^\ell):=\bigcup_{\substack{i=0\\i \text{~even}}}^\ell \mu_{\ell, i}(T^\ell\setminus T_*^\ell),
$$
which means that $y\in \tau^{\ell+2}(\mu_{\ell+j, *}(T^{\ell+j}\setminus T_*^{\ell+j})\cap S^{\ell+2})\subseteq B_Q$ as per~\dualeqref{eq:bpm-def-prime}, as required.
\end{proof}

\subsection{Proof of Theorem~\ref{thm:main}}\duallabel{sec:main-proof}
We now give

\begin{proofof}{Theorem~\ref{thm:main}} Now putting~\dualeqref{eq:malg-large} together with~Lemma~\dualref{lm:vertex-cover}, we get
\begin{equation*}
\begin{split}
|M_{ALG}\cap (A_P\times (Q\setminus B_Q))|&\geq |M_{ALG}|-\left(\frac1{1+\ln 2}|P|+O(|P|/L)\right)\\
&\geq \left(\frac1{1+\ln 2}+\eta-O(1/K)\right)|P|-\left(\frac1{1+\ln 2}|P|+O(|P|/L)\right)\\
&\geq (\eta-O(1/L))|P|\\
&\geq (\eta/2)|P|,
\end{split}
\end{equation*}
where we used~\dualeqref{eq:malg-large} in the third transition and the fact that $L\leq \sqrt{K}$ by~\dualref{p3} in the forth transition, and assumed that $L=\sqrt{K}$ is larger than an absolute constant that depends on $\eta$ in the last transition. Thus,
\begin{equation}\duallabel{eq:success-prob}
\prob_{\wh{G}\sim \mathcal{D}}\left[|M_{ALG}\cap (A_P\times (Q\setminus B_Q))|\geq (\eta/4)|P|\text{~and~} M_{ALG}\subseteq \wh{E}\right]\geq 1/2.
\end{equation}
Note that the second condition above, namely $M_{ALG}\subseteq \wh{E}$ enforces the constraint that the algorithm does not output non-edges\footnote{The analysis generalizes easily to the setting where the algorithm is allowed to output a small fraction of non-edges, but this is a rather non-standard assumption, and we prefer to operate under the more standard model where $M_{ALG}$ must be a subset of $\wh{E}$ with a good probability.}. We do not add this condition explicitly in calculations below to simplify notation (one can think of $|M_{ALG}|$ as being defined as zero when $M_{ALG}$ contains non-edges). Now recall that by Lemma~\dualref{lm:special-edges} we have
\begin{equation*}
\begin{split}
M_{ALG}\cap (A_P\times (Q\setminus B_Q))&\subseteq \bigcup_{\ell\in [L], k\in [K/2]} \tau^\ell(E^\ell_{k, \J^\ell_k})=\bigcup_{\ell\in [L], k\in [K/2]} \wh{E}^\ell_{k, \J^\ell_k}.
\end{split}
\end{equation*}
Thus, there exist $\ell^*\in [L], k^*\in [K/2]$ such that
\begin{equation}\duallabel{eq:special-lk}
\begin{split}
\prob\left[|M_{ALG}\cap \wh{E}^{\ell^*}_{k^*, \j^\ell_{k^*}}|\geq \frac{\eta}{2KL} |P|\right]\geq \frac1{KL}.
\end{split}
\end{equation}
Indeed, otherwise one would have
\begin{equation*}
\begin{split}
\prob[|M_{ALG}\cap (A_P\times (Q\setminus B_Q))|&\geq (\eta/4)|P|]\\
&\leq  \prob\left[\text{exist~} \ell\in [L]\text{~and~} k\in [K/2] \text{~such that~}|M_{ALG}\cap E^\ell_{k, \j^\ell_k}|\geq \frac{\eta}{2LK}|P|\right]\\
&\leq \sum_{\ell\in [L]} \sum_{k\in [K/2]} \prob\left[|M_{ALG}\cap E^\ell_{k, \j^\ell_k}|\geq \frac{\eta}{2KL}|P|\right]\\
&< \sum_{\ell\in [L]} \sum_{k\in [K/2]} \frac1{KL}\\
&= (KL/2)\cdot \frac1{KL}\\
&= 1/2, 
\end{split}
\end{equation*}
a contradiction with~\dualeqref{eq:success-prob}.

To simplify notation, we let $\ell=\ell^*, k=k^*$. Let that by Definition~\dualref{def:glk} we write $\wh{G}_{<(\ell, k)}$ to denote the subgraph of $\wh{G}$ that arrives up to the $k$-th phase of the $\ell$-th round. Also recall that {\bf (a)} $\wh{G}_{\leq (\ell, k)}$ is fully determined by $\Lambda_{<(\ell, k)}$ and $X^\ell_k$ (see Definition~\dualref{def:lambda}) and {\bf (b)} conditioned on $\Lambda_{<(\ell, k)}$ and $X^\ell_k$ one has $\j^\ell_k\sim UNIF(\ac{\B}^\ell_k)$.  For simplicity of notation we write 
$$ 
\B=\B^\ell_k\text{,~~}\ac{\B}=\ac{\B}^\ell_k\text{~~and~~}\j=\j^\ell_k.
$$
Recall that $\ac{\B}^\ell_k=\B^\ell_k\setminus \Ext^\ell_k\cup \{\q^\ell_k, \r^\ell\}$. We also let
$$
S_k=S^\ell_k\text{~~and~~}X:=X^\ell_k.
$$

\paragraph{Lower bounding the space usage of ALG.} In what follows we show that since $M_{ALG}$ often returns many edges from $\wh{G}_{(\ell, k)}$ as per~\dualeqref{eq:special-lk}, the conditional entropy of $X^\ell_k$ given $\Pi$ and $\Lambda_{\leq(\ell, k)}$ is low, which gives the desired lower bound on $s$. Let $\Pi\in \{0, 1\}^s$ denote the state of ALG after it has been presented with $\wh{G}_{\leq (\ell, k)}$. Then finish running ALG on $\wh{G}_{> (\ell, k)}$ starting with state $\Pi$. Let $M_{ALG}$ denote the matching output by ALG. We have 
\begin{equation}\duallabel{eq:93t9239hdsagasasaaaa}
\begin{split}
s=|\Pi|&\geq H(\Pi)\\
&\geq H(\Pi| \Lambda_{<(\ell, k)})\\
&\geq I(\Pi; X | \Lambda_{<(\ell, k)})\\
&\geq \sum_{\i\in \ac{\B}} I(\Pi; X_\i | \Lambda_{<(\ell, k)})\\
&=\sum_{\i\in \ac{\B}} I(\Pi; X_\i | \Lambda_{<(\ell, k)}, \{\j=\i\})\\
&\geq \sum_{\i\in \ac{\B}} I(M_{ALG}; X_\i | \Lambda_{<(\ell, k)}, \{\j=\i\})\\
\end{split}
\end{equation}
The second transition uses the fact that conditioning does not increase entropy, the forth transition uses the fact that $X_\i$'s are independent conditioned on $\Lambda_{<(\ell, k)}$, the forth transition uses the fact that $\j$ is independent of $\Pi$ and $X_\i$ conditioned on $\Lambda_{<(\ell, k)}$.   The final transition is by the data processing inequality:
\begin{lemma} \emph{(Data Processing Inequality)} \label{thm:dpi}
    For any random variables $(X,Y,Z)$ such that $X \to Y \to Z$ forms a Markov chain, we have $I(X;Z) \le I(X;Y)$.
\end{lemma}

Recall that we let $\ell=\ell^*$ and $k=k^*$, where $\ell^*$ and $k^*$ satisfy~\dualeqref{eq:special-lk}, and let $\j=\j^\ell_k$, to simplify notation.
We now lower bound 
\begin{equation}\duallabel{eq:20iabibafihf0ASAS}
\begin{split}
\sum_{\i\in \ac{\B}} I(M_{ALG}; X_\i | \Lambda_{<(\ell, k)}, \{\j=\i\})&=\sum_{\i\in \ac{\B}} H(X_\i| \Lambda_{<(\ell, k)}, \{\j=\i\})-H(X_\i | M_{ALG}, \Lambda_{<(\ell, k)}, \{\j=\i\})\\
&=\sum_{\i\in \ac{\B}} H(X_\i)-H(X_\i | M_{ALG}, \Lambda_{<(\ell, k)}, \{\j=\i\}).
\end{split}
\end{equation}

We now upper bound $H(X_\i | M_{ALG}, \Lambda_{<(\ell, k)}, \{\j=\i\})$ on the rhs of~\dualeqref{eq:20iabibafihf0ASAS}. Let 
\begin{equation}\duallabel{eq:8923yty9y239ruyURadF}
\begin{split}
\E:=&\left\{|M_{ALG}\cap E^\ell_{k, \j}|\geq \frac{\eta}{2KL} |P|\text{~and~}M_{ALG}\subseteq \wh{E}\right\}
\end{split}
\end{equation}
and  let $Z$ denote the indicator of $\E$. Note that $\expect[Z]=\prob[\E]\geq \frac1{KL}$ by~\dualeqref{eq:special-lk}. We have
\begin{equation}\duallabel{eq:plus-z-198yugudfg}
\begin{split}
H(X_\i| M_{ALG}, \Lambda_{<(\ell, k)}, \{\j=\i\})&\leq H(X_\i, Z| M_{ALG}, \Lambda_{<(\ell, k)}, \{\j=\i\})\\
&\leq H(Z)+H(X_\i| M_{ALG}, \Lambda_{<(\ell, k)}, \{\j=\i\}, Z)\\
&\leq 1+H(X_\i| M_{ALG}, \Lambda_{<(\ell, k)}, \{\j=\i\}, Z),\\
\end{split}
\end{equation}
where we used the fact that $H(Z)\leq 1$, as $Z$ is a binary variable. At the same time, since 
$\expect[Z]=\expect_{\i\sim UNIF(\ac{\B})}\left[Z|\{\j=\i\}\right]\geq \frac1{KL}$ by~\dualeqref{eq:special-lk},
and $\j\sim UNIF(\ac{\B})$, there exists a subset $\mathcal{J}\subseteq \ac{\B}$ such that $|\mathcal{J}|\geq \frac1{KL}|\ac{\B}|$ and for every $\i\in \mathcal{J}$ one has 
$\expect[Z| \{\j=\i\}]\geq \frac1{KL}.$ For every $\i\in \mathcal{J}$ one has
\begin{equation}\duallabel{eq:93g9g7g89gFGYFGSA}
\begin{split}
H(X_\i| M_{ALG}, &\Lambda_{<(\ell, k)}, \{\j=\i\}, Z)\\
&=H(X_\i| M_{ALG}, \Lambda_{<(\ell, k)}, \{\j=\i\wedge Z=1\})\cdot \prob[Z=1| \{\j=\i\}]\\
&+H(X_\i| M_{ALG}, \Lambda_{<(\ell, k)}, \{\j=\i\wedge Z=0\})\cdot \prob[Z=0| \{\j=\i\}]\\
\end{split}
\end{equation}

We now bound both terms on the rhs in~\dualeqref{eq:93g9g7g89gFGYFGSA}. For the second term we have
\begin{equation}
\begin{split}
H(X_\i| M_{ALG}, \Lambda_{<(\ell, k)}, \{\j=\i\wedge Z=0\})&\leq \expect_{\Lambda_{<(\ell, k)}}[|S_k|] \cdot H_2(1-1/K)\\
&\leq (1+\sqrt{\e})\frac1{K}|T| \cdot H_2(1-1/K)\\
\end{split}
\end{equation}
where the first transition is because $\sum_{y\in S_k} X_\i(y)=\lceil (1-\frac1{K}) |S_k|\rceil$ by definition of $X_\i$ and the second transition is by Lemma~\dualref{lm:size-bounds}, {\bf (2)}.

For the first term on the rhs in~\dualeqref{eq:93g9g7g89gFGYFGSA}  we note that since $M_{ALG}\subseteq \wh{E}$ as we are conditioning on the event $\mathcal{E}$ (by conditioning on $\{Z=1\}$) for every $y\in S_k$ that is matched by $M_{ALG}$ one has $X_\i(y)=1$. 
By conditioning on $\{Z=1\wedge \j=\i\}$, we get by~\dualeqref{eq:8923yty9y239ruyURadF} $|M_{ALG}\cap E^\ell_{k, \i}|\geq \frac{\eta |P|}{2KL}$,
and hence 
$$
\gamma:=\frac{|M_{ALG}^\i|}{|S_k|}\geq \frac{\eta |P|}{2KL |S_k|}\geq \frac{\eta |T|}{4K |S_k|}\geq \eta/8,
$$
where we let $M_{ALG}^\i=M_{ALG}\cap E^\ell_{k, \i}$ for convenience.
For every fixing $\lambda$ of $\Lambda_{<(\ell, k)}$ one has, 
$$
H(X_\i| M_{ALG}, \{\Lambda_{<(\ell, k)}=\lambda \wedge \j=\i\wedge Z=1\})\leq (1-\gamma) |S_k| H_2\left(1-\frac1{K(1-\gamma)}\right),
$$
since conditioned on $M_{ALG}$, $\lambda, \j=\i$ and the success event $Z=1$ there are exactly $(1-\gamma)|S_k|$ values of $y\in S_k\setminus M_{ALG}^\i$ such that $X_\i(y)=1$, and hence the conditional entropy of $X_\i$ is bounded by
\begin{equation*}
\begin{split}
\log_2 { |S_k\setminus M_{ALG}| \choose (1-\frac1{K}-\gamma) |S_k|}&=\log_2 { (1-\gamma) |S_k| \choose (1-\frac1{K}-\gamma) |S_k|}\\
&=\log_2 { (1-\gamma) |S_k| \choose (1-\frac1{K(1-\gamma)}) (1-\gamma) |S_k|}\\
&\leq (1-\gamma) |S_k| H_2\left(1-\frac1{K(1-\gamma)}\right),
\end{split}
\end{equation*}
where the last transition is by subadditivity of entropy. Recalling that $\gamma\geq \eta/8$ and $\eta>0$ is a small constant we bound the rhs above by
\begin{equation}\duallabel{eq:8g823ggfaib}
\begin{split}
(1-\gamma) |S_k| H_2\left(1-\frac1{K(1-\gamma)}\right)&\leq (1-\eta/8) |S_k| H_2\left(1-\frac1{K(1-\eta/8)}\right)\\
&\leq (1+\sqrt{\e}) \frac1{K}|T|\cdot  (1-\eta/8) H_2\left(1-\frac1{K(1-\eta/8)}\right),
\end{split}
\end{equation}
where in the second transition we also used the fact that by Lemma~\dualref{lm:size-bounds}, {\bf (2)}, we have $|S_k|\leq (1+\sqrt{\e})\frac1{K}|T|$.  At this point we also note that 
\begin{equation*}
\begin{split}
(1-\eta/8) H_2\left(1-\frac1{K(1-\eta/8)}\right)&=\frac1{K}\log_2 K+\frac1{K\ln 2}-\frac1{K} \log\frac1{1-8/\eta}+O(1/K^2)\\
&\leq H_2(1-1/K)-\frac1{K} \log\frac1{1-\eta/8}+O(1/K^2).
\end{split}
\end{equation*}
since $H_2(1-1/K)=\frac1{K}\log_2 K+\frac{1}{K \ln 2}+O(1/K^2)$ and $K$ is larger than a constant. Putting the above bounds together, we get, assuming that $K$ is larger than $1/\eta$ by a large constant factor,
 \begin{equation*}
\begin{split}
H(X_\i| M_{ALG}, \Lambda_{<(\ell, k)}, \{\j=\i\wedge Z=1\})\leq (1+\sqrt{\e})\frac1{K}|T|\cdot H_2(1-1/K)-\Omega(\eta/K) |T|.
\end{split}
\end{equation*}
for every $\i\in \mathcal{J}$, which by~\dualeqref{eq:93g9g7g89gFGYFGSA} implies for $\i\in \mathcal{J}$
\begin{equation}\duallabel{eq:mathcalj}
\begin{split}
H(X_\i| M_{ALG}, &\Lambda_{<(\ell, k)}, \{\j=\i\}, Z)\leq (1+\sqrt{\e})\frac1{K}|T|\cdot H_2(1-1/K)-\Omega\left(\frac{\eta}{K^2 L^2}\right) |P|\\
\end{split}
\end{equation}

Finally, for $\i\in \ac{\B}\setminus \mathcal{J}$ we have the bound 
\begin{equation}\duallabel{eq:9h329ty3utjbbfsf}
\begin{split}
H(X_\i| M_{ALG}, \Lambda_{<(\ell, k)}, \{\j=\i\}, Z)\leq (1+\sqrt{\e})\frac1{K}|T|\cdot H_2(1-1/K),
\end{split}
\end{equation}
since the number of nonzeros in $X_\i$ is exactly $\lceil (1-1/K) |S_k\rceil$. Putting ~\dualeqref{eq:mathcalj} and~\dualeqref{eq:9h329ty3utjbbfsf} together with~\dualeqref{eq:20iabibafihf0ASAS} and using~\dualeqref{eq:plus-z-198yugudfg}, we get
\begin{equation*}
\begin{split}
H(X| \Pi, \Lambda_{<(\ell, k)})&\leq \sum_{\i\in \ac{\B}} H(X_\i|  M_{ALG}, \Lambda_{<(\ell, k)}, \{\j=\i\})\\
&\leq \sum_{\i\in \ac{\B}} (1+H(X_\i|  M_{ALG}, \Lambda_{<(\ell, k)}, \{\j=\i\}, Z))\\
&\leq \sum_{\i\in \mathcal{J}} \left(H(X_\i|\Lambda_{<(\ell, k)})-\Omega\left(\frac{\eta}{K^2 L^2}\right) |P|\right)+\sum_{\i\in \ac{\B}\setminus \mathcal{J}} H(X_\i|\Lambda_{<(\ell, k)})\\
&=\sum_{\i\in \ac{\B}} H(X_\i|\Lambda_{<(\ell, k)})- |\mathcal{J}| \cdot  \Omega\left(\frac{\eta}{K^2 L^2}\right) |P|.\\
\end{split}
\end{equation*}
On the other hand, since $|S_k|\geq (1-\sqrt{\e})\frac1{K}|T|$ for all choices of $\Lambda_{<(\ell, k)}$ by Lemma~\dualref{lm:size-bounds}, {\bf (2)}, we get, since the nonzeros of $X_\i$ are a uniformly random set of size $\lceil (1-1/K)|S_k|\rceil$, that
$$
H(X|\Lambda_{<(\ell, k)})\geq (1-\sqrt{\e})\frac1{K}|T|\cdot |\ac{\B}|\cdot (1-o_N(1)) H_2(1-1/K).
$$

Substituting this into~\dualeqref{eq:20iabibafihf0ASAS}, we get
\begin{equation*}
\begin{split}
s=|\Pi|&\geq \Omega\left(\frac{\eta}{K^2L^2}\right) |\mathcal{J}| \cdot  |P|-O(\sqrt{\e})\frac1{K}|T|\cdot |\ac{\B}|\cdot H_2(1-1/K)\\
&\geq \Omega\left(\frac{\eta}{K^2L^2}\right) |\mathcal{J}| \cdot  |P| \text{~~~~~~~(since $\e<K^{-100K^2}$ by ~\dualref{p6},~\dualref{p5} and \dualref{p3})}\\
&\geq \Omega\left(\frac{\eta}{K^3L^3}\right) |\ac{\B}| \cdot  |P|\text{~~~~~~~(since $|\mathcal{J}|\geq |\ac{\B}|/(KL)$)}\\
&\geq \Omega_K(|\ac{\B}| \cdot  |P|).\\
\end{split}
\end{equation*}
Now note that $|\ac{\B}|\geq (1/2)|\B|$ since $|\Ext^\ell_k|\leq K$ and $n$ is sufficiently large as a function of $K$.  Finally, recall that by~\dualref{p0}
$$
N=m^n=n^{20n},
$$
and therefore
$$
|\B|\geq |\F|/(KL)=2^{\Omega(\e^2 n)}=N^{\Omega_\e(1/\log\log N)}.
$$ 
To summarize, since $|P|=O(L) N$ and $L$ is an absolute constant, we get a lower bound of
$$
s=\Omega_K(|\B| \cdot  |P|)=|P|^{1+\Omega(1/\log\log |P|)},
$$
as required. 

\end{proofof}

\section*{Acknowledgements}
This project has received funding from the European Research Council (ERC) under the European Union?s Horizon 2020 research and innovation programme (grant agreement No 759471).

\begin{appendix}
\section{Proof of Lemma~\ref{lm:alg-gen-online-intro}}\label{app:gen-online}

\begin{proofof}{Lemma~\ref{lm:alg-gen-online-intro}}
Fix $\ell\in [L]$, and let $G^\ell=(S^\ell, T^\ell, E^\ell)$ denote the $\ell$-th gadget graph.  Let $(E')^\ell$ denote a subset of $E^\ell$ that contains every edge independently with probability $C/(\e^2 n)$ for an absolute constant $C>0$. We show that with high probability over the choice of $(E')^\ell$ the edge set $(E')^\ell$ contains a matching of at least a $1-\e$ fraction of $S^\ell$ to $T^\ell\setminus T_*^\ell$. We drop the superscript $\ell$ to simplify notation.

Now note that for every subset $A\subseteq S$ and $B\subseteq T\setminus T_*$ such that $|A|\geq |B|-\e n$ one has 
\begin{equation}\label{eq:8328tg32t}
|E\cap (A \times (T\setminus (T_*\cup B)))|\geq (\e n/2)^2.
\end{equation}
Indeed, sort elements of $A=\{a_1,\ldots, a_r\}, r=|A|,$ so that $\pi(a_1)\leq \pi(a_2)\leq \ldots, \pi(a_r)$. We have for every $i=1,\ldots, r$ that $\pi(a_i)\geq n/2-r+i$. Since $a_i$ has an edge to every $j\in T$ such that $j\geq \pi(a_i)$, we have that the degree of $a_i$ in $E$ is lower bounded by $n/2+r-i$. At most $|T_*\cup B|=|T_*|+|B|\leq n/2+(r-\e n)$ of these edges go to $T_*\cup B$ (this is where we use that $|B|\leq |A|+\e n$), and therefore the $i$-th vertex in $A$ contributes at least $(n/2+r-i)-(n/2+(r-\e n))\geq \e n-i$. Thus, the first $\e n/2$ vertices in $A$ have degree at least $\e n/2$ outside of $T_*\cup B$, which proves~\eqref{eq:8328tg32t}. The probability that none of these edges are included in the sample $E'$ is bounded by 
$$
\left(1-\frac{C}{\e^2 n}\right)^{(\e n/2)^2}=\left(1-\frac{C}{\e^2 n}\right)^{\e^2 n^2/4}\leq \exp(-C n/4)\leq 2^{-4n}.
$$
Taking a union bound over all choices of $A\subseteq S, B\subseteq T\setminus T_*$ (at most $2^{2n}$ choices), we get that with high probability for every $A\subseteq S$, every $B\subseteq T\setminus T_*$ such that $|A|\geq |B|+\e n$ one has
$$
E'\cap (A \times (T\setminus (T_*\cup B)))\neq \emptyset.
$$
This precludes the existence of a vertex cover in $E'\cap (S\times (T\setminus T_*))$ of size smaller than $|S|-\e n$, and thus there exists a matching of all but $\e n$ vertices in $S$ to $T\setminus T_*$, as required. Combining these matchings over all gadgets gives a $1-O(\e)$-approximation to the maximum matching in $\wh{G}=(P, Q, \wh{E})$.
\end{proofof}


\section{Proofs omitted from Section~\ref{sec:toy-construction}}

\subsection{Proof of Lemma~\ref{lm:line-properties}}\label{app:line-properties}

\begin{proof}
We start by proving {\bf (1)}. Due to the assumption that $y\in T_k$ we have
\begin{equation*}
\begin{split}
\text{line}_j(y)&=\left\lbrace y'\in [m]^n: (y'-y)_s=0\text{~for all~}s\neq j \right\rbrace\\
&=\left\lbrace y'\in [m]^n: (y'-y)_s=0\text{~for all~}s\neq j\right\rbrace. \\
\end{split}
\end{equation*}
Since there are exactly $m$ possible values for $y'_j$ one has $|\text{line}_j(y)|=m$. Also  note that $j\not \in J_{<k}$, since $j\in \B_k$, $J_{<k}\in \B_{<k}$ and $\B_{<k}\cap \B_k=\emptyset$. Thus, every $y'\in \text{line}_j(y)$ coincides with $y$ on all coordinates $s\in J_{<k}$, so $y'_{j_s}/m\in \left[0, 1-\frac{1}{K-s}\right)$ for all $s\in \{0, 1, \ldots, k-1\}$ per~\eqref{eq:def-tk-allconstraints}, and hence we have $y'\in T_k$ and $\text{line}_j(y)\subseteq T_k$.

For {\bf (2}), we note that since $(K-s)|m$ for every $s\in [K/2]$ by~\ref{p0} and~\ref{p1}, there are exactly $m/(K-k)$ values for $y'_j$, namely $\{m/(K-k), m/(K-k)+1,\ldots, m-1\}$, that result in  $y'\not \in T_k^j$, by definition of $T_k^j$ (see \eqref{eq:def-tkj}).

For {\bf (3}), we recall that by~\eqref{eq:def-sk} 
$$
S_k=\left\{x\in T_k: \weight(x) \in \left[0, \frac{1}{K-k}\right)\cdot W \pmod{ W}\right\}.
$$
For every $x\in T_k$ one has
\begin{equation*}
\begin{split}
\text{line}_j(x)\cap S_k&=\left\lbrace x'\in [m]^n: x'_{-j}=x_{-j}\text{~and~}\weight(x') \in \left[0, \frac{1}{K-k}\right)\cdot W \pmod{ W} \right\rbrace\\
\end{split}
\end{equation*}
Write $x'=(x'_{-j} x'_j)$, where $x'_{-j}\in [m]^{[n]\setminus \{j\}}$. Note that by definition of $\weight(x')$ (Definition~\ref{def:weight})
$$
\weight(x')=\weight(x'_{-j})+x'_j.
$$
We thus get
\begin{equation*}
\begin{split}
|\text{line}_j(x)\cap S_k|&=\left|\left\lbrace x'\in [m]^n: x'_{-j}=x_{-j}\text{~and~}\weight(x') \in \left[0, \frac{1}{K-k}\right)\cdot W \pmod{ W} \right\rbrace\right|\\
&=\left|\left\lbrace x'\in [m]^n: x'_{-j}=x_{-j}\text{~and~}\weight(x'_{-j})+x'_j \in \left[0, \frac{1}{K-k}\right)\cdot W \pmod{ W} \right\rbrace\right|\\
&=\frac1{K-k}\left|\left\lbrace x'\in [m]^n: x'_{-j}=x_{-j}\right\rbrace\right|\\
&=\frac1{K-k}\left|\text{line}_j(x)\right|,
\end{split}
\end{equation*}
where the last equality uses the fact that since $W \mid m$ and $(K-k) \mid W$ by~\ref{p0} and~\ref{p1}, exactly a $\frac1{K-k}$ fraction of settings of $x_j\in [m]$ result in 
$$
\weight(x')=\weight(x'_{-j})+x'_j \in \left[0, \frac{1}{K-k}\right)\cdot W \pmod{ W}.
$$
This establishes {\bf (3)}.

For {\bf (4)}, we first recall that by~\eqref{eq:def-tkj} 
$$
S_k^j=\left\{x\in S_k: x_j/m\in \left[0, 1-\frac1{K-k}\right)\right\}.
$$
Thus for every $x\in S_k^j$ one has
\begin{equation*}
\begin{split}
\left|\text{line}_j(x)\cap S_k^j\right|&=\left|\left\lbrace x'\in [m]^n: x'_{-j}=x_{-j}\text{~and~}\weight(x') \in \left[0, \frac{1}{K-k}\right)\cdot W \pmod{ W}\right.\right.\\
&\text{~~~~~~~~~~~~~~~~~~~~~~~~~~~and~}\\
&\left.\left.\text{~~~~~~~~~~~~~~~~~~~~~~~~~~}x'_j/m\in \left[0, 1-\frac1{K-k}\right) \right\rbrace\right|\\
&=\frac1{K-k}\left|\left\lbrace x'\in [m]^n: x'_{-j}=x_{-j}\text{~and~}x'_j/m\in \left[0, 1-\frac1{K-k}\right)\right\rbrace\right|
\end{split}
\end{equation*}
where the last equality uses the fact that exactly $\frac1{K-k}$ fraction of settings of $x'_j\in \left\{0, 1,\ldots, (1-\frac1{K-k})m-1\right\}$ lead to 
$$
\weight(x')=\weight(x'_{-j})+x'_j \in \left[0, \frac{1}{K-k}\right)\cdot W \pmod{ W}.
$$
since $(K-k) \mid m$ and $W \mid m/(K-k)$ by~\ref{p0} and~\ref{p1}. We now note that since $K-k\mid m$, we get
$$
\left|\left\lbrace x'\in [m]^n: x'_{-j}=x_{-j}\text{~and~}x'_j/m\in \left[0, 1-\frac1{K-k}\right)\right\rbrace\right|=\left(1-\frac1{K-k}\right)|\text{line}_j(x)|.
$$
Putting the two bounds together yields the result. \end{proof}

\subsection{Proof of Lemma~\ref{lm:size-bounds}}\label{app:size-bounds}

\begin{proofof}{Lemma~\ref{lm:size-bounds}}
We start by proving {\bf (1)}:
\begin{equation*}
\begin{split}
|T_k|&=|T_0|\cdot \prob_{y\sim UNIF([m]^n)}\left[y_{j_s/m}\in \left[0, 1-\frac{1}{K-s}\right)\text{~for every~}s=0,\ldots, k-1\right]\\
&=|T_0|\cdot \prod_{s=0}^{k-1} \prob_{y\sim UNIF([m]^n)}\left[y_{j_s}/m\in \left[0, 1-\frac{1}{K-s}\right)\right]\\
&=|T_0|\cdot \prod_{s=0}^{k-1} \left(1-\frac{1}{K-s}\right)\text{~~~~~(since $K-s$ divides $m$ for all $s\in [K/2]$ by assumption)}\\
&=|T_0|\cdot \prod_{s=0}^{k-1} \frac{K-s-1}{K-s}\\
&=|T_0|\cdot  \frac{K-(k-1)-1}{K}\\
&=|T_0|\cdot  (1-k/K),\\
\end{split}
\end{equation*}
as required.

We now prove {\bf (2)}. Pick any coordinate $r \in \B_k$, and recall that $T_k$ does not depend on $r$, i.e. for every $x_{-r}\in [m]^{[n]\setminus \{r\}}$ such that $(x_r, x_{-r}) \in T_k$ for some $x_r\in [m]$ one has $(x_r, x_{-r})$ for every $x_r\in [m]$. This is because by~\eqref{eq:def-tk-allconstraints} $T_k$ only depends on coordinates in $\B_{<k}$. This means that
\begin{equation}\label{eq:923hgefC}
\begin{split}
|S_k|&=\prob_{x\sim UNIF(T_k)}\left[\sum_{s\in [n]} x_s \in \left[0, \frac{1}{K-k}\right)\cdot W \pmod{ W} \right]\\
&=\expect_{x_{-r}\sim UNIF(T_k)}\left[\prob_{x_r\sim UNIF([m])}\left[\sum_{s\in [n]} x_s \in \left[0, \frac{1}{K-k}\right)\cdot W \pmod{ W} \right]\right],
\end{split}
\end{equation}
where we used the fact that $T_k$ is independent of $r$ to conclude that $x_r\sim UNIF([m])$ in the inner probability regardless of the choice of $x_{-r}$. For the inner probability we get
\begin{equation*}
\begin{split}
&\prob_{x_r\sim UNIF([m])}\left[\sum_{s\in [n]} x_s \in \left[0, \frac{1}{K-k}\right)\cdot W \pmod{ W} \right]\\
&=\prob_{x_r\sim UNIF([m])}\left[x_r \in \left[\left[0, \frac{1}{K-k}\right)\cdot W -\sum_{s\in [n]\setminus \{r\}} x_s \right) \pmod{ W} \right]\\
&=\frac{1}{K-k},
\end{split}
\end{equation*}
where the last line uses the assumption that $K-k \mid W$ and $W \mid m$. Substituting this into~\eqref{eq:923hgefC}, we get
$|S_k|=\frac{1}{K-k}|T_k|$. Since by {\bf (1)} one has $|T_k|=|T_0|\cdot  (1-k/K)$, this implies that 
\begin{equation*}
\begin{split}
|S_k|&=\frac{1}{K-k}|T_k|=\frac{1}{K-k}\cdot (1-k/K)|T_0|=\frac1{K}|T_0|,
\end{split}
\end{equation*}
as required.

\end{proofof}


\section{Proofs omitted from Section~\dualref{sec:main-construction}}\label{app:size-bounds-full}

\subsection{Construction of the set $\F$}\duallabel{app:F-construction}

\begin{lemma}\label{lem:code}
For any $\e\in (0, 1)$, any integers  $m\geq 1$ and $w=(\e/2)m$, there exists a collection $\mathcal F_{m, w, \e} \subset \bool^m$  of vectors of Hamming weight $w$ with $\log |\mathcal{F}_{m, w,\e}| = \Omega(\e^2 m)$ such that for all $\u\neq \u'\in \mathcal F_{w,\e}$, $(\u, \u') < \e w$.
\end{lemma}
\begin{proof}
The proof is via the probabilistic method. Partition $[m]$ into $w$ subsets $I_1, \ldots, I_w$, with $|I_s|=m/w$ for $s=1,\ldots, w$. We pick $\u_1,\ldots,\u_N$ independently as follows. For every $j=1,\ldots, N$, the vector $\u_j$ includes exactly one random element of $I_s$ for each $s=1,\ldots, w$. This ensures that the Hamming weight of each $\u_j$ is exactly $w$.  

We now show that the vectors have small intersection size with high probability. Fix $i\neq j\in[N]$. Imagine $\u_i$ being fixed and picking the $w$ elements of $\u_j$ one by one. Let $X_s$ denote the indicator random variable for the event that the $s$th element of $\u_j$ (picked from $I_s$) is also in $S_i$. Then $(\u_i, \u_j) = \sum_{s=1}^w X_k$, and we set $\mu:= \expect[(\u_i, \u_j)]$. Note that $\mu=(w/m)\cdot w$, since for every $s=1,\ldots, w$ the vector $\u_i$ has exactly one nonzero coordinate in $I_s$, and the probability that $\u_j$ chooses the same coordinate is $1/|I_s|=w/m$. 
We have $\prob[(\u_i, \u_j) \ge \e w] = \prob[\sum_{s=1}^w X_s \ge 2\mu]$  The random variables $X_s$ are independent and thus the Chernoff bound yields
$$
\prob[(\u_i, \u_j) \geq  2\mu) \le \left(\frac{e}{4}\right)^\mu \le e^{-\Omega((w/m) w)}\leq e^{-c\e^2 m}
$$
for a constant $c>0$. 
Setting $N = 2^{(\ln_2 e) c \e^2 m/2}$ so that ${N \choose 2}<N^2=2^{(\ln_2 e) c \e^2 m}=e^{c \e^2 m}$, by a union bound with positive probability $|\u_i\cap \u_j| < \e w$ for all $i\neq j$, simultaneously, as desired. Note for this choice of $N$, we have $\log|\mathcal \F_{m, w,\e}| = \log N = \Theta(\e^2 m)$.
\end{proof}

\subsection{Proofs of Lemma~\dualref{lm:rect-size} and Lemma~\dualref{lm:rect-subspace-size}}

\begin{claim}\duallabel{cl:divisibility}
For every $x\in [m]^n$, every $\j\in \F$, every pair of integers $c, d$, $c\leq d$ such that $(W/w)\mid (d-c)$, if $\lambda$ divides $W/w$, 
\begin{equation*}
\begin{split}
&\left|\left\{c\leq t<d:  \weight(x+t\cdot \j) \pmod{W}\in [0, 1/\lambda)\cdot W\right\}\right|=\frac1{\lambda}\cdot (d-c).
\end{split}
\end{equation*}
\end{claim}
\begin{proof} First, we write 
$$
t=u\cdot (W/w)+v,
$$ where $u=\lfloor t/(W/w)\rfloor$ and $v=t\pmod{W/w}$, so that 
\begin{equation}\duallabel{eq:394hgh23NCKFkd}
\begin{split}
\weight(x+t\cdot \j) \pmod{W}&=(\weight(x)+t\cdot w)\pmod{W}\\
&=(\weight(x)+(u\cdot (W/w)+v)\cdot w)\pmod{W}\\
&=(\weight(x) \pmod{W}+v\cdot w)\pmod{W}.\\
\end{split}
\end{equation}

Similarly, write
\begin{equation*}
\begin{split}
c&=f\cdot (W/w)+e\\
d&=g\cdot (W/w)+e,
\end{split}
\end{equation*} 
where $f=\lfloor c/(W/w)\rfloor$, $g=\lfloor c/(W/w)\rfloor$ and $e=c\pmod{W/w}=d\pmod{W/w}$ (the last equality is justified by the assumption that $ (W/w) \mid (d-c)$). With this notation in place, using~\dualeqref{eq:394hgh23NCKFkd}, we can express the set in question conveniently as
\begin{equation}
\begin{split}
&\left\{c\leq t<d:  \weight(x+t\cdot \j) \pmod{W}\in [0, 1/\lambda)\cdot W\right\}\\
&=\left\{f\cdot (W/w)+e\leq t<g\cdot (W/w)+e:\right.\\
&\left. \text{~~~~~~~~~~~~~~~~~~~~~~~~~~~~~~~~~~~~}(\weight(x) \pmod{W}+v\cdot w) \pmod{W}\in [0, 1/\lambda)\cdot W\right\}\\
=&\left\{f\cdot (W/w)\leq u\cdot (W/w)+v-e<g\cdot (W/w):\right.\\
&\left. \text{~~~~~~~~~~~~~~~~~~~~~~~~~~~~~~~~~~~~}(\weight(x) \pmod{W}+v\cdot w) \pmod{W}\in [0, 1/\lambda)\cdot W\right\}\\
\end{split}
\end{equation}
Note that for every $u$ such that 
\begin{equation}\duallabel{eq:u-bound-wh93hg3g}
f+1\leq u<g
\end{equation}
one has
\begin{equation}\duallabel{eq:923yth9dhfHF}
f\cdot (W/w)\leq u\cdot (W/w)+v-e<g\cdot (W/w)
\end{equation}
for all $v\in [W/w]=\{0, 1,\ldots, W/w-1\}$, since $e\in [W/w]=\{0, 1,\ldots, W/w-1\}$ by definition of $e$. Now since $\lambda \mid W/w$ by assumption, using~\dualeqref{eq:394hgh23NCKFkd} we get that for every $u$ that satisfies~\dualeqref{eq:u-bound-wh93hg3g} exactly $\frac1{\lambda}\cdot (W/w)$ choices for $v\in [W/w]$ lead to
\begin{equation}\duallabel{eq:8g3tgBABFBJFxmL}
(\weight(x) \pmod{W}+v\cdot w)\pmod{W}\in [0, 1/\lambda)\cdot W.
\end{equation}

It remains to note that for $u=f$ the condition in~\dualeqref{eq:923yth9dhfHF} is satisfied if and only if $e\leq v< W/w$, and for 
$u=g$ the condition in~\dualeqref{eq:923yth9dhfHF} is satisfied if and only if $0\leq v< e$ . Since $\{e, e+1,\ldots, W/w-1\}\cup \{0,1,\ldots, e-1\}=[W/w]$, we again get that overall exactly $\frac1{\lambda}\cdot (W/w)$ choices of $v$ satisfy~\dualeqref{eq:8g3tgBABFBJFxmL}.  This establishes the claim.
\end{proof}

\begin{lemma}[Intersection of a cube with a subspace]\duallabel{lm:cube-subspace-size}
For every $\I, \J\subset \F, |\I|, |\J|\leq K^2$, every $\a\in \deltagridInt^\J$, if $R=\rect(\J, \a)$, the following conditions hold.

\begin{description}
\item[(1)] For every $x\in [m]^n\setminus B$ one has
$$
(1-\e^{2/3}) \cdot \Delta^{|\I\cap \J|}\cdot G \leq \left|\subspace_\I(x)\cap R\right|\leq (1+\e^{2/3})\cdot \Delta^{|\I\cap \J|}\cdot G,
$$
where $G=(M/w)^{|\I|}$. 

\item[(2)] For every positive integer $\lambda\leq K$, if 
$$
R'=\{x\in R: \weight(x)\pmod{W}\in [0, 1/\lambda)\cdot W\},
$$
one has for every $x\in [m]^n\setminus B$
$$
(1-\e^{2/3})\cdot \frac1{\lambda}\cdot \Delta^{|\I\cap \J|}\cdot G \leq \left|\subspace_\I(x)\cap R'\right|\leq (1+\e^{2/3})\cdot \frac1{\lambda}\cdot \Delta^{|\I\cap \J|}\cdot G,
$$
where $G=(M/w)^{|\I|}$. 
\end{description}
\end{lemma}
\begin{proof}
For $t\in \mathbb{Z}^\I$ with $||t||_\infty\leq 2M/w$ consider
\begin{equation}\duallabel{eq:0923hg9hg94FEGF}
x'=x+\sum_{\i\in \I} t_\i \cdot \i,
\end{equation}
and note that every vertex in $\subspace_\I(x)$ can be written in this form by Definition~\dualref{def:subspace}. Since $x$ is not a boundary point, i.e. $x\in [m]^n\setminus B$ (see Definition~\dualref{def:boundary}), one has $x'\in [m]^n$ for every such $t$. Thus, it suffices to bound the number of choices of such coefficients $t$ that result in both $\block_\I(x')=\block_\I(x)$ and $x'\in R$ to prove {\bf (1)} and similarly bound the number of choices of $t$ that result in both $\block_\I(x')=\block_\I(x)$ and $x'\in R'$ to prove {\bf (2)}. We do this in what follows.

\paragraph{Notation and basic properties of $x'$.} We start by noting some basic properties of $x'$. First note that for every $\k\in \I\cap \J$
\begin{equation}\duallabel{eq:3hg9034hgHDH}
\begin{split}
\left|\langle x', \k\rangle-(\langle x, \k\rangle+t_\k\cdot w)\right|&=\left|\langle x+\sum_{\i\in \I} t_\i\cdot \i, \k\rangle-(\langle x, \k\rangle+t_\k\cdot w)\right|\\
&=\left|\sum_{\i\in \I\setminus \{\k\}} t_\i\cdot \langle \i, \k\rangle\right|\\
&\leq \sum_{\i\in \I\setminus \{\k\}} t_\i\cdot |\langle \i, \k\rangle|\\
&\leq \e |\I|\cdot ||t||_\infty\cdot w\\
&\leq (2\e |\I|)\cdot M.\\
\end{split}
\end{equation}
and for every $\k\in \J\setminus \I$ 
\begin{equation}\duallabel{eq:3hg9034hgHDH23ihf}
\begin{split}
\left|\langle x', \k\rangle-\langle x, \k\rangle\right|&=\left|\langle x+\sum_{\i\in \I} t_\i\cdot \i, \k\rangle-\langle x, \k\rangle\right|\\
&=\left|\sum_{\i\in \I} t_\i\cdot \langle \i, \k\rangle\right|\\
&\leq \sum_{\i\in \I} t_\i\cdot |\langle \i, \k\rangle|\\
&\leq \e |\I|\cdot ||t||_\infty\cdot w\\
&\leq (2 \e |\I|) M.\\
\end{split}
\end{equation}

For every $\k\in \I$ define 
\begin{equation}\duallabel{eq:qk-def-249gt9243}
q_\k=\left\lfloor \frac{1}{w}\left(\langle x, \k\rangle \pmod{M}\right)\right\rfloor
\end{equation}
for convenience, and note that 
\begin{equation}\duallabel{eq:9024ht92htas4D}
0\leq \langle x, \k\rangle \pmod{M}-w\left\lfloor \frac{1}{w}\left(\langle x, \k\rangle \pmod{M}\right)\right\rfloor< w.
\end{equation}

Fix $\eta\in (0, 1/10)$, and assume that $\eta$ satisfies
\begin{equation}\duallabel{eq:92cnnjdbfsfs}
\eta>2w/M\text{~~and~~}\eta\geq 5 \e |\I|.
\end{equation}

\paragraph{Lower bound.}  We now prove that any  $t$ such that
\begin{equation}\duallabel{eq:t-range}
-q_\k+\frac{M}{w}\cdot (\a_\k+\eta) \leq t_\k<-q_\k+\frac{M}{w}\cdot (\a_\k+\Delta-\eta)
\end{equation}
for all $\k\in \I\cap \J$ and 
\begin{equation}\duallabel{eq:t-range-unconstrained}
-q_\k+\frac{M}{w}\cdot \eta \leq t_\k<-q_\k+\frac{M}{w}\cdot (1-\eta)
\end{equation}
for $\k\in \I\setminus \J$ satisfies 
\begin{description}
\item[(a)] $x'=x+\sum_{\i\in \I} t_\i \cdot \i\in \subspace_\I(x)\cap R$ as long as $\eta$ is not too small (recall that $q_\k$ is defined in~\dualeqref{eq:qk-def-249gt9243});
\item[(b)] $\block_\I(x')=\block_\I(x)$.
\end{description}
The two bounds above show that any $t$ that satisfies both~\dualeqref{eq:t-range} and~\dualeqref{eq:t-range-unconstrained} leads to $x'\in \subspace_\I(x)\cap R$. Counting the number of settings of $t$ that satisfy these constraints, we get
\begin{equation}\duallabel{eq:lb-8gf8gfsf}
\begin{split}
|\subspace_\I(x)\cap R|&\geq (M/w)^{|\I|} (1-4\eta)^{|\I\setminus \J|} (\Delta-4\eta)^{|\I\cap \J|}\\
&\geq (M/w)^{|\I|} \Delta^{|\I\cap \J|}(1-4\eta/\Delta)^{|\I|},
\end{split}
\end{equation}
where we used the fact that since $\eta>2w/M$ by assumption, we have $\lceil \eta M/w\rceil\leq 2\eta M/w$.

We start with {\bf (a)}. We verify that the dot product of every $x'$ as above with $\k\in \I\cap \J$ satisfies
\begin{equation}\duallabel{eq:8t8g38tg8gt893gt89ga8g8fg8gf}
\langle x',\k\rangle \pmod{M}\in  [\a_k, \a_k+\Delta)\cdot M.
\end{equation} 
First, for $k\in \I\cap \J$, using~\dualeqref{eq:3hg9034hgHDH}, it suffices to show that 
$$
\langle x, \k\rangle \pmod{M}+t_\k\cdot w\in [\a_k, \a_k+\Delta)\cdot M,
$$
as well as show that the quantity on the lhs above does not fall too close to the boundary of the interval on the rhs (to ensure that the error terms in~\dualeqref{eq:3hg9034hgHDH} can be absorbed). 

We have using the upper bound on $t_\k$ from~\dualeqref{eq:t-range} as well as~\dualeqref{eq:9024ht92htas4D}
\begin{equation*}
\begin{split}
\langle x, \k\rangle\pmod{M}+t_\k\cdot w&\leq \langle x, \k\rangle \pmod{M}-w\cdot \q_\k+(\a_\k+\Delta-\eta)\cdot M\text{~~~~~~~~~~~(by~\dualeqref{eq:t-range})}\\
&= (\langle x, \k\rangle \pmod{M}-w\cdot \q_\k)+(\a_\k+\Delta-\eta)\cdot M\\
&\leq  w+(\a_\k+\Delta-\eta)\cdot M\text{~~~~~~~~~~~~~~~~~~~~~~~~~~~~~~~~~~~~~~~~~~~~~~~~~~~~~~~(by~\dualeqref{eq:9024ht92htas4D})}\\
&= (\a_\k+\Delta+\frac{w}{M}-\eta)\cdot M.\\
\end{split}
\end{equation*}
We also have using the lower bound on $t_\k$ from~\dualeqref{eq:t-range}as well as~\dualeqref{eq:9024ht92htas4D}
\begin{equation*}
\begin{split}
\langle x, \k\rangle\pmod{M}+t_\k\cdot w&\geq  \langle x, \k\rangle \pmod{M}-w\cdot \q_\k+(\a_\k+\eta)\cdot M\text{~~~~~~~~~~~~~~~~~~~~~~~~~~~(by~\dualeqref{eq:t-range})}\\
&= (\langle x, \k\rangle \pmod{M}-w\cdot \q_\k)+(\a_\k+\eta)\cdot M\\
&\geq (\a_\k+\eta)\cdot M.\text{~~~~~~~~~~~~~~~~~~~~~~~~~~~~~~~~~~~~~~~~~~~~~~~~~~~~~~~~~~~~~~~~~~~~~~~~~~(by~\dualeqref{eq:9024ht92htas4D})}\\
\end{split}
\end{equation*}
The two bounds together imply that for all $t$ satisfying ~\dualeqref{eq:t-range} one has
$$
(\a_\k+\eta)\cdot M \leq (\langle x, \k\rangle +t_\k\cdot w)\pmod{M}\leq (\a_\k+\Delta+\frac{w}{M}-\eta)\cdot M.
$$
We also note that $\q_\k\in [0, M/w)$, implying that one has $|t_\k|\leq 2M/w$ for all $\k\in \I$ for every $t$ satisfying~\dualeqref{eq:t-range} and~\dualeqref{eq:t-range-unconstrained} as long as $\eta<1$. Combining this with~\dualeqref{eq:3hg9034hgHDH}, we get that for every $t$ satisfying~\dualeqref{eq:t-range} the point $x'=x+\sum_{\i\in \I} t_\i\cdot \i$ satisfies ~\dualeqref{eq:8t8g38tg8gt893gt89ga8g8fg8gf} for $t\in \I\cap \J$ by~\dualeqref{eq:92cnnjdbfsfs}.
Similarly, we get using~\dualeqref{eq:3hg9034hgHDH23ihf} that for every $t$ satisfying~\dualeqref{eq:t-range} and~\dualeqref{eq:t-range-unconstrained} the point $x'=x+\sum_{\i\in \I} t_\i\cdot \i$ satisfies ~\dualeqref{eq:8t8g38tg8gt893gt89ga8g8fg8gf} for $t\in \J\setminus \I$ as long as 
~\dualeqref{eq:92cnnjdbfsfs} holds.

We now establish {\bf (b)}. Note that for every $t$ satisfying~\dualeqref{eq:t-range} and~\dualeqref{eq:t-range-unconstrained} and every $\k\in \I$ one has, using~\dualeqref{eq:3hg9034hgHDH} and~\dualeqref{eq:3hg9034hgHDH23ihf} that
\begin{equation*}
\begin{split}
-q_\k\cdot w+\eta M-(2\e |\I|)M\leq \langle x', \k\rangle-\langle x, k\rangle&\leq -q_\k\cdot w+(1-\eta)M+(2\e \I)M.
\end{split}
\end{equation*}
Indeed, this follows directly from~\dualeqref{eq:t-range-unconstrained} for $\k\in \I\setminus \J$, and follows from~\dualeqref{eq:t-range} for $\k\in \I\cap \J$ by recalling that $\a_\k\in \deltagridInt\subseteq [0, 1-\Delta]$. Rearranging the terms and using~\dualeqref{eq:92cnnjdbfsfs} , we get
\begin{equation*}
\begin{split}
\langle x, \k\rangle-q_\k\cdot w< \langle x', \k\rangle&< \langle x, \k\rangle-q_\k\cdot w+M.
\end{split}
\end{equation*}
By definition of $q_\k$ (see~\dualeqref{eq:qk-def-249gt9243}) we have $0\leq q_\k\cdot w< (\langle x, \k\rangle \pmod{M})$. Thus, the above implies
\begin{equation*}
\begin{split}
\left\lfloor \frac1{M}\langle x', \k\rangle\right \rfloor&=\left\lfloor \frac1{M}\langle x, \k\rangle\right \rfloor
\end{split}
\end{equation*}
for all $\k\in \I$, and therefore $\block_\I(x')=\block_\I(x)$. Since $||t||_\infty\leq 2M/w$, we get that $x'\in \subspace_\I(x)$.

\paragraph{Upper bound.} We now upper bound the number of choices for $t$ such that $x'$ as in~\dualeqref{eq:0923hg9hg94FEGF} belongs to $\subspace_\I(x)\cap R$. We first note that every such $t$ that leads to $x'\in \subspace_\I(x)\cap R$ must satisfy
\begin{description}
\item[(a)] for all $\k\in \I$
\begin{equation}\duallabel{eq:t-box}
-\q_\k-\eta \frac{M}{w}\leq t_\k\leq -q_\k+(1+\eta) \frac{M}{w}
\end{equation}

\item[(b)] for all $k\in \I\cap \J$
\begin{equation}\duallabel{eq:t-range-ex-ij}
t_\k\not \in \left[-q_\k+\eta \frac{M}{w}, -q_\k+\frac{M}{w}\cdot (\a_\k-\eta)\right)\cup \left(-q_\k+\frac{M}{w}\cdot (\a_\k+\Delta+\eta), -q_\k+(1-\eta)\frac{M}{w}\right]
\end{equation}
\end{description}

We start by proving {\bf (a)}. Suppose that~\dualeqref{eq:t-box} is not true for some $\k\in \I$. We assume that $t_\k\leq -\q_\k-\eta \frac{M}{w}$ (the other case is analogous). Then one has
\begin{equation*}
\begin{split}
\langle x', \k\rangle&\leq \langle x, \k\rangle-\q_\k-\eta M+\sum_{\i\in \I\setminus \{\k\}} |t_\i|\cdot \langle \i, \k\rangle\\
&\leq (w/M-\eta+(2\e |\I|)) M,
\end{split}
\end{equation*}
where we upper bounded the difference of the first two terms on the rhs by $w$ as per~\dualeqref{eq:9024ht92htas4D}, and used the fact that $||t||_\infty\leq 2M/w$ to upper bound the last term.  Since $\eta>2w/M$ and $\eta\geq 5\e |\I|$,  we get
$\lfloor \langle x', \k\rangle/M\rfloor<\lfloor \langle x, \k\rangle/M\rfloor$, and hence $\block_\I(x')\neq \block_\I(x)$.

We now prove {\bf (b)}. We consider two cases. 

\noindent{\bf Case 1:} Suppose that $-\q_\k+\eta \frac{M}{w}\leq t_\k<-\q_\k+\frac{M}{w}\cdot (\a_\k-\eta)$. Since
\begin{equation*}
\begin{split}
\langle x, \k\rangle \pmod{M}-w\cdot \q_\k+(\a_\k-\eta)\cdot M&\leq (\a_\k+\frac{w}{M}-\eta)\cdot M\text{~~~~~~~~~~~~~(by~\dualeqref{eq:9024ht92htas4D})}\\
&\text{~and}\\
\langle x, \k\rangle \pmod{M}-w\cdot \q_\k+\eta\cdot M&\geq \eta\cdot M,\\
\end{split}
\end{equation*}
we get, using~\dualeqref{eq:3hg9034hgHDH} and ~\dualeqref{eq:92cnnjdbfsfs} together with the assumption that $\eta<1/10$ and the fact that $\Delta\leq 1/2$ by~\dualref{p3}, that $\langle x', \k\rangle\pmod{M}\not \in [\a_\k, \a_\k+\Delta)\cdot M$.\

\noindent{\bf Case 2:} Suppose that $-q_\k+\frac{M}{w}\cdot (\a_\k+\Delta+\eta)< t_\k \leq -q_\k+(1-\eta) \frac{M}{w}$. Then we have
\begin{equation*}
\begin{split}
\langle x, \k\rangle \pmod{M}-w\cdot \q_\k+(\a_\k+\Delta+\eta)\cdot M&\geq (\a_\k+\Delta+\eta)M\text{~~~~~~~~~~~~~(by~\dualeqref{eq:9024ht92htas4D})}\\
&\text{~and}\\
\langle x, \k\rangle \pmod{M}-w\cdot \q_\k+(1+\eta)\cdot M&\leq (1+w/M-\eta)\cdot M,
\end{split}
\end{equation*}
and hence using~\dualeqref{eq:3hg9034hgHDH} and~\dualeqref{eq:92cnnjdbfsfs} together with the assumption that $\eta<1/10$ we get $\langle x', \k\rangle\not \in [\a_\k, \a_\k+\Delta)\cdot M$.

Counting the number of settings for $t$ that satisfy both {\bf (a)} and {\bf (b)}, we get
\begin{equation}\duallabel{eq:lb-8gf8gfsf}
|\subspace_\I(x)\cap R|\leq (M/w)^{|\I|}(1+4\eta)^{|\I\setminus \J|}(\Delta+4\eta)^{|\I\cap \J|}\leq (M/w)^{|\I|} \Delta^{|\I\cap \J|}(1+4\eta/\Delta)^{|\I|}.
\end{equation}

\paragraph{Gathering bounds and setting the parameter $\eta$.}
We now let
\begin{equation}\duallabel{eq:setting-of-eta}
\eta=\e\cdot K^3,
\end{equation} 
so that
\begin{equation}\duallabel{eq:3904h9h9thewgFF}
\begin{split}
(1+4\eta/\Delta)^{|\I|}&\leq (1+4\eta/\Delta)^{K^2}\text{~~~~~~~~~~~~~~~~~~~~~(since $|\I|\leq K^2$)}\\
&\leq (1+4\e K^3/\Delta)^{K^2}\text{~~~~~~~~~~~~~~~~(by setting of $\eta$)}\\
&\leq (1+4\e K^3 \cdot K^K)^{K^2}\text{~~~~~~~~~~~~(since $\Delta\geq K^{-K}$ by~\dualref{p3})}\\
&\leq 1+8\e K^{K+5}\text{~~~~~~~~~~~~~~~~~~~~~~~~(since $4\e K^5 \cdot K^K<1$ by~\dualref{p6}, ~\dualref{p5} and~\dualref{p3})}\\
&\leq 1+\e^{2/3} (8\e^{1/3} K^{K+5})\\
&\leq 1+\e^{2/3}/3\text{~~~~~~~~~~~~~~~~~~~~~~~~~~~~(by~\dualref{p6})}.
\end{split}
\end{equation}
The last transition uses the fact that 
\begin{equation*}
\begin{split}
8\e^{1/3} K^{K+5}&\leq 8\delta^{2/3} K^{K+5}\text{~~~~~~~~~~~~~~~~~~~~~~~~~~~~(by~\dualref{p5})}\\
&\leq 8K^{-50K^2} K^{K+5}\text{~~~~~~~~~~~~~~~~~~~~~(by~\dualref{p3} and~\dualref{p5})}\\
&\leq 1/3,
\end{split}
\end{equation*}
since $K$ is larger than an absolute constant.
Similarly we have $(1-4\eta/\Delta)^{|\I|}\geq 1-\e^{2/3}/3$. We also verify that our setting of $\eta$ in~\dualeqref{eq:setting-of-eta} satisfies conditions in~\dualeqref{eq:92cnnjdbfsfs}. First, we have $\eta=\e K^3\geq 5\e |\I|$ since $|\I|\leq K^2$ by assumption and $K$ is larger than an absolute constant. We also have $\eta>2w/M$ by~\dualref{p7}. This completes the proof of {\bf (1)}.

We now prove {\bf (2)}, the second bound of the lemma. We consider two cases, depending on whether $\I\cap \J\neq \emptyset$.

\paragraph{Case 1: $\I\cap \J\neq \emptyset$.} Consider any choice of $t$ that satisfies~\dualeqref{eq:t-range} and~\dualeqref{eq:t-range-unconstrained}, which by our analysis above leads to $x'\in \subspace_\I(x)$. Now select $\k_*\in \I\cap \J$ arbitrarily, and let $t_{\k_*}$ vary in the range
\begin{equation}\duallabel{eq:t-constraints}
-q_\k+\frac{M}{w}\cdot \a_\k \leq t_{\k_*}<-q_\k+\frac{M}{w}\cdot (\a_\k+\Delta).
\end{equation}
We have per~\dualeqref{eq:0923hg9hg94FEGF} together with the fact that $|\u|=w$ for all $\u\in \F$ that  
$$
\weight(x')=\weight(x)+\sum_{\k\in \I} w\cdot t_\k.
$$ 
Letting $z=x'+\sum_{\k\in \I\setminus \{\k_*\}} \k\cdot t_\k$, we get by Claim~\dualref{cl:divisibility}
\begin{equation*}
\begin{split}
&|\{-q_{\k_*}+\a_{\k_*}\cdot \frac{M}{w}\leq t_{\k_*}<-q_{\k_*}+(\a_{\k_*}+\Delta)\cdot \frac{M}{w}:\\
&~~~~~~~~~~~~~~~~~~~~~~~~~~~~~~~~~~~\left.\left. \weight(z+t\cdot \k_*) \pmod{W}\in [0, 1/\lambda)\cdot W\right\}\right|=\frac1{\lambda}\cdot \Delta\cdot \frac{M}{W}.
\end{split}
\end{equation*}
Note that the preconditions of Claim~\dualref{cl:divisibility} are satisfied since $\frac{W}{w} \mid \Delta \cdot \frac{M}{w}$  by~\dualref{p2} and $\lambda \mid \frac{W}{w}$ since $\lambda \leq K$ is a positive integer by assumption of the lemma as well as~\dualref{p1}.

Since our analysis above shows that every $t$ that satisfies~\dualeqref{eq:t-range} and~\dualref{eq:t-range-unconstrained} leads to $x' \in R$, we get
\begin{equation}\duallabel{eq:823g8gasgfu}
\begin{split}
\left|\subspace_\I(x)\cap R'\right|\geq \left(\frac1{\lambda}\Delta-4\eta\right)\cdot (M/w)^{|\I|}(1-4\eta)^{|\I\setminus \J|}(\Delta+4\eta)^{|\I\cap \J|-1},
\end{split}
\end{equation}
where the $\frac1{\lambda}\cdot \Delta-4\eta$ term above is due to the fact that by~\dualeqref{eq:823g8gasgfu} one has
\begin{equation*}
\begin{split}
&|\{-q_{\k_*}+(\a_{\k_*}+\eta)\cdot \frac{M}{w}\leq t_{\k_*}<-q_{\k_*}+(\a_{\k_*}+\Delta-\eta)\cdot \frac{M}{w}:\\
&~~~~~~~~~~~~~~~~~~~~~~~~~~~~~~~~~~~\left.\left. \weight(z+t\cdot \k_*) \pmod{W}\in [0, 1/\lambda)\cdot W\right\}\right|\\
\geq &|\{-q_{\k_*}+\a_{\k_*}\cdot \frac{M}{w}\leq t_{\k_*}<-q_{\k_*}+(\a_{\k_*}+\Delta)\cdot \frac{M}{w}:\\
&~~~~~~~~~~~~~~~~~~~~~~~~~~~~~~~~~~~\left.\left. \weight(z+t\cdot \k_*) \pmod{W}\in [0, 1/\lambda)\cdot W\right\}\right|-4\eta M/w\\
&\geq \left(\frac1{\lambda}\cdot \Delta-4\eta\right)\cdot \frac{M}{w},
\end{split}
\end{equation*}
as  the assumption $\eta>2w/M$ implies that $\lceil \eta M/w\rceil\leq 2\eta M/w$. Similarly, we get
\begin{equation*}
\begin{split}
\left|\subspace_\I(x)\cap R'\right|\leq \left(\frac1{\lambda}\cdot \Delta+4\eta\right)\cdot (M/w)^{|\I|}(1+4\eta)^{|\I\setminus \J|}(\Delta+4\eta)^{|\I\cap \J|-1}.
\end{split}
\end{equation*}
Similarly to~\dualeqref{eq:3904h9h9thewgFF} we get
$$
\frac1{\lambda}\cdot \Delta+4\eta\leq (1+\e^{2/3}/4)\frac1{\lambda} \text{~~~~and~~~}\frac1{\lambda}\cdot \Delta-4\eta\geq (1-\e^{2/3}/4)\frac1{\lambda} 
$$
since $\lambda$ is a positive integer bounded by $K$ by assumption. Thus,
\begin{equation*}
\begin{split}
(1-\e^{2/3})\frac1{\lambda}\cdot (M/w)^{|\I|} \Delta^{|\I\cap \J|} \leq \left|\subspace_\I(x)\cap R'\right|\leq (1+\e^{2/3}) \frac1{\lambda}\cdot (M/w)^{|\I|}\Delta^{|\I\cap \J|}
\end{split}
\end{equation*}
as required.

\paragraph{Case 2: $\I\cap \J=\emptyset$.} The proof is similar to {\bf Case 1} above. Consider any choice of $t$ that satisfies~\dualeqref{eq:t-range} and~\dualeqref{eq:t-range-unconstrained}, which by our analysis above leads to $x'\in \subspace_\I(x)$. Now select $\k_*\in \I$ arbitrarily, and let $t_{\k_*}$ vary in the range
\begin{equation}\duallabel{eq:t-constraints}
-q_\k \leq t_{\k_*}<-q_\k+\frac{M}{w}. 
\end{equation}
We have per~\dualeqref{eq:0923hg9hg94FEGF} that  $\weight(x')=\weight(x)+\sum_{\k\in \I} w\cdot t_\k$. Letting $z=x'+\sum_{\k\in \I\setminus \{\k_*\}} \k\cdot t_\k$, we get by Claim~\dualref{cl:divisibility}
\begin{equation*}
\begin{split}
&|\{-q_{\k_*}\leq t_{\k_*}<-q_{\k_*}+\frac{M}{w}:\\
&~~~~~~~~~~~~~~~~~~~~~~~~~~~~~~~~~~~\left.\left. \weight(z+t\cdot \k_*) \pmod{W}\in [0, 1/\lambda)\cdot W\right\}\right|=\frac1{\lambda}\cdot \frac{M}{W}.
\end{split}
\end{equation*}
Note that the preconditions of Claim~\dualref{cl:divisibility} are satisfied since $\frac{W}{w} \mid  \frac{M}{w}$  by~\dualref{p2} and $\lambda \mid \frac{W}{w}$ since $\lambda \leq K$ is a positive integer by assumption of the lemma as well as~\dualref{p1}.

Similarly to {\bf Case 1}, we now get
\begin{equation*}
\begin{split}
\left|\subspace_\I(x)\cap R'\right|&\leq \left(\frac1{\lambda}+4\eta\right)\cdot (M/w)^{|\I|}(1+4\eta/\Delta)^{|\I|}\\
\left|\subspace_\I(x)\cap R'\right|&\geq \left(\frac1{\lambda}-4\eta\right)\cdot (M/w)^{|\I|}(1-4\eta/\Delta)^{|\I|},
\end{split}
\end{equation*}
and
\begin{equation*}
\begin{split}
(1-\e^{2/3})\frac1{\lambda}\cdot (M/w)^{|\I|} \leq \left|\subspace_\I(x)\cap R'\right|\leq (1+\e^{2/3}) \frac1{\lambda}\cdot (M/w)^{|\I|}
\end{split}
\end{equation*}
as required.
\end{proof}

We now give a proof of Lemma~\dualref{lm:rect-subspace-size}, restated here for convenience of the reader:

\noindent{\em {\bf Lemma~\dualref{lm:rect-subspace-size}} (Restated)

For every $\I, \J\subset \F, |\I|, |\J|\leq K^2$, every $\a, \b\in \deltagrid^\J, \a<\b$, if $R=\rect(\J, \a, \b)$ is a rectangle such that $\gamma:=\prod_{\i\in \I\cap \J} (\b_\i-\a_\i)$, the following conditions hold.

\begin{description}
\item[(1)] For every $x\in [m]^n\setminus B$ one has
$$
(1-\e^{2/3}) \cdot \gamma\cdot G \leq \left|\subspace_\I(x)\cap R\right|\leq (1+\e^{2/3})\cdot \gamma\cdot G,
$$
where $G=(M/w)^{|\I|}$. 

\item[(2)] For every positive integer $\lambda\leq K$ such that $\lambda \mid W/w$, if 
$$
R'=\{x\in R: \weight(x)\pmod{W}\in [0, 1/\lambda)\cdot W\},
$$
one has for every $x\in [m]^n\setminus B$
$$
(1-\e^{2/3})\cdot \frac1{\lambda}\cdot \gamma\cdot G \leq \left|\subspace_\I(x)\cap R'\right|\leq (1+\e^{2/3})\cdot \frac1{\lambda}\cdot \gamma\cdot G,
$$
where $G=(M/w)^{|\I|}$. 
\end{description}}
\begin{proof}
We have $R=\rect(\I, \a, \b)=\bigcup_{\q\in Q} \rect(\I, \q)$ for a subset $Q$ of $\deltagrid^\I$ by Claim~\dualref{cl:subcubes-decomp}, and hence by Lemma~\dualref{lm:cube-subspace-size} one has
\begin{equation*}
\begin{split}
\left|\rect(\I, \a, \b)\cap \subspace_\I(x)\right|&=\left|\bigcup_{\q\in Q} \rect(\I, \q)\cap \subspace_\I(x)\right|\\
&=\sum_{\q\in Q} |\rect(\I, \q)\cap \subspace_\I(x)|.
\end{split}
\end{equation*}
The result now follows by Lemma~\dualref{lm:cube-subspace-size}. 
\end{proof}

We now give a proof of Lemma~\dualref{lm:rect-size}, restated here for convenience of the reader:

\noindent{\em {\bf Lemma~\dualref{lm:rect-size}} (Bounds on sizes of rectangles)
For every $\I\subseteq \F$ such that $|\I|\leq K^2$, for every $\c, \d\in (\deltagrid)^\I, \c< \d$, $\gamma:=\prod_{\i\in \I} (\d_\i-\c_\i)$, $R=\rect(\I, \c, \d)$,  for every positive integer $\lambda\leq K$ such that $\lambda \mid W/w$, if 
$$
R'=\{x\in R: \weight(x)\pmod{W}\in [0, 1/\lambda)\cdot W\},
$$
the following conditions hold:

\begin{description}
\item[(1)] the cardinality of $R$ is bounded as
$$
(1-\sqrt{\e})\gamma \leq |R|/m^n\leq (1+\sqrt{\e})\cdot \gamma
$$
\item[(2)] the cardinality of $R'$ is bounded as
$$
\frac1{\lambda}\cdot (1-\sqrt{\e})\gamma \leq \left|R'\right|/m^n\leq \frac1{\lambda}\cdot (1+\sqrt{\e}) \gamma.
$$
\end{description}
}

\begin{proof} Let $C\subset [m]^n$ denote a minimial $\I$-subspace cover as per Definition~\dualref{def:subspace-cover}, i.e. a collection of $x$ such that $\bigcup_{x\in C} \subspace_\I(x)=[m]^n$ and $\subspace_\I(x)\cap \subspace_\I(x')=\emptyset$ for $x, x'\in C$, $x\neq x'$. We have 
\begin{equation}\duallabel{eq:90yt9h9h9h9fehHF}
\begin{split}
|\rect(\I, \c, \d)|&=\sum_{x\in C} |\rect(\I, \c, \d)\cap \subspace_\I(x)|\\
&=\sum_{x\in C\setminus B} |\rect(\I, \c, \d)\cap \subspace_\I(x)|+\sum_{x\in B} |\rect(\I, \c, \d)\cap \subspace_\I(x)|\\
\end{split}
\end{equation}
First note that the second term above is upper bounded by
\begin{equation}\duallabel{eq:0923hg9h2gifhDIHds}
|B| \cdot (5M/w)^{|\I|}\leq \frac1{n^{10}}\cdot m^n\cdot (5M/w)^{|\I|}\leq \e^{2/3}\cdot \gamma \cdot m^n,
\end{equation}
where we used that fact that for every $x\in [m]^n$ and every $\I\subseteq \F$ one has $|\subspace_\I(x)|\leq (5M/w)^{|\I|}$ due to the fact that the integer vector $t$ of coefficients in the definition of $\subspace_\I(x)$ is constrained to be bounded by $2M/w$ coordinatewise, as well as the fact that $n$ is sufficiently large as a function of $M, W, K, L, \Delta, \delta$ and $\e$.

For $x\in C\setminus B$ we have  by Lemma~\dualref{lm:rect-subspace-size}
$$
G\cdot (1-\e^{2/3}) \leq \left|\subspace_\I(x)\cap \rect(\I, \c, \d)\right|\leq G\cdot (1+\e^{2/3}),
$$
where $G=(M/w)^{|\I|} \prod_{\i\in \I} (\d_\i-\c_\i)$. Summing over all $x\in C$ as per~\dualeqref{eq:90yt9h9h9h9fehHF} and using the upper bound on the second term of~\dualeqref{eq:90yt9h9h9h9fehHF} provided by~\dualeqref{eq:0923hg9h2gifhDIHds} gives the result.
\end{proof}

\subsection{Proof of Lemma~\dualref{lm:partition}}\duallabel{app:partition-lemma}

\begin{proofof}{Lemma~\dualref{lm:partition}}
We start by noting that by~\dualeqref{eq:p-def} and~\dualeqref{eq:q-def}
\begin{equation}\duallabel{eq:pqs}
P\cup Q=S^0\cup \left(\bigcup_{\ell\geq 0} T^\ell\right)\cup \Upsilon_{even}\cup \Upsilon_{odd}.
\end{equation}
For every $\ell\in [L]$ let $Z^\ell$ be as in Lemma~\dualref{lm:t-star-ell-rec}, so that
\begin{equation}\duallabel{eq:923yt9haIEG}
T_*^\ell=Z^\ell\cup \left(\nu_{L-1, L-1-\ell}(T_*^{L-1})\cup \bigcup_{j=1}^{L-1-\ell} \nu_{\ell+j, j}(T^{\ell+j}\setminus T_*^{\ell+j})\right),
\end{equation}
where
$$
Z^\ell=T_*^\ell\setminus \left(\nu_{L-1, L-1-\ell}(T_*^{L-1})\cup \bigcup_{j=1}^{L-1-\ell} \nu_{\ell+j, j}(T^{\ell+j}\setminus T_*^{\ell+j})\right).
$$
and
\begin{equation}\duallabel{eq:8ttFTxjiHFHFX}
|Z^\ell|\leq K^L \delta^{1/4} \cdot |P|
\end{equation}
by Lemma~\dualref{lm:t-star-ell-rec}. Adding the $j=0$ term to the rhs of~\dualeqref{eq:923yt9haIEG} and $T^\ell\setminus T_*^\ell$ to the lhs, we get, recalling that $\nu_{\ell, 0}$ is the identity map,
\begin{equation}\duallabel{eq:923yt9haIEGwiehg}
T^\ell=Z^\ell\cup \left(\nu_{L-1, L-1-\ell}(T_*^{L-1})\cup \bigcup_{j=0}^{L-1-\ell} \nu_{\ell+j, j}(T^{\ell+j}\setminus T_*^{\ell+j})\right).
\end{equation}

Putting~\dualeqref{eq:pqs} and ~\dualeqref{eq:923yt9haIEGwiehg} together and letting $D=\bigcup_{\ell=0}^{L-1} \nu_{L-1, L-1-\ell}(T_*^{L-1})$  to simplify notation, we get
\begin{equation}\duallabel{eq:49h9g4g}
\begin{split}
P\cup Q&=S^0\cup \left(\bigcup_{\ell\in [L]} T^\ell\right)\cup \Upsilon_{even}\cup \Upsilon_{odd}\\
&=S^0\cup D\cup \left( \bigcup_{\ell\in [L]} \bigcup_{\substack{j=0}}^\ell \nu_{\ell+j, j}(T^{\ell+j}\setminus T_*^{\ell+j})\right)\cup \left(\bigcup_{\ell\in [L]} Z^\ell\right)\cup \Upsilon_{even}\cup \Upsilon_{odd}\\
&=S^0\cup D\cup \left( \bigcup_{\ell\in [L]} \bigcup_{\substack{j=0}}^\ell \nu_{\ell, j}(T^\ell\setminus T_*^\ell)\right)\cup \left(\bigcup_{\ell\in [L]} Z^\ell\right)\cup \Upsilon_{even}\cup \Upsilon_{odd}\\
&=S^0\cup D\cup  \left(\bigcup_{\substack{\ell\in [L]\\ \ell\text{~even}}} \bigcup_{\substack{j=0}}^\ell \nu_{\ell, j}(T^\ell\setminus T_*^\ell)\right)\cup \left(\bigcup_{\substack{\ell\in [L]\\ \ell\text{~odd}}} \bigcup_{\substack{j=0}}^\ell \nu_{\ell, j}(T^\ell\setminus T_*^\ell)\right)\cup \left(\bigcup_{\ell\in [L]} Z^\ell\right)\cup \Upsilon_{even}\\
&\cup \Upsilon_{odd}.
\end{split}
\end{equation}

We first upper bound $|\Upsilon_{even}|$ and $|\Upsilon_{odd}|$. By Lemma~\dualref{lm:t-star-minus-tau-sp} we have that for every $\ell\in [L], \ell>0$, $|T_*^{\ell-1}\setminus \tau^\ell(S^\ell)|\leq \delta^{1/4} |T_0^\ell|=\delta^{1/4}N$. At the same time, $\tau^\ell$ is injective by Claim~\dualref{cl:pi-star-injective} and 
$$
|S^\ell|=\sum_{k\in [K/2]} |S^\ell_k|=(1\pm \e^{1/4}) N/2
$$ Lemma~\dualref{lm:size-bounds}, {\bf (2)},  and $|T^{\ell-1}_*|=(1\pm \e^{1/2}) N/2$ by Lemma~\dualref{lm:size-bounds}, {\bf (1)}, together with the fact that $T_*^\ell=T_{K/2}^\ell$ by Definition~\dualref{def:terminal-subcube}. Putting these bounds together, we get 
$$
\left|\{s\in S^\ell: \tau^\ell(s)\text{~is not defined}\}\right|=O(\e^{1/4}+\delta^{1/4}) N.
$$
Thus,
\begin{equation}\duallabel{eq:upsilon-small}
|\Upsilon_{even}\cup \Upsilon_{odd}|=O(L\cdot (\e^{1/4}+\delta^{1/4})) N=O(N),
\end{equation}
since
\begin{equation*}
\begin{split}
L(\e^{1/4}+\delta^{1/4})&\leq L \delta^{1/4}\text{~~~~~~~~(by~\dualref{p6})}\\
&\leq L\cdot \Delta^{(100/4)K^2}\text{~~~~~~~~(by~\dualref{p5})}\\
&\leq L\cdot K^{-25K^2}\text{~~~~~~~~(since $\Delta\leq 1/K$ by~\dualref{p4})}\\
&\leq K\cdot K^{-25K^2}\text{~~~~~~~~(since $L\leq K$ by~\dualref{p3})}\\
&\leq  1,
\end{split}
\end{equation*}
where the last transition uses the fact that $K$ is larger than an absolute constant. Using Corollary~\dualref{cor:rect-nu-j} we have  
\begin{equation}\duallabel{eq:dn92ugYFGYGDF}
\begin{split}
|D|&=\left|\bigcup_{\ell=0}^{L-1} \nu_{L-1, L-1-\ell}(T_*^{L-1})\right|\\
&\leq \sum_{\ell=0}^{L-1} \left|\nu_{L-1, L-1-\ell}(T_*^{L-1})\right|\\
&\leq \sum_{\ell=0}^{L-1} (\ln 2+C/K)^\ell\cdot |T_*^{L-1}|\text{~~~~~(by Corollary~\dualref{cor:rect-nu-j})}\\
&\leq |T_*^{L-1}| \cdot \sum_{\ell=0}^{\infty} (\ln 2+0.0001)^\ell\\
&=O(N)
\end{split}
\end{equation}
so since $|S^0|=N/2$ by definition, it suffices to show that the union of the third and the forth terms on the last line of~\dualeqref{eq:49h9g4g} above equals $A_P\cup A_Q\cup B_P\cup B_Q$.

To that effect we recall that by Definition~\dualref{def:nu} for every $\ell=0,\ldots, L-1$ and $j=0,\ldots, \ell$
$\nu_{\ell, j+1}(T^\ell\setminus T_*^\ell)=\tau^{\ell-j}(\downset^{\ell-j}(\nu_{\ell, j}(T^\ell\setminus T_*^\ell))).$ Thus the union of the first and the second terms on the last line of~\dualeqref{eq:49h9g4g} can be rewritten as
\begin{equation}\duallabel{eq:9239h92ghIBFISBF}
\begin{split}
&\bigcup_{\substack{\ell\in [L]\\ \ell\text{~even}}} \bigcup_{\substack{j=0}}^\ell \nu_{\ell, j}(T^\ell\setminus T_*^\ell)\\
&=\bigcup_{\substack{\ell\in [L]\\ \ell\text{~even}}} \bigcup_{\substack{j=0\\j~\text{even}}}^\ell \left(\nu_{\ell, j}(T^\ell\setminus T_*^\ell)\cup \tau^{\ell-j}(\downset^{\ell-j}(\nu_{\ell, j}(T^\ell\setminus T_*^\ell)))\right)\\
&= \left(\bigcup_{\substack{\ell\in [L]\\ \ell\text{~even}}} \bigcup_{\substack{j=0\\j~\text{even}}}^\ell \nu_{\ell, j}(T^\ell\setminus T_*^\ell)\right)\cup \left(\bigcup_{\substack{\ell\in [L]\\ \ell\text{~even}}}  \bigcup_{\substack{j=0\\j~\text{even}}}^\ell \tau^{\ell-j}(\downset^{\ell-j}(\nu_{\ell, j}(T^\ell\setminus T_*^\ell))))\right)\\
&=A_P\cup B_Q\cup \Delta_0,
\end{split}
\end{equation}
where the last transition is by~\dualeqref{eq:apm-def} and~\dualeqref{eq:bpm-def} and we let
$$
\Delta_0=\bigcup_{\substack{\ell\in [L]\\ \ell\text{~even}}} \nu_{\ell, *}(\Ext_\delta(T_*^\ell)\setminus T_*^\ell).
$$
By Lemma~\dualref{lm:rect-int-size} one has 
$$
|\Ext_\delta(T_*^\ell)\setminus T_*^\ell|\leq \sqrt{\delta} |T_*^\ell|,
$$
which implies, since $\nu_{\ell, *}$ maps every vertex to at most $K^{L+1}\leq K^K$ vertices, that 
\begin{equation}\duallabel{eq:HUHFUGEFXVV0}
|\Delta_0|\leq 2L\cdot K^K\cdot \sqrt{\delta} |T_*^\ell|\leq \delta^{1/4} N,
\end{equation}
where we used~\dualref{p5} to conclude that $2L\cdot K^K\cdot \sqrt{\delta}\leq \delta^{1/4}$.

Similarly, we get
\begin{equation}\duallabel{eq:9823ht8g8gALSND}
\begin{split}
&\bigcup_{\substack{\ell\in [L]\\ \ell\text{~odd}}} \bigcup_{\substack{j=0}}^\ell \nu_{\ell, j}(T^\ell\setminus T_*^\ell)\\
&=\left(\bigcup_{\substack{\ell\in [L]\\ \ell\text{~odd}}} \bigcup_{\substack{j=0\\j~\text{even}}}^\ell \nu_{\ell, j}(T^\ell\setminus T_*^\ell)\right)\cup \left(\bigcup_{\substack{j=0\\j~\text{odd}}}^\ell \nu_{\ell, j}(T^\ell\setminus T_*^\ell)\right)\\
&=A_Q\cup B_P\cup \Delta_1,
\end{split}
\end{equation}
where we used~\dualeqref{eq:apm-def} and~\dualeqref{eq:bpm-def}, and let
$$
\Delta_1=\bigcup_{\substack{\ell\in [L]\\ \ell\text{~odd}}} \nu_{\ell, *}(\Ext_\delta(T_*^\ell)\setminus T_*^\ell).
$$
As before, by Lemma~\dualref{lm:rect-int-size} one has 
$$
|\Ext_\delta(T_*^\ell)\setminus T_*^\ell|\leq \sqrt{\delta} |T_*^\ell|,
$$
which implies, since $\nu_{\ell, *}$ maps every vertex to at most $K^{L+1}\leq K^K$ vertices, that 
\begin{equation}\duallabel{eq:HUHFUGEFXVV1}
|\Delta_1|\leq 2L\cdot K^K\cdot \sqrt{\delta} |T_*^\ell|\leq \delta^{1/4} N,
\end{equation}
where we used~\dualref{p5} to conclude that $2L\cdot K^K\cdot \sqrt{\delta}\leq \delta^{1/4}$. Putting~\dualeqref{eq:8ttFTxjiHFHFX}, ~\dualeqref{eq:HUHFUGEFXVV0} and~\dualeqref{eq:HUHFUGEFXVV1} together with~\dualeqref{eq:dn92ugYFGYGDF} gives the result.
\end{proofof}

\renewcommand{\duallabel}[1]{\label{#1-app-d}}
\renewcommand{\dualref}[1]{\ref{#1-app-d}}
\renewcommand{\dualeqref}[1]{\eqref{#1-app-d}}
\section{Lower bound of $1-e^{-1}$ using basic gadgets}\label{app:D}

We now outline how the $1-e^{-1}+\Omega(1)$ hardness result of~\cite{Kapralov13} can be obtained using our basic gadgets above.  The bound is somewhat weaker in that it does not prove, for every $K\geq 2$, hardness of $(1-(1-1/K)^K+\Omega(1))$-approximation when the input graph is shared by $K$ parties, as the bound of~\cite{Kapralov13} does. Our construction is a slight simplification, and gets  $(1-e^{-1}+O(1/K))$-hardness when the number of parties is $K$, which still converges to $1-e^{-1}$ with $K$ getting large. One also notes that the sets $S_0, S_1,\ldots$ in~\cite{Kapralov13} have geometrically decreasing size, whereas in our case they are all of size about $|T|/K$ -- this is due to a reparameterization, which is more convenient for our main result, i.e. the $\frac1{1+\ln 2}$ lower bound. 

First, we partition $\mathcal{F}$ into disjoint subsets  of equal size, letting 
$$
\mathcal{F}=\B_0\cup \B_1\cup \ldots \cup \B_{\wt{K}},
$$ 
where $\B_i\cap \B_j=\emptyset$ if $i\neq j$, and $\wt{K}:=\lfloor (1-e^{-1})K\rfloor$. We let
\begin{equation}\duallabel{eq:j-def}
\J\in \B_0\times \ldots\times \B_{\wt{K}-1},
\end{equation}
i.e., $\j_k\in \B_k$ for $k\in [\wt{K}]$. We extend the definition of $T_k$ and $S_k$ for $k\in [\wt{K}]$ (as opposed to just $k\in [K/2+1]$). Let $T_0=T$, and for every $k\in [\wt{K}]$ let 
\begin{equation}\duallabel{eq:def-tk}
\begin{split}
T_{k+1}:=\left\{y\in T_k: \langle y, \j_k\rangle \pmod{M}\in \left[0, 1-\frac{1}{K-k}\right)\cdot M\right\}.\\
\end{split}
\end{equation}
Note that, as above, $T_0\supset T_1\supset\ldots\supset T_{\wt{K}}$ form a nested sequence, and for every $k\in [\wt{K}+1]$ one has 
\begin{equation}\duallabel{eq:def-tk-all-constraints}
\begin{split}
T_k:=\left\{y\in T_0: \langle y, \j_s\rangle \pmod{M}\in \left[0, 1-\frac{1}{K-s}\right)\cdot M\text{~for all~}s\in \{0, 1,\ldots, k-1\}\right\}.\\
\end{split}
\end{equation}
The innermost set in this sequence is again a central object of our construction:
\begin{definition}[Terminal subcube]\duallabel{def:terminal-subcube}
We refer to $T_*:=T_{\wt{K}}$ as the {\em terminal subcube}.
\end{definition}

Define
\begin{equation}\duallabel{eq:def-sk}
\begin{split}
S_k&:\delequal \left\{x\in T_k: \weight(x) \in \left[0, \frac{1}{K-k}\right)\cdot W \pmod{ W} \right\},
\end{split}
\end{equation}
The above stands for $S_k$ being a set of vertices such that $S_k\delequal \wt{T}_k$, where
\begin{equation*}
\wt{T}_k:=\left\{x\in T_k: \weight(x) \in \left[0, \frac{1}{K-k}\right)\cdot W \pmod{ W} \right\}
\end{equation*}
is the set of vertices in $T_k$ whose weight modulo $W$ belongs to a certain range. We stress here that unlike the collection of sets $T_k$, the sets $S_k$ are disjoint. We also let, for every $k\in [\wt{K}]$ and $\i\in \B_k$ 
\begin{equation}\duallabel{eq:skj}
\begin{split}
T_k^\j&=\left\{y\in T_k: \langle y, \j\rangle \pmod{M}\in \left[0, 1-\frac{1}{K-k}\right)\cdot M\right\}\\
S_k^\j&=\left\{x\in S_k: \langle x, \j\rangle \pmod{M}\in \left[0, 1-\frac{1}{K-k}\right)\cdot M\right\}.\\
\end{split}
\end{equation}

First, we note that size bounds in Lemma~\ref{lm:size-bounds-full} and Lemma~\ref{lm:large-matching-full} extend for all $k=0,1,\ldots, \wt{K}$, i.e. apply to the sets defined above (the changes to the proof amount to extending the range of $k$ appropriately). We state them here for convenience of the reader.

\begin{lemma}\duallabel{lm:size-bounds}
One has
\begin{itemize}
\item[{\bf (1)}] For every $k\in [\wt{K}+1]$ one has $|T_k|=(1\pm \sqrt{\e})\cdot |T_0|(1-k/K)$;
\item[{\bf (2)}] For every $k\in [\wt{K}]$ one has  $|S_k|=(1\pm \sqrt{\e}) \cdot |T_0|/K$;
\item[{\bf (3)}] For every $k\in [\wt{K}]$, every $\i\in \B_k$ one has $|S^\j_k|=(1\pm \sqrt{\e})(1-\frac1{K-k})|T_0|/K$.
\item[{\bf (4)}] For every $k\in [\wt{K}]$, every $\i\in \B_k$ one has $|T^\j_k|=(1\pm \sqrt{\e})(1-\frac{k+1}{K})|T_0|$.
\end{itemize}
\end{lemma}

\begin{lemma}\duallabel{lm:large-matching}
There exists a matching of a $(1-O(1/K))$ fraction of vertices in $S$ to $T\setminus T_*$.
\end{lemma}

We now define the edges of $G=(S, T, E)$ incident on $S_k$ for every $k\in [\wt{K}]$. For every $\i\in \B_k$ 
let 
\begin{equation}\duallabel{eq:c-j-def}
C_\j\subset [m]^n
\end{equation}
be a minimal $\j$-line cover as per Definition~\ref{def:line-cover-full}.
For every $y\in C$, we include a complete bipartite graph between $\text{line}_\j(y)\cap \Int_\delta(S_k^\j)$ and $\text{line}_\j(y)\cap (T_k\setminus T_k^\j)$: let
\begin{equation}\duallabel{eq:edges-ei-def}
E_k=\bigcup_{\i\in \B_k} E_{k, \j},
\end{equation}
where
\begin{equation}\duallabel{eq:edges-eij-def}
E_{k, \j}=\bigcup_{y\in C_\j} (\text{line}_\j(y)\cap \Int_\delta(S_k^\j)) \times (\text{line}_\j(y) \cap (T_k\setminus T_k^\j)).
\end{equation}
We let $E=\bigcup_{k\in [\wt{K}]} E_k$.

\begin{remark}
Note that the edge set $E_k$ is fully defined by the prefix $\J_{<k}$.
\end{remark}

\begin{remark}
We note that the edge set defined in~\dualeqref{eq:edges-eij-def} does not depend on the specific choice of a cover $C_\j$ used, i.e. any minimal $\j$-line cover produces the same edge set as per~\dualeqref{eq:edges-eij-def}.
\end{remark}

\paragraph{Input distribution $\mathcal{D}$.} For every $k\in [\wt{K}]$ sample $\j_k$, the $k$-th element of $\J$, independently and uniformly from $\B_k$, so that 
$$
\J\sim \text{UNIF}\left(\B_0\times \ldots\times \B_{\wt{K}-1}\right).
$$

For every $k\in [\wt{K}]$, $\i\in \B_k$ and $y\in S_k$ let
\begin{equation}\duallabel{eq:x-ell-def}
X_{k, \j}(y)=\text{Bernoulli}(1-1/K)
\end{equation}
denote independent Bernoulli random variables conditioned on $\sum_{y\in S_k} X_{k, \j}(y)=\lceil (1-\frac1{K})|S_k|\rceil$ for all $k$ and $\j$.  We use these variables to sample edges of the graph $G$ as follows. Define
\begin{equation*}
\wt{E}_{k, \j}=\bigcup_{y\in C_\j} \left\{u\in \text{line}_\j(y)\cap \Int_\delta(S_k^\j): X_{k, \j}(u)=1\right\} \times (\text{line}_\j(y) \cap (T_k\setminus T_k^\j)),
\end{equation*}
where $C_\j$ is a minimal $\j$-line cover, and let 
$$
\wt{E}_k=\bigcup_{\i\in \B_k}  \wt{E}_{k, \j}.
$$
Comparing this to the definition of the edge set of $G$ in~\dualeqref{eq:edges-ei-def},  one observes that we subsample edges of $G$ in a somewhat dependent way -- the set $\wt{E}_k$ contains, for every direction $\i\in \B_k$ and $y\in C_\j$, a complete bipartite graph between vertices $u$ in $\text{line}_\j(y)\cap \Int_\delta(S_k^\j)$ that were sampled by $X_{k, \j}(u)$ and $\text{line}_\j(y) \cap (T_k\setminus T_k^\j)$.  The fact that randomness is provided by the vertices $u\in S_k$ as opposed to edges themselves will not be a problem since we are interested in concentration of matching size in $G$ and do not need to reason about arbitrary edge sets -- see proof of Lemma~\dualref{lm:subsampling-matchings-gl} below. Let 
$$
\wt{G}=(S\cup S_*, T, \wt{E}\cup M_*),
$$ 
where $S_*$ is a disjoint set of nodes of size equal to the size of $T_*$, and $M_*$ is a perfect matching between $T_*$ and $S_*$.
Note that the subsampling operation used to produce $\wt{E}$ from $E$ has the effect of making it hard to store edges of $\wt{G}$ (since the algorithm intuitively must remember which edge of $G$ was included and which was not), but at the same time ensures that $\wt{G}$ contains a nearly perfect matching.

\begin{lemma}[Large matching in $\wt{G}$]\duallabel{lm:subsampling-matchings-gl}
With probability at least $1-1/N$ there exists a matching of $S\cup S_*$ to $T$ of size at least $(1-O(1/K))|T|$.
\end{lemma}
\begin{proof}
For every edge  $e\in E$ define the random variable 
\begin{equation}\duallabel{eq:z-def}
Z_e=\left\{
\begin{array}{ll}
1&\text{~if~}e\in \wt{E}\\
0&\text{~o.w.}\\
\end{array}
\right.
\end{equation}
Note that for every matching $M\subseteq E$ random variables $\{Z_e\}_{e\in M}$ are negatively dependent, since a matching $M$ touches every vertex at most once.

By Lemma~\dualref{lm:large-matching}  applied to $G=(S, T, E)$ there exists a matching of a $(1-O(1/K))$ fraction of vertices in $S$ to $T\setminus T_*$ -- denote this matching by $M$.  Let 
$$
\wt{M}:=M\cap \wt{E}=\{e\in M: Z_e=1\}
$$
denote the subset of the edges of $M$ that are included in $\wt{E}$. Note that $\wt{M}$ is a matching between a subset of $S$ and a subset of $T\setminus T_*$, and we have 
$$
\expect[|\wt{M}|]=\sum_{e\in M} \prob[e\in \wt{E}]=\sum_{e\in M} \expect[Z_e]=(1-1/K) |M|
$$
by definition of $Z_e$ in~\dualeqref{eq:z-def} and the fact that every edge in $E$ is included in $\wt{E}$ with probability $1-1/K$ by~\dualeqref{eq:x-ell-def}.
Since the random variables $\{Z_e\}_{e\in M}$ are negatively dependent, we have by an application to the Chernoff bound (for negatively dependent random variables)
$$
\prob[|\wt{M}|<(1-2/K)|M|]\leq \exp(-\Omega(|M|/K)).
$$
Since $M$ matches at least a constant fraction of $S$, we get that $|M|=\Omega(N)$, and therefore 
$$
\prob[|\wt{M}|<(1-2/K)|M|]\leq \exp(-\Omega(N/K))\leq N^{-2},
$$
where $N$ is the number of vertices in our graph instance. 
\end{proof}

\paragraph{Ordering of edges of $\wt{G}$ in the stream.} The graph $\wt{G}$ is presented in the stream over $\wt{K}+1$ {\em phases} as follows. For every $k\in [\wt{K}]$, the edges in $\wt{E}_k=\wt{E}\cap E_k$ are presented in the stream (the ordering of edges within a phase is arbitrary). Finally a perfect matching between $T_*$ and a disjoint set of nodes $S_*$ is presented in the stream.
 
\begin{definition}\duallabel{def:gk}
For $k\in [\wt{K}]$ we write $\wt{G}_{<k}=(T, S_0\cup \ldots \cup S_{k-1}, \wt{E}_{<k})$, where $\wt{E}_{<k}=\bigcup_{s\in [k]} \wt{E}_s$.
\end{definition}

\begin{definition}\duallabel{def:lambda}
For every $k\in [\wt{K}]$ let $\Lambda_k=(X_k, \J_k)$. We write $\Lambda_{<k}=\left(\Lambda_{s}\right)_{0\leq s<k}$.
\end{definition}

\begin{remark}\duallabel{rm:lambda}
Note that $\wt{G}_{\leq k}$ is fully determined by $\Lambda_{<k}$ and $X_k$, and $\j_k$ is uniformly random in $\B_k$ conditioned on $\Lambda_{<k}$ and $X_k$.
\end{remark}

\subsection{Upper and lower bounds on matchings in $\wt{G}$}\duallabel{sec:malg}
We first prove
\begin{lemma}[Large matching in $\wt{G}$]\duallabel{lm:large-matching-g-prime}
With probability at least $1-1/N$ there exists a matching in $\wt{G}$ of size at least $(1-O(1/K))|T|$.
\end{lemma}
\begin{proof}
By Lemma~\dualref{lm:subsampling-matchings-gl} with probability at least $1-N^{-1}$ there exists a matching of a $1-O(1/K)$ fraction of $S$ to $T\setminus T_*$ in $\wt{G}$. Since $\wt{G}$ also contains a perfect matching of $T_*$ to a disjoint set of vertices $S_*$, this gives the result.
\end{proof}

We now turn to upper bounding the performance of a small space streaming algorithm on our input distribution $\mathcal{D}$. Since the input is sampled from a distribution, we may assume by Yao's minimax principle that the streaming algorithm ALG is deterministic. Let ALG denote a deterministic streaming algorithm that uses $s$ bits of space and at the end of the stream outputs a matching $M_{ALG}$ in $\wt{G}$ such that 
\begin{equation*}
\prob_{\wt{G}\sim \mathcal{D}} \left[|M_{ALG}|\geq \left(1-e^{-1}+\eta\right)|M_{OPT}|\right]\geq 3/4
\end{equation*}
for some positive $\eta\in (0, 1)$, where $M_{OPT}$ is a maximum matching in $\wt{G}$. Note that we are assuming that with probability at least $3/4$ both $M_{ALG}$ is a matching in $\wt{G}$ (i.e., in particular, the algorithm does not output edges that are not in $\wt{G}$) and the size of $M_{ALG}$ is large as above. At the same time by Lemma~\dualref{lm:large-matching-g-prime} one has
$$
\prob_{\wt{G}\sim \mathcal{D}} \left[|M_{OPT}|< (1-O(1/K))|T|\right]\leq N^{-1}.
$$ 
Putting the two bounds above together, we get
\begin{equation}\duallabel{eq:malg-large}
\begin{split}
\prob_{\wt{G}\sim \mathcal{D}} \left[|M_{ALG}|\geq \left(1-e^{-1}+\eta-O(1/K)\right)|T|\right]\geq 1/2.
\end{split}
\end{equation}
In what follows we show that any algorithm that achieves~\dualeqref{eq:malg-large} must essentially remember, for many edges of $G=(S, T, E)$ whether they were included in $\wt{G}$.

\paragraph{Upper bounding $|M_{ALG}|$.}  

\begin{lemma}\duallabel{lm:vertex-cover}
For every matching $M\subseteq \wt{E}$ one has 
$$
|M|\leq |M\cap ((T\setminus \Ext_\delta(T_*))\times \downset(T_*))|+(1-e^{-1})|T|+O(|T|/K).
$$
\end{lemma}
\begin{proof}
We exhibit a vertex cover of appropriate size for $M$. Specifically, we add to the vertex cover one endpoint of every edge in
$$
M\cap ((T\setminus \Ext_\delta(T_*))\times \downset(T_*)),
$$
as well as all vertices in  $S\setminus \downset(T_*)$ and $\Ext_\delta(T_*)$. Note that this is indeed a vertex cover for $\wt{G}$. The size of the vertex cover is 
\begin{equation}\duallabel{eq:23t77Gy8uGBBYVF}
\begin{split}
&|M\cap ((T\setminus \Ext_\delta(T_*))\times \downset(T_*))|+|S\setminus \downset(T_*)|+|\Ext_\delta(T_*)|.
\end{split}
\end{equation}
We now bound the second and third terms above. First, by Lemma~\dualref{lm:size-bounds}, {\bf (2)}, we have
$$
|S|=\sum_{k\in [\wt{K}]} |S_k|\leq (1+\sqrt{\e})\cdot \wt{K}\cdot |T|/K\leq (1+O(1/K)) (1-e^{-1})|T|,
$$
since $\wt{K}=\lfloor (1-e^{-1})K\rfloor\leq (1-e^{-1})K$ and $\sqrt{\e}=O(1/K)$ by~\ref{p3-full},~\ref{p5-full} and~\ref{p6-full}. At the same time we have 
$$
|\downset(T_*)|=\left|\bigcup_{k\in [\wt{K}]} \downset_k(T_*)\right|=\sum_{k\in [\wt{K}]} |\downset_k(T_*)|.
$$
For every $k\in [\wt{K}]$ we now apply Lemma~\ref{lm:rect-size-full}, {\bf (2)}, to lower bound $|\downset_k(T_*)|$ (noting, crucially, that $T_*\subseteq T_k$ for all $k\in [\wt{K}]$).  For that note that $T_*=\rect(\I, \c, \d)$, where $\I=\{\j_k\}_{k\in [\wt{K}]}$, 
$\c_{\j_k}=0$ and $\d_{\j_k}=1-\frac{1}{K-k}$ for $k\in [\wt{K}]$. We thus apply Lemma~\ref{lm:rect-size-full}, {\bf (2)} with $\lambda=K-k$ and 
\begin{equation*}
\begin{split}
\gamma&=\prod_{\i\in \I} (\d_\i-\c_\i)=\prod_{k=0}^{\wt{K}-1} \left(1-\frac1{K-k}\right)=\prod_{k=0}^{\wt{K}-1} \frac{K-k-1}{K-k}=\frac{K-\wt{K}}{K}= e^{-1}+O(1/K),
\end{split}
\end{equation*}
since $\wt{K}=\lfloor (1-e^{-1})K\rfloor$. We thus get, since $\sqrt{\e}=O(1/K)$ by~\ref{p3-full},~\ref{p5-full} and~\ref{p6-full}, that
$$
\left|\downset_k(T_*)\right|\geq \frac1{\lambda}\cdot (1-\sqrt{\e})\gamma |T|\geq e^{-1}(1-O(1/K))\frac1{K-k}\cdot |T|.
$$
Summing over $k\in [\wt{K}]$, we get
$$
|\downset(T_*)|\geq \sum_{k\in [\wt{K}]} |\downset_k(T_*)|\geq e^{-1}(1-O(1/K)) \left(\sum_{k\in [\wt{K}]} \frac1{K-k}\right)\cdot |T|.
$$
Since 
$$
\sum_{k\in [\wt{K}]} \frac1{K-k}\geq \sum_{k=1}^{\lfloor (1-e^{-1})K\rfloor} \frac1{K-k}\geq \int_0^{1-e^{-1}-O(1/K)} \frac1{1-x} dx=1-O(1/K),
$$
we get 
$$
|\downset(T_*)|\geq e^{-1}(1-O(1/K))\cdot |T|.
$$ Finally, we have by Lemma~\ref{lm:rect-int-size-full} that $|\Ext_\delta(T_*)\setminus T_*|\leq \sqrt{\delta} |T_*|$, and by Lemma~\ref{lm:rect-size-full}, {\bf (1)}, using the calculation for $\gamma$ above, we have $|T_*|=(1+O(1/K))e^{-1} \cdot |T|$, and therefore $|\Ext_\delta(T_*)|=(1+\sqrt{\delta}) (1+O(1/K))e^{-1} \cdot |T|=(1+O(1/K))e^{-1} \cdot |T|$ by~\ref{p3-full} and~\ref{p5-full}. Putting these bounds together, we get 
\begin{equation*}
\begin{split}
|S\setminus \downset(T_*)|+|\Ext_\delta(T_*)|&\leq (1+O(1/K)) (1-e^{-1})|T|-e^{-1}(1-O(1/K))\cdot |T|\\
&+(1+O(1/K))e^{-1} \cdot |T|\\
&\leq (1+O(1/K)) (1-e^{-1}) |T|,
\end{split}
\end{equation*}
as required. This together with~\dualeqref{eq:23t77Gy8uGBBYVF} gives the result of the lemma. \end{proof}

We now prove 
\begin{lemma}\duallabel{lm:special-edges}
For every matching $M\subseteq \wt{E}$ one has 
$$
M\cap ((T\setminus \Ext_\delta(T_*))\times \downset(T_*)) \subseteq \bigcup_{k\in [\wt{K}]} E_{k, \j_k}.
$$
\end{lemma}
\begin{proof}
Fix $k\in [\wt{K}]$. Consider $(x, y)\in E_k$, where $x\in T_k, y\in S_k$, such that $(x, y)\in E\cap ((T\setminus \Ext_\delta(T_*))\times \downset(T_*))$. Since $x\in T_k\cap (T\setminus \Ext_\delta(T_*))=T_k\setminus \Ext_\delta(T_*)$ (see Definition~\ref{def:exterior-full}), there exists $s\in \{k,\ldots, \wt{K}-1\}$ such that
\begin{equation}\duallabel{eq:0834gh3g9jdkoasjd}
\langle x, \j_s\rangle \pmod{M}\in \left[1-\frac{1}{K-s}+\delta, 1-\delta\right)\cdot M.
\end{equation}

Since $(x, y)\in E_k$, one has $y=x+\lambda \cdot \u$ for some $\u\in \B_k$ and integer $\lambda$ satisfying $|\lambda|\leq 2M/w$. Suppose towards a contradiction that $\u\neq \j_k$. In that case one has 
\begin{equation*}
\begin{split}
\left|\langle y, \j_s\rangle -\langle x, \j_s\rangle\right|&=\left|\langle x+\lambda \cdot \u, \j_s\rangle -\langle x, \j_s\rangle\right|\\
&=|\lambda|\cdot \langle \u, \j_s\rangle\\
&\leq |\lambda|\cdot \e\cdot w\\
&\leq 2\e\cdot M\\
&< \delta,\\
\end{split}
\end{equation*}
where we used the fact that $\u\neq \j_s$, since $\u\in \B_k$, $\B_k\cap \J=\{\j_k\}$ and $\u\neq \j_k$ by assumption. The last transition is by~\ref{p6-full}. We thus get by combining the above with~\dualeqref{eq:0834gh3g9jdkoasjd} that  
\begin{equation*}
\langle y, \j_s\rangle \pmod{M}\not \in \left[0, 1-\frac{1}{K-s}\right)\cdot M,
\end{equation*}
and therefore $y\not \in \downset(T_*)$. Thus, we have $\u=\j_k$, and therefore $(x, y)\in E_{k, \j_k}$, as required.
\end{proof}

\subsection{Proof of Theorem~\ref{thm:main}}\duallabel{sec:main-proof}
We now give

\begin{proofof}{Theorem~\ref{thm:main}} Now putting~\dualeqref{eq:malg-large} together with~Lemma~\dualref{lm:vertex-cover}, we get
\begin{equation*}
\begin{split}
|M\cap ((T\setminus \Ext_\delta(T_*))\times \downset(T_*))|&\geq |M_{ALG}|-\left((1-e^{-1})|T|+O(|T|/K)\right)\\
&\geq \left(1-e^{-1}+\eta-O(1/K)\right)|T|-\left((1-e^{-1})|T|+O(|T|/K)\right)\\
&\geq (\eta-O(1/K))|T|\\
&\geq (\eta/2)|T|,
\end{split}
\end{equation*}
with probability at least $1/2$, where we assumed that $K$ is larger than an absolute constant that depends on $\eta$ in the last transition. Thus,
\begin{equation}\duallabel{eq:success-prob}
\prob_{\wh{G}\sim \mathcal{D}}\left[|M_{ALG}\cap ((T\setminus \Ext_\delta(T_*))\times \downset(T_*))|\geq (\eta/2)|T|\text{~and~} M_{ALG}\subseteq \wt{E}\right]\geq 1/2.
\end{equation}
Note that the second condition above, namely $M_{ALG}\subseteq \wt{E}$ enforces the constraint that the algorithm does not output non-edges\footnote{The analysis generalizes easily to the setting where the algorithm is allowed to output a small fraction of non-edges, but this is a rather non-standard assumption, and we prefer to operate under the more standard model where $M_{ALG}$ must be a subset of $\wh{E}$ with a good probability.}. We do not add this condition explicitly in calculations below to simplify notation (one can think of $|M_{ALG}|$ as being defined as zero when $M_{ALG}$ contains non-edges). Now recall that by Lemma~\dualref{lm:special-edges} we have
\begin{equation*}
\begin{split}
M_{ALG}\cap ((T\setminus \Ext_\delta(T_*))\times \downset(T_*))&\subseteq \bigcup_{k\in [\wt{K}]} E_{k, \j_k}
\end{split}
\end{equation*}
Thus, there exists $k^*\in [\wt{K}]$ such that 
\begin{equation}\duallabel{eq:special-lk}
\begin{split}
\prob\left[|M_{ALG}\cap \wt{E}_{k^*, \j_{k^*}}|\geq \frac{\eta}{2K} |T|\right]\geq \frac1{2K}.
\end{split}
\end{equation}
Indeed, otherwise one would have
\begin{equation*}
\begin{split}
&\prob[|M_{ALG}\cap ((T\setminus \Ext_\delta(T_*))\times \downset(T_*))|\geq (\eta/2)|T|]\\
&\leq  \prob\left[\text{exists~} k\in [\wt{K}] \text{~such that~}|M_{ALG}\cap E_{k, \j_k}|\geq \frac{\eta}{2K}|T|\right]\\
&\leq \sum_{k\in [\wt{K}]} \prob\left[|M_{ALG}\cap E_{k, \j_k}|\geq \frac{\eta}{2K}|P|\right]\\
&< \sum_{k\in [\wt{K}]} \frac1{2K}\\
&\leq K\cdot \frac1{2K}\\
&= 1/2, 
\end{split}
\end{equation*}
a contradiction with~\dualeqref{eq:success-prob}.

To simplify notation, we let $k=k^*$. Recall that {\bf (a)} $\wt{G}_{\leq k}$ is fully determined by $\Lambda_{<k}$ and $X_k$ (see Definition~\dualref{def:lambda}) and {\bf (b)} conditioned on $\Lambda_{<k}$  and $X_k$ one has $\j_k\sim UNIF(\B_k)$.  For simplicity of notation we write 
$ \B=\B_k$,  $\j=\j_k$ and $X=X_k$. 

\paragraph{Lower bounding the space usage of ALG.} In what follows we show that since $M_{ALG}$ often returns many edges from $\wt{E}_k$ as per~\dualeqref{eq:special-lk}, the conditional entropy of $X_k$ given $\Pi$ and $\Lambda_{\leq k}$ is low, which gives the desired lower bound on $s$. Let $\Pi\in \{0, 1\}^s$ denote the state of ALG after it has been presented with $\wh{G}_{\leq k}$. Then finish running ALG on $\wh{G}_{> k}$ starting with state $\Pi$. Let $M_{ALG}$ denote the matching output by ALG. We have 
\begin{equation}\duallabel{eq:93t9239hdsagasasaaaa}
\begin{split}
s=|\Pi|&\geq H(\Pi)\\
&\geq H(\Pi| \Lambda_{<k})\\
&\geq I(\Pi; X | \Lambda_{<k})\\
&\geq \sum_{\i\in \B} I(\Pi; X_\i | \Lambda_{<k})\\
&=\sum_{\i\in \B} I(\Pi; X_\i | \Lambda_{<k}, \{\j=\i\})\\
&\geq \sum_{\i\in \B} I(M_{ALG}; X_\i | \Lambda_{<k}, \{\j=\i\})\\
\end{split}
\end{equation}
The second transition uses the fact that conditioning does not increase entropy, the forth transition uses the fact that $X_\i$'s are independent conditioned on $\Lambda_{<k}$, the forth transition uses the fact that $\J$ is independent of $\Pi$ and $X_\i$ conditioned on $\Lambda_{<k}$.   The final transition is by the data processing inequality:
\begin{lemma} \emph{(Data Processing Inequality)} \label{thm:dpi}
    For any random variables $(X,Y,Z)$ such that $X \to Y \to Z$ forms a Markov chain, we have $I(X;Z) \le I(X;Y)$.
\end{lemma}

We now lower bound 
\begin{equation}\duallabel{eq:20iabibafihf0ASAS}
\begin{split}
\sum_{\i\in \B} I(M_{ALG}; X_\i | \Lambda_{<k}, \{\j=\i\})&=\sum_{\i\in \B} H(X_\i| \Lambda_{<k}, \{\j=\i\})-H(X_\i | M_{ALG}, \Lambda_{<k}, \{\j=\i\})\\
&=\sum_{\i\in \B} H(X_\i)-H(X_\i | M_{ALG}, \Lambda_{<k}, \{\j=\i\}).
\end{split}
\end{equation}

We now upper bound $H(X_\i | M_{ALG}, \Lambda_{<k}, \{\j=\i\})$ on the rhs of~\dualeqref{eq:20iabibafihf0ASAS}. Let 
\begin{equation}\duallabel{eq:8923yty9y239ruyURadF}
\begin{split}
\E:=&\left\{|M_{ALG}\cap E_{k, \j}|\geq \frac{\eta}{2K} |P|\text{~and~}M_{ALG}\subseteq \wh{E}\right\}
\end{split}
\end{equation}
and  let $Z$ denote the indicator of $\E$. Note that $\expect[Z]=\prob[\E]\geq \frac1{2K}$ by~\dualeqref{eq:special-lk}. We have
\begin{equation}\duallabel{eq:plus-z-198yugudfg}
\begin{split}
H(X_\i| M_{ALG}, \Lambda_{<k}, \{\j=\i\})&\leq H(X_\i, Z| M_{ALG}, \Lambda_{<k}, \{\j=\i\})\\
&\leq H(Z)+H(X_\i| M_{ALG}, \Lambda_{<k}, \{\j=\i\}, Z)\\
&\leq 1+H(X_\i| M_{ALG}, \Lambda_{<k}, \{\j=\i\}, Z),\\
\end{split}
\end{equation}
where we used the fact that $H(Z)\leq 1$, as $Z$ is a binary variable. At the same time, since 
$\expect[Z]=\expect_{\i\sim UNIF(\B)}\left[Z|\{\j=\i\}\right]\geq \frac1{2K}$ by~\dualeqref{eq:special-lk},
and $\J\sim UNIF(\B)$, there exists a subset $\mathcal{J}\subseteq \B$ such that $|\mathcal{J}|\geq \frac1{4K}|\B|$ and for every $\i\in \mathcal{J}$ one has 
$\expect[Z| \{\j=\i\}]\geq \frac1{4K}.$ For every $\j\in \mathcal{J}$ one has
\begin{equation}\duallabel{eq:93g9g7g89gFGYFGSA}
\begin{split}
H(X_\i| M_{ALG}, &\Lambda_{<k}, \{\j=\i\}, Z)\\
&=H(X_\i| M_{ALG}, \Lambda_{<k}, \{\j=\i\wedge Z=1\})\cdot \prob[Z=1| \{\j=\i\}]\\
&+H(X_\i| M_{ALG}, \Lambda_{<k}, \{\j=\i\wedge Z=0\})\cdot \prob[Z=0| \{\j=\i\}]\\
\end{split}
\end{equation}

We now bound both terms on the rhs in~\dualeqref{eq:93g9g7g89gFGYFGSA}. For the second term we have
\begin{equation}
\begin{split}
H(X_\i| M_{ALG}, \Lambda_{<k}, \{\j=\i\wedge Z=0\})&\leq \expect_{\Lambda_{<k}}[|S_k|] \cdot H_2(1-1/K)\\
&\leq (1+\sqrt{\e})\frac1{K}|T| \cdot H_2(1-1/K)\\
\end{split}
\end{equation}
where the first transition is because $\sum_{y\in S_k} X_\i(y)=\lceil (1-\frac1{K}) |S_k|\rceil$ by definition of $X_\i$ and the second transition is by Lemma~\dualref{lm:size-bounds}, {\bf (2)}.

For the first term on the rhs in~\dualeqref{eq:93g9g7g89gFGYFGSA}  we note that since $M_{ALG}\subseteq \wt{E}$ as we are conditioning on the event $\mathcal{E}$ (by conditioning on $\{Z=1\}$) for every $y\in S_k$ that is matched by $M_{ALG}$ one has $X_\i(y)=1$. 
By conditioning on $\{Z=1\wedge \j=\i\}$, we get by~\dualeqref{eq:8923yty9y239ruyURadF} $|M_{ALG}\cap E_{k, \j}|\geq \frac{\eta |P|}{2K}$,
and hence 
$$
\gamma:=\frac{|M_{ALG}^\j|}{|S_k|}\geq \frac{\eta |P|}{2K |S_k|}\geq \frac{\eta |T|}{4K |S_k|}\geq \eta/8,
$$
where we let $M_{ALG}^\j:=M_{ALG}\cap E_{k, \j}$ to simplify notation.
For every fixing $\lambda$ of $\Lambda_{<k}$ one has, 
$$
H(X_\i| M_{ALG}, \{\Lambda_{<k}=\lambda \wedge \j=\i\wedge Z=1\})\leq (1-\gamma) |S_k| H_2\left(1-\frac1{K(1-\gamma)}\right),
$$
since conditioned on $M_{ALG}$, $\lambda, \j=\i$ and the success event $Z=1$ there are exactly $(1-\gamma)|S_k|$ values of $y\in S_k\setminus M_{ALG}$ such that $X_\i(y)=1$, and hence the conditional entropy of $X_\i$ is bounded by
\begin{equation*}
\begin{split}
\log_2 { |S_k\setminus M_{ALG}^\j| \choose (1-\frac1{K}-\gamma) |S_k|}&=\log_2 { (1-\gamma) |S_k| \choose (1-\frac1{K}-\gamma) |S_k|}\\
&=\log_2 { (1-\gamma) |S_k| \choose (1-\frac1{K(1-\gamma)}) (1-\gamma) |S_k|}\\
&\leq (1-\gamma) |S_k| H_2\left(1-\frac1{K(1-\gamma)}\right),
\end{split}
\end{equation*}
where the last transition is by subadditivity of entropy. Recalling that $\gamma\geq \eta/8$ and $\eta>0$ is a small constant we bound the rhs above by
\begin{equation}\duallabel{eq:8g823ggfaib}
\begin{split}
(1-\gamma) |S_k| H_2\left(1-\frac1{K(1-\gamma)}\right)&\leq (1-\eta/8) |S_k| H_2\left(1-\frac1{K(1-\eta/8)}\right)\\
&\leq (1+\sqrt{\e}) \frac1{K}|T|\cdot  (1-\eta/8) H_2\left(1-\frac1{K(1-\eta/8)}\right),
\end{split}
\end{equation}
where in the second transition we also used the fact that by Lemma~\dualref{lm:size-bounds}, {\bf (2)}, we have $|S_k|\leq (1+\sqrt{\e})\frac1{K}|T|$.  At this point we also note that 
\begin{equation*}
\begin{split}
(1-\eta/8) H_2\left(1-\frac1{K(1-\eta/8)}\right)&=\frac1{K}\log_2 K+\frac1{K \ln 2}-\frac1{K} \log\frac1{1-8/\eta}+O(1/K^2)\\
&\leq H_2(1-1/K)-\frac1{K} \log\frac1{1-\eta/8}+O(1/K^2).
\end{split}
\end{equation*}
since $H_2(1-1/K)=\frac1{K}\log_2 K+\frac1{K\ln 2}+O(1/K^2)$ and $K$ is larger than a constant. Putting the above bounds together, we get, assuming that $K$ is larger than $1/\eta$ by a large constant factor,
 \begin{equation*}
\begin{split}
H(X_\i| M_{ALG}, \Lambda_{<k}, \{\j=\i\wedge Z=1\})\leq (1+\sqrt{\e})\frac1{K}|T|\cdot H_2(1-1/K)-\Omega(\eta/K) |T|.
\end{split}
\end{equation*}
for every $\i\in \mathcal{J}$, which by~\dualeqref{eq:93g9g7g89gFGYFGSA} implies for $\i\in \mathcal{J}$
\begin{equation}\duallabel{eq:mathcalj}
\begin{split}
H(X_\i| M_{ALG}, &\Lambda_{<k}, \{\j=\i\}, Z)\leq (1+\sqrt{\e})\frac1{K}|T|\cdot H_2(1-1/K)-\Omega\left(\frac{\eta}{K^2}\right) |T|\\
\end{split}
\end{equation}

Finally, for $\i\in \B\setminus \mathcal{J}$ we have the bound 
\begin{equation}\duallabel{eq:9h329ty3utjbbfsf}
\begin{split}
H(X_\i| M_{ALG}, \Lambda_{<k}, \{\j=\i\}, Z)\leq (1+\sqrt{\e})\frac1{K}|T|\cdot H_2(1-1/K),
\end{split}
\end{equation}
since the number of nonzeros in $X_\i$ is exactly $\lceil (1-1/K) |S_k\rceil$. Putting ~\dualeqref{eq:mathcalj} and~\dualeqref{eq:9h329ty3utjbbfsf} together with~\dualeqref{eq:20iabibafihf0ASAS} and using~\dualeqref{eq:plus-z-198yugudfg}, we get
\begin{equation*}
\begin{split}
H(X| \Pi, \Lambda_{<k})&\leq \sum_{\i\in \B} H(X_\i|  M_{ALG}, \Lambda_{<k}, \{\j=\i\})\\
&\leq \sum_{\i\in \B} (1+H(X_\i|  M_{ALG}, \Lambda_{<k}, \{\j=\i\}, Z))\\
&\leq \sum_{\i\in \mathcal{J}} \left(H(X_\i|\Lambda_{<k})-\Omega\left(\frac{\eta}{K^2}\right) |T|\right)+\sum_{\i\in \B\setminus \mathcal{J}} H(X_\i|\Lambda_{<k})\\
&=\sum_{\i\in \B} H(X_\i|\Lambda_{<k})- |\mathcal{J}| \cdot  \Omega\left(\frac{\eta}{K^2}\right) |T|.\\
\end{split}
\end{equation*}
On the other hand, since $|S_k|\geq (1-\sqrt{\e})\frac1{K}|T|$ for all choices of $\Lambda_{<k}$ by Lemma~\dualref{lm:size-bounds}, {\bf (2)}, we get, since the nonzeros of $X_\i$ are a uniformly random set of size $\lceil (1-1/K)|S_k|\rceil$, that
$$
H(X|\Lambda_{<k})\geq (1-\sqrt{\e})\frac1{K}|T|\cdot |\B|\cdot (1-o_N(1)) H_2(1-1/K).
$$

Substituting this into~\dualeqref{eq:20iabibafihf0ASAS}, we get
\begin{equation*}
\begin{split}
s=|\Pi|&\geq \Omega\left(\frac{\eta}{K^2}\right) |\mathcal{J}| \cdot  |T|-O(\sqrt{\e})\frac1{K}|T|\cdot |\B|\cdot H_2(1-1/K)\\
&\geq \Omega\left(\frac{\eta}{K^2}\right) |\mathcal{J}| \cdot  |T| \text{~~~~~~~(since $\e<K^{-100K^2}$ by ~\ref{p6-full},~\ref{p5-full} and \ref{p3-full})}\\
&\geq \Omega\left(\frac{\eta}{K^3}\right) |\B| \cdot  |T|\text{~~~~~~~(since $|\mathcal{J}|\geq |\B|/(4K)$)}\\
&\geq \Omega_K(|\B| \cdot  |T|).\\
\end{split}
\end{equation*}
Finally, recall that by~\ref{p0-full}
$$
N=m^n=n^{20n},
$$
and therefore
$$
|\B|\geq |\F|/K=2^{\Omega(\e^2 n)}=N^{\Omega_\e(1/\log\log N)}.
$$ 
To summarize, we get a lower bound of
$$
s=\Omega_K(|\B| \cdot  |T|)=|T|^{1+\Omega(1/\log\log |T|)},
$$
as required. 

\end{proofof}

\end{appendix}
\newcommand{\etalchar}[1]{$^{#1}$}

\end{document}